\documentclass[a4paper,14pt,oneside,openany]{memoir}
\pdfoutput=1

\makeatletter
\@ifundefined{c@draft}{
  \newcounter{draft}
  \setcounter{draft}{0}  
}{}
\makeatother

\makeatletter
\@ifundefined{c@showmarkup}{
  \newcounter{showmarkup}
  \setcounter{showmarkup}{0}  
}{}
\makeatother

\makeatletter
\@ifundefined{c@englishthesis}{
  \newcounter{englishthesis}
  \setcounter{englishthesis}{1}  
}{}
\makeatother

\makeatletter
\@ifundefined{c@usealtfont}{
  \newcounter{usealtfont}
  \setcounter{usealtfont}{1}    
}{}
\makeatother

\makeatletter
\@ifundefined{c@fontfamily}{
  \newcounter{fontfamily}
  \setcounter{fontfamily}{1}  
}{}
\makeatother

\makeatletter
\@ifundefined{c@bibliosel}{
  \newcounter{bibliosel}
  \setcounter{bibliosel}{1}   
}{}
\makeatother

\makeatletter
\@ifundefined{c@mediadisplay}{
  \newcounter{mediadisplay}
  \setcounter{mediadisplay}{3}   
}{}
\makeatother

\makeatletter
\@ifundefined{c@imgprecompile}{
  \newcounter{imgprecompile}
  \setcounter{imgprecompile}{0}   
}{}
\makeatother
\newif\ifxetexorluatex   
\ifxetex
    \xetexorluatextrue
\else
    \ifluatex
        \xetexorluatextrue
    \else
        \xetexorluatexfalse
    \fi
\fi

\newif\ifsynopsis           

\usepackage{etoolbox}[2015/08/02]   
\providebool{presentation}

\usepackage{comment}    

\usepackage{pdflscape}  
\usepackage{geometry}   

\usepackage{amsthm,amsmath,amscd}   
\usepackage{amsfonts,amssymb}       
\usepackage{mathtools}              
\usepackage{xfrac}                  
\usepackage[
    locale = DE,
    list-separator       = {;\,},
    list-final-separator = {;\,},
    list-pair-separator  = {;\,},
    list-units           = single,
    range-units          = single,
    range-phrase={\text{\ensuremath{-}}},
    fraction-function    = \sfrac,
    separate-uncertainty,
    ]{siunitx}[=v2]                 
\ifnumequal{\value{englishthesis}}{0}{
    \patchcmd{\subequations}{}
    {}
    {\typeout{subequations patched}}{\typeout{subequations not patched}}
}{}

\newlength{\curtextsize}
\newlength{\bigtextsize}

\ifxetexorluatex
    \ifpresentation
    \fi
    \PassOptionsToPackage{no-math}{fontspec}    
    \usepackage{polyglossia}[2014/05/21]        
\else
    \ifnumequal{\value{usealtfont}}{0}{}{
        \input glyphtounicode.tex
        \input glyphtounicode-cmr.tex 
        \pdfgentounicode=1
    }
    \usepackage{cmap}   
    \ifnumequal{\value{usealtfont}}{2}{}{
        \defaulthyphenchar=127  
    }
    \usepackage{textcomp}
    \usepackage[T1,T2A]{fontenc}                    
    \ifnumequal{\value{usealtfont}}{1}{
        \IfFileExists{pscyr.sty}{\usepackage{pscyr}}{}  
    }{}
    \usepackage[utf8]{inputenc}[2014/04/30]         
    \ifnumequal{\value{englishthesis}}{0}{
        \usepackage[english, russian]{babel}[2014/03/24]
    }{
        \usepackage[english]{babel}[2014/03/24]
    }

    \makeatletter\AtBeginDocument{\let\@elt\relax}\makeatother 
    \ifnumequal{\value{usealtfont}}{2}{
        \usepackage[scaled=0.914]{XCharter}[2017/12/19] 
        \usepackage[charter, vvarbb, scaled=1.048]{newtxmath}[2017/12/14]
        \ifpresentation
        \else
            \setDisplayskipStretch{-0.078}
        \fi
    }{}
\fi

\ifpresentation
\else
    \indentafterchapter     
    \usepackage{indentfirst}
\fi

\ifpresentation
\else
    \usepackage[dvipsnames, table, hyperref]{xcolor} 
\fi

\usepackage{longtable,ltcaption} 
\usepackage{multirow,makecell}   
\usepackage{tabu, tabulary}      
\makeatletter
\@ifpackagelater{longtable}{2020/02/07}{
\def\tabuendlongtrial{%
    \LT@echunk  \global\setbox\LT@gbox \hbox{\unhbox\LT@gbox}\kern\wd\LT@gbox
                \LT@get@widths
}%
}{}
\makeatother

\usepackage{threeparttable}      

\usepackage{soulutf8}
\usepackage{icomma}  

\IfFileExists{impnattypo.sty}{
    \ifluatex
        \ifnumequal{\value{draft}}{1}{
            \usepackage[hyphenation, lastparline, nosingleletter, homeoarchy,
            rivers, draft]{impnattypo}
        }{
            \usepackage[hyphenation, lastparline, nosingleletter]{impnattypo}
        }
    \else
        \usepackage[hyphenation, lastparline]{impnattypo}
    \fi
}{}

\usepackage{tikz}                   
\usetikzlibrary{chains}             
\usetikzlibrary{shapes.geometric}   
\usetikzlibrary{shapes.symbols}     
\usetikzlibrary{arrows}             

\ifxetexorluatex
    
\fi
\usepackage{hyperref}[2012/11/06]

\usepackage{graphicx}[2014/04/25]   
\usepackage{caption}                
\usepackage{subcaption}             
\usepackage{pdfpages}               

\usepackage{aliascnt}
\usepackage[figure,table]{totalcount}   
\usepackage{totcount}   
\usepackage{totpages}   

\ifpresentation
\else
    \ifnumequal{\value{englishthesis}}{0}{
        \usepackage[russian]{cleveref} 
    }{
        \usepackage{cleveref}
    }
    \usepackage{kvsetkeys}
    \makeatletter
    \let\org@@cref\@cref
    \renewcommand*{\@cref}[2]{%
        \edef\process@me{%
            \noexpand\org@@cref{#1}{\zap@space#2 \@empty}%
        }\process@me
    }
    \makeatother
\fi

\usepackage{placeins} 

\ifnumequal{\value{draft}}{1}{
    \usepackage[firstpage]{draftwatermark}
    \SetWatermarkText{DRAFT}
    \SetWatermarkFontSize{14pt}
    \SetWatermarkScale{15}
    \SetWatermarkAngle{45}
}{}

\ifdefined\tikzexternalrealjob
\fi

\ifnumequal{\value{imgprecompile}}{1}{
    \usetikzlibrary{external}   
    \tikzexternalize[prefix=images/cache/,optimize command away=\includepdf] 
    \ifxetex
        \tikzset{external/up to date check={diff}}
    \fi
}{}

\usepackage{physics}
\usepackage{bbm}
\usepackage{qcircuit}
\usepackage{tabularx}
\usepackage[export]{adjustbox}
\usepackage{oplotsymbl}
\synopsisfalse                      

\usepackage{afterpage}

\usepackage{enumitem}

\usepackage[intoc]{nomencl}
\makenomenclature
\usepackage{fr-longtable}    

\usepackage{fancyvrb}
\usepackage{listings}
\lccode`\~=0\relax 

\usepackage{upgreek} 

\ifnumequal{\value{draft}}{1}{
   \IfFileExists{.git/gitHeadInfo.gin}{
      \usepackage[mark,pcount]{gitinfo2}

   }{}
}{}   

\newcounter{intvl}
\newcounter{otstup}
\newcounter{contnumeq}
\newcounter{contnumfig}
\newcounter{contnumtab}
\newcounter{pgnum}
\newcounter{chapstyle}
\newcounter{headingdelim}
\newcounter{headingalign}
\newcounter{headingsize}

\settocdepth{subsection}    
\setsecnumdepth{subsection} 

\newcommand{\tabindent}{0cm}

\newcommand{\tabformat}{plain}

\newcommand{\tabsinglecenter}{false}

\newcommand{\tabjust}{justified}

\newcommand{\tablabelsep}{~\cyrdash\ }

\newcommand{\splitformatlabel}{\raggedleft}

\newcommand{\splitformattext}{\raggedright}

\newcommand{\figlabelsep}{~\cyrdash\ }  

\definecolor{linkcolor}{rgb}{0.9,0,0}
\definecolor{citecolor}{rgb}{0,0.6,0}
\definecolor{urlcolor}{rgb}{0,0,1}

\newcommand{\bibtitlefull}{Список литературы} 

\newcommand{\mc}[1]{\mathcal{#1}}
\newcommand{\id}{\mathbbm{1}}
\newcommand{\rmi}{\mathrm{i}}
\newcommand{\placeholder}{\ast}

\DeclareMathOperator{\Var}{Var}

\tikzset{blank/.style={rectangle,inner sep=0pt,draw=none,fill=none,minimum
    size=0pt} }         

\newcommand{\thesisAuthorLastName}{Уваров}
\newcommand{\thesisAuthorOtherNames}{Алексей Викторович}
\newcommand{\thesisAuthorInitials}{А.\,В.}
\newcommand{\thesisAuthorLastNameEn}{Uvarov}
\newcommand{\thesisAuthorOtherNamesEn}{Alexey Viktorovich}
\newcommand{\thesisAuthorInitialsEn}{A.\,V.}

\newcommand{\thesisAuthor}             
{%
    \texorpdfstring{
        \thesisAuthorLastName~\thesisAuthorOtherNames
    }{%
        \thesisAuthorLastName, \thesisAuthorOtherNames
    }
}
\newcommand{\thesisAuthorEn}             
{%
    \texorpdfstring{
        \thesisAuthorLastNameEn~\thesisAuthorOtherNamesEn
    }{%
        \thesisAuthorLastNameEn, \thesisAuthorOtherNamesEn
    }
}

\newcommand{\thesisAuthorShort}        
{\thesisAuthorInitials~\thesisAuthorLastName}
\newcommand{\thesisAuthorShortEn}        
{\thesisAuthorInitialsEn~\thesisAuthorLastNameEn}

\newcommand{\thesisTitle}              
{Вариационные квантовые алгоритмы для решения задач минимизации локальных гамильтонианов}
\newcommand{\thesisTitleEn}              
{Variational quantum algorithms for local Hamiltonian problems}
\newcommand{\thesisSpecialtyNumber}    
{01.04.02}
\newcommand{\thesisSpecialtyTitle}     
{Теоретическая физика}
\newcommand{\thesisSpecialtyTitleEn}     
{Theoretical physics}
\newcommand{\thesisDegree}             
{кандидата физико-математических наук}
\newcommand{\thesisDegreeEn}             
{Candidate of Physical and Mathematical Sciences}
\newcommand{\thesisDegreeShort}        
{канд.~физ.-мат.~наук}
\newcommand{\thesisDegreeShortEn}        
{cand.~of phys.-math.~sciences}
\newcommand{\thesisCity}               
{Москва}
\newcommand{\thesisCityEn}               
{Moscow}
\newcommand{\thesisYear}               
{2022}
\newcommand{\thesisOrganization}       
{Автономная некоммерческая образовательная организация высшего профессионального образования <<Сколковский институт науки и технологий>>}
\newcommand{\thesisOrganizationEn}       
{Autonomous non-profit organization for higher education \\
Skolkovo Institute of Science and Technology}
\newcommand{\thesisOrganizationShort}  
{Сколтех}
\newcommand{\thesisOrganizationShortEn}  
{Skoltech}

\newcommand{\thesisInOrganization}     
{Автономной некоммерческой образовательной организации высшего профессионального образования <<Сколковский институт науки и технологий>>}

\newcommand{\supervisorFio}              
{Биамонте Джейкоб Дэниел}
\newcommand{\supervisorRegalia}          
{доктор физико-математических наук, профессор}
\newcommand{\supervisorFioShort}         
{Дж.\,Д.~Биамонте}
\newcommand{\supervisorRegaliaShort}     
{д.~ф.-м.~н, проф.}

\newcommand{\supervisorFioEn}              
{Biamonte Jacob Daniel}
\newcommand{\supervisorRegaliaEn}          
{Doctor of Physical and Mathematical Sciences, Professor}
\newcommand{\supervisorFioShortEn}         
{\fixme{J.\,D.~Biamonte}}
\newcommand{\supervisorRegaliaShortEn}     
{\fixme{D.Sc., Prof}}

\newcommand{\opponentOneFio}           
{\fixme{Фамилия Имя Отчество}}
\newcommand{\opponentOneRegalia}       
{\fixme{доктор физико-математических наук, профессор}}
\newcommand{\opponentOneJobPlace}      
{\fixme{Не очень длинное название для места работы}}
\newcommand{\opponentOneJobPost}       
{\fixme{старший научный сотрудник}}

\newcommand{\opponentTwoFio}           
{\fixme{Фамилия Имя Отчество}}
\newcommand{\opponentTwoRegalia}       
{\fixme{кандидат физико-математических наук}}
\newcommand{\opponentTwoJobPlace}      
{\fixme{Основное место работы c длинным длинным длинным длинным названием}}
\newcommand{\opponentTwoJobPost}       
{\fixme{старший научный сотрудник}}

\newcommand{\opponentOneFioEn}           
{\fixme{Name name name}}
\newcommand{\opponentOneRegaliaEn}       
{\fixme{Doctor of physcial and mathematical sciences, Professor}}
\newcommand{\opponentOneJobPlaceEn}      
{\fixme{Somewhat long job place name}}
\newcommand{\opponentOneJobPostEn}       
{\fixme{Senior research scientist}}

\newcommand{\opponentTwoFioEn}           
{\fixme{Name name name}}
\newcommand{\opponentTwoRegaliaEn}       
{\fixme{Doctor of physcial and mathematical sciences}}
\newcommand{\opponentTwoJobPlaceEn}      
{\fixme{Main job place with a long long long long long long, reeeeally long title}}
\newcommand{\opponentTwoJobPostEn}       
{\fixme{Senior research scientist}}

\newcommand{\leadingOrganizationTitle} 
{\fixme{Федеральное государственное бюджетное образовательное учреждение высшего
профессионального образования с~длинным длинным длинным длинным названием}}
\newcommand{\leadingOrganizationTitleEn} 
{\fixme{Federal state organization tratatatatata}}

\newcommand{\defenseDate}              
{\fixme{DD mmmmmmmm YYYY~г.~в~XX часов}}
\newcommand{\defenseDateEn}              
{\fixme{DD mmmmmmmm YYYY~at~XX:YY}}
\newcommand{\defenseCouncilNumber}     
{\fixme{Д\,123.456.78}}
\newcommand{\defenseCouncilNumberEn}     
{\fixme{D\,123.456.78}}
\newcommand{\defenseCouncilTitle}      
{\fixme{Название учреждения}}
\newcommand{\defenseCouncilTitleEn}      
{\fixme{Defence council title}}
\newcommand{\defenseCouncilAddress}    
{\fixme{Адрес}}
\newcommand{\defenseCouncilAddressEn}    
{\fixme{Address}}
\newcommand{\defenseCouncilPhone}      
{\fixme{+7~(0000)~00-00-00}}

\newcommand{\defenseSecretaryFio}      
{\fixme{Фамилия Имя Отчество}}
\newcommand{\defenseSecretaryRegalia}  
{\fixme{д-р~физ.-мат. наук}}            

\newcommand{\defenseSecretaryFioEn}      
{\fixme{Lastname First Middle}}
\newcommand{\defenseSecretaryRegaliaEn}  
{\fixme{DSc.}}    

\newcommand{\synopsisLibrary}          
{\fixme{Название библиотеки}}
\newcommand{\synopsisDate}             
{\fixme{DD mmmmmmmm}\the\year~года}

\newcommand{\synopsisLibraryEn}          
{\fixme{Library title}}
\newcommand{\synopsisDateEn}             
{\fixme{DD mmmmmmmm}\the\year}

\providecommand{\keywords}
{}
\ifxetexorluatex
    \ifnumequal{\value{englishthesis}}{0}{
        \setmainlanguage[babelshorthands=true]{russian}
        \setotherlanguage{english}
    }{
        \setmainlanguage[babelshorthands=true]{english}
        \setotherlanguage{russian}
    }

    \ifnumequal{\value{fontfamily}}{1}{
        \IfFontExistsTF{Times New Roman}{}{\setcounter{fontfamily}{0}}
    }{}
    \ifnumequal{\value{fontfamily}}{2}{
        \IfFontExistsTF{LiberationSerif}{}{\setcounter{fontfamily}{0}}
    }{}

    \ifnumequal{\value{fontfamily}}{0}{                    
        \setmonofont{CMU Typewriter Text}                  
        \newfontfamily\cyrillicfonttt{CMU Typewriter Text} 
        \defaultfontfeatures{Ligatures=TeX}                
        \setmainfont{CMU Serif}                            
        \newfontfamily\cyrillicfont{CMU Serif}             
        \setsansfont{CMU Sans Serif}                       
        \newfontfamily\cyrillicfontsf{CMU Sans Serif}      
    }

    \ifnumequal{\value{fontfamily}}{1}{                    
        \setmonofont{Courier New}                          
        \newfontfamily\cyrillicfonttt{Courier New}         
        \defaultfontfeatures{Ligatures=TeX}                
        \setmainfont{Times New Roman}                      
        \newfontfamily\cyrillicfont{Times New Roman}       
        \setsansfont{Arial}                                
        \newfontfamily\cyrillicfontsf{Arial}               
    }

    \ifnumequal{\value{fontfamily}}{2}{                    
        \setmonofont{LiberationMono}[Scale=0.87] 
        \newfontfamily\cyrillicfonttt{LiberationMono}[     
            Scale=0.87]
        \defaultfontfeatures{Ligatures=TeX}                
        \setmainfont{LiberationSerif}                      
        \newfontfamily\cyrillicfont{LiberationSerif}       
        \setsansfont{LiberationSans}                       
        \newfontfamily\cyrillicfontsf{LiberationSans}      
    }

\else
    \ifnumequal{\value{usealtfont}}{1}{
        \IfFileExists{pscyr.sty}{}{}
    }{}
\fi
\ifnumequal{\value{englishthesis}}{1}{
\DeclareRobustCommand{\fixme}{\textcolor{red}}  
\AtBeginDocument{%
    \setlength{\parindent}{2.5em}                   
}

\DeclareCaptionLabelSeparator{tabsep}{\tablabelsep} 
\DeclareCaptionFormat{split}{\splitformatlabel#1\par\splitformattext#3}

\captionsetup[table]{
        format=\tabformat,                
        font=normal,                      
        skip=.0pt,                        
        parskip=.0pt,                     
        position=above,                   
        justification=\tabjust,           
        indent=\tabindent,                
        labelsep=tabsep,                  
        singlelinecheck=\tabsinglecenter, 
}

\DeclareCaptionLabelSeparator{figsep}{\figlabelsep} 

\captionsetup[figure]{
        format=plain,                     
        font=normal,                      
        skip=.0pt,                        
        parskip=.0pt,                     
        position=below,                   
        singlelinecheck=true,             
        justification=centerlast,         
        labelsep=figsep,                  
}

\DeclareCaptionSubType{figure}
\ifnumequal{\value{englishthesis}}{0}{
}{}
\ifsynopsis
\DeclareCaptionFont{norm}{\fontsize{10pt}{11pt}\selectfont}
\newcommand{\subfigureskip}{2.pt}
\else
\DeclareCaptionFont{norm}{\fontsize{14pt}{16pt}\selectfont}
\newcommand{\subfigureskip}{0.pt}
\fi

\captionsetup[subfloat]{
        labelfont=norm,                 
        textfont=norm,                  
        labelsep=space,                 
        labelformat=brace,              
        justification=centering,        
        singlelinecheck=true,           
        skip=\subfigureskip,            
        parskip=.0pt,                   
        position=below,                 
}

\ifpresentation
\else
    \crefdefaultlabelformat{#2#1#3}

    \crefformat{equation}{(#2#1#3)} 
    \labelcrefformat{equation}{(#2#1#3)} 
    \crefrangeformat{equation}{(#3#1#4) \cyrdash~(#5#2#6)} 
    \labelcrefrangeformat{equation}{(#3#1#4) \cyrdash~(#5#2#6)} 
    \crefmultiformat{equation}{(#2#1#3)}{ и~(#2#1#3)}{, (#2#1#3)}{ и~(#2#1#3)} 
    \labelcrefmultiformat{equation}{(#2#1#3)}{ и~(#2#1#3)}{, (#2#1#3)}{ и~(#2#1#3)} 

    \crefformat{subequation}{(#2#1#3)} 
    \labelcrefformat{subequation}{(#2#1#3)} 
    \crefrangeformat{subequation}{(#3#1#4) \cyrdash~(#5#2#6)} 
    \labelcrefrangeformat{subequation}{(#3#1#4) \cyrdash~(#5#2#6)} 
    \crefmultiformat{subequation}{(#2#1#3)}{ и~(#2#1#3)}{, (#2#1#3)}{ и~(#2#1#3)} 
    \labelcrefmultiformat{subequation}{(#2#1#3)}{ и~(#2#1#3)}{, (#2#1#3)}{ и~(#2#1#3)} 

    \crefformat{chapter}{#2#1#3} 
    \labelcrefformat{chapter}{#2#1#3} 
    \crefrangeformat{chapter}{#3#1#4 \cyrdash~#5#2#6} 
    \labelcrefrangeformat{chapter}{#3#1#4 \cyrdash~#5#2#6} 
    \crefmultiformat{chapter}{#2#1#3}{ и~#2#1#3}{, #2#1#3}{ и~#2#1#3} 
    \labelcrefmultiformat{chapter}{#2#1#3}{ и~#2#1#3}{, #2#1#3}{ и~#2#1#3} 

    \crefformat{section}{#2#1#3} 
    \labelcrefformat{section}{#2#1#3} 
    \crefrangeformat{section}{#3#1#4 \cyrdash~#5#2#6} 
    \labelcrefrangeformat{section}{#3#1#4 \cyrdash~#5#2#6} 
    \crefmultiformat{section}{#2#1#3}{ и~#2#1#3}{, #2#1#3}{ и~#2#1#3} 
    \labelcrefmultiformat{section}{#2#1#3}{ и~#2#1#3}{, #2#1#3}{ и~#2#1#3} 

    \crefformat{appendix}{#2#1#3} 
    \labelcrefformat{appendix}{#2#1#3} 
    \crefrangeformat{appendix}{#3#1#4 \cyrdash~#5#2#6} 
    \labelcrefrangeformat{appendix}{#3#1#4 \cyrdash~#5#2#6} 
    \crefmultiformat{appendix}{#2#1#3}{ и~#2#1#3}{, #2#1#3}{ и~#2#1#3} 
    \labelcrefmultiformat{appendix}{#2#1#3}{ и~#2#1#3}{, #2#1#3}{ и~#2#1#3} 

    \crefformat{figure}{#2#1#3} 
    \labelcrefformat{figure}{#2#1#3} 
    \crefrangeformat{figure}{#3#1#4 \cyrdash~#5#2#6} 
    \labelcrefrangeformat{figure}{#3#1#4 \cyrdash~#5#2#6} 
    \crefmultiformat{figure}{#2#1#3}{ и~#2#1#3}{, #2#1#3}{ и~#2#1#3} 
    \labelcrefmultiformat{figure}{#2#1#3}{ и~#2#1#3}{, #2#1#3}{ и~#2#1#3} 

    \crefformat{table}{#2#1#3} 
    \labelcrefformat{table}{#2#1#3} 
    \crefrangeformat{table}{#3#1#4 \cyrdash~#5#2#6} 
    \labelcrefrangeformat{table}{#3#1#4 \cyrdash~#5#2#6} 
    \crefmultiformat{table}{#2#1#3}{ и~#2#1#3}{, #2#1#3}{ и~#2#1#3} 
    \labelcrefmultiformat{table}{#2#1#3}{ и~#2#1#3}{, #2#1#3}{ и~#2#1#3} 

    \crefformat{lstlisting}{#2#1#3} 
    \labelcrefformat{lstlisting}{#2#1#3} 
    \crefrangeformat{lstlisting}{#3#1#4 \cyrdash~#5#2#6} 
    \labelcrefrangeformat{lstlisting}{#3#1#4 \cyrdash~#5#2#6} 
    \crefmultiformat{lstlisting}{#2#1#3}{ и~#2#1#3}{, #2#1#3}{ и~#2#1#3} 
    \labelcrefmultiformat{lstlisting}{#2#1#3}{ и~#2#1#3}{, #2#1#3}{ и~#2#1#3} 

    \crefformat{ListingEnv}{#2#1#3} 
    \labelcrefformat{ListingEnv}{#2#1#3} 
    \crefrangeformat{ListingEnv}{#3#1#4 \cyrdash~#5#2#6} 
    \labelcrefrangeformat{ListingEnv}{#3#1#4 \cyrdash~#5#2#6} 
    \crefmultiformat{ListingEnv}{#2#1#3}{ и~#2#1#3}{, #2#1#3}{ и~#2#1#3} 
    \labelcrefmultiformat{ListingEnv}{#2#1#3}{ и~#2#1#3}{, #2#1#3}{ и~#2#1#3} 
\fi

\ifluatex
    \hypersetup{
        unicode,                
    }
\fi

\hypersetup{
    linktocpage=true,           
    plainpages=false,           
    colorlinks,                 
    linkcolor={linkcolor},      
    citecolor={citecolor},      
    urlcolor={urlcolor},        
    pdftitle={\thesisTitle},    
    pdfauthor={\thesisAuthor},  
    pdfsubject={\thesisSpecialtyNumber\ \thesisSpecialtyTitle},      
    pdfkeywords={\keywords},    
    pdflang={ru},
}
\ifnumequal{\value{draft}}{1}{
    \hypersetup{
        draft,
    }
}{}

\setlist{nosep,
    labelindent=\parindent,leftmargin=*
}

\makeatletter
\def\formtotal#1#2#3#4#5{%
    \newcount\@c
    \@c\totvalue{#1}\relax
    \newcount\@last
    \newcount\@pnul
    \@last\@c\relax
    \divide\@last 10
    \@pnul\@last\relax
    \divide\@pnul 10
    \multiply\@pnul-10
    \advance\@pnul\@last
    \multiply\@last-10
    \advance\@last\@c
    #2%
    \ifnum\@pnul=1#5\else%
    \ifcase\@last#5\or#3\or#4\or#4\or#4\else#5\fi
    \fi
}
\makeatother

\ifboolexpr{ (test {\ifnumequal{\value{draft}}{1}}) or (test {\ifnumequal{\value{showmarkup}}{1}})}{
        \newrobustcmd{\todo}[1]{\textcolor{red}{#1}}
        \newrobustcmd{\note}[2][]{\ifstrempty{#1}{#2}{\textcolor{#1}{#2}}}
        \newenvironment{commentbox}[1][]%
        {\ifstrempty{#1}{}{\color{#1}}}%
        {}
}{
        \newrobustcmd{\todo}[1]{}
        \newrobustcmd{\note}[2][]{}
        \excludecomment{commentbox}
}
}{
\DeclareRobustCommand{\fixme}{\textcolor{red}}  
\AtBeginDocument{%
    \setlength{\parindent}{2.5em}                   
}

\DeclareCaptionLabelSeparator{tabsep}{\tablabelsep} 
\DeclareCaptionFormat{split}{\splitformatlabel#1\par\splitformattext#3}

\captionsetup[table]{
        format=\tabformat,                
        font=normal,                      
        skip=.0pt,                        
        parskip=.0pt,                     
        position=above,                   
        justification=\tabjust,           
        indent=\tabindent,                
        labelsep=tabsep,                  
        singlelinecheck=\tabsinglecenter, 
}

\DeclareCaptionLabelSeparator{figsep}{\figlabelsep} 

\captionsetup[figure]{
        format=plain,                     
        font=normal,                      
        skip=.0pt,                        
        parskip=.0pt,                     
        position=below,                   
        singlelinecheck=true,             
        justification=centerlast,         
        labelsep=figsep,                  
}

\DeclareCaptionSubType{figure}
\ifnumequal{\value{englishthesis}}{0}{
}{}
\ifsynopsis
\DeclareCaptionFont{norm}{\fontsize{10pt}{11pt}\selectfont}
\newcommand{\subfigureskip}{2.pt}
\else
\DeclareCaptionFont{norm}{\fontsize{14pt}{16pt}\selectfont}
\newcommand{\subfigureskip}{0.pt}
\fi

\captionsetup[subfloat]{
        labelfont=norm,                 
        textfont=norm,                  
        labelsep=space,                 
        labelformat=brace,              
        justification=centering,        
        singlelinecheck=true,           
        skip=\subfigureskip,            
        parskip=.0pt,                   
        position=below,                 
}

\ifpresentation
\else
    \crefdefaultlabelformat{#2#1#3}

    \crefformat{equation}{(#2#1#3)} 
    \labelcrefformat{equation}{(#2#1#3)} 
    \crefrangeformat{equation}{(#3#1#4) \cyrdash~(#5#2#6)} 
    \labelcrefrangeformat{equation}{(#3#1#4) \cyrdash~(#5#2#6)} 
    \crefmultiformat{equation}{(#2#1#3)}{ и~(#2#1#3)}{, (#2#1#3)}{ и~(#2#1#3)} 
    \labelcrefmultiformat{equation}{(#2#1#3)}{ и~(#2#1#3)}{, (#2#1#3)}{ и~(#2#1#3)} 

    \crefformat{subequation}{(#2#1#3)} 
    \labelcrefformat{subequation}{(#2#1#3)} 
    \crefrangeformat{subequation}{(#3#1#4) \cyrdash~(#5#2#6)} 
    \labelcrefrangeformat{subequation}{(#3#1#4) \cyrdash~(#5#2#6)} 
    \crefmultiformat{subequation}{(#2#1#3)}{ и~(#2#1#3)}{, (#2#1#3)}{ и~(#2#1#3)} 
    \labelcrefmultiformat{subequation}{(#2#1#3)}{ и~(#2#1#3)}{, (#2#1#3)}{ и~(#2#1#3)} 

    \crefformat{chapter}{#2#1#3} 
    \labelcrefformat{chapter}{#2#1#3} 
    \crefrangeformat{chapter}{#3#1#4 \cyrdash~#5#2#6} 
    \labelcrefrangeformat{chapter}{#3#1#4 \cyrdash~#5#2#6} 
    \crefmultiformat{chapter}{#2#1#3}{ и~#2#1#3}{, #2#1#3}{ и~#2#1#3} 
    \labelcrefmultiformat{chapter}{#2#1#3}{ и~#2#1#3}{, #2#1#3}{ и~#2#1#3} 

    \crefformat{section}{#2#1#3} 
    \labelcrefformat{section}{#2#1#3} 
    \crefrangeformat{section}{#3#1#4 \cyrdash~#5#2#6} 
    \labelcrefrangeformat{section}{#3#1#4 \cyrdash~#5#2#6} 
    \crefmultiformat{section}{#2#1#3}{ и~#2#1#3}{, #2#1#3}{ и~#2#1#3} 
    \labelcrefmultiformat{section}{#2#1#3}{ и~#2#1#3}{, #2#1#3}{ и~#2#1#3} 

    \crefformat{appendix}{#2#1#3} 
    \labelcrefformat{appendix}{#2#1#3} 
    \crefrangeformat{appendix}{#3#1#4 \cyrdash~#5#2#6} 
    \labelcrefrangeformat{appendix}{#3#1#4 \cyrdash~#5#2#6} 
    \crefmultiformat{appendix}{#2#1#3}{ и~#2#1#3}{, #2#1#3}{ и~#2#1#3} 
    \labelcrefmultiformat{appendix}{#2#1#3}{ и~#2#1#3}{, #2#1#3}{ и~#2#1#3} 

    \crefformat{figure}{#2#1#3} 
    \labelcrefformat{figure}{#2#1#3} 
    \crefrangeformat{figure}{#3#1#4 \cyrdash~#5#2#6} 
    \labelcrefrangeformat{figure}{#3#1#4 \cyrdash~#5#2#6} 
    \crefmultiformat{figure}{#2#1#3}{ и~#2#1#3}{, #2#1#3}{ и~#2#1#3} 
    \labelcrefmultiformat{figure}{#2#1#3}{ и~#2#1#3}{, #2#1#3}{ и~#2#1#3} 

    \crefformat{table}{#2#1#3} 
    \labelcrefformat{table}{#2#1#3} 
    \crefrangeformat{table}{#3#1#4 \cyrdash~#5#2#6} 
    \labelcrefrangeformat{table}{#3#1#4 \cyrdash~#5#2#6} 
    \crefmultiformat{table}{#2#1#3}{ и~#2#1#3}{, #2#1#3}{ и~#2#1#3} 
    \labelcrefmultiformat{table}{#2#1#3}{ и~#2#1#3}{, #2#1#3}{ и~#2#1#3} 

    \crefformat{lstlisting}{#2#1#3} 
    \labelcrefformat{lstlisting}{#2#1#3} 
    \crefrangeformat{lstlisting}{#3#1#4 \cyrdash~#5#2#6} 
    \labelcrefrangeformat{lstlisting}{#3#1#4 \cyrdash~#5#2#6} 
    \crefmultiformat{lstlisting}{#2#1#3}{ и~#2#1#3}{, #2#1#3}{ и~#2#1#3} 
    \labelcrefmultiformat{lstlisting}{#2#1#3}{ и~#2#1#3}{, #2#1#3}{ и~#2#1#3} 

    \crefformat{ListingEnv}{#2#1#3} 
    \labelcrefformat{ListingEnv}{#2#1#3} 
    \crefrangeformat{ListingEnv}{#3#1#4 \cyrdash~#5#2#6} 
    \labelcrefrangeformat{ListingEnv}{#3#1#4 \cyrdash~#5#2#6} 
    \crefmultiformat{ListingEnv}{#2#1#3}{ и~#2#1#3}{, #2#1#3}{ и~#2#1#3} 
    \labelcrefmultiformat{ListingEnv}{#2#1#3}{ и~#2#1#3}{, #2#1#3}{ и~#2#1#3} 
\fi

\ifluatex
    \hypersetup{
        unicode,                
    }
\fi

\hypersetup{
    linktocpage=true,           
    plainpages=false,           
    colorlinks,                 
    linkcolor={linkcolor},      
    citecolor={citecolor},      
    urlcolor={urlcolor},        
    pdftitle={\thesisTitle},    
    pdfauthor={\thesisAuthor},  
    pdfsubject={\thesisSpecialtyNumber\ \thesisSpecialtyTitle},      
    pdfkeywords={\keywords},    
    pdflang={ru},
}
\ifnumequal{\value{draft}}{1}{
    \hypersetup{
        draft,
    }
}{}

\makeatletter
\AddEnumerateCounter{\asbuk}{\russian@alph}{щ}      
\makeatother

\setlist{nosep,
    labelindent=\parindent,leftmargin=*
}

\ifxetexorluatex
    \makeatletter
    \def\russian@Alph#1{\ifcase#1\or
       А\or Б\or В\or Г\or Д\or Е\or Ж\or
       И\or К\or Л\or М\or Н\or
       П\or Р\or С\or Т\or У\or Ф\or Х\or
       Ц\or Ш\or Щ\or Э\or Ю\or Я\else\xpg@ill@value{#1}{russian@Alph}\fi}
    \def\russian@alph#1{\ifcase#1\or
       а\or б\or в\or г\or д\or е\or ж\or
       и\or к\or л\or м\or н\or
       п\or р\or с\or т\or у\or ф\or х\or
       ц\or ш\or щ\or э\or ю\or я\else\xpg@ill@value{#1}{russian@alph}\fi}
    \def\cyr@Alph#1{\ifcase#1\or
        А\or Б\or В\or Г\or Д\or Е\or Ж\or
        И\or К\or Л\or М\or Н\or
        П\or Р\or С\or Т\or У\or Ф\or Х\or
        Ц\or Ш\or Щ\or Э\or Ю\or Я\else\xpg@ill@value{#1}{cyr@Alph}\fi}
    \def\cyr@alph#1{\ifcase#1\or
        а\or б\or в\or г\or д\or е\or ж\or
        и\or к\or л\or м\or н\or
        п\or р\or с\or т\or у\or ф\or х\or
        ц\or ш\or щ\or э\or ю\or я\else\xpg@ill@value{#1}{cyr@alph}\fi}
    \makeatother
\else
    \makeatletter
    \if@uni@ode
      \def\russian@Alph#1{\ifcase#1\or
        А\or Б\or В\or Г\or Д\or Е\or Ж\or
        И\or К\or Л\or М\or Н\or
        П\or Р\or С\or Т\or У\or Ф\or Х\or
        Ц\or Ш\or Щ\or Э\or Ю\or Я\else\@ctrerr\fi}
    \else
      \def\russian@Alph#1{\ifcase#1\or
        \CYRA\or\CYRB\or\CYRV\or\CYRG\or\CYRD\or\CYRE\or\CYRZH\or
        \CYRI\or\CYRK\or\CYRL\or\CYRM\or\CYRN\or
        \CYRP\or\CYRR\or\CYRS\or\CYRT\or\CYRU\or\CYRF\or\CYRH\or
        \CYRC\or\CYRSH\or\CYRSHCH\or\CYREREV\or\CYRYU\or
        \CYRYA\else\@ctrerr\fi}
    \fi
    \if@uni@ode
      \def\russian@alph#1{\ifcase#1\or
        а\or б\or в\or г\or д\or е\or ж\or
        и\or к\or л\or м\or н\or
        п\or р\or с\or т\or у\or ф\or х\or
        ц\or ш\or щ\or э\or ю\or я\else\@ctrerr\fi}
    \else
      \def\russian@alph#1{\ifcase#1\or
        \cyra\or\cyrb\or\cyrv\or\cyrg\or\cyrd\or\cyre\or\cyrzh\or
        \cyri\or\cyrk\or\cyrl\or\cyrm\or\cyrn\or
        \cyrp\or\cyrr\or\cyrs\or\cyrt\or\cyru\or\cyrf\or\cyrh\or
        \cyrc\or\cyrsh\or\cyrshch\or\cyrerev\or\cyryu\or
        \cyrya\else\@ctrerr\fi}
    \fi
    \makeatother
\fi

\makeatletter
\def\formtotal#1#2#3#4#5{%
    \newcount\@c
    \@c\totvalue{#1}\relax
    \newcount\@last
    \newcount\@pnul
    \@last\@c\relax
    \divide\@last 10
    \@pnul\@last\relax
    \divide\@pnul 10
    \multiply\@pnul-10
    \advance\@pnul\@last
    \multiply\@last-10
    \advance\@last\@c
    #2%
    \ifnum\@pnul=1#5\else%
    \ifcase\@last#5\or#3\or#4\or#4\or#4\else#5\fi
    \fi
}
\makeatother

\ifboolexpr{ (test {\ifnumequal{\value{draft}}{1}}) or (test {\ifnumequal{\value{showmarkup}}{1}})}{
        \newrobustcmd{\todo}[1]{\textcolor{red}{#1}}
        \newrobustcmd{\note}[2][]{\ifstrempty{#1}{#2}{\textcolor{#1}{#2}}}
        {\ifstrempty{#1}{}{\color{#1}}}%
        {}
}{
        \newrobustcmd{\todo}[1]{}
        \newrobustcmd{\note}[2][]{}
        \excludecomment{commentbox}
}
}

\graphicspath{{images/}{Dissertation/images/}}         

\OnehalfSpacing*    

\vfuzz \hfuzz
\newlength{\otstuplen}
\ifnumequal{\value{headingalign}}{0}{
    \newcommand{\hdngalign}{\centering}                
    \newcommand{\hdngaligni}{}
    \setlength{\otstuplen}{0pt}
}{%
    \newcommand{\hdngalign}{}                 
    \newcommand{\hdngaligni}{\hspace{\otstuplen}}      
} 

\setrmarg{2.55em plus1fil}                             

\ifnumgreater{\value{headingdelim}}{0}{%
}{}
\ifnumgreater{\value{headingdelim}}{1}{%
    \AfterEndPreamble{
        \setsecnumformat{\csname the#1\endcsname.\space}
    }
}{%
    \AfterEndPreamble{
        \setsecnumformat{\csname the#1\endcsname\quad}
    }
}

\makeevenhead{plain}{}{\rmfamily\thepage}{}
\makeoddhead{plain}{}{\rmfamily\thepage}{}
\makeevenfoot{plain}{}{}{}
\makeoddfoot{plain}{}{}{}
\newif\ifendTOC

\ifnumequal{\value{englishthesis}}{0}{
    \newcommand*{\tocheader}{
        \ifnumequal{\value{pgnum}}{1}{%
            \ifendTOC\else\hbox to \linewidth%
              {\noindent{}~\hfill{Стр.}}\par%
              \ifnumless{\value{page}}{3}{}{%
                \vspace{0.5\onelineskip}
              }
              \afterpage{\tocheader}
            \fi%
        }{}%
        }%
}{
    \newcommand*{\tocheader}{
        \ifnumequal{\value{pgnum}}{1}{%
            \ifendTOC\else\hbox to \linewidth%
              {\noindent{}~\hfill{Page}}\par%
              \ifnumless{\value{page}}{3}{}{%
                \vspace{0.5\onelineskip}
              }
              \afterpage{\tocheader}
            \fi%
        }{}%
        }%
}

\newcommand{\basegostsectionfont}{\fontsize{14pt}{16pt}\selectfont\bfseries}

\makechapterstyle{thesisgost}{%
    \chapterstyle{default}
    \setlength{\beforechapskip}{0pt}
    \setlength{\midchapskip}{0pt}
    \setlength{\afterchapskip}{\theintvl\curtextsize}
    \renewcommand*{\chapnamefont}{\basegostsectionfont}

    \ifnumgreater{\value{headingdelim}}{0}{%
        \renewcommand*{\afterchapternum}{.\space}   
    }{%
        \renewcommand*{\afterchapternum}{\quad}     
    }
    
    \renewcommand*{\printchaptername}{}
    
}

\makeatletter
\makechapterstyle{thesisgostchapname}{%
    \chapterstyle{thesisgost}
    
    \renewcommand*{\printchaptername}{\hdngaligni\hdngalign\chapnamefont \@chapapp} %
}
\makeatother

\chapterstyle{thesisgost}

\setsecheadstyle{\basegostsectionfont\hdngalign}
\setsecindent{\otstuplen}

\setsubsecheadstyle{\basegostsectionfont\hdngalign}
\setsubsecindent{\otstuplen}

\setsubsubsecheadstyle{\basegostsectionfont\hdngalign}
\setsubsubsecindent{\otstuplen}

\sethangfrom{\noindent #1} 

\ifnumequal{\value{chapstyle}}{1}{%
    \chapterstyle{thesisgostchapname}
}{}%

\setbeforesecskip{\theintvl\curtextsize}
\setaftersecskip{\theintvl\curtextsize}
\setbeforesubsecskip{\theintvl\curtextsize}
\setaftersubsecskip{\theintvl\curtextsize}
\setbeforesubsubsecskip{\theintvl\curtextsize}
\setaftersubsubsecskip{\theintvl\curtextsize}

\ifnumequal{\value{headingsize}}{1}{
    \renewcommand{\basegostsectionfont}{\large\bfseries}
    \renewcommand*{\chapnamefont}{\Large\bfseries}

}{}

\ifnumequal{\value{contnumeq}}{1}{%
    \counterwithout{equation}{chapter} 
}{}
\ifnumequal{\value{contnumfig}}{1}{%
    \counterwithout{figure}{chapter}   
}{}
\ifnumequal{\value{contnumtab}}{1}{%
    \counterwithout{table}{chapter}    
}{}

\AfterEndPreamble{
    \regtotcounter{totalcount@figure}
    \regtotcounter{totalcount@table}       
    \regtotcounter{TotPages}               
    \newtotcounter{totalappendix}
    \newtotcounter{totalchapter}
}
\def\zz{\ifx\[$\else\aftergroup\zzz\fi}
\def\zzz{\setbox0\lastbox
\dimen0\dimexpr\extrarowheight + \ht0-\dp0\relax
\setbox0\hbox{\raise-.5\dimen0\box0}%
\ht0=\dimexpr\ht0+\extrarowheight\relax
\dp0=\dimexpr\dp0+\extrarowheight\relax
\box0
}

\theoremstyle{definition}
\newtheorem{definition}{Definition}[section]

\theoremstyle{definition}
\newtheorem{example}{Example}[section]

\theoremstyle{plain}
\newtheorem{theorem}{Theorem}[section]

\newtheorem{proposition}{Proposition}[section]

\theoremstyle{remark}
\newtheorem*{remark}{Remark}

\lstdefinelanguage{Renhanced}%
{keywords={abbreviate,abline,abs,acos,acosh,action,add1,add,%
        aggregate,alias,Alias,alist,all,anova,any,aov,aperm,append,apply,%
        approx,approxfun,apropos,Arg,args,array,arrows,as,asin,asinh,%
        atan,atan2,atanh,attach,attr,attributes,autoload,autoloader,ave,%
        axis,backsolve,barplot,basename,besselI,besselJ,besselK,besselY,%
        beta,binomial,body,box,boxplot,break,browser,bug,builtins,bxp,by,%
        c,C,call,Call,case,cat,category,cbind,ceiling,character,char,%
        charmatch,check,chol,chol2inv,choose,chull,class,close,cm,codes,%
        coef,coefficients,co,col,colnames,colors,colours,commandArgs,%
        comment,complete,complex,conflicts,Conj,contents,contour,%
        contrasts,contr,control,helmert,contrib,convolve,cooks,coords,%
        distance,coplot,cor,cos,cosh,count,fields,cov,covratio,wt,CRAN,%
        create,crossprod,cummax,cummin,cumprod,cumsum,curve,cut,cycle,D,%
        data,dataentry,date,dbeta,dbinom,dcauchy,dchisq,de,debug,%
        debugger,Defunct,default,delay,delete,deltat,demo,de,density,%
        deparse,dependencies,Deprecated,deriv,description,detach,%
        dev2bitmap,dev,cur,deviance,off,prev,,dexp,df,dfbetas,dffits,%
        dgamma,dgeom,dget,dhyper,diag,diff,digamma,dim,dimnames,dir,%
        dirname,dlnorm,dlogis,dnbinom,dnchisq,dnorm,do,dotplot,double,%
        download,dpois,dput,drop,drop1,dsignrank,dt,dummy,dump,dunif,%
        duplicated,dweibull,dwilcox,dyn,edit,eff,effects,eigen,else,%
        emacs,end,environment,env,erase,eval,equal,evalq,example,exists,%
        exit,exp,expand,expression,External,extract,extractAIC,factor,%
        fail,family,fft,file,filled,find,fitted,fivenum,fix,floor,for,%
        For,formals,format,formatC,formula,Fortran,forwardsolve,frame,%
        frequency,ftable,ftable2table,function,gamma,Gamma,gammaCody,%
        gaussian,gc,gcinfo,gctorture,get,getenv,geterrmessage,getOption,%
        getwd,gl,glm,globalenv,gnome,GNOME,graphics,gray,grep,grey,grid,%
        gsub,hasTsp,hat,heat,help,hist,home,hsv,httpclient,I,identify,if,%
        ifelse,Im,image,\%in\%,index,influence,measures,inherits,install,%
        installed,integer,interaction,interactive,Internal,intersect,%
        inverse,invisible,IQR,is,jitter,kappa,kronecker,labels,lapply,%
        layout,lbeta,lchoose,lcm,legend,length,levels,lgamma,library,%
        licence,license,lines,list,lm,load,local,locator,log,log10,log1p,%
        log2,logical,loglin,lower,lowess,ls,lsfit,lsf,ls,machine,Machine,%
        mad,mahalanobis,make,link,margin,match,Math,matlines,mat,matplot,%
        matpoints,matrix,max,mean,median,memory,menu,merge,methods,min,%
        missing,Mod,mode,model,response,mosaicplot,mtext,mvfft,na,nan,%
        names,omit,nargs,nchar,ncol,NCOL,new,next,NextMethod,nextn,%
        nlevels,nlm,noquote,NotYetImplemented,NotYetUsed,nrow,NROW,null,%
        numeric,\%o\%,objects,offset,old,on,Ops,optim,optimise,optimize,%
        options,or,order,ordered,outer,package,packages,page,pairlist,%
        pairs,palette,panel,par,parent,parse,paste,path,pbeta,pbinom,%
        pcauchy,pchisq,pentagamma,persp,pexp,pf,pgamma,pgeom,phyper,pico,%
        pictex,piechart,Platform,plnorm,plogis,plot,pmatch,pmax,pmin,%
        pnbinom,pnchisq,pnorm,points,poisson,poly,polygon,polyroot,pos,%
        postscript,power,ppoints,ppois,predict,preplot,pretty,Primitive,%
        print,prmatrix,proc,prod,profile,proj,prompt,prop,provide,%
        psignrank,ps,pt,ptukey,punif,pweibull,pwilcox,q,qbeta,qbinom,%
        qcauchy,qchisq,qexp,qf,qgamma,qgeom,qhyper,qlnorm,qlogis,qnbinom,%
        qnchisq,qnorm,qpois,qqline,qqnorm,qqplot,qr,Q,qty,qy,qsignrank,%
        qt,qtukey,quantile,quasi,quit,qunif,quote,qweibull,qwilcox,%
        rainbow,range,rank,rbeta,rbind,rbinom,rcauchy,rchisq,Re,read,csv,%
        csv2,fwf,readline,socket,real,Recall,rect,reformulate,regexpr,%
        relevel,remove,rep,repeat,replace,replications,report,require,%
        resid,residuals,restart,return,rev,rexp,rf,rgamma,rgb,rgeom,R,%
        rhyper,rle,rlnorm,rlogis,rm,rnbinom,RNGkind,rnorm,round,row,%
        rownames,rowsum,rpois,rsignrank,rstandard,rstudent,rt,rug,runif,%
        rweibull,rwilcox,sample,sapply,save,scale,scan,scan,screen,sd,se,%
        search,searchpaths,segments,seq,sequence,setdiff,setequal,set,%
        setwd,show,sign,signif,sin,single,sinh,sink,solve,sort,source,%
        spline,splinefun,split,sqrt,stars,start,stat,stem,step,stop,%
        storage,strstrheight,stripplot,strsplit,structure,strwidth,sub,%
        subset,substitute,substr,substring,sum,summary,sunflowerplot,svd,%
        sweep,switch,symbol,symbols,symnum,sys,status,system,t,table,%
        tabulate,tan,tanh,tapply,tempfile,terms,terrain,tetragamma,text,%
        time,title,topo,trace,traceback,transform,tri,trigamma,trunc,try,%
        ts,tsp,typeof,unclass,undebug,undoc,union,unique,uniroot,unix,%
        unlink,unlist,unname,untrace,update,upper,url,UseMethod,var,%
        variable,vector,Version,vi,warning,warnings,weighted,weights,%
        which,while,window,write,\%x\%,x11,X11,xedit,xemacs,xinch,xor,%
        xpdrows,xy,xyinch,yinch,zapsmall,zip},%
    otherkeywords={!,!=,~,$,*,\%,\&,\%/\%,\%*\%,\%\%,<-,<<-},
    alsoother={._$},
    sensitive,%
    morecomment=[l]\#,%
    morestring=[d]",%
    morestring=[d]'
}%

\newlength{\twless}
\newlength{\lmarg}


\DefineVerbatimEnvironment
{Verb}{Verbatim}
{fontsize=\fontsize{12pt}{14pt}\selectfont}

\newfloat[chapter]{ListingEnv}{lol}{Листинг}

\makeatletter
\AfterEndPreamble{
    \let\c@ListingEnv\relax 
    \newaliascnt{ListingEnv}{lstlisting} 
    \let\ftype@lstlisting\ftype@ListingEnv 
}
\makeatother

\newcounter{rowcnt}

\ifnumequal{\value{englishthesis}}{0}{
    
    \renewcommand{\leq}{\ensuremath{\leqslant}}
    
    \renewcommand{\geq}{\ensuremath{\geqslant}}

    \renewcommand{\epsilon}{\ensuremath{\upvarepsilon}}   
    \renewcommand{\phi}{\ensuremath{\upvarphi}}
    \renewcommand{\alpha}{\upalpha}
    \renewcommand{\beta}{\upbeta}
    \renewcommand{\gamma}{\upgamma}
    \renewcommand{\delta}{\updelta}
    \renewcommand{\varepsilon}{\upvarepsilon}
    \renewcommand{\zeta}{\upzeta}
    \renewcommand{\eta}{\upeta}
    \renewcommand{\theta}{\uptheta}
    \renewcommand{\vartheta}{\upvartheta}
    \renewcommand{\iota}{\upiota}
    \renewcommand{\kappa}{\upkappa}
    \renewcommand{\lambda}{\uplambda}
    \renewcommand{\mu}{\upmu}
    \renewcommand{\nu}{\upnu}
    \renewcommand{\xi}{\upxi}
    \renewcommand{\pi}{\uppi}
    \renewcommand{\varpi}{\upvarpi}
    \renewcommand{\rho}{\uprho}
    \renewcommand{\sigma}{\upsigma}
    \renewcommand{\tau}{\uptau}
    \renewcommand{\upsilon}{\upupsilon}
    \renewcommand{\varphi}{\upvarphi}
    \renewcommand{\chi}{\upchi}
    \renewcommand{\psi}{\uppsi}
    \renewcommand{\omega}{\upomega}
}{} 

\ifnumequal{\value{bibliosel}}{0}{%

\usepackage{cite}                                   

\bibliographystyle{BibTeX-Styles/utf8gost71u}    

\makeatletter
\renewcommand{\@biblabel}[1]{#1.}   
\makeatother

\newcommand*{\autocite}[1]{}  

\newcommand*{\insertbibliofull}{
\bibliography{biblio/external,biblio/author}         
}

\newtotcounter{citenum}
\def\oldcite{}
\let\oldcite=\bibcite
\def\bibcite{\stepcounter{citenum}\oldcite}
}{

\usepackage{csquotes} 
\makeatletter
\ifnumequal{\value{draft}}{0}{
\usepackage[%
backend=biber,
bibencoding=utf8,
sorting=none,
style=gost-numeric,
language=autobib,
autolang=other,
clearlang=true,
defernumbers=true,
sortcites=true,
doi=false,
isbn=false,
]{biblatex}[2016/09/17]
\ltx@iffilelater{biblatex-gost.def}{2017/05/03}%
{\toggletrue{bbx:gostbibliography}%
}{}
}{
\usepackage[%
backend=biber,
bibencoding=utf8,
sorting=none,
language=autobib,
autolang=other,
clearlang=true,
]{biblatex}[2016/09/17]%
}
\makeatother

\providebool{blxmc} 
\boolfalse{blxmc} 
\ifxetexorluatex
\else
    \DefineBibliographyStrings{russian}{number={\textnumero}}

    \ifdefmacro{\ExplSyntaxOn}{}{\usepackage{expl3}}
    \makeatletter
    \ltx@ifpackagelater{biblatex}{2020/02/23}{
    }{
        \ltx@ifpackagelater{biblatex}{2019/12/01}{
            \usepackage{expl3}[2020/02/25]
            \@ifpackagelater{expl3}{2020/02/25}{
                \booltrue{blxmc} 
            }{}
        }{}
    }
    \makeatother
    \ifblxmc
        \typeout{Assuming this version of biblatex defines MakeCapital
        incorrectly}
        \usepackage{xparse}
        \makeatletter
        \ExplSyntaxOn
        \NewDocumentCommand \blx@maketext@lowercase {m}
          {
            \text_lowercase:n {#1}
          }

        \NewDocumentCommand \blx@maketext@uppercase {m}
          {
            \text_uppercase:n {#1}
          }

        \RenewDocumentCommand \MakeCapital {m}
          {
            \text_titlecase_first:n {#1}
          }
        \ExplSyntaxOff

        \protected\def\blx@biblcstring#1#2#3{%
          \blx@begunit
          \blx@hyphenreset
          \blx@bibstringsimple
          \lowercase{\edef\blx@tempa{#3}}%
          \ifcsundef{#2@\blx@tempa}
            {\blx@warn@nostring\blx@tempa
             \blx@endnounit}
            {#1{\blx@maketext@lowercase{\csuse{#2@\blx@tempa}}}%
             \blx@endunit}}

        \protected\def\blx@bibucstring#1#2#3{%
          \blx@begunit
          \blx@hyphenreset
          \blx@bibstringsimple
          \lowercase{\edef\blx@tempa{#3}}%
          \ifcsundef{#2@\blx@tempa}
            {\blx@warn@nostring\blx@tempa
             \blx@endnounit}
            {#1{\blx@maketext@uppercase{\csuse{#2@\blx@tempa}}}%
             \blx@endunit}}
        \makeatother
    \fi
\fi

\ifsynopsis
\ifnumgreater{\value{usefootcite}}{0}{
    \ExecuteBibliographyOptions{autocite=footnote}
    \newbibmacro*{cite:full}{%
        \printtext[bibhypertarget]{%
            \usedriver{%
                \DeclareNameAlias{sortname}{default}%
            }{%
                \thefield{entrytype}%
            }%
        }%
        \usebibmacro{shorthandintro}%
    }
    \DeclareCiteCommand{\smartcite}[\mkbibfootnote]{%
        \usebibmacro{prenote}%
    }{%
        \usebibmacro{citeindex}%
        \usebibmacro{cite:full}%
    }{%
        \multicitedelim%
    }{%
        \usebibmacro{postnote}%
    }
}{}
\fi

\addbibresource[label=bl-external]{biblio/external_rebiber.bib}
\addbibresource[label=bl-author]{biblio/author.bib}
\addbibresource[label=bl-registered]{biblio/registered.bib}

\DeclareSourcemap{ 
    \maps{
        \map{
            \step[fieldsource=language, fieldset=langid, origfieldval, final]
            \step[fieldset=language, null]
        }
        \map{
            \step[fieldsource=numpages, fieldset=pagetotal, origfieldval, final]
            \step[fieldset=numpages, null]
        }
        \map{
            \step[fieldsource=pagestotal, fieldset=pagetotal, origfieldval, final]
            \step[fieldset=pagestotal, null]
        }
        \map[overwrite]{
            \step[fieldsource=shortjournal, final]
            \step[fieldset=journal, origfieldval]
            \step[fieldset=shortjournal, null]
        }
        \map[overwrite]{
            \step[fieldsource=shortbooktitle, final]
            \step[fieldset=booktitle, origfieldval]
            \step[fieldset=shortbooktitle, null]
        }
        \map{
            \step[fieldsource=medium,
            match=\regexp{Электронный\s+ресурс},
            final]
            \step[fieldset=media, fieldvalue=eresource]
            \step[fieldset=medium, null]
        }
        \map[overwrite]{
            \step[fieldset=issn, null]
        }
        \map[overwrite]{
            \step[fieldsource=abstract]
            \step[fieldset=abstract,null]
        }
        \map[overwrite]{ 
            \step[fieldsource=urldate,
            match=\regexp{([0-9]{2})\.([0-9]{2})\.([0-9]{4})},
            replace={$3-$2-$1$4}, 
            final]
        }
        \map[overwrite]{ 
            \step[fieldsource=keywords]
            \step[fieldset=keywords,null]
        }
        \map[overwrite]{ 
            \step[fieldsource=authorvak,final=true]
            \step[fieldset=keywords,fieldvalue={,biblioauthorvak},append=true]
        }
        \map[overwrite]{ 
            \step[fieldsource=authorscopus,final=true]
            \step[fieldset=keywords,fieldvalue={,biblioauthorscopus},append=true]
        }
        \map[overwrite]{ 
            \step[fieldsource=authorwos,final=true]
            \step[fieldset=keywords,fieldvalue={,biblioauthorwos},append=true]
        }
        \map[overwrite]{ 
            \step[fieldsource=authorconf,final=true]
            \step[fieldset=keywords,fieldvalue={,biblioauthorconf},append=true]
        }
        \map[overwrite]{ 
            \step[fieldsource=authorother,final=true]
            \step[fieldset=keywords,fieldvalue={,biblioauthorother},append=true]
        }
        \map[overwrite]{ 
            \step[fieldsource=authorpatent,final=true]
            \step[fieldset=keywords,fieldvalue={,biblioauthorpatent},append=true]
        }
        \map[overwrite]{ 
            \step[fieldsource=authorprogram,final=true]
            \step[fieldset=keywords,fieldvalue={,biblioauthorprogram},append=true]
        }
        \map[overwrite]{ 
            \perdatasource{biblio/external.bib}
            \step[fieldset=keywords, fieldvalue={,biblioexternal},append=true]
        }
        \map[overwrite]{ 
            \perdatasource{biblio/author.bib}
            \step[fieldset=keywords, fieldvalue={,biblioauthor},append=true]
        }
        \map[overwrite]{ 
            \perdatasource{biblio/registered.bib}
            \step[fieldset=keywords, fieldvalue={,biblioregistered},append=true]
        }
        \map[overwrite]{ 
            \step[fieldset=keywords, fieldvalue={,bibliofull},append=true]
        }
        \map[overwrite]{
            \step[typesource=online, typetarget=online, final]
            \step[fieldsource=howpublished, fieldset=organization, origfieldval]
            \step[fieldset=howpublished, null]
        }
    }
}

\ifnumequal{\value{mediadisplay}}{1}{
    \DeclareSourcemap{
        \maps{%
            \map{
                \step[fieldset=media, fieldvalue=text]
            }
        }
    }
}{}
\ifnumequal{\value{mediadisplay}}{2}{
    \DeclareSourcemap{
        \maps{%
            \map[overwrite]{
                \step[fieldset=media, null]
            }
        }
    }
}{}
\ifnumequal{\value{mediadisplay}}{3}{
    \DeclareSourcemap{
        \maps{
            \map[overwrite]{
                \step[fieldsource=media,match={text},final]
                \step[fieldset=media, null]
            }
        }
    }
}{}
\ifnumequal{\value{mediadisplay}}{4}{
    \DeclareSourcemap{
        \maps{
            \map[overwrite]{
                \step[fieldsource=media,match={eresource},final]
                \step[fieldset=media, null]
            }
        }
    }
}{}

\ifsynopsis
\else
\DeclareSourcemap{ 
    \maps{
        \map[overwrite]{
            \perdatasource{biblio/author.bib}
            \step[fieldset=addendum, null] 
        }
    }
}
\fi

\ifpresentation
\DeclareSourcemap{
    \maps{
    \map[overwrite,foreach={%
        address,%
        chapter,%
        edition,%
        editor,%
        eid,%
        howpublished,%
        institution,%
        key,%
        month,%
        note,%
        number,%
        organization,%
        pages,%
        publisher,%
        school,%
        series,%
        type,%
        media,%
        url,%
        doi,%
        location,%
        volume,%
    }]{
        \perdatasource{biblio/author.bib}
        \step[fieldset=$MAPLOOP,null]
    }
    }
}
\fi

\defbibfilter{vakscopuswos}{%
    keyword=biblioauthorvak or keyword=biblioauthorscopus or keyword=biblioauthorwos
}

\defbibfilter{scopuswos}{%
    keyword=biblioauthorscopus or keyword=biblioauthorwos
}

\defbibfilter{papersregistered}{%
    keyword=biblioauthor or keyword=biblioregistered
}

\DefineBibliographyStrings{english}{%
    pages = {p\adddot} 
}

\DefineBibliographyExtras{russian}{%
}
\DefineBibliographyExtras{english}{%
}

\setcounter{abbrvpenalty}{10000} 

\setcounter{highnamepenalty}{10000} 
\setcounter{lownamepenalty}{10000}

\makeatletter
    \newtotcounter{citenum}
    \defbibenvironment{counter}
        {\setcounter{citenum}{0}\renewcommand{\blx@driver}[1]{}} 
        {} 
        {\stepcounter{citenum}} 

    \newtotcounter{citeauthorvak}
    \defbibenvironment{countauthorvak}
        {\setcounter{citeauthorvak}{0}\renewcommand{\blx@driver}[1]{}}
        {}
        {\stepcounter{citeauthorvak}}

    \newtotcounter{citeauthorscopus}
    \defbibenvironment{countauthorscopus}
        {\setcounter{citeauthorscopus}{0}\renewcommand{\blx@driver}[1]{}}
        {}
        {\stepcounter{citeauthorscopus}}

    \newtotcounter{citeauthorwos}
    \defbibenvironment{countauthorwos}
        {\setcounter{citeauthorwos}{0}\renewcommand{\blx@driver}[1]{}}
        {}
        {\stepcounter{citeauthorwos}}

    \newtotcounter{citeauthorother}
    \defbibenvironment{countauthorother}
        {\setcounter{citeauthorother}{0}\renewcommand{\blx@driver}[1]{}}
        {}
        {\stepcounter{citeauthorother}}

    \newtotcounter{citeauthorconf}
    \defbibenvironment{countauthorconf}
        {\setcounter{citeauthorconf}{0}\renewcommand{\blx@driver}[1]{}}
        {}
        {\stepcounter{citeauthorconf}}

    \newtotcounter{citeauthor}
    \defbibenvironment{countauthor}
        {\setcounter{citeauthor}{0}\renewcommand{\blx@driver}[1]{}}
        {}
        {\stepcounter{citeauthor}}

    \newtotcounter{citeauthorvakscopuswos}
    \defbibenvironment{countauthorvakscopuswos}
        {\setcounter{citeauthorvakscopuswos}{0}\renewcommand{\blx@driver}[1]{}}
        {}
        {\stepcounter{citeauthorvakscopuswos}}

    \newtotcounter{citeauthorscopuswos}
    \defbibenvironment{countauthorscopuswos}
        {\setcounter{citeauthorscopuswos}{0}\renewcommand{\blx@driver}[1]{}}
        {}
        {\stepcounter{citeauthorscopuswos}}

    \newtotcounter{citeregistered}
    \defbibenvironment{countregistered}
        {\setcounter{citeregistered}{0}\renewcommand{\blx@driver}[1]{}}
        {}
        {\stepcounter{citeregistered}}

    \newtotcounter{citeauthorpatent}
    \defbibenvironment{countauthorpatent}
        {\setcounter{citeauthorpatent}{0}\renewcommand{\blx@driver}[1]{}}
        {}
        {\stepcounter{citeauthorpatent}}

    \newtotcounter{citeauthorprogram}
    \defbibenvironment{countauthorprogram}
        {\setcounter{citeauthorprogram}{0}\renewcommand{\blx@driver}[1]{}}
        {}
        {\stepcounter{citeauthorprogram}}

    \newtotcounter{citeexternal}
    \defbibenvironment{countexternal}
        {\setcounter{citeexternal}{0}\renewcommand{\blx@driver}[1]{}}
        {}
        {\stepcounter{citeexternal}}
\makeatother

\defbibheading{nobibheading}{} 
\defbibheading{pubgroup}{\section*{#1}} 
\defbibheading{pubsubgroup}{\noindent\textbf{#1}} 

\newcommand*{\insertbibliofull}{
    \ifnumequal{\value{englishthesis}}{0}{
        \printbibliography[keyword=bibliofull,section=0, title=\bibtitlefull]
    }{
        \printbibliography[keyword=bibliofull,section=0]    
    }
    \ifnumequal{\value{draft}}{0}{
      \printbibliography[heading=nobibheading,env=counter,keyword=bibliofull,section=0]
    }{}
}

}

\makeatletter
\let\blx@rerun@biber\relax
\makeatother

\begin{document}
\ifnumequal{\value{englishthesis}}{0}{
    \gappto\captionsrussian{
\renewcommand{\contentsname}{Оглавление}
\renewcommand{\figurename}{Рисунок}
\renewcommand{\tablename}{Таблица}
\renewcommand{\listfigurename}{Список рисунков}%
\renewcommand{\listtablename}{Список таблиц}%
\renewcommand{\bibname}{\bibtitlefull}%
\renewcommand{\nomname}{Список сокращений и условных обозначений}%
\renewcommand{\eqdeclaration}[1]{, см.~(#1)}%
\renewcommand{\pagedeclaration}[1]{, стр.~#1}%
\renewcommand{\nomAname}{Латинские буквы}%
\renewcommand{\nomGname}{Греческие буквы}%
\renewcommand{\nomXname}{Верхние индексы}%
\renewcommand{\nomZname}{Индексы}%
\unskip} 
\renewcommand{\contentsname}{Оглавление}
\renewcommand{\figurename}{Рисунок}
\renewcommand{\tablename}{Таблица}
\renewcommand{\listfigurename}{Список рисунков}%
\renewcommand{\listtablename}{Список таблиц}%
\renewcommand{\bibname}{\bibtitlefull}%
\renewcommand{\nomname}{Список сокращений и условных обозначений}%
\renewcommand{\eqdeclaration}[1]{, см.~(#1)}%
\renewcommand{\pagedeclaration}[1]{, стр.~#1}%
\renewcommand{\nomAname}{Латинские буквы}%
\renewcommand{\nomGname}{Греческие буквы}%
\renewcommand{\nomXname}{Верхние индексы}%
\renewcommand{\nomZname}{Индексы}%
}{}

\ifnumequal{\value{englishthesis}}{1}{
\thispagestyle{empty}
\begin{center}
\thesisOrganizationEn
\end{center}
\vspace{0pt plus4fill} 
\IfFileExists{figures/Skoltech_Logo-eps-converted-to.pdf}{
  \begin{minipage}[b]{0.5\linewidth}
    \begin{flushleft}
      \includegraphics[height=0.75cm]{figures/Skoltech_Logo-eps-converted-to}
    \end{flushleft}
  \end{minipage}%
  \begin{minipage}[b]{0.5\linewidth}
    \begin{flushright}
      As a manuscript\\
    \end{flushright}
  \end{minipage}
}{
\begin{flushright}
As a manuscript

\end{flushright}
}
\vspace{0pt plus6fill} 
\begin{center}
{\large \thesisAuthorEn}
\end{center}
\vspace{0pt plus1fill} 
\begin{center}
\textbf {\large 
\thesisTitleEn}

\vspace{0pt plus2fill} 
{
Speciality \thesisSpecialtyNumber\ ---

<<\thesisSpecialtyTitleEn>>
}

\ifdefined\thesisSpecialtyTwoNumber
{
Специальность \thesisSpecialtyTwoNumber\ ---

<<\thesisSpecialtyTwoTitle>>
}
\fi

\vspace{0pt plus2fill} 
Dissertation for a degree of

\thesisDegreeEn
\end{center}
\vspace{0pt plus4fill} 
\begin{flushright}
\ifdefined\supervisorTwoFio
Scientific advisors:

\supervisorRegaliaEn

\ifdefined\supervisorDead
\framebox{\supervisorFioEn}
\else
\supervisorFioEn
\fi

\supervisorTwoRegalia

\ifdefined\supervisorTwoDead
\framebox{\supervisorTwoFio}
\else
\supervisorTwoFio
\fi
\else
Scientific advisor:

\supervisorRegaliaEn

\ifdefined\supervisorDead
\framebox{\supervisorFioEn}
\else
\supervisorFioEn
\fi
\fi

\end{flushright}
\vspace{0pt plus4fill} 
{\centering\thesisCityEn\ --- \thesisYear\par}
}{}   
\thispagestyle{empty}
\begin{center}
\thesisOrganization
\end{center}
\vspace{0pt plus4fill} 
\IfFileExists{figures/Skoltech_Logo-eps-converted-to.pdf}{
  \begin{minipage}[b]{0.5\linewidth}
    \begin{flushleft}
      \includegraphics[height=0.75cm]{figures/Skoltech_Logo-eps-converted-to}
    \end{flushleft}
  \end{minipage}%
  \begin{minipage}[b]{0.5\linewidth}
    \begin{flushright}
      На правах рукописи\\
    \end{flushright}
  \end{minipage}
}{
\begin{flushright}
На правах рукописи

\end{flushright}
}
\vspace{0pt plus6fill} 
\begin{center}
{\large \thesisAuthor}
\end{center}
\vspace{0pt plus1fill} 
\begin{center}
\textbf {\large 
\thesisTitle}

\vspace{0pt plus2fill} 
{
Специальность \thesisSpecialtyNumber\ ---

<<\thesisSpecialtyTitle>>
}

\ifdefined\thesisSpecialtyTwoNumber
{
Специальность \thesisSpecialtyTwoNumber\ ---

<<\thesisSpecialtyTwoTitle>>
}
\fi

\vspace{0pt plus2fill} 
Диссертация на соискание учёной степени

\thesisDegree
\end{center}
\vspace{0pt plus4fill} 
\begin{flushright}
\ifdefined\supervisorTwoFio
Научные руководители:

\supervisorRegalia

\ifdefined\supervisorDead
\framebox{\supervisorFio}
\else
\supervisorFio
\fi

\supervisorTwoRegalia

\ifdefined\supervisorTwoDead
\framebox{\supervisorTwoFio}
\else
\supervisorTwoFio
\fi
\else
Научный руководитель:

\supervisorRegalia

\ifdefined\supervisorDead
\framebox{\supervisorFio}
\else
\supervisorFio
\fi
\fi

\end{flushright}
\vspace{0pt plus4fill} 
{\centering\thesisCity\ --- \thesisYear\par}
\ifdefmacro{\microtypesetup}{\microtypesetup{protrusion=false}}{} 
\tableofcontents*
\addtocontents{toc}{\protect\tocheader}
\endTOCtrue
\ifdefmacro{\microtypesetup}{\microtypesetup{protrusion=true}}{}        
\ifnumequal{\value{contnumfig}}{1}{}{\counterwithout{figure}{chapter}}
\ifnumequal{\value{contnumtab}}{1}{}{\counterwithout{table}{chapter}}


\addcontentsline{toc}{chapter}{Introduction}
\chapter*{Introduction}


\subsection*{Relevance of the work} 

This work considers the quantum model of computation. In this model, classical bits -- modelled by Boolean variables -- are replaced by qubits, modelled by two-dimensional complex vector spaces. The key effects enabling the quantum computations are quantum superposition and quantum entanglement. While the state space of $n$ bits is a Cartesian product of individual state spaces, the state space of $n$ qubits is the \textit{tensor} product of individual spaces. Because of this, a state of the qubits can in principle be not factorizable into single-qubit states. This, on the one hand, enables quantum computers to process many different classical bit states in parallel, and on the other hand, makes them extremely difficult to simulate by classical means: naively simulating a quantum device with $n$ qubits requires working with $2^n$-dimensional complex vectors.

The complexity class of problems efficiently solved by a quantum computer is called $\mathbf{BQP}$. While it is not known how $\mathbf{BQP}$ relates to $\mathbf{P}$, the class of problems efficiently solved by a classical computer, the former contains problems that do not yet have an efficient classical solutions. The most famous example of that was found with the discovery of Shor's algorithm, which enables quick factorization of integer numbers in their prime factors and, consequently, the defeat of hitherto unbreakable algorithms of encryption \cite{nielsen_quantum_2010}.

The appeal of quantum computers is offset by the notorious difficulty of making them in practice. While the current devices are far more advanced than the first proof-of-concept devices, they are still a long way from the characteristics that would enable them to, e.g.,~run the Shor's algorithm for any practical means. 
The limitations of physical quantum computers come from extreme fragility of quantum states. The quantum logic gates are not performed with perfect accuracy, the qubits themselves lose coherence over time, and there is even a probability of incorrect measurement, meaning that even the correct quantum state can be read with errors. For this reason, the current generation of quantum computers is referred to as noisy intermediate-scale quantum (NISQ) devices \cite{bharti_noisy_2021}.

While large-scale, fault-tolerant quantum computers are a long way ahead, there is research effort towards developing algorithms that can work within the limitations of NISQ devices. One such family of quantum algorithms is called variational quantum algorithms \cite{cerezo_variational_2020}. Such algorithms are designed to solve optimization problems. The broad idea of these algorithms is to encode the solution to the problem in a parametrized quantum state (called an ansatz), so that a series of measurements on that state can give an estimate of the cost function to be minimized. Then the parameters of the state are tuned so as to deliver a minimum to the cost function. The range of problems that can be addressed by this approach is rather extensive: from classical optimization problems, like the travelling salesepson problem, to problems that arise in quantum chemistry and quantum physics, e.g.~investigation of the potential landscape of chemical reactions~\cite{reiher_elucidating_2017} and ground state properties of electronic lattices~\cite{cade_strategies_2019}. One of the main algorithms studied in the dissertation is the variational quantum eigensolver (VQE) algorithm~\cite{peruzzo_variational_2014} which uses the variational approach to solve ground state problems from quantum physics and quantum chemistry. 
Another important branch of study is devoted to quantum machine learning with variational algotihms~\cite{skolik_layerwise_2020,havlicek_supervised_2019,schuld_circuit-centric_2020}. In that approach, a quantum device acts as a classifier that is trained to partition data points encoded as quantum states. 

Despite the simplicity of the approach, the variational quantum algorithms are far from being understood well. Since they involve optimization, there are lots of questions concerning the convergence of this optimization, the properties of the cost function landscape, the choice of the circuits, etc. This dissertation contributes to the literature both in numerical experiments and in theoretical results.

\subsection*{Dissertation goals} 

The goal of this dissertation is to investigate the performance of variational quantum algorithms and to further the theoretical understanding of their behavior. To achieve this goal, we set up and performed the following tasks:

\begin{enumerate}
    \item Develop a numerical implementation the variational quantum eigensolver algorithm and study its convergence for physical problems. Study its dependence on the depth of the circuit, number of qubits, and other relevant parameters.
    \item Develop a numerical implementation of a quantum classifier and investigate its properties. In particular, study its ability to learn on quantum data, namely on quantum states obtained using VQE.
    \item Study the physically relevant properties of the VQE-optimized solutions for the next-nearest-neighbor Hubbard model. Analyze the behavior of gradients in the associated optimization problem.
    \item Calculate a bound on the variance of derivatives of Hamiltonian cost functions with respect to random selection of circuit parameters. Estimate the onset of the barren plateau regime with the increasing depth and connectivity of a parametrized quantum circuit.
    \item Develop an algorithm for bounding the fidelity of the Greenberger-Horne-Zeilinger state and other Clifford states based on their parent Hamiltonian.
\end{enumerate}

\subsection*{Statements defended}

\begin{enumerate}

    \item We applied Adiabatically-assisted VQE (AAVQE) to the ground state problem for the Heisenberg XXZ model, i.e.~the family of spin Hamiltonians of the type 
    \begin{equation*}
        H(h) = \sum_{i=1}^{n} X_i X_{i+1} + Y_i Y_{i+1}+ h Z_i Z_{i+1},
    \end{equation*}
    where $X_i, Y_i$, and $Z_i$ are respective Pauli matrices acting on $i$'th spin. AAVQE consists in using the solution for $H(h)$ as a starting point of the optimization process for $H(h + \delta h)$, eventually finding an approximate ground state for a range of parameters $[h_{\mathrm{min}}, h_{\mathrm{max}}]$. Numerical experiments show that Adiabatically-assisted VQE finds different solutions depending on the direction of whether the optimization starts from $h_{\mathrm{min}}$ and increases $h$, or starts from $h_{\mathrm{max}}$ and gradually decreases $h$. Specifically, the largest difference occures near the critical point of the model ($h = 1$). In contrast, the transverse-field Ising model shows no such behavior near the critical point: both methods of running AAVQE yield the same ground state energies up to the tolerance of the optimizer.
    \item The approximation found by VQE shows the largest energy error in the vicinity of the critical point of either spin model. For the best ans\"atze considered, the maximum energy error is $0.12$ units for the Ising model and $1.1$ units for the Heisenberg model. We argue that this is consistent with the fact that critical states demonstrate a logarithmic scaling of entanglement entropy, 
    unattainable by a fixed-depth ansatz circuit.
    \item We proposed a quantum classifier which demonstrates that quantum machine learning models can partition quantum data. The classifier was tested for ten qubits on the VQE-approximated ground states of transverse-field Ising model and was able to predict the phase of the system with 97\% accuracy. A similar experiment conducted for the Heisenberg XXZ model yielded 93\% accuracy.

    \item Variational quantum eigensolver was numerically applied to a Hubbard-like model, with ansatz depth ranging from one to ten layers. For 4 qubits and fewer, VQE found the exact solution up to machine precision with 1-2 layers. For 5-11 qubits, the error scaling with depth was fit to exponential decay (correlation coefficient between depth and $\log_{10} \Delta E$ was $-0.92$ or less in all cases).
    
    \item We estimated the variance of partial derivatives of the energy cost function for the next-nearest-neighbor Hubbard model in 4-10 qubits by random sampling of the ansatz parameters. The model was transformed to a qubit Hamiltonian by Jordan--Wigner and Bravyi--Kitaev transformations. 
    We observed that, under the Jordan--Wigner transformation, the dependence of the variance on circuit depth becomes constant after at most 15 layers. The behavior of the variance as a function of the number of qubits is consistent with the prediction of exponential decay given in \cite{mcclean_barren_2018}.
    For the Bravyi--Kitaev transform we observed a longer transient behavior with the circuit depth.
    \item We derived a lower bound on the variance of derivatives for parametrized quantum circuits composed of blocks that constitute local 2-designs. Let the block that depends on a real parameter $\theta$ be situated in layer $l_c$ out of $l$ layers total. Then the variance $\Var \partial_\theta E$ behaves as
    \begin{equation*}
        \Var \partial_\theta E \in \sum_h \Omega\left(3^{-|C(h)|} \left(\frac{3}{4}\right)^{l-l_c}\right).
    \end{equation*}
    Here the sum is taken over every Pauli string $h$ in the Pauli decomposition of the cost function, and $|C(h)|$ is the size of the causal cone of this Pauli string.
    \item We developed an algorithm to estimate the fidelity of Clifford states using at most $n$ series of Pauli measurements, where $n$ is the number of qubits. 
\end{enumerate}

\subsection*{Scientific novelty}
\begin{enumerate}
    \item In the original proposal \cite{garcia-saez_addressing_2018}, the adiabatically-assisted VQE algorithm was used for a family of Hamiltonians $H(\tau)$ such that $H(0)$ is an easily-solved problem, while $H(1)$ encodes a difficult problem. For the transverse-field Ising and Heisenberg models, the ground state problem dependence on the parameter can be characterized as <<easy, hard, easy>>: $H(0)$ is easy to solve, $H(1)$ is difficult, then in the limit of $\tau \rightarrow \infty$ the problem $H(\tau)$ again becomes easy.
    \item We demonstrated that a quantum classifier can be trained on intrinsically quantum data. Prior art appears to focus on classification problems based on classical input data.    
    \item For fermionic Hamiltonians, the onset of barren plateaus phenomenon was shown to occur differently depending on the choice of fermion-to-qubit mapping, which was previously not considered in the literature.
    \item We derived a lower bound on the variance of the derivatives of cost function with respect to ansatz parameters. Compared to the existing literature \cite{mcclean_barren_2018,cerezo_cost-function-dependent_2020}, our estimate depends only on the size of the causal cone of the circuit and is applicable to quantum circuits of arbitrary topology (as long as they consist of two-qubit blocks that constitute approximate 2-designs).
    \item We proposed a method of estimating fidelity for the GHZ state that only involves two series of measurements. The method relies on the relation between the energy of a state with respect to a Hamiltonian and the fidelity between that state and the ground state (the stability lemma).
\end{enumerate}

\subsection*{Theoretical and practical significance} Classification of quantum phases by quantum means is potentially useful for condensed matter physics. The results on the convergence of VQE are important for the design of quantum optimization algorithms. The proposed method for validating the GHZ state is potentially useful for evaluating the properties of quantum devices in a simple manner.
Practical significance of the work is supported by the usage of the results in delivering the RFBR grant No.\ 19-31-90159 ``Aspiranty''.

The \textbf{validity of the work} is supported by numerical experiments and rigorous mathematical proofs, where applicable.

\subsection*{Presentations and validation of the results}
The main results of the work have been reported in the following scientific conferences and workshops:

\begin{enumerate}
    \item International Conference on Quantum Technologies (July 15-19, 2019, Moscow, poster);
    \item 62nd MIPT conference (Nov 18-24, 2019, Moscow, talk);
    \item Inaugural Symposium for Computational Materials Program of Excellence (September 4-6, 2019, Moscow, poster);
    \item International School on Quantum Technologies (March 1-7, 2020, Sochi, poster);
    \item International Conference on Quantum Technologies (July 12-16, 2021, Moscow, online, poster);
    \item QuOne workshop on Quantum Machine Learning (Feb 22-24, 2022, Tehran, online, talk).
\end{enumerate}

\subsection*{Structure of the dissertation} The dissertation consists of introduction, six chapters, conclusions, bibliography, list of symbols and abbreviations, list of tables, list of figures, and supplemental material.

\subsection*{Publications}
The work in this thesis is based on the following publications:

\begin{enumerate}
    \item \fullcite{uvarov_machine_2020}
    \item \fullcite{uvarov_variational_2020}
    \item \fullcite{uvarov_barren_2021}
\end{enumerate}

\subsection*{Acknowledgments} I thank my advisor, Jacob Biamonte, for providing guidance and assistance at all stages of my doctoral study. I also thank my coauthors, Dmitry Yudin and Andrey Kardashin for a productive collaboration. I am grateful to Akshay Vishwanathan, Richik Sengupta, Soumik Adhikary, Daniil Rabinovich, Hariphan Philathong, Oksana Borzenkova, Ernesto Campos, Aly Nasrallah, and Mauro Morales for numerous insightful discussions. Last but not least, I thank my family and my friends for immense support during my scientific journey.

\chapter{Main concepts in quantum computation}
\label{chap:quantum_basics}

In this chapter, we will briefly review the basic mathematical concepts necessary to reason about quantum computers. For a more extensive introduction, we refer the reader to the classic textbooks: \cite{nielsen_quantum_2010,kitaev_classical_2002}.

\section{Definitions and notation}

A \textit{qubit} is a quantum system with two controllable states. We will model such a system with a two-dimensional state space $\mathbb{C}^2$ with a distinguished orthonormal basis, called the \textit{computational basis} and denoted as $\{ \ket{0}, \ket{1} \}$. The Hilbert space\footnote{Throughout the thesis, we will only work in finite-dimensional Hilbert spaces.} $\mathcal{H}$ of $n$ qubits is the complex vector space $(\mathbb{C}^2)^{\otimes n}$. \textit{Pure states} of the qubits are the vectors $|\psi \rangle \in \mathcal{H}$ of norm one (vectors denoted like that are called \textit{ket vectors}). The elements of the dual space $\mc{H}^*$ are depicted as so-called \textit{bra vectors} $\bra{\psi}$. The inner product between two states $\ket{\phi}$ and $\ket{\psi}$ is denoted simply as a concatenation of the corresponding bra and ket vectors $\braket{\phi}{\psi}$.

The standard basis for the registry of $n$ qubits, also called the computational basis, is the tensor product of individual computational bases. 
When writing such basis states, we will omit the tensor product: for example, $\ket{0} \otimes \ket{0} \otimes \ket{1} \equiv \ket{001}$.

A key feature of quantum states is entanglement:

\begin{definition}
    A state $\ket{\psi} \in \mathcal{H}_1 \otimes \mathcal{H}_2$ is \textit{entangled} if there are no states $\ket{\phi_1} \in \mathcal{H}_1, \ \ket{\phi_2} \in \mathcal{H}_2$ such that $\ket{\psi} = \ket{\phi_1} \otimes \ket{\phi_2}$. Otherwise, the state is called \textit{separable} with respect to the bipartition $\mathcal{H} = \mathcal{H}_1 \otimes \mathcal{H}_2$. If an $n$-qubit state can be presented as a tensor product of $n$ single-qubit states, it is called a \textit{separable} or a \textit{product state}.
\end{definition}

\begin{example}
    The \textit{Bell state} $\ket{\Phi} = \frac{1}{\sqrt{2}}(\ket{00} + \ket{11})$ is an example of an entangled state of two qubits. This is easy to check: an arbitrary separable state of two qubits has the form $(\alpha \ket{0} + \beta \ket{1}) \otimes (\gamma \ket{0} + \delta \ket{1})$. By comparing the coefficients, we arrive to the following system of equations that has no solutions:
    \begin{equation}
        \left\{
            \begin{array}{rl}
                \alpha \gamma = & \frac{1}{\sqrt{2}} \\
                \beta \delta = & \frac{1}{\sqrt{2}} \\
                \alpha \delta = & 0 \\
                \beta \gamma = & 0 \\
            \end{array}
        \right.     
    \end{equation}
    The $n$-qubit analog of the Bell state $\ket{\Psi} = \frac{1}{\sqrt{2}}(\ket{0...0} + \ket{1...1})$ is known as the \textit{Greenberger--Horne--Zeilinger state} or as \textit{Schr\"odinger's cat state}.
\end{example}

The space of linear operators on $\mc{H}$ can be equipped with a Hilbert-Schmidt inner product:
\begin{equation}
(A, B) = \frac{1}{\mathrm{dim} \mc{H}} \Tr A^\dagger B.
\end{equation}
Relative to this inner product, the space of linear operators on $\mc{H}$ admits an orthogonal basis of \textit{Pauli strings}, which are tensor products of \textit{Pauli matrices}:

\begin{definition}[Pauli matrices]
    \textit{Pauli matrices} are the $2 \times 2$ identity matrix $\id$ and the following three $2 \times 2$ matrices:
    \begin{equation}
        X = \begin{pmatrix}
            0 & 1 \\ 1 & 0 
        \end{pmatrix}, \ 
        Y = \begin{pmatrix}
            0 & -\mathrm{i} \\ \mathrm{i} & 0 
        \end{pmatrix}, \ 
        Z = \begin{pmatrix}
            1 & 0 \\ 0 & -1
        \end{pmatrix}.
    \end{equation}
\end{definition}

\begin{definition}[Pauli string]
A \emph{Pauli string} is a tensor product of $n$ Pauli matrices $\{\id, X, Y, Z \}$. The $n$-qubit identity operator $\id \otimes \id \otimes ... \otimes \id$ is the \emph{trivial} or \textit{unit} string. The \textit{algebraic locality} or just \textit{locality} of a Pauli string is the number of non-identity Pauli matrices contained in the string.
\end{definition}

Since the product of two Pauli strings is, up to a multiple of $\rmi$, a Pauli string, it is sometimes important to consider the \emph{Pauli group} consisting of all Pauli strings, possibly multiplied by $\rmi$, $-1$, or $-\rmi$.

The expansion of the operator in the basis of Pauli strings will be called its \textit{Pauli decomposition}. For a Hermitian operator $H$ on $\mc{H}$ in a given orthonormal basis, we will call its \textit{locality} $\operatorname{loc} H$ the maximum locality of its terms in the decomposition and its \textit{cardinality} $\operatorname{card} H$ the number of Pauli strings in the decomposition. Note that for a Hermitian operator, the coefficients of the Pauli decomposition will be real.

All Pauli strings have eigenvalues $\pm 1$. The eigenbases of individual Pauli matrices are also important. The computational basis is the eigenbasis for the $Z$ matrix. The eigenstates of the $X$ matrix are known as the plus state and the minus state: $\ket{+} = \frac{1}{\sqrt{2}}(\ket{0} + \ket{1}), \ \ket{-} = \frac{1}{\sqrt{2}}(\ket{0} - \ket{1})$. The eigenstates of $Y$ are $\ket{y_\pm} = \frac{1}{\sqrt{2}}(\ket{0} \pm \mathrm{i} \ket{1})$.

In the basis of Pauli strings, the Gram matrix of the Hilbert-Schmidt inner product is the identity matrix. In fact, $\operatorname{Herm}(2^n)$ behaves like $\mathbb{R}^{4^n}$ equipped with the standard inner product. Because of that, for an operator $A = \sum c_\alpha \sigma_\alpha$, written in its Pauli decomposition, we will occasionally be interested in $p$-norms of the vector of its coefficients $||\{c_\alpha\}||_p$.

\begin{proposition}
    \label{prop:ad_is_pauli_orthogonal}
    Let $U$ be a unitary operator on $\mc{H}$. Denote $\operatorname{Ad}_U$ the operator of conjugation by $U$: 
    \begin{align}
        \operatorname{Ad}_U: \operatorname{End}(\mc{H}) &\rightarrow  \operatorname{End}(\mc{H}) \\
        \operatorname{Ad}_U(X) &= U^\dagger X U
    \end{align}
    This operator is unitary on $\operatorname{End}(\mc{H})$ and orthogonal on $\operatorname{Herm}(\mc{H})$  (when the latter is considered as a real vector space spanned by Pauli strings).
\end{proposition}
\begin{proof}
    If $H$ is Hermitian, $U^\dagger HU$ is also Hermitian, so $\operatorname{Herm}(\mc{H})$ is an invariant subspace of $\operatorname{Ad}_U$. The Hilbert--Schmidt product $(A, B) = \Tr(A^\dagger B)$ is preserved under the conjugation of both factors: $\operatorname{Tr} (U^\dagger A^\dagger U U^\dagger B U) = \operatorname{Tr} (A^\dagger B)$.
\end{proof}

In particular, it follows that the vector 2-norm of a Hermitian operator (i.e.~the sum of the squares of its Pauli coefficients) is preserved under $\operatorname{Ad}_U$.

\section{Quantum model of computation}


What operations can be performed on idealized qubits? The basic building blocks of any quantum algorithm are (i) quantum gates and (ii) measurements. 

\subsubsection{Quantum gates}

Quantum gates are represented by unitary matrices:
\begin{equation}
    \mc{U}(d) : = \{ U \in \mathbb{C}^{d \times d} | UU^\dagger = \id \}.
\end{equation}
We say that a quantum gate $U$ is an $m$-qubit gate if the gate is supported only on $m$ qubits. That is, there is a subset of $m$ qubits such that, under proper ordering of the tensor factors, $U$ can be represented as $V \otimes \id$ for some $V \in \mc{U}(2^m)$.

A quantum computer is assumed to be able to produce few-qubit quantum gates from some fixed set. A sequence of quantum gates, together with measurements (see next paragraph) and possibly gates conditioned on classical measurements, is called a \textit{quantum circuit}. 

We highlight two key features of unitary maps:

\begin{enumerate}
    \item Unitary means reversible. We cannot erase information. Thus, we cannot naively implement quantum analogs of AND and OR gates. However, with some extra effort, classical logic gates can still be embedded in quantum gates.
    \item Multi-qubit gates can produce entanglement.
\end{enumerate}

Sometimes a quantum gate will depend on a real parameter. In this case, the circuit containing such gates will be called a \textit{parametrized quantum circuit}. We will assume that this dependence is infinitely differentiable. In the context of variational quantum algorithms, a parametrized quantum circuit is sometimes called an \textit{ansatz}.

\begin{example}
    A single-qubit rotation $R_Y(\theta) = e^{-\rmi Y \theta / 2}$ is a parametrized quantum gate. The fact that $Y^2 = \id$ enables the following decomposition: $R_Y(\theta) = \cos (\theta/2) \id - \rmi \sin (\theta / 2) Y$.
\end{example}

\subsubsection{Measurements}

In general, when we measure a quantum system, we measure some kind of \textit{observable}, i.e.~a Hermitian operator. Let $A$ be some observable with eigenvalues $\lambda_i$ and corresponding eigenstates $\ket{\lambda_i}$. Let the system be in the state $\ket{\psi} = \sum c_i \ket{\lambda_i}$. Then, upon measurement, we will observe the value $\lambda_i$ with probability $|c_i|^2$. The expected value of the measurement outcome is $\sum_i |c_i|^2 \lambda_i$. Importantly, the expected value of $A$ with respect to $\ket{\psi}$ is also expressed as
\begin{equation}
    \langle A \rangle = \bra{\psi} A \ket{\psi}.
\end{equation}
If the measurement does not destroy the system (e.g.~in photonic quantum computing, measuring a qubit effectively destroys it), the state after measurement with outcome $\lambda$ can be described by the following \textit{density operator} (see Section \ref{subsec:mixed}):
\begin{equation}
    \label{eq:postmeasurement}
    \rho_{\text{post}} = \frac{\Pi_\lambda \ket{\psi}\bra{\psi} \Pi_\lambda}{\Tr (\ket{\psi}\bra{\psi} \Pi_\lambda)},
\end{equation}
where $\Pi_\lambda$ is the projector on the eigenspace of $\lambda$. If the eigenspace is one-dimensional, the language of density operators can be avoided: in such case, the post-measurement state is a pure state $\ket{\psi}$.

Different observables can be measured simultaneously if and only if the corresponding operators have a shared eigenbasis. Sometimes we only care about the basis itself, so the observables we measure are projectors on the basis states: $P_1 = \ket{e_1} \bra{e_1}, ..., P_k = \ket{e_k} \bra{e_k}$. We then say that we measure in the basis $(e_1, ..., e_k)$.

Without any extra effort, a quantum computer makes a $Z$ measurement on each qubit. To measure qubits in different bases, one should apply single-qubit or multi-qubit unitaries before measurement. For example, if we want to measure the qubits in the $\{ \ket{+},\ket{-} \}$ basis, we should apply a Hadamard gate to the qubit. If we would like to measure the qubits in some entangled basis, we would have to apply some multi-qubit gates before the measurement. Since single-qubit operations are easier to perform than multi-qubit operations, we mostly consider measurements of observables that don't need entangled bases, such as Pauli strings.

\section{Noise and decoherence}

\subsection{Mixed states and density matrices}
\label{subsec:mixed}
By the very nature of quantum computation as external manipulation of the qubits, it is not enough to consider the qubits as isolated physical entities. In order to focus on the behavior of qubits, though, one can get rid of the irrelevant degrees of freedom, while preserving their action on the qubits. To do that, we represent quantum states not as vectors in $\mc{H}$, but as a certain kind of linear operators on $\mc{H}$. 

\begin{definition}
    A \textit{density operator} (or \textit{density matrix}) on $\mc{H}$ is an operator $\rho \in \operatorname{End}(\mc{H})$ with the following properties:
    \begin{enumerate}
        \item $\rho$ is Hermitian;
        \item All eigenvalues of $\rho$ are non-negative;
        \item $\Tr \rho = 1$.
    \end{enumerate}
\end{definition}
Any density matrix can be diagonalized, i.e.~there is a basis $e_i$ such that
\begin{equation}
    \rho = \sum_i p_i \ket{e_i}\bra{e_i}.
\end{equation}
When such a diagonalization has only one term, we call $\rho$ a pure state, otherwise it is called a \textit{mixed state}. Note that the state is pure if and only if $\rho^2 = \rho$, that is, when the matrix is a projector $\ket{\psi} \bra{\psi}$. 


One way to arrive to density matrix formalism is to consider the situation when there is an entangled pure state in a Hilbert space $\mc{H} = \mc{H}_1 \otimes \mc{H}_2$, and the observer only has access to one subsystem. Denote $\{\ket{e_i}\}$ the basis of $\mc{H}_1$ and $\{\ket{f_j}\}$ the basis of $\mc{H}_2$. If we measure some observable $A$ only on the first subsystem, that is equivalent to measuring $A \otimes \id$ on the entire system. If the state is equal to $\sum_{ij} c_{ij} \ket{e_i} \ket{f_j}$, then the expected value of $A$ is equal to 
\begin{equation}
    \label{eq:expect_partial}
    \langle A \rangle = \sum_{ijkl} c^*_{ij} c_{kl} (\bra{e_i} \bra{f_j}) A (\ket{e_k} \ket{e_l}) = \sum_{ijkl} c^*_{ij} c_{kl} \bra{e_i} A \ket{e_k} \delta_{jl}.
\end{equation}
Define the \textit{partial trace} of the state as follows:
\begin{equation}
    [\rho_{\mc{H}_1}]_{ik} := \sum_j c^*_{ij} c_{kj} \ket{e_k} \bra{e_i}.
\end{equation}
The expected value (\ref{eq:expect_partial}) is now equal to 
\begin{equation}
    \langle A \rangle = \Tr (\rho_{\mc{H}_1} A).
\end{equation}
This partial trace is now a density matrix called \textit{reduced density matrix}. Of course, reduced density matrices are defined not only for pure states, but also for mixed states.


There are several metrics to estimate the mixedness of a state. The simplest one is \textit{purity}, which we define as $\Tr{\rho^2}$. A state $\rho$ is pure if and only if $\rho$ is a projector, i.e. $\rho^2 = \rho$. In this case, the purity is equal to 1, which is the maximum possible purity. The minimum possible purity can be obtained from the relation between different vector norms. Let $\lambda_i, i = 1,..., \operatorname{dim} \mc{H}$ be the eigenvalues of $\rho$. Their sum is the one-norm of the vector $\boldsymbol{\lambda} = (\lambda_1, ..., \lambda_{\operatorname{dim} \mc{H}})$. The purity is the sum of squares of $\lambda_i$, i.e. the square of the two-norm. From the relations between $p$-norms of vectors we know that $||\lambda||_2 \geq ||\lambda||_2 / \sqrt{\operatorname{dim} \mc{H}}$. This implies that minimum possible purity is equal to $1/\operatorname{dim} \mc{H}$. This minimum is reached by the \textit{maximally mxied state} $\rho = \id / \operatorname{dim} \mc{H}$.

A more complicated measure of mixedness is the \textit{von Neumann entropy}:

\begin{equation}
    S(\rho) = - \Tr \rho \log \rho.
\end{equation}

This formula makes sense even when $\rho$ has zero eigenvalues because $\lim_{x \rightarrow 0^+} x \log x = 0$. The units of $S(\rho)$ are called \textit{nats} if the formula uses the natural logarithm, and \textit{bits} or \textit{ebits} if the logarithm is base two. 

The Bell state $\ket{\Phi} = \frac{1}{\sqrt{2}}(\ket{00} + \ket{11})$ has exactly one ebit of entropy: the reduced density matrix of any qubit is simply $\frac12 \id$, hence the entropy is equal to $-\log \frac{1}{2}$ nats or $-\log_2 \frac{1}{2} = 1$ ebit. For that reason, the states that are obtaied from the Bell state by application of local unitary operators are called \textit{maximally entangled states} in two qubits.

The notion of entanglement can also be extended to mixed states. A density matrix is called entangled if it cannot be represented as a convex combination of product states, i.e.~in the form $\rho = \sum p_i \rho_1^{(i)} \otimes \rho_2^{(i)}$, where $p_i$ are positive~\cite{werner_quantum_1989}. However, for mixed states, deciding whether a given state is entangled or not is already a complicated problem, which is only solved for small systems~\cite{horodecki_five_2022}.

\subsection{Noise models}

When we introduced the density matrix formalism, we assumed that there are two subsystems, and found that the reduced density matrix can provide information about the entanglement between the subsystems. However, this also works if one of the subsystem is the entire quantum registry, and the other subsystem is the rest of the Universe. This way, the state of the quantum device is described by a density matrix.

A perfect quantum gate $U$ will act on a mixed quantum state $\rho$ by conjugation: $\rho \mapsto U \rho U^\dagger$. However, when a unitary maps acts by conjugation on a composite system, its action on reduced density matrices does not necessarily reduce to conjugation by unitary maps.

\begin{example}
    Consider a quantum circuit consisting of two gates: a Hadamard gate acting on the first qubit and a CNOT gate acting on qubits $(1, 2)$. Starting from the state $\ket{00}$, this circuit prepares the Bell state $\frac{1}{\sqrt{2}}(\ket{00} + \ket{11})$. Before the application of the gates, the reduced density matrix of the first qubit was equal to $\ket{0} \bra{0}$, while after the gates, it is equal to $\frac12 \id$. These matrices have different rank, therefore they cannot be conjugate to each other.
\end{example}

In general, physically meaningful operations that can be performed on density matrices are described by a special kind of maps called \textit{quantum channels}.


\begin{definition}
    A quantum channel is a linear map $\Phi: \operatorname{End}(\mc{H}_1) \rightarrow \operatorname{End}(\mc{H}_2)$, where $\mc{H}_1$ and $\mc{H}_2$ are Hilbert spaces, such that:
    \begin{enumerate}
        \item $\Phi$ is completely positive: if $\rho$ is positive semidefinite, then so is $\Phi(\rho)$;
        \item $\Phi$ is trace preserving: $\Tr(\rho) = \Tr(\Phi(\rho))$.
    \end{enumerate}
\end{definition}

We will not go deep into the theory of quantum channels, instead only highlighting a few interesting examples. For a more detailed account we refer the reader to Ref.~\cite{watrous_theory_2018}.
\begin{enumerate}
    \setcounter{enumi}{-1}
    \item The identity channel $\rho \mapsto \rho$ is of course a quantum channel;
    \item Conjugation by a unitary $\mathrm{Ad}_U(\rho) = U \rho U^\dagger$;
    \item A \textit{depolarizing channel} with noise strength $p$ is used to simulate uniform, featureless noise:
    \begin{equation}
        \Phi(\rho) = (1-p) \rho + \frac{p}{\operatorname{dim} \mc H} \id.
    \end{equation}
    \item For single-qubit denisty matrices, one can consider \textit{bit flip}    and \textit{phase flip} channels:
    \begin{align}
        \Psi_{\mathrm{bit}}(\rho) &= (1-p) \rho + p \cdot X \rho X \\
        \Psi_{\mathrm{phase}}(\rho) &= (1-p) \rho + p \cdot Z \rho Z 
    \end{align}

\end{enumerate}

Let us give a few more observations about the depolarizing channel. First, we note that for one qubit, this channel has the following form:
\begin{equation}
    \Phi(\rho) = (1-\frac{3}{4}p) \rho + \frac{p}{4}(X \rho X + Y \rho Y + Z \rho Z).
\end{equation}
Second, we observe that, as long as their supports are the same, the depolarizing channel and the conjugation channel commute. This can be verified directly:
\begin{gather}
    \mathrm{Ad}_U \circ \Phi_p (\rho) = \mathrm{Ad}_U ((1 - p) \rho + \frac{p}{\dim{\mc{H}}}\id) = (1 - p) U \rho U^\dagger + \frac{p}{\dim{\mc{H}}}\id; \\
    \Phi_p \circ \mathrm{Ad}_U (\rho) = \Phi_p (U \rho U^\dagger) =  (1 - p) U \rho U^\dagger + \frac{p}{\dim{\mc{H}}}\id.
\end{gather}
This is no longer the case when the support of $U$ is larger than the support of $\Phi$. For example, a two-qubit swap channel
\begin{equation}
    \mathrm{SWAP}: \rho \otimes \eta \mapsto \eta \otimes \rho
\end{equation}
does not commute with the depolarizing channel acting on the first qubit.

\section{Tensor networks}

The theory of quantum computation extensively uses the notion of tensor product, and thus greatly benefits from the diagrammatic language of \textit{tensor networks}. In fact, the diagrams of quantum circuits themselves can be thought of as tensor network diagrams. Detailed exposition of tensor networks for quantum physics and quantum information can be found in Refs.~\cite{biamonte_lectures_2020,bridgeman_hand-waving_2017,orus_practical_2014}.

Let $V$ be the vector space of interest (most frequently we will deal with the qubit state space $\mathbb{C}^2$) and $V^*$ its dual space. A $(p, q)$\textit{-tensor} is a linear function from $V^q$ to $V^p$. The space of such tensors $T^p_q (V)$ is isomorphic to $V^* \otimes ... \otimes V^* \otimes V \otimes ... \otimes V$ with $q$ copies of $V^*$ and $p$ copies of $V$. In tensor network diagrams, such a tensor is depicted as a box with $p$ wires going right and $q$ wires going left.

\begin{equation}
    M^{i_1 ... i_p}_{j_1 ... j_q} = \vcenter{\hbox{\includegraphics{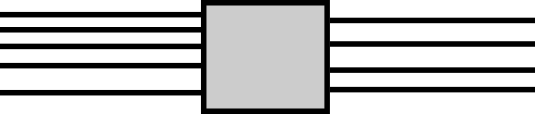}}}
\end{equation}
For example, a linear map $A: V \rightarrow V$ is depicted as a box with one wire on each side. The box itself can be annotated as well:
\begin{equation}
    A^i_j = \vcenter{\hbox{\includegraphics{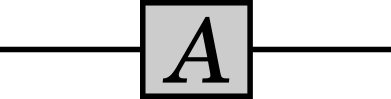}}}
\end{equation}
We will draw a vector $v \in V$ as a triangle-shaped box with one wire sprouting to the left, while a covector $\delta \in V^*$ is depicted as a triangle-shaped box with one wire sprouting to the right. A projector $\ket{\psi} \bra{\psi}$ can then be depicted as follows:
\begin{equation}
    \ket{\psi} \bra{\psi} = \vcenter{\hbox{\includegraphics{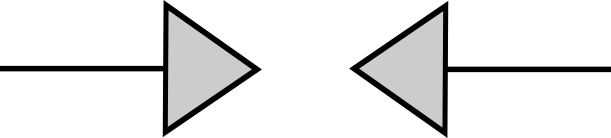}}}
\end{equation}

The tensor product of (arbitrary) tensors is depicted by drawing the two tensors one above the other:
\begin{equation}
    A \otimes B = \vcenter{\hbox{\includegraphics{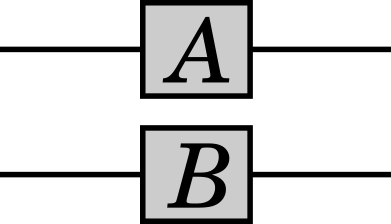}}}
\end{equation}
The contraction of tensors is depicted by connecting the corresponding wires:
\begin{equation}
    A \ket{v} = \vcenter{\hbox{\includegraphics{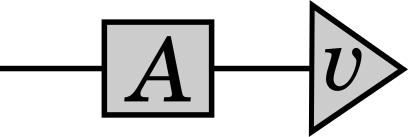}}}
\end{equation}
As long as the wires connect the correct inputs and outputs, the curves depicting the contractions of indices can be completely arbitrary. For example, taking the trace of a matrix requires connecting its input and output indices, which implies that we need to curve the wire into a loop (note that $\Tr A$ is a scalar, so the tensor network diagram has no loose wires):
\begin{equation}
    \Tr A = \vcenter{\hbox{\includegraphics{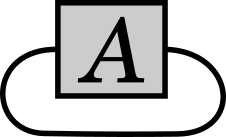}}}
\end{equation}
An identity operator $\id: V \rightarrow V$ is depicted simply as a line segment:
\begin{equation}
    \id = \vcenter{\hbox{\includegraphics{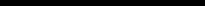}}}
\end{equation}
If we bend this line segment one way or another, we can get a $(2, 0)$ or a $(0, 2)$ tensor. This corresponds to a tensor proportional to the maximally entangled state in $V \otimes V$ or its dual:
\begin{align}
    \sum_i \ket{e_i} \ket{e_i} &= \vcenter{\hbox{\includegraphics{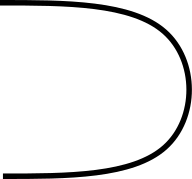}}} \\
    \sum_i \bra{e_i} \bra{e_i} &= \vcenter{\hbox{\includegraphics[angle=180]{figures/inkscape/bell.png}}}
\end{align}
Importantly, the when these diagrams occur on their own, they become basis dependent. Similarly, the transpose of a matrix is performed by turning both its wires in the opposite direction:
\begin{equation}
    A^T = \vcenter{\hbox{\includegraphics{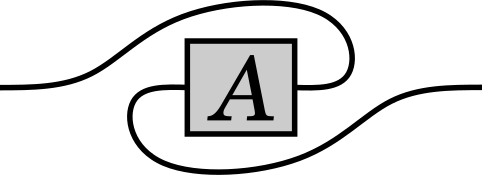}}}
\end{equation}










One must take care of the Hermitian conjugation operation (which is the combination of complex conjugation and transpose): we only denote it by explicitly marking tensors with the dagger sign ($\dagger$). 

It is also possible to consider linear functions over different vector spaces. For instance, when talking about mixed states, we can treat separately the internal degrees of freedom and the environment. In this case, one should just keep track of which wires belong to which space, either by subscripting them or by drawing them in different thickness, color, etc.

The convention on the direction of wires is completely arbitrary. The direction we choose here aligns with the mathematical notation of matrix multiplication, where the order they act on a vector is right to left. Quantum circuits are typically written in the opposite direction, so that the quantum gates are to be read left to right (the only exception we are aware of is \cite{kitaev_classical_2002}).

One of the most important things about tensor networks is that the diagrams can be arbitrarily deformed, and as long as the structure is preserved, the diagram will depict the same tensor. Formally tensor network diagrams are treated in Ref.~\cite{joyal_geometry_1991}. To prove the invariance of diagrams under deformations, Joyal and Street treat tensor network diagrams as generalizations of graphs with certain additional structure (some nodes are marked as outer nodes, incoming and outgoing edges for each node are ordered, etc.).

More generally, tensor diagrams can be used in different settings other than quantum theory. The key requirement is that the objects of interest and relations between them form a category in which the idea of tensor products makes sense. For example, Ref.~\cite{coecke_physics_2010} shows how tensor diagrams can be applied in topology, logic, and theory of computation.

An interesting case of tensor networks are tensor network states, i.e.~quantum states that are best described as a contraction of certain tensors. For example, quantum circuits can be treated as such: wires are wires, and boxes are quantum gates. Suppose each qubit in a tensor network state is depicted by a distinct wire. Knowing the structure of the network, we can provide an upper bound to the amount of entanglement that will be found across any bipartition of the qubits. To do that, we split the network in two halves, so that each half only contains the outgoing wires of the respective partition of the qubit registry. There is no unique way to do this, and the bound will depend on the partition of the tensors. In any case, the halves will be connected by a number of wires. For the special case of quantum circuits we can assume that these are $w$ wires, each having the dimension equal to two. Schematically this situation in shown in Fig.~\ref{fig:cut}. If we now group the qubit spaces in bipartitions into two spaces, the tensor will essentially be a matrix with rank not exceeding $2^w$. The reduced density matrix of one of the halves is obtained from two copies of the tensor by contraction over a half of the open wires. The resulting tensor --- the density matrix --- also has the rank of $\leq 2^w$. It is then easy to see that the von Neumann entropy of this matrix is bound by $w$ ebits. Thus, the number of ebits in a bipartition is bounded by the number of cuts you need to make in order to split the tensor network diagram according to the bipartition.

\begin{figure}
    \centering
    \includegraphics[width=0.5\linewidth]{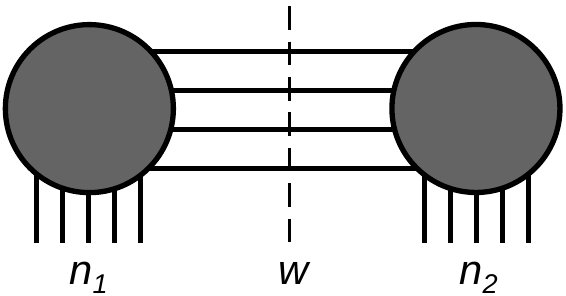}
    \caption{A quantum circuit can be treated as a tensor network state with all bonds having dimension 2. If a bipartition cuts $w$ wires, the total bond dimension is at most $2^w$, while the cut separates at most $w$ ebits of entanglement. Reprinted from \cite{uvarov_machine_2020}.}
    \label{fig:cut}
\end{figure}

\section{Complexity-theoretical aspects of quantum computation}


How do quantum computers compare to classical computers? What is so different about them? To answer these questions, we must make a detour to the field of computational complexity. 




% ...

\subsection{Classical complexity classes}

When we talk about problems that computers can or cannot solve, we mean the following. A \textit{problem} is a generic question of the kind ``Given $X$, find $Y$'', and an \textit{instance} of the problem is a specific value of $X$. When $Y$ is either yes or no, we call an instance a ``yes'' instance or a ``no'' instance, repsectively. Such yes or no problems are called \textit{decision problems}. 

Decision problems are grouped into complexity classes. Strictly speaking, the definitions of complexity classes rely on the notion of a Turing machine. However, for our purposes, we can equivalently define the most important complexity classes without having to introduce Turing machines. Complexity classes are usually defined by specifying how much time or space it takes to solve the problem as a function of the size of the input.

The complexity class $\mathbf{P}$ is defined as a class of problems that can be solved on a classical computer in polynomial time. This class includes many everyday problems solved by classical computers, such as sorting arrays of numbers, finding an element in an array, multiplying two matrices, etc. 

The class $\mathbf{NP}$ is defined as a class of decision problems where a solution to a ``yes'' instance of the problem can be efficiently verified with a classical computer. That is, if the answer to the question is ``yes'' and we are presented with some evidence of that (called \textit{witness}), we can quickly convince ourselves that the answer is indeed yes.

\begin{example}[\textsc{Travelling salesperson}, decision version]
    Consider an edge-weighted, connected, undirected graph $G=(V, E)$ and a positive real number $L$. Decide if there exists a path that visits every vertex of the graph, such that the sum of the edges taken in the path is less than or equal to $L$. This is a problem in $\mathbf{NP}$ because the ``yes'' answer can be verified by providing such a path.
\end{example}

\begin{example}[\textsc{Circuit-SAT}]
    Given a Boolean circuit, decide if (``yes'') there is such an input that the first bit of the output evaluates to `1', or (``no'') for any input, the first bit evaluates to `0'. This problem is also in $\mathbf{NP}$ because the ``yes'' answer can be quickly verified by providing the correct input.
\end{example}

These two classes are the subject of the most famous open question of complexity theory, namely the $\mathbf{P} \overset{?}{=} \mathbf{NP}$ question. The distinction of these classes is important for the following reason. The class $\mathbf{P}$ is often considered to be the class of decision problems that can be reasonably solved by a classical computer. However, many important problems actually lie in class $\mathbf{NP}$.

The $\mathbf{NP}$ class is large and diverse, but there is a subclass of problems that stand out as the most difficult in this class. By relative difficulty we mean their \textit{Carp reducibility} to one another. Let $a$ is an instance of problem $A$, and let $f$ be a procedure that maps instances of problem $A$ to instances of problem $B$, taking ``yes'' instances to ``yes'' instances and ``no'' instances to ``no'' instances. In addition, let $f$ take polynomial time in the size of $a$. In this situation, we say that $A$ \textit{can be reduced to} $B$. 

If we now find an efficient procedure for solving $B$, this means that we can also efficiently solve $A$ by first mapping instances of $A$ to instances of $B$ and then solving the latter. That is why if $A$ is reduced to $B$, then $B$ is considered a more difficult problem.

A problem in $\mathbf{NP}$ such that all problems in $\mathbf{NP}$ can be reduced to it, is called an $\mathbf{NP}$-complete problem. The primary example of an $\mathbf{NP}$-complete problem is \textsc{Circuit-SAT}: the verifier of any problem in $\mathbf{NP}$ is a polynomial-time algorithm for a deterministic Turing machine, which can also be mapped to a polynomial-sized Boolean circuit. The completeness of \textsc{Circuit-SAT} is known as the Cook---Levin theorem \cite{cook_complexity_1971,levin_universal_1973}.

Finally, we should mention the complexity class $\mathbf{BPP}$. This class is defined as the class of problems that can be solved in polynomial time by a classical computer that also has access to a polynomial amount of random bits. Since such programs rely on randomness, we also demand that the probability of solving the problem wrong is bounded by $1/3$ regardless of the size of the input (this probability is actually arbitrary as long as it is fixed and strictly less than $1/2$). In this case, running the same algorithm $N$ times reduces the chance of error to $3^{-N}$.

The only known relation between these classes is that $\mathbf{P} \subseteq \mathbf{NP}$. It is generally believed that $\mathbf{P} \neq \mathbf{NP}$ and that $\mathbf{P} = \mathbf{BPP}$.

\subsection{Quantum complexity classes}

The first inquiries in quantum computation used the language of quantum Turing machines. In this model of computation, the elements of the tape are essentially qudits, the internal state of the machine is a quantum state, and the function that defines the next step is a unitary transformation \cite{bernstein_quantum_1997}.


The quantum Turing machine model is equivalent to the model of quantum circuits, and thus we can define $\mathbf{BQP}$ as the class of problems that can be solved by polynomial-sized quantum circuits. The most famous example of a $\mathbf{BQP}$ problem is the problem of factoring integers into prime factors. The Shor's algorithm \cite{nielsen_quantum_2006} puts it into $\mathbf{BQP}$, while it is not known if it belongs to $\mathbf{P}$. It is not known how $\mathbf{BQP}$ compares to $\mathbf{P}$ or $\mathbf{NP}$. It is not difficult to show that $\mathbf{P} \in \mathbf{BQP}$, but the other direction is unknown. It is suspected that $\mathbf{BQP}$ is not included in $\mathbf{NP}$.

Just like $\mathbf{BQP}$ can be considered a quantum counterpart of $\mathbf{P}$ (or more precisely, $\mathbf{BPP}$), the class $\mathbf{NP}$ has a quantum counterpart, too. This class is called $\mathbf{QMA}$ and contains problems that can be efficiently verified with a quantum computer. The name of the class $\mathbf{QMA}$ is an acronym for ``Quantum Merlin--Arthur''. This name comes from a hypothetical scenario featuring Arthur, king of the Britons, and Merlin, a mighty wizard and advisor to Arthur. Merlin can solve any problem, but Arthur doesn't trust him too much and wants to be able to verify Merlin's suggestions. To convince Arthur that the ``yes'' instance is indeed a ``yes'' instance, Merlin can provide a polynomial-sized quantum state. More precisely, for ``yes'' instances, Arthur should be able to correctly verify it with probability above $2/3$, while for the ``no'' instance, Merlin should not be able to fool Arthur into believing that it's a ``yes'' instance with probability above $1/3$. Again, here the probabilities are with respect to Arthur's source of randomness, so by running his algorithm more than once, he can achive any desired level of confidence.

In fact, all complexity classes above (and some more) can be specified in terms of what Arthur can do and what Merlin can provide him as proof of his words. The resources of Arthur and Merlin are summarized in Table \ref{tab:complexity_classes}.

\begin{table}
    \centering
    \begin{tabularx}{\textwidth}{|lXl|}
        \hline
        \textbf{Class} & \textbf{Arthur has a...} & \textbf{Merlin can provide...} \\
        \hline
        $\mathbf{P}$ & classical computer & nothing \\
        $\mathbf{BPP}$ & classical computer and a source of randomness & nothing\\
        $\mathbf{NP}$ & classical computer & a polynomially-long string \\
        $\mathbf{MA}$ & classical computer and a source of randomness &  a polynomially-long string \\
        $\mathbf{BQP}$ & quantum computer & nothing \\
        $\mathbf{QMA}$ & quantum computer & a quantum state in $poly(n)$ qubits\\
        $\mathbf{QCMA}$ & quantum computer & a polynomially-long string\\
        \hline
   \end{tabularx}
   \caption{\label{tab:complexity_classes}Complexity classes described as an interaction between Arthur and Merlin.}
\end{table}

A canonical $\mathbf{QMA}$-complete problem is \textsc{Quantum circuit-SAT}: given a quantum circuit, decide if there is any input quantum state such that the first qubit of the output is measured in the state $\ket{1}$. This problem is $\mathbf{QMA}$-complete almost by definition: it is in $\mathbf{QMA}$ because the input state is a witness, and it is complete for $\mathbf{QMA}$ because for any problem in $\mathbf{QMA}$, you can construct a verifier circuit, and now if you ask if this verifier circuit ever returns a $\ket{1}$, you obtain an instance of \textsc{Quantum circuit-SAT}.

An ingenious construction due to Kitaev and coworkers \cite{kitaev_classical_2002} proves that the problem of finding a low-energy state of a local Hamiltonian is also $\mathbf{QMA}$-complete.

\begin{definition}[$k$-\textsc{Local Hamiltonian}]
    Given: $(H, a, b)$, where:
    \begin{enumerate}
        \item $H = \sum_i H_i$ is a Hamiltonian acting on a product of $n$ subsystems. The terms $H_i$ have operator norm bounded by a constant and support bounded by $k \in \mathbb{N}$. For each subsystem, the number of terms acting nontrivially on that subsystem is also bounded by a constant in $n$.
        \item $a, b \in \mathbb{R}$, such that $a < b$ and $b - a \geq O(1/ poly(n))$.
    \end{enumerate}
    Decide if there is a state $\ket{\psi}$ such that $\bra{\psi}H\ket{\psi} \leq a$ or if all states have energy above $\bra{\psi}H\ket{\psi} \geq b$, promised that either is the case.
\end{definition}


This problem was first proven to be $\mathbf{QMA}$-complete for $k \geq 5$ \cite{kitaev_classical_2002}. Then, it was gradually improved to lower values of $k$ \cite{kempe_3-local_2003}, until it was ultimately proven complete even for $k = 2$ \cite{kempe_complexity_2006}. Note that for $k = 1$ this problem is in $\mathbf{P}$: each term can just be minimized independently.

In the following chapters, we will mostly deal with instances of \textsc{Local Hamiltonian}. 
A plausible complexity-theoretic assumption is that $\mathbf{BQP} \neq \mathbf{QMA}$, which means we may never be able to solve this problem efficiently, even with scalable, fault-tolerant quantum computers. 
However, not all is lost. In classical computing, the fact that a problem is $\mathbf{NP}$-complete only implies that the worst instances are hard to solve, while many real instances turn out to have good approximate solutions. Arguably, the entire effort of variational quantum algorithms is devoted to finding good heuristic algorithms for $\mathbf{QMA}$-complete problems.

\ifnumequal{\value{contnumfig}}{1}{\counterwithout{figure}{chapter}
}{\counterwithin{figure}{chapter}}
\ifnumequal{\value{contnumtab}}{1}{\counterwithout{table}{chapter}
}{\counterwithin{table}{chapter}}
\chapter{Variational quantum algorithms} 
\label{chap:vqas}


As we mentioned earlier, present-day quantum computers are a long way from cracking encryption. Nonetheless, there are efforts to find use for quantum computers in the NISQ era. An important line of research is the research of variational quantum algorithms (VQAs). These algorithms rely on a feedback loop between the quantum device and the classical computer. Most often, the quantum device evaluates some kind of cost function (or its gradient), while the classical device performs optimization using the data from the quantum device. In this chapter we focus on this approach to quantum computation. This chapter is mostly devoted to variational quantum eigensolver, however, many ideas are applicable for any variational quantum algorithm.


\section{Variational principle}

Consider a Hermitian operator $H$. We do not know its spectrum, but we have some kind of guess about its ground state $\ket{\psi_{\text{guess}}}$. We know that the guess is not the true state, but we hope that some variation of this guess is close to the truth. So, we can introduce some family of the states $\ket{\psi(\theta)}$. Now, since any state is a linear combination of the ground state and excited states, the energy of that state is never below the ground state energy:

\begin{equation}
    \label{eq:variational_principle}
    E(\theta) = \frac{\bra{\psi(\theta)} H \ket{\psi(\theta)}}
         {\braket{\psi(\theta)}{\psi(\theta)}} \geq E_{gs}.
\end{equation}

This means that if we pick $\theta$ so that $E(\theta)$ is minimal, we can find a state that is the best approximation of the ground state in the sense of energy.


\subsection{Toy example: spin model}

Consider the following Hamiltonian describing two interacting spins in an external magnetic field:

\begin{equation}
    \label{eq:tfim_simple}
    H = J Z \otimes Z + h(X \otimes \id + \id \otimes X) \equiv J Z_1 Z_2 + h(X_1 + X_2).
\end{equation}

For $h = 0, J<0$, the ground state space of this system is obviously spanned by two product states: $\ket{00}$ and $\ket{11}$. 
For $J = 0, h>0$, the unique ground state is a product state $\ket{++} \equiv \frac{1}{2} (\ket{0} + \ket{1}) \otimes (\ket{0} + \ket{1})$. 
For generic values of $J$ and $h$, the ground state is not necessarily a product state. 
Intuitively, the reason is that the two terms in the Hamiltonian compete with each other, and neither of the product states is a good solution.
What if there are $n$ spins in the magnetic field? The Hamiltonian of this model (the \textit{transverse field Ising model}, TFI) is then as follows:

\begin{equation}
    \label{eq:tfim_toy}
    H = J \sum_{i=1}^{n-1} Z_i Z_{i+1} + h \sum_{i=1}^n X_i.
\end{equation}

The same reasoning still applies. In the limiting cases, there are simple product ground states, while generic case is more complicated. In fact, the case $J=h$ is not even amenable for perturbation theory, since both terms have comparable impacts.



Still, we can try to approximate the ground state using the variational principle. For now, we will consider the case $J < 0$. A simple guess is a translation-invariant state: $\ket{\psi(\theta, \varphi)} = (\cos(\theta) \ket{0} + e^{\rmi \varphi} \sin(\theta) \ket{1})^{\otimes n} \equiv \ket{\varphi}^{\otimes n}$. In this case, the calculations reduce to evaluation of $\langle \varphi | Z | \varphi \rangle$ and $\langle \varphi | X | \varphi \rangle$:

\begin{equation}
    \langle \varphi | Z | \varphi \rangle = \cos 2 \theta; \ 
    \langle \varphi | X | \varphi \rangle = \sin 2 \theta \cos \phi. 
\end{equation}
The total energy is then equal to 
\begin{equation}
    \label{eq:toy_tfim_energy}
    E = Jn \cos^2 2 \theta + hn \sin 2 \theta \cos \phi.
\end{equation}

By differentiating (\ref{eq:toy_tfim_energy}) and finding all stationary points, we find that the lowest energy obtainable is $-n \min{|J|, |h|}$. Thus, by optimizing over $\theta$ and $\phi$, we found the unentangled state which is closest to the ground state in terms of energy error.

The exact solution can be found by means of Jordan--Wigner transform \cite{lieb_two_1961,pfeuty_one-dimensional_1970}. The difference between exact solution and our variational solution is shown in Fig. \ref{fig:tfim_rank_one}.

\begin{figure}
    \centering
    \includegraphics[width=0.7\textwidth]{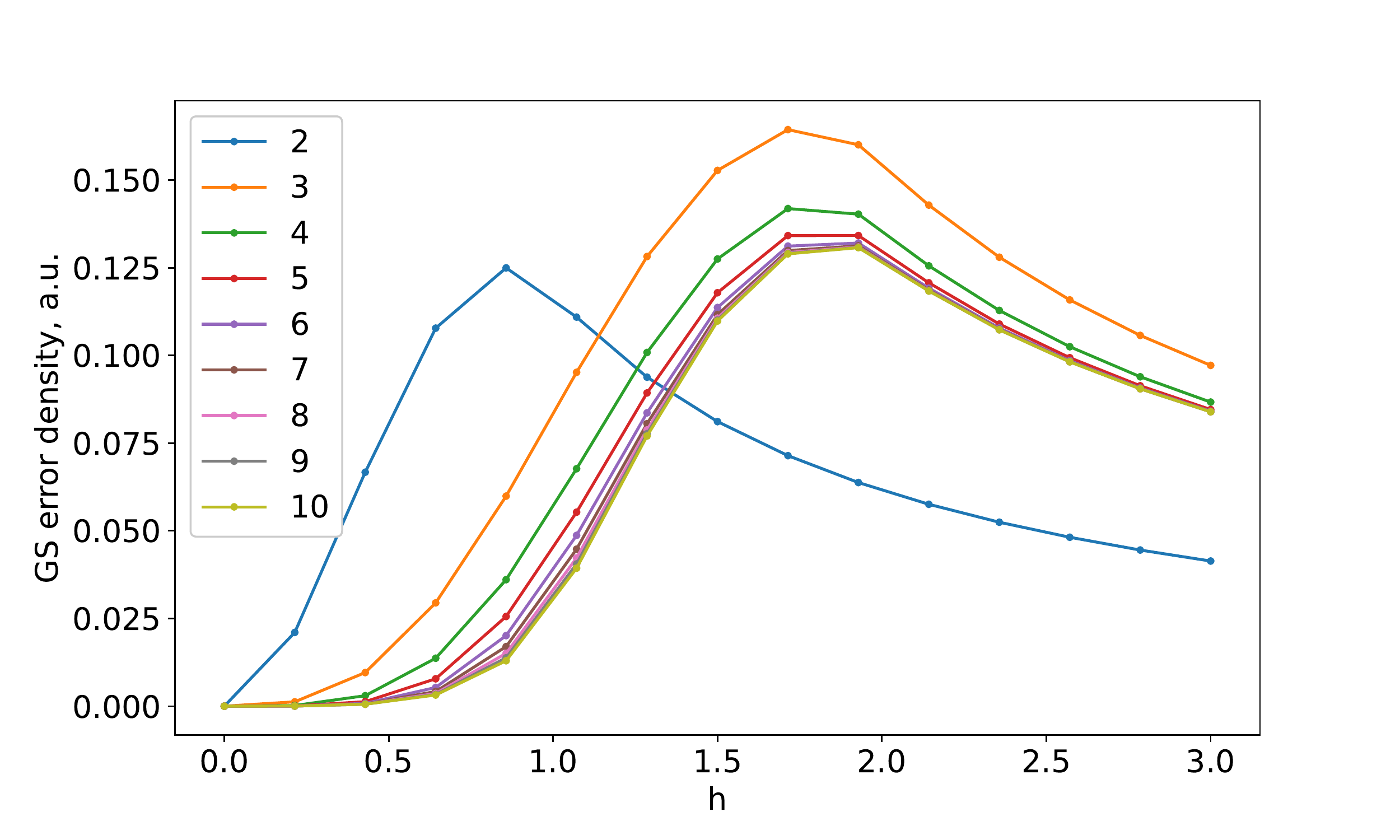}
    \caption{Error in ground state energy per qubit for ferromagnetic ($J = -1$) TFI model as a function of $h$. The approximate solutions are calculated using the unentangled ansatz state.}
    \label{fig:tfim_rank_one}
\end{figure}

\section{Variational quantum eigensolver}

\subsection{General description}

The variational principle inspired an algorithm for optimization of quantum Hamiltonians called variational quantum eigensolver (VQE). In short, we want to be able to estimate values like $\bra{\psi(\theta)} H \ket{\psi(\theta)}$ using a quantum computer. To do that, we need $\ket{\psi({\theta})}$ to be a state that can be prepared in the quantum processor (that is, a parametrized quantum circuit applied to some fixed reference state $\ket{\varphi}$), and $H$ needs to be an operator that acts in the space of qubits. 

Continuing the example of the spin model, we have a Hamiltonian consisting of $Z$ and $X$ operators. To do something about this model using a quantum computer, we identify each spin with a qubit (i.e.~for a spin chain with $n$ spins we need $n$ qubits). Since operators $Z$ and $X$ do not commute, they cannot be simultaneously measured. Hence, one has to measure them separately, in independent experiments. So, the subroutine for estimating $\bra{\psi(\theta)} H \ket{\psi(\theta)}$ is the following:

\begin{enumerate}
    \item Repeatedly prepare $\ket{\psi(\theta)}$ and measure each qubit in the standard basis. After enough measurements, estimate $\bra{\psi(\theta)} Z_i Z_{i+1} \ket{\psi(\theta)}$.
    \item Repeatedly prepare $\ket{\psi(\theta)}$ and measure each qubit in the $\{ \ket{+}, \ket{-} \}$ basis. After enough measurements, estimate $\bra{\psi(\theta)} X_i \ket{\psi(\theta)}$.
    \item Set $\bra{\psi(\theta)} H \ket{\psi(\theta)} := J \sum \bra{\psi(\theta)} Z_i Z_{i+1} \ket{\psi(\theta)} + h \sum \bra{\psi(\theta)} X_i \ket{\psi(\theta)}$.
\end{enumerate}

The trick here is that the qubits in the physical device do not have to have interactions prescribed by the model; all we need is to be able to prepare states and measure them in the standard basis (recall that measuring in the $\{ \ket{+}, \ket{-} \}$ basis is not much more difficult: we only need to apply a Hadamard gate to each qubit and then measure in the standard basis).

Now, assume that $\ket{\psi(\theta)}$ is a differentiable function of $\theta \in \mathbb{R}^k$ for some $k$, and $H \in \operatorname{Herm} (2^n)$. Then, the broad description of the VQE algorithm is the following loop:

\begin{enumerate}
    \item Estimate $E(\theta) = \bra{\psi(\theta)} H \ket{\psi(\theta)}$. 
    \item Use a classical optimization routine to find the next value of $\theta$. If the new value of $\theta$ is different from the old value by a threshold smaller than $\epsilon$, stop.
\end{enumerate}

A few remarks are in order:

\begin{itemize}
    \item In step 1, the estimate of $E(\theta)$ is sometimes replaced by the estimate of $\nabla E$, depending on the optimization routine in step 2. In the next subsection, we will give more detail the different choices of the optimization routine.
    \item The value $\bra{\psi(\theta)} H \ket{\psi(\theta)}$ is not easy to evaluate for any Hamiltonian. To be able to do that, we need to put additional restrictions on $H$. More precisely, we need $H$ to consist of a polynomial number of operators $h_i$ such that every expectation $\bra{\psi(\theta)} h_i \ket{\psi(\theta)}$ can be measured in polynomial time. An example of a Hamiltonian the conforms to this restriction is a Hamiltonian that consists of a polynomial number of Pauli strings. Alternatively, the Hamiltonian can consist of a polynomial number of projectors on product states.
\end{itemize}

\subsection{Optimization subroutine}

\paragraph{Stochastic graident descent.}
\textit{Gradient descent} is an algorithm of finding a local minimum of a differentiable function $f: \mathbb{R}^n \rightarrow \mathbb{R}$. The simplest variant of gradient descent goes as follows. Fix a \textit{learning rate} $\gamma > 0$ and pick an initial guess $x_0$. Then find the solution by repeating the following iteration:
\begin{equation}
    x_{i+1} = x_i - \gamma \nabla f(x_i).
\end{equation}
Stochastic gradient descent (SGD) appears when $\nabla f(x_i)$ is replaced by a random vector that estimates $\nabla f(x_i)$. In the context of supervised machine learning, stochastic gradient descent appears naturally: to calculate $\nabla f(x_i)$ exactly, one needs to iterate over all training samples, which is resource-intensive. Instead, the gradient is calculated over randomly chosen subsets of the training set called minibatches. The size of the minibatch seems to affect the quality of the solution. In paticular, there is evidence that with large batches SGD tends to find sharper minima that exhibit poorer generalization \cite{keskar_large-batch_2017}.

The learning rate $\gamma$ is also an important hyperparameter\footnote{A hyperparameter is a parameter of the model that is not updated during the training iterations, unlike e.g.~the weights of the neural network.} of the model. On the one hand, large learning rate means that the gradient descent may have trouble converging because large steps can overshoot the minimum. On the other hand, a small learning rate means more iterations. A simple change of the SGD that takes the best of both worlds consists in making $\gamma$ dependent on the iteration number, making it large in early steps and small in later steps. Of course, there are much more advanced variants of SGD with more complicated update rules, but reviewing them is outside the scope of this thesis. An interested reader can find the reviews of different SGD variants in Refs.~\cite{ruder_overview_2017,mehta_high-bias_2019}.

\paragraph{Gradient-free and gradient-based optimizers in VQAs.}

The optimization routines used for VQAs can be either gradient-based or gradient-free. Many experiments, both numerical and physical, used either class of methods for the optimization. Since there is no need to calculate gradient values in such mehtods, there is a reason to believe that they are more resilient to noise. Indeed, a naive implementation of gradient-based optimization would use finite differences to evaluate the gradient $\nabla E$:

\begin{equation}
    \label{eq:finite_difference}
    (\nabla E)_i = \frac{\partial E}{\partial \theta_i} \approx \frac{E(\theta_1, ..., \theta_i + \delta, ..., \theta_k) - E(\theta_1, ..., \theta_k)}{\delta}.
\end{equation}

However, this method is extremely unstable to the measurement errors. Indeed, if the variance of energy measurement is $\Delta$, then the variance of this difference is $\Delta^2 / \delta$. For this reason, many experimental implementations rely on gradient-free methods such as Nelder-Mead, Powell's method, or other techniques \cite{peruzzo_variational_2014,kokail_self-verifying_2019}. Another method used to mitigate the noisy measurements is to use the simultaneous perturbation stochastic algorithm (SPSA) \cite{spall_multivariate_1992} which evaluates the gradient by taking a finite-difference directional derivative in a random direction \cite{kandala_hardware-efficient_2017}. Despite even larger variance in the gradient, the appeal of this algorithm lies in the low cost of each individual iteration: for each step, the cost function needs to be evaluated just in two points. 


As good as the gradient-free methods are, gradient-based optimization methods have better upper bounds on convergence. Ref.~\cite{harrow_low-depth_2019} adapts these bounds to the quantum variational aglorithms and introduces a gradient acquisition tehcnique based on the Hadamard test. Let the $H = \sum h_l \sigma_l$ be the Pauli decomposition of the cost function, and the ansatz be a product of unitaries whose generators also admit Pauli decompositions:
\begin{equation}
    \ket{\psi} = U_p ... U_1 \ket{\psi_0} = e^{-\rmi \frac{A_p \theta_p}{2}} ... e^{-\rmi \frac{A_1 \theta_1}{2}} \ket{\psi_0}; \quad A_j = \sum_{k=1}^{m_j} \beta_k Q_k.
\end{equation}
Then the partial derivative can be expressed as follows:
\begin{equation}
    \frac{\partial E}{\partial \theta_j} 
    = \sum_k \sum_l \beta_k h_l \operatorname{Im}
    \bra{\psi_0} 
    U_{1}^\dagger ... U_j^\dagger
    Q_k
    U_{j+1}^\dagger ... U_p^\dagger \sigma_l U_p ... U_1 \ket{\psi_0}.
\end{equation}
The imaginary value in the right-hand side can be evaluated from the output of the following circuit:
\begin{equation*}
    \Qcircuit @C=1em @R=1em {
    \lstick{\ket{+}} 
    & \qw 
    & \ctrl{1}
    & \qw
    & \ctrl{1}
    & \measureD{Y}
    \\
    \lstick{\ket{0...0}} 
    & \gate{U_j ... U_1} 
    & \gate{Q_k}
    & \gate{U_p ... U_{j+1}} 
    & \gate{\sigma_l}
    & \qw
    \\
    }
\end{equation*}

The variance of this measurement no longer has the dependence on a small step $\delta$ and is much more robust to experimental errors. However, the downside of that method is the requirement to implement $2n$ more control gates than in the ansatz circuit and that every partial derivative requires $m_j \operatorname{Card} H$ measurements.

Gradient-based methods became more prominent after the introduction of another estimation technique called parameter-shift rule \cite{mitarai_quantum_2018,schuld_evaluating_2019}. This rule is not harder to implement than the original ansatz circuit itself, at the cost of having to make more measurements. If the dependence on $\theta_i$ is realized as a quantum gate of the form $\exp(-\rmi G \theta_i/2)$, where $G$ has spectrum of $(-1, 1)$, then the exact value of the derivative can be obtained as follows:
\begin{equation}
    \label{eq:parameter_shift}
    \frac{\partial E}{\partial \theta_i} = \frac{1}{2} (E(\theta_1, ..., \theta_i + \pi/2, ..., \theta_k) - E(\theta_1, ..., \theta_i - \pi/2, ..., \theta_k)).
\end{equation}
To see why this is the case, let us consider the dependence of the ansatz on one variable $\theta$ while hiding all other unitaries in suitably redefined $\ket{\phi}$ or in $H$:
\begin{equation}
    E = \bra{\phi} e^{\rmi G \theta / 2} H e^{-\rmi G \theta / 2} \ket{\phi}.
\end{equation}
The partial derivative is then equal to 
\begin{equation}
    \frac{\partial E}{\partial \theta} = \bra{\phi} e^{\rmi G \theta / 2} \rmi F H e^{-\rmi G \theta / 2} \ket{\phi} - \bra{\phi} e^{\rmi G \theta / 2} \rmi HF e^{-\rmi G \theta / 2} \ket{\phi}.
\end{equation}
If we write down the values $E(\theta \pm \pi / 2)$ at shifted parameters, using the fact that $e^{\rmi G \pi / 4} = \frac{1}{\sqrt 2} (1 + \rmi G)$, we will get the desired formula by simple algebra.
The parameter-shift rule has been successfully used in VQAs \cite{sweke_stochastic_2019,barison_efficient_2021}, and recently it was generalized to the gates with arbitrary spectrum \cite{kyriienko_generalized_2021}.

The value of the gradient obtained in the experiment is a random variable that estimates the true value of $\frac{\partial E}{\partial \theta_i}$. Therefore, the optimization process is a random process itself, which means that formally the optimization process is an instance of stochastic gradient descent.

A good property of the abovementioned estimator is that it is unbiased, i.e.~the expected value of the estimate $\mathbb{E}\overline{\frac{\partial E}{\partial \theta_i}}$ is equal to the true value. However, its variance depends on the number of measurements made for the estimator. There is a question of efficiency: on the one hand, more measurements per step mean better estimation of the gradient, on the other hand, fewer measurements per step mean that the same budget of calls to the quantum device enables more steps of the gradient descent. What is then an optimal number of measurements per optimization step? This question was studied in Ref.~\cite{sweke_stochastic_2019}. The experiments of Sweke et al.~have shown that this number does not have to be large: in some of the tests, the optimal number of measurements per expected value per step is \textit{one}.

\subsection{Choice of the ansatz}

The ansatz $\ket{\psi(\theta)}$ can be constructed in a number of ways. Ignoring the variants of VQE which construct the ansatz iteratively (more on them later), the popular approaches are the problem-inspired ans\"atze and hardware-tailored ans\"atze.

\subsubsection{Problem-inspired ans\"atze}

A simple ansatz of that sort is called the Hamiltonian variational ansatz \cite{wecker_progress_2015}, which uses the operators comprising the target Hamiltonian as the generators for the quantum gates in the ansatz. For the transverse-field Ising model, such an ansatz could look like this:

\begin{equation}
    \ket{\psi(\theta)} = e^{i\theta_{2n} Z_n Z_1} ... e^{i\theta_{n+1} Z_1 Z_2} e^{i\theta_n X_n} ... e^{i\theta_1 X_1}\ket{\psi_0}.
\end{equation}

To increase the ``power'' of the ansatz (i.e.~the amount of states that it can prepare), this pattern of gates can be repeated. If we denote $U_{ZZ} (\xi_1, ..., \xi_n) = e^{i\xi_n Z_n Z_1} ... e^{i\xi_1 Z_1 Z_2}$, and $U_X(\xi_1, ..., \xi_n) = e^{i\xi_n X_n} ... e^{i\xi_1 X_1}$, then an $L$-layered Hamiltonian variational ansatz is as follows:
\begin{multline}
    \ket{\psi(\theta)} = 
    U_{ZZ} (\theta_{(2L-1)n+1}, ..., \theta_{2Ln})
    U_X (\theta_{(2L-2)n+1}, ..., \theta_{(2L-1)n})... \\
    ...
    U_{ZZ} (\theta_{n+1}, ..., \theta_{2n})
    U_X (\theta_1, ..., \theta_n) \ket{\psi_0}.
\end{multline}


Another popular problem-inspired ansatz, called the unitary coupled cluster \cite{taube_new_2006}, is used for quantum chemistry problems \cite{wecker_progress_2015,barkoutsos_quantum_2018,omalley_scalable_2016,shen_quantum_2017,xia_coupled_2020,xu_test_2020}. The ansatz is formulated in terms of the electronic structure problem and later translated to the language of qubits and circuits via different fermion-to-qubit-transforms. The problem is that fermionic operators obey the anticommutation relations, while no such relations are in place for qubits (i.e.~spins of distinguishable particles). In Section \ref{sec:fermion-transforms} we give an overview of such transforms.

The ansatz consists in approximately implementing a unitary operator comprising elementary electronic excitations:
\begin{equation}
\label{eq:ucc}
    \ket{\psi(\boldsymbol{\theta})} = 
    e^{T(\boldsymbol{\theta}) - T^\dagger(\boldsymbol{\theta})}
    \ket{\psi_0},
\end{equation}
where the operator $T$ is a sum of operators $T_1, ..., T_k$ corresponding to 1 to $k$ electronic transitions:
\begin{align}
    T_1 &= \sum_{i,j} \theta_{ij} a^\dagger_i a_j \\
    T_2 &= \sum_{i, j, k, l} \theta_{ijkl} 
    a^\dagger_i a^\dagger_j a_k a_l \\
    & ... \\
    T_k &= \sum_{i_1, ...,i_{2k}} \theta_{i_1 ... i_{2k}} 
    a^\dagger_{i_1} a^\dagger_{i_k} a_{i_{k+1}} ... a_{2k}.
\end{align}
The variational parameters are the real coefficients $\theta_{1_i ... i_m}$. Most often, the series of operators is truncated at $k = 2$, in which case the ansatz is called UCCSD (unitary coupled cluster, single and double).

There are certain difficulties with implementing this method on NISQ devices. The main trouble is that naive implementation of this ansatz leads to circuits whose depth is outside the reach of current NISQ hardware. Nonetheless, the UCC is a promising and actively studied subject (see \cite{anand_quantum_2022} for a review).

\subsubsection{Hardware-inspired ans\"atze}

One of the first things that one observes about the UCCSD ansatz is that its generic form is quite long. NISQ devices, however, prefer much shorter circuits. The variational principle works regardless of which ansatz you use, so one can try circuits that are coming from the restrictions of the hardware. 

One of the ans\"atze of that kind is now known as the hardware-efficient ansatz (HEA)\cite{kandala_hardware-efficient_2017}. The idea is simple: apply single-qubit gates to every qubit (this is easy), and then apply a sequence of entangling two-qubit gates.  Repeat this for a few rounds. What we obtain this way is a parametrized quantum circuit whose entangling power is adjusted by adding more layers. The two-qubit gates don't even have to be continuously parametrized, although they of course can be \cite{campos_abrupt_2020}. Typically, these are CNOT or CZ gates applied in a circular fashion to qubits $(i, i+1) \mod n$ (for example, see \cite{mcclean_barren_2018}). Sometimes the entangling gates are applied in an all-to-all fashion \cite{skolik_layerwise_2020}.

Another ansatz that is used for VQE is the so-called alternating layered ansatz, also known as the checkerboard ansatz \cite{uvarov_machine_2020,bravo-prieto_scaling_2020,cerezo_cost-function-dependent_2020}. This ansatz assumes line or ring connectivity of the qubits. Then it fixes some kind of two-qubit ansatz block and applies it in a checkerboard fashion: first, apply gates to qubit pairs $(2i, 2i+1)$, then, to pairs $(2i -1 , 2i)$. The block can be anything as long as it generates entanglement. For such circuits, it is easy to estimate the amount of entanglement generated across any bipartition using the theory of tensor networks and SVD \cite{biamonte_lectures_2020}. Another nice property is that, under some assumptions about the individual blocks, one can attempt to make some assertions about the whole circuit (moslty concerning its abitily to produce a random unitary operator \cite{brandao_local_2016}). The ansatz of that kind will be analyzed in depth in Chapter \ref{chap:plateaus}.

\section{Example: studying the electronic structure of a molecule}

In this section, we outline the whole procedure of applying VQE to a molecular Hamiltonian. We start with the nonrelativistic\footnote{In the second quantized picture though, the relativistic corrections end up having the same qualitative structure \cite{veis_relativistic_2012}.} Schr\"odinger equation for the electrons and nuclei of the molecule:

\begin{multline}
    H(\mathbf{r}_1, ..., \mathbf{r}_{N_e}, \mathbf{R}_1, ..., \mathbf{R}_{N_n}) = 
    \sum_{i=1}^{N_n} \frac{\hbar^2 \nabla^2}{2 M_i}
    + \sum_{i=1}^{N_e} \frac{\hbar^2 \nabla^2}{2 m_e} +\\
    + \sum_{i \neq j} \frac{Q_i Q_j}{|\mathbf{R}_i - \mathbf{R}_j|}
    + \sum_{i \neq j} \frac{1}{|\mathbf{r}_i - \mathbf{r}_j|}
    - \sum_{i=1}^{N_n} \sum_{j=1}^{N_e} \frac{Q_i}{|\mathbf{R}_i - \mathbf{r}_j|}.
\end{multline}

Since the mass of an electron is about $1/2000$'th of the mass of a proton, it is reasonable to introduce the Born--Oppenheimer approximation: for caluclation of the electronic structure, we consider the positions of the nuclei fixed, and the nuclei are treated as classical particles with potential interaction governed by the eigenstates of the electronic structure\footnote{This approximation is not always valid (for instance, if there is substantial coupling between electronic and vibrational degrees of freedom). For this situation, more advanced techniques are required \cite{yalouz_state-averaged_2021}.}.

Now, let's fix the nuclear positions $\mathbf{R}_i$ and only consider the electronic degrees of freedom. What we need to introduce now is the set of basis states $\psi_j(\mathbf{r})$. One could use the eigenstates of the hydrogen atom, but modern quantum chemistry provides us with a whole range of advanced basis sets. The one found most frequently in the VQE literature is the minimal STO-3G basis set: every state is a sum of three functions proportional to $x^i y^j z^k \exp(-\alpha (x^2 + y^2 + z^2))$. The word ``minimal'' means that there are as many basis states as there are atomic orbitals. More complex basis sets include more states than atomic orbitals (two per orbital, three per orbital, etc.) Such basis sets are called double-zeta, triple-zeta, and so on. In any case, such basis sets contain an infinite number of states, so we must limit ourselves to a finite subset. This truncation mainly depends on the number of qubits at our disposal.

Electrons cannot be together in the same state, so valid states of electrons are spanned by Slater determinants constructed out of single-electron states $\phi_i$ (also taking care of the spin degree of freedom $s_i$):

\begin{equation}
    \label{eq:slater}
    \Psi(\mathbf{r}_1, ..., \mathbf{r}_{N_e}, s_1,..., s_{N_e}) = \operatorname{det}
    \begin{pmatrix}
        \phi_1(\mathbf{r}_1, s_1) & ... & \phi_1(\mathbf{r}_{N_e}, s_{N_e}) \\
        ... & ... & ... \\
        \phi_{N_e}(\mathbf{r}_1, s_1) & ... & \phi_{N_e}(\mathbf{r}_{N_e}, s_{N_e}) \\
    \end{pmatrix}.
\end{equation}

Via Slater determinants, the basis set of single-electron functions induces a basis in the space of multi-electron wave functions called the Fock basis. The states in this basis can be denoted as $\ket{n_{0, \uparrow}, n_{0, \downarrow}, ...,  n_{k, \uparrow}, n_{k, \downarrow}, ...}$, where $n_{k, \uparrow}, n_{k, \uparrow} \in \{0, 1\}$ mark if the basis state is occupied. In the following, we will drop the spin degree of freedom from the notation and will treat the orbitals with different spin state as two distinct states. 

Now we need to rewrite the Hamiltonian in this basis. This is done using creation and annihilation operators. An \textit{annihilation operator} $a_i$ is defined as follows:
\begin{gather}
    \label{eq:fermi_annihilation}
    a_i \ket{n_1, ..., n_{i-1}, 0_i, ...} = 0; \\ 
    a_i \ket{n_1, ..., n_{i-1}, 1_i, ...} = (-1)^{n_1 + ... + n_{i-1}} 
    \ket{n_1, ..., n_{i-1}, 0_i, ...}.
\end{gather}
Its Hermitian conjugate $a^\dagger_i$ is called a \textit{creation operator}. Since the interactions we consider cannot create or destroy particles, these operators always come in pairs $a^\dagger_i a_j$, which means taking an electron from the site $j$ to the site $i$. Due to the fact that electrons are fermions, these creation-annihilation operators obey the following anticommutation relations: $\{a_i, a_j \} = 0, \{a_i, a^\dagger_j \} = \delta_{ij}$. 

The matrix elements of $H$ can be rewritten in terms of such operators as follows. For any pair of basis states, we can uniquely write down the combination of creation and annihilation operators that maps one state to the other (up to the ordering of the operators and excluding pairs like $a^\dagger_i a_i$), while mapping the other basis states to zero. This means that we can express any projector $\ket{\alpha} \bra{\beta}$ in such a form. The matrix elements $\bra{\Psi_i} H \ket{\Psi_j}$ are then calculated by integrating the respective wavefunctions over the three-dimensional space.

After calculating the coefficients and expressing everything in the secons quantized form, we arrive at the following Hamiltonian:
\begin{equation}
    H = \sum_{pq} h_{pq} a^\dagger_p a_q + \sum_{ijkl} V_{ijkl} a^\dagger_i a^\dagger_j a_k a_l.
\end{equation}
Observe that there are no terms of degree higher than four. The reason for that is the two-body nature of the Coulomb interaction. The matrix elements between Slater determinants consists of integrals of the form
\begin{equation}
    \int \phi^*_{m_1}(x_1) ... \phi^*_{m_k} (x_k) \frac{1}{x_i - x_j}
    \phi_{n_1}(x_1) ... \phi_{n_k} (x_k) dx_1 ... dx_k. 
\end{equation}
In such an integral, all terms not involving $x_i$ and $x_j$ can be integrated separately, while the orthonormality of the basis states ensures that the resulting factor is either zero or one, with the latter being the case when $m_l = n_l$. So any possible difference can only happen in positions $i$ and $j$.

Only a few steps remain for this Hamiltonian to become a problem accessible to VQE. First, we assume that some of the low-lying states are always occupied, and that the states that are higher than a certain energy threshold are always empty. The space of states that we keep considering is called the \textit{active space}. The larger that space is, the better is the quality of the approximation. Finally, we need to identify electronic orbitals with qubits by using either the Jordan--Wigner or Bravyi--Kitaev transformation (Section \ref{sec:fermion-transforms}). Now the problem is completely ready for solution via VQE. We choose the ansatz, the optimization method, and run the algorithm.

\section{Fermion-to-qubit transformations}
\label{sec:fermion-transforms}

Fermionic operators are described in terms of creation and annihilation operators. Due to the permutation antisymmetry of fermionic wavefunctions, these operators obey peculiar anticommutation relations: $\{a_i, a^\dagger_j\} = \delta_{ij}, \{a_i, a_j\} = 0$. Qubits, on the other hand, are distinguishable particles. A mapping from fermions to qubits must respect these relations. Below are a few different methods to achieve that. In this section, we follow the original work \cite{bravyi_fermionic_2002} and the detailed exposition of \cite{seeley_bravyi-kitaev_2012}.

\subsection{Jordan--Wigner transform}

The simplest way to map fermions to qubits is to identify each qubit with a fermionic mode. To preserve the relations, an annihilation operator $a_i$ is mapped to an operator acting on qubit $i$, times a trail of operators counting the parity of the population of the preceding modes:
\begin{equation}
    a_i \mapsto f_i := Z_1 ... Z_{i-1} \otimes 
    \begin{pmatrix}
        0 & 1 \\
        0 & 0 \\
    \end{pmatrix} = Z_1 ... Z_{i-1} (X_i + \rmi Y_i) / 2.
\end{equation}
Comparing this with Eq. (\ref{eq:fermi_annihilation}), we can see that $f_i$ acts on the qubit registry exactly like $a_i$ acts on fermions.

\subsection{Bravyi--Kitaev transform}

The Bravyi--Kitaev transform \cite{bravyi_fermionic_2002} maps $m$ fermions to $m$ qubits by storing the parity information in a more elaborate way. The idea is to encode both the population information and parity information in a way that would require a logarithmic number of qubits to calculate.

Denote an $m$-site fermion state $\ket{n_0, ..., n_{m-1}}$ and the corresponding $m$-qubit state $\ket{x_0} \ket{x_2} ... \ket{x_{m-1}}$ 
To encode one state into the other means to provide a function that maps $(n_0, ..., n_{m-1})$ to $(x_0, ..., x_{m-1})$. Bravyi and Kitaev provide the following function:
\begin{equation}
    \label{eq:bk_formula}
    x_j = n_j + \sum_{s \prec j} n_s \ (\operatorname{mod} 2),
\end{equation}
where the partial order $\prec$ is introduced on as follows. Two numbers are written in their binary expansions $\alpha_{t-1} ... \alpha_0$, $\beta_{t-1} ... \beta_{0}$. We say that $\alpha_{t-1} ... \alpha_0 \preceq \beta_{t-1} ... \beta_{0}$ if there is a position $l_0$ such that (1) $\beta_l = 1 \ \forall l < l_0$ and (2) $\alpha_l = \beta_l \ \forall l \geq l_0$. That is, if the binary expansions of $a$ and $b$ match up to a certain position, and all lower positions of $b$ are equal to one, then $a \preceq b$.

For example, if we expand (\ref{eq:bk_formula}) as a matrix-vector equation over $\mathbb{F}_2$ for the example of $m = 8$, we will get the following (zeros are written as blank spaces for visual clarity):

\begin{equation}
    \label{eq:bk_matrix}
    \begin{pmatrix}
        x_0 \\ x_1 \\ x_2 \\ x_3 \\ x_4 \\
        x_5 \\ x_6 \\ x_7 \\
    \end{pmatrix} = 
    \begin{pmatrix}
    1 &   &   &   &   &   &   &   \\   
    1 & 1 &   &   &   &   &   &   \\
      &   & 1 &   &   &   &   &   \\
    1 & 1 & 1 & 1 &   &   &   &   \\
      &   &   &   & 1 &   &   &   \\
      &   &   &   & 1 & 1 &   &   \\
      &   &   &   &   &   & 1 &   \\
    1 & 1 & 1 & 1 & 1 & 1 & 1 & 1 \\
    \end{pmatrix} 
    \begin{pmatrix}
        n_0 \\ n_1 \\ n_2 \\ n_3 \\ n_4 \\
        n_5 \\ n_6 \\ n_7 \\
    \end{pmatrix} 
    .
\end{equation}
This is somewhat opaque, but observe the following. If we need to know the parity of $a_j$, i.e.~$P_j := n_0 + ... + n_{j-1}$, we need to sum up the $x_k$ such that $k+1$ is obtained by taking the expansion of $j$ in the powers of two and throwing away several lowest terms. For example, $P_7 = x_6 + x_5 + x_3$, and $P_6 = x_5 + x_3$. In any case, the parity information is stored in no more than $\log_2 m$ qubits\footnote{Also observe that the matrix in (\ref{eq:bk_matrix}) looks somewhat like the Sierpi\'nski triangle.}. For each $j$, we will denote as $P(j)$ the set of qubits that is required to find the parity.

Next, we need to find qubits which contain $n_j$. Again, by construction of the sets $S(j)$ there is a logarithmic number of such qubits. We will denote this set $U(j)$.

The application of $j$ depends on whether it is odd or even. For $j$ even, $a_j$ maps to the following operator:
\begin{equation}
    \label{eq:aj_bk_even}
    a_j \mapsto \left(\bigotimes_{i \in P(j)} Z_i\right) \otimes \left(\bigotimes_{k \in U(j)} X_k \right) \otimes 
    \begin{pmatrix}
    0 & 1 \\
    0 & 0 \\    
    \end{pmatrix}_j.
\end{equation}
For $j$ odd, the calculation is somewhat more involved. The reason for this is that, for $j$ even, the mapping is very simple: $x_j = n_j$. For $j$ odd, it is somewhat more laborious to account for the population $n_j$ correctly. Denote $F(j)$ the set of qubits that have the same parity as the orbital $j$. Then define an operator $\Pi_j$ as
\begin{equation}
    \Pi^-_j = \frac12 \left(X_j \otimes \bigotimes_{i \in F(j)} Z_i + \rmi Y_j\right).
\end{equation}
Finally, for $j$ odd, the annihilation operator maps to:
\begin{equation}
    \label{eq:aj_bk_odd}
    a_j \mapsto \left(\bigotimes_{i \in P(j) \backslash F(j)} Z_i\right) \otimes \left(\bigotimes_{k \in U(j)} X_k \right) \otimes 
    \Pi^-_j.
\end{equation}
The creation operator $a^\dagger_j$ can be obtained from Eqs.\,(\ref{eq:aj_bk_even}, \ref{eq:aj_bk_odd}) by taking the usual Hermitian conjugate. The sizes of $P(j)$ and $U(j)$ are logarithmic in $m$, so the operators will have logarithmic locality.

\section{Potential applications of VQE}

\subsection{Quantum chemistry}
The most prominent application of VQE so far is quantum chemistry. Indeed, many proof-of-concept experiments \cite{peruzzo_variational_2014,omalley_scalable_2016,kandala_hardware-efficient_2017,hempel_quantum_2018,shen_quantum_2017}, as well as numerical simulations \cite{parrish_quantum_2019,romero_strategies_2017}, consider small molecules as target Hamiltonians. In the Born--Oppenheimer approximation, VQE returns the energy of the electronic ground state for the specific locations of the nuclei. Solving the same problem for different positions provides an energy landscape for the interaction of the nuclei. This way, one could potentially learn the pathways taken by molecules during chemical reactions. To do that, one needs to resolve the energy level up to \textit{chemical accuracy}, which is typically considered to be 1 kcal/mole. Since the chemical reactions often involve going over potential barriers, the rates of reactions are exponentially sensitive to the height of said barriers and hence to the error of the computation. To do well, we need to include many energy levels, whcih may lead to unfavorable scaling of simulation time just because of the sheer number of parameters to be optimized in a generic ansatz \cite{elfving_how_2020}. Nonetheless, advanced techniques in designing the VQE experiment --- discussed in the next section --- may prove useful in dealing with this problem.

\subsection{Condensed matter physics and lattice QFT}

One can apply VQE to lattice electronic problems (e.g.~Hubbard model) in a rather straightforward way \cite{cade_strategies_2019,uvarov_variational_2020}. The approach is the same as for the chemical Hamiltonians, except that instead of orbitals, the basis single-electron states are the lattice sites. Then one can proceed with Bravyi--Kitaev or Jordan--Wigner transformation.

In an experiment in \cite{kokail_self-verifying_2019}, the authors apply VQE to a model derived from lattice quantum field theory. Each site of a lattice can contain either an electron, or a positron, or both. For that reason, every lattice site is in fact simulated by two lattice sites: even-numbered sites are populated with electrons, while odd-numbered sites are populated with positrons. Other than that, the mapping to qubits is mostly the usual Jordan--Wigner transform.

There is also an alternative approach to variational investigation of quantum systems which involves variational calculation of the Green's function of the system~\cite{endo_calculation_2020}.

\subsection{Classical optimization problems}

Classical optimization problems can be embedded into quantum Hamiltonians to be then solved via VQE. For example, it is easy to embed the problem called \textsc{MaxCut}. Since this problem is \textbf{NP}-complete, this already enables the solution of many other classical problems.

The problem is formulated as follows. Let $G = (V, E)$ be an undirected graph. A \textit{cut} is a partition of the vertex set $V$ into two subsets. The \textit{size of the cut} is defined to be the number of edges connecting vertices from different subsets. In other words, if we group the vertices in two piles, the size of the cut is the number of edges connecting the two piles. The task is to find the maximum cut. 

This task can be expressed as a Hamiltonian minimization problem as follows. identify a qubit with each vertex, with $\ket{0}$ and $\ket{1}$ being the labels of this vertex belonging to piles one and two, respectively. For every edge $(i, j) \in E$, we would like to introduce a penalty if the edge is not cut. This can be done by a term $Z_i Z_j$. Summing up over the edges, we end up with an antiferromagnetic Ising model defined on the graph $G$:

\begin{equation}
    \label{eq:maxcut_ising}
    H = \sum_{(i, j) \in E} Z_i Z_j.
\end{equation}

This embedding of classical problems into VQE is related to the idea of continuous relaxation of problems [McClean low-depth 2021]. In particular, the analogy is quite vivid if the ansatz is a very simplistic one that doesn't introduce entanglement \cite{bittel_training_2021}. 

Continuous relaxation means the following: instead of $\{-1, 1\}$, the vertices are assigned points on a high-dimensional sphere. The solution to the continuous problem (in case of MaxCut) can be found by means of semidefinite programming. This solution is then mapped back to the discrete solution by simply cutting the sphere into random halves. For MaxCut, this approach is known as the Goemans--Williamson algorithm. Under certain plausible assumptions, it is $\mathbf{NP}$-hard to solve MaxCut with a better approximation ratio than with this algorithm.

At current state it seems that the outlook of VQE for classical problems is not very promising: entanglement-free ans\"atze can be simulated without a quantum computer, and simple entangled ans\"atze, at least in a numerical experiment in Ref.~\cite{nannicini_performance_2019}, do not seem to yield any substantial improvement in optimization results. The quantum approximate optimization algorithm (QAOA, see Section~\ref{sec:vqa_related}), which is a variant of VQE with a peculiar ansatz, is also actively investigated as a tool to solve classical problems.

\section{Variants of VQE}
\label{sec:vqe_variants}

The basic variant of VQE is quite generic. It of course invites many different modifications. Here we will survey some of the prominent proposals.

\subsection{Variants that modify the cost function}

\paragraph{Adiabatically-assisted VQE.} Since the optimization version of the \textsc{local Hamiltonian} problem is $\mathbf{QMA}$-hard, one could expect that the optimization landscape is very nonconvex, and convergence to suboptimal minima is a very possible thing. To counter this, Ref.~\cite{garcia-saez_addressing_2018} borrows an idea from adiabatic quantum computing. Let $H$ be the problem Hamiltonian and $H_{\text{init}}$ some Hamiltonian whose ground state is easy to prepare. Denote $H_s = (1-s)H_{\text{init}} + sH$. Find the ground state of $H_0$, and then run VQE for each $H_{s + \Delta s}$, using the solution of $H_{s}$ as a starting point. This may look very inefficient, since now you have to solve many VQE problems instead of one, but these intermediate problems may converge much faster because the starting point is already close to an exact solution. In Chapter \ref{chap:vqe_numerics}, we adopt this variant to solve the ground state problem for the transverse-field Ising model at different field strengths.

\paragraph{Meta-VQE.} Recall that when we discussed quantum chemistry, we introduced the Born-Oppenheimer approximation that essentially decouples the electronic motion and the nuclear motion. In order to predict chemical reactions and coupling energies, one needs to find e.g.~optimal distances between atoms. This implies that one has to solve many similar VQE problems for different values nuclei positions $\{\mathbf{R}_i\}$. To do this, Meta-VQE lets some parameters of the ansatz be functions of the nuclei positions. This way, when the ansatz is optimized, it gives some approximation to all problems simultaneously. This approximation is far from chemical precision, but it provides a good inital point for further VQE solution of individual problems, which quite important in the view of the barren plateaus phenomenon discussed in Chapter~\ref{chap:plateaus}.

\paragraph{Divide-and-conquer.} Ref.~\cite{fujii_deep_2020} proposes a variant of VQE --- dubbed Deep VQE --- that is reminiscent of the famous DMRG method. In Deep VQE, the quantum system is split into several subsystems. The Hamiltonian is then naturally separated into local terms and subsystem-subsystem interaction terms. The algorithm runs VQE for each subsystem separately, finding the ground states and a number of low-energy exctitations. Then the interaction terms are projected to the low-energy subspaces of the subsystems. Finally, this new problem is mapped to qubits and solved again via VQE. As a result, the solution is found using a smaller number of qubits, possibly at a cost of increased energy error.

\subsection{Variants that dynamically update the ansatz structure}

\paragraph{Adapt-VQE.} Optimizing large circuits requires many estimations of derivatives, that is why it is a good idea to start with a small ansatz and increase it along the way. This is the idea of Adapt-VQE \cite{grimsley_adaptive_2019}: first, one fixes some pool of Hermitian operators $F_i$ that one wishes to use in the ansatz. For chemical problems, such operators are one- and two-particle excitations. Then, for each operator in the pool, one can estimate the derivative that one would get if one appended a parametrized gate $\exp{\mathrm{i} \theta F_i}$ to the end of the circuit. The gate with the largest derivative (by magnitude) is then chosen and appended to the ansatz. Optimize over the ansatz parameters and repeat. For small molecules, this seems to provide better results than the UCCSD ansatz while using fewer gates. 

Some modifications of this algorithm were proposed in the literature. One is qubit-adapt-VQE \cite{tang_qubit-adapt-vqe_2021}, which replaces fermionic excitations by two-qubit Pauli operators, and batched adapt-VQE \cite{sapova_variational_2021}, which adds more than one generator per iteration.

\paragraph{Ansatz pruning.} This idea was independently explored in \cite{bilkis_semi-agnostic_2021} and \cite{sim_adaptive_2021}. After constructing the ansatz, one can find that some gates are close to the identity. Such gates can often be eliminated without affecting the quality of the solution too much. This procedure can be combined with iterative ansatz growing techniques like ADAPT-VQE.

\paragraph{Genetic evolution of ans\"atze.} Ref.~\cite{chivilikhin_mog-vqe_2020} proposes updating the ansatz with a genetic algorithm. Genetic algorithms are inspired by evolution of species in nature: a pool of candidate solutions is sorted according to the fitness score (energy obtained by minimization plus a penalty depending on the depth of the circuit). The next generation of candidate solutions is produced by combining the successful solutions (this is done by a crossover of two circuits: two lists of gates are split at a random point, then the head of one is concatenated with the tail of the other and vice versa) and introducing mutations in the form of randomly inserted/removed gates. The proposed algorithm optimizes both for the energy of the solution and for the number of gates in the ansatz circuit.

\subsection{Other}

\paragraph{Quantum subspace expansion.} Suppose after optimizing over the problem Hamiltonian $H$, the algorithm came up with a state $\ket{\psi}$. Then quantum the subspace expansion (QSE) algorithm \cite{mcclean_hybrid_2017,colless_computation_2018} gives a prescription on how to evaluate not only the energy $\bra{\psi} H \ket{\psi}$, but also values of the kind $\bra{\psi_i} H \ket{\psi_j}$ for some closely-related states $\ket{\psi_i}$. The idea is that, if $\sigma_i$ and $\sigma_j$ are some Pauli strings, then the Pauli decomposition of $\sigma_i H \sigma_j$ has the same cardinality as the Pauli decomposition of $H$. This means that the value can $\bra{\psi} \sigma_i H \sigma_j \ket{\psi}$ can be estimated by a standard procedure. Physically the operators $\sigma_i$ are meant to mimic some few-particle excitation operators, which means that the states $\ket{\psi_i}$ should not be too far in energy from the state $\ket{\psi}$. In any case, after this procedure we can build the matrix elements of $H$ projected onto the subspace spanned by $\ket{\psi_i}$ and diagonalize it classically. This may improve the energy estimate of the ground state and also provide estimations of the energies of low-lying excited states.

A related proposal is described in Ref.~\cite{bharti_iterative_2020}. This algorithm also approximates the ground state of a target Hamiltonian, but does so by taking inner products of candidate states. Essentially, it runs the quantum subspace expansion for a fixed input state, then takes second-order perturbations by taking products of two Pauli strings, and so on until the desired precision is reached.

\paragraph{Qubit coupled cluster.} In VQAs, there is heavy use for alternating between the Schr\"odinger picture of quantum mechanics and the Heisenberg picture. In the latter, the quantum state is fixed, and quantum evolution acts by conjugation on the observables. The approach suggested in Refs. \cite{ryabinkin_iterative_2020,ryabinkin_qubit_2018} exploits the Heisenberg picture by applying some generators of the ansatz directly to the target Hamiltonian instead of the quantum state. At the cost of increased Hamiltonian cardinality, this method may potentially simplify the ansatz state, making it better suited for running on NISQ devices.

\section{Related algorithms}
\label{sec:vqa_related}

\subsection{Quantum neural network training}

An interesting direction of research lies at the intersection of variational quantum algorithms and machine learning. A parametrized quatnum circuit can be treated as a device that accepts a quantum state and returns a probability distribution of measurements. One can use this approach to solve a classification problem for quantum states. This is covered in more detail in Chapter \ref{chap:qml}.

\subsection{QAOA}

Quantum-assisted optimization algorithm (QAOA) \cite{farhi_quantum_2014} is an algorithm that can be considered a variant of VQE with a fixed ansatz structure. Let $H_0$ be the Hamiltonian of an easy problem, and let $H$ be the problem of interest. Then the QAOA ansatz is the alternating sequence of evolution operators\footnote{Later on, this circuit structure became known as ``Quantum alternating operator ansatz'', the acronym being the same as before.} generated by $H_0$ and $H$, acting on the ground state of $H_0$:

\begin{equation}
    \ket{\psi(\gamma_1, ..., \gamma_p, \beta_1, ..., \beta_p)}
    = e^{-\rmi \beta_p H_0} e^{-\rmi \gamma_p H} ... e^{-\rmi \beta_1 H_0} e^{-\rmi \gamma_1 H} \ket{\psi_0}.
\end{equation}

In practice, the Hamiltonian $H_0$ is taken to be the sum of single-qubit $X$ operators: $H_0 = \sum X_i$. Its ground state is the product of plus states $\ket{+}^{\otimes n}$. The initial proposals and investigations considered QAOA for classical optimization problems like MaxCut \cite{farhi_quantum_2017}, but later on it was also used for optimal preparation of quantum gates \cite{kiani_learning_2020} and Hamiltonian minimization \cite{ho_efficient_2019}.

\subsection{Optimized preparation of quantum gates}

While VQE finds an optimal state, one can also adopt the variational paradigm to the task of finding a circuit preparing a desired gate. This algorithm is called variational quantum gate optimization (VQGO) \cite{heya_variational_2018}. The variational circuit $U(\boldsymbol{\theta})$ is optimized to maximize the fidelity between the experimental output $U(\boldsymbol{\theta}) \rho U^\dagger (\boldsymbol{\theta})$ and the (simulated) output $U_\text{ideal} \rho U^\dagger_\text{ideal}$, averaged over random input states $\rho$.

An alternative algorithm for the similar task is known as quantum-assisted quantum compiling \cite{khatri_quantum-assisted_2019}. This algorithm uses two quantum registers. First, all qubits of these registers are coupled in Bell pairs ($i$'th qubit of register A to $i$'th qubit of register B). Then, a fixed gate $U$ is prepared in the first register, and a variational circuit $V^*$ is prepared in the second register. By measuring in the Bell basis, one can evaluate $|\Tr UV^\dagger|^2$ and optimize it with a classical computer. As a result, the conjugate $V$ will be as close to $U$ as possble. Naturally, there is a question: if we can already prepare a gate $U$, why bother with trying to make its copy with a variational ansatz? There are several propositions. First, one can try to compress quantum circuits this way. Second, if $U$ is implemented by an unknown process, one can recover a circuit describing this process. Finally, if the qubits in different registers have different quality, this procedure can be used for comparison or for noise correction.

\subsection{Variational state preparation}

Variational approach can be used to optimize the preparation of a state described by a known quantum circuit. A way to do that is described in \cite{biamonte_universal_2021}. The idea is to construct a cost function that is minimized by the target state and run VQE against that cost function. One such cost function is a the so-called telescope construction. It is simple, but its downside is that its cardinality grows exponentially with the number of non-Clifford gates. Another construction is already known as the Kitaev clock construction \cite{kitaev_classical_2002}. 






\section{VQE vs. phase estimation algorithm}

Variational quantum eigensolver is not the first quantum algorithm proposed for solution of hard optimization problems. An earlier proposal is known as the phase estimation algorithm (PEA). This algorithm is based on the quantum Fourier transform and the ability to make a controlled execution of $U = \exp(\mathrm{i} \tau H)$.

PEA works as follows. Prepare an approximate ground state $\ket{\psi}$ in one register, and the plus state $\ket{+}^{\otimes m}$ in an $m$-qubit ancilla register. 

Then apply $U$ controlled on the first qubit, $U^2$ controlled onthe second qubit, $U^4$ controlled on the third qubit, and so on, doubling the power for each qubit. The register with $\psi$ will remain unentangled with the control register. If $\ket{\psi}$ is an eigenstate of $U$ with $U\ket{\psi}=e^{2\pi \rmi \phi}$, the control register will now have the Fourier transform of $\phi$. Inverse FT yields the answer.

If $\ket{\psi}$ is not an eigenstate of $H$, then we'll have a superposition of states. Let $U = e^{i \tau H}$, and $|\psi \rangle = \sum c_i | \lambda_i \rangle$. Then after PEA we will obtain a state $| \phi \rangle = \sum c_i (e^{i \tau \lambda_i} |\lambda_i \rangle \otimes |\text{QFT}_i \rangle)$. Then with probability $|c_i|^2$, we will obtain the necessary eigenvalue.

Solution of the ground state problem using PEA has been experimentally demonstrated a number of times \cite{whitfield_ground_2012,lanyon_towards_2010,bauer_hybrid_2016,omalley_scalable_2016}. Compared to VQE, this algorithm has an advantage that the number of runs required to find the ground state energy to precision $\epsilon$ is estimated as $O(\epsilon^{-1})$, whereas for VQE the evaluation of energy \textit{in each iteration} is about $O(\epsilon^{-2})$ (however, one doesn't necessarily have to obtain such precision in each iteration: the paper \cite{sweke_stochastic_2019} argues evaluation of derivatives in small experimental samples is also a viable strategy; nonetheless, the final answer has to be evaluated to good precision). 

The downsides are the complexity of the required quantum circuit and the requirement that the initial state has significant overlap with the true ground state. The latter is sometimes referred to as ``orthogonality catastrophe'' \cite{kohn_nobel_1999}. While complexity-theoretical intuition suggests that this catastrophe cannot be avoided completely --- otherwise we would be able to solve $\mathbf{QMA}$-complete problems in quantum polynomial time --- some remedies for practical problems are investigated in the literature \cite{tubman_postponing_2018}.

Even with the better asymptotic scaling, PEA may still be quite computationally expensive for practical problems. Ref.~\cite{reiher_elucidating_2017} estimates the resource requirement for calculating the chemical reactions involving the so-called iron-molybdenum cofactor (FeMoco). This molecule is a critical element of the enzyme that takes part in biological nitrogen fixation, i.e.~the process of transforming atmospheric nitrogen into ammonia. Reiher et al.~conclude that the resources required for simulating FeMoco are comparable to those required for factoring 4096-bit integers using the Shor's algorithm. Will VQE fare even worse? Who knows.

The comparison of these algorithms makes some authors suggest that VQE is an algorithm for noisy devices with low available depth, and in the era of fault-tolerant, large-scale QCs it may be gradually replaced by PEA \cite{elfving_how_2020}. In the meantime, there is a proposal to combine the two approaches by using both kinds of measurements in a varying proportion \cite{wang_accelerated_2019}.
\chapter{VQE for physical models}
\label{chap:vqe_numerics}

In this chapter we start investigating the algorithm called variational quantum eigensolver (VQE), which is arguably the flagship algorithm of the VQA family. We will give an overview of the known physical an numerical experiments, as well as present our own results in investigating the properties of VQE solutions.

\section{Ansatz dependence and depth scaling}

In this section, we will summarize what is known about the performance of VQE on a number of model problems. Typically, these problems are lattice problems because it is easy to study scaling with the number of qubits $n$ by just increasing the lattice. For molecular Hamiltonians, the only easy way to scale the problem up is to increase the size of the active space. Even changing the basis of orbitals leads to a different family of problems, strictly speaking.

A number of papers studied how the error of VQE behaves as a function of ansatz depth for different models. Quite often, the error was found to decrease exponentially with the depth of the ansatz. This was observed for the Hubbard model by Cade et al.~\cite{cade_strategies_2019} and also observed by us for the Hubbard model with the next-nearest neighbor Coulomb interactions \cite{uvarov_variational_2020}. Later, the same behavior was also observed for the Heisenberg model on the kagome lattice \cite{kattemolle_variational_2021,bosse_probing_2021}. 

On the other hand, some problems exhibit a threshold in the optimization error with increasing depth. For example, Ref.~\cite{bravo-prieto_scaling_2020} studies transverse field Ising and Heisenberg models for a 1D chain with open boundary conditions and finds two regimes in the behavior. Before the threshold depth, the error decays polynomially in depth, whereas after the threshold the convergence is exponential. The location of the threshold itself also moves with the number of qubits, seemingly linear with $n$. This threshold appears not only in the hardware efficient ansatz, like in Ref.~\cite{bravo-prieto_scaling_2020}, but also in the Hamiltonian variational ansatz as well \cite{wiersema_exploring_2020}. A similar threshold behavior was also observed in optimizing the ansatz to get to a target unitary \cite{kiani_learning_2020,campos_abrupt_2020}.

Recent findings connect this kind of threshold behavior to the dimensionality of the Lie algebra associated with the ansatz circuit \cite{larocca_theory_2021}. When the effective dimension\footnote{The ansatz can be treated as a differentiable function from $\mathbb
{R}^k \rightarrow \mathbb{C}^{2^n}$. One way to define effective dimension is to take the maximum rank of its Jacobian. Alternatively, the dimension is defined in terms of the Fisher information matrix, however that is essentially the same.} of the ansatz exceeds said dimension, the ansatz is said to be overparametrized. Apparently, in the overparametrized regime, the optimization landscape loses all the suboptimial minima, although the good behavior of the landscape may very well happen before the overparametrization. 

\section{VQE for spin models}

\subsection{Transverse field Ising model}

The first simple example of an interesting system for VQE testing is the transverse field Ising (TFI) model with periodical boundary conditions:

\begin{equation}\label{eq:tfim}
    H_\mathrm{TFIM}=J\sum\limits_{i=1}^n Z_i Z_{i+1} + h\sum\limits_{i=1}^n X_i, \; J>0, \; h>0.
\end{equation}

This model has trivial product state solutions for $h = 0$ and for $J = 0$, but for other values of the parameters the solution is an entangled state. The model is still integrable: the Jordan--Wigner transformation maps it to a system of free fermions \cite{lieb_two_1961}. This fact makes it possible to use the TFI model to test VQE and its variants even outside numerical simulations. 

\begin{figure}
    \centering
    \includegraphics[width=0.7\textwidth]{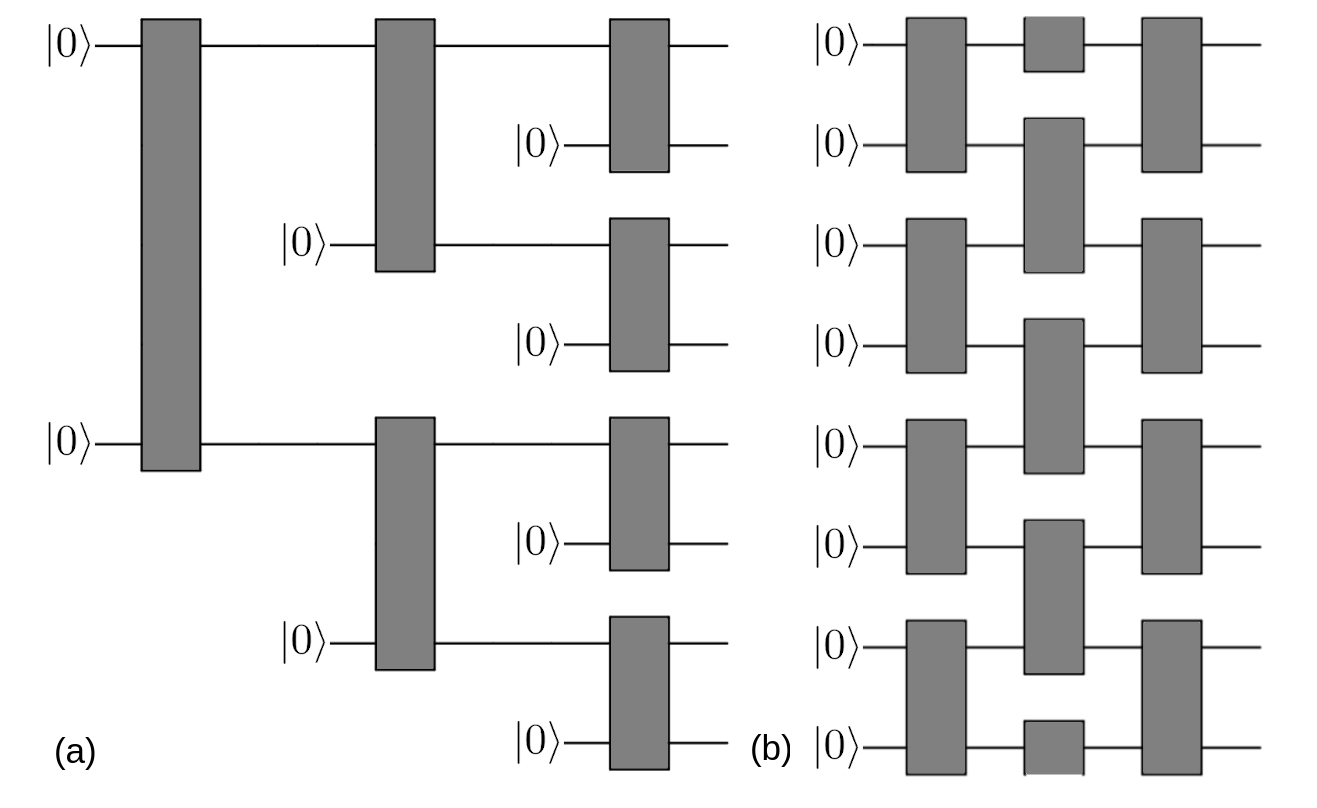}
    \caption{(a) Tree tensor network state. (b) Checkerboard tensor
    network state. In both cases, the quantum register is instantiated in
    the $\ket{0}^{\otimes n}$ state and subject to entangling gates. Black boxes indicate
    two-qubit gates specified by Fig. \ref{fig:entangler}. Each block is
    parametrized independently. Reprinted from \cite{uvarov_machine_2020}}.
    \label{fig:tree_and_checkerboard}
\end{figure}

We studied this model numerically for 10 qubits. We tested VQE solutions for a variety of different ansatz circuits:

\begin{enumerate}
    \item The \textit{rank-1 ansatz} simply consists of $R_Y$ and $R_Z$ gates applied to each qubit. The name stems from the fact that the state does not generate any entanglement: whatever qubit partition one selects, the Schmidt rank of the bipartition will always be equal to one.
    \item The \textit{tree tensor network ansatz} is depicted in Figure \ref{fig:tree_and_checkerboard}a. This ansatz provides entanglement between even the farthest of qubits, but the entanglement across any bipartition does not exceed $O(1)$ ebits: any contiguous region can be isolated by $O(1)$ cuts. It is possible to contract $\bra{\psi_{tree}} A \ket{\psi_{tree}}$ classically with $A$ being a local observable and $\ket{\psi_{tree}}$ being a tree tensor network state. 
    \item The \textit{checkerboard ansatz} consists of two-qubit blocks arranged in a checkerboard pattern Figure \ref{fig:tree_and_checkerboard}b. This ansatz can be expanded to obtain arbitrary precision, at the cost of having more parameters to optimize.
\end{enumerate}

Figure \ref{fig:tree_and_checkerboard} depicts the last two ans\"atze are with placeholder blocks. That is because they merely outline the structure of the circuit, while the content of these blocks may be more or less arbitrary. There is some tradeoff in this regard: if the blocks have more controllable parameters, they typically enable a richer set of states, but they as well become harder to optimize. Conversely, simpler blocks lead to poorer ansatz circuits which might be easier to optimize. In this experiment, we used a moderate-sized two-body ansatz that mostly uses the operators from the problem at hand, see Figure \ref{fig:entangler}. The unitary prepared by this block is described by the following equation:
\begin{equation}
    \label{eq:two_qubit_block1}
    U(\boldsymbol{\tilde{\theta}}) = (R_z(\tilde{\theta}_5) \otimes R_z(\tilde{\theta}_4)) \
    \circ R_{zz}(\tilde{\theta}_3) \ \circ  \ (R_x(\tilde{\theta}_1) \otimes R_x(\tilde{\theta}_2)),
\end{equation}
where $R_z (\theta) = e^{i\theta Z/2}$, $\ R_x (\theta) = e^{i\theta X/2}$, and $R_{zz} (\theta) = e^{i\theta Z\otimes Z/2}$. Thus, a complete ansatz would have five free parameters per two-qubit block.

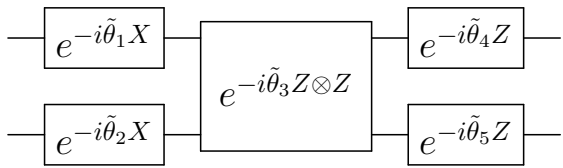
\begin{figure}
    \centering
    \mbox{
    \Qcircuit @C=1.0em @R=1.0em {
           & \gate{e^{-i \tilde{\theta}_1 X}} & \multigate{1}{e^{-i \tilde{\theta}_3 {Z} \otimes {Z}}} & \gate{e^{-i \tilde{\theta}_4 Z}} & \qw \\
           & \gate{e^{-i \tilde{\theta}_2 X}} & \ghost{e^{-i \tilde{\theta}_3 Z \otimes Z}} & \gate{e^{-i \tilde{\theta}_5 Z}} & \qw \\
       }
    }
    \caption{Two-qubit entangler gate used in preparation of the states. The circuit is to be read left to right.}
    \label{fig:entangler}
\end{figure}

In our numerical implementation, we use Qiskit \cite{aleksandrowicz_qiskit:_2019} to simulate quantum circuits and the limited Broyden–Fletcher–Goldfarb–Shanno method (L-BFGS-B) to update the parameters during the classical step of VQE. We scan values of $h$ from $0$ to $2J$. For $h=0$, the optimization process started from a random point, then each additional point begins from the previous solution. This procedure is also known as AAVQE \cite{garcia-saez_addressing_2018}, except that, unlike AAVQE, we are also interested in all the intermediate points of the trajectory. To eliminate any obviously sub-optimal solutions, we also ran the scanning in the opposite direction, and for each value of the field we kept the better result.

The results are plotted in Figure \ref{fig:dE_ising}. All ans\"atze demonstrate an increase in the error in the vicinity of the phase transition point. Interestingly enough, tree tensor network ansatz does not yield any improvement against the rank-1 ansatz, nor does one layer of the checkerboard ansatz. As should be expected, increasing the depth of the checkerboard ansatz substantially improves the energy error of the solution. At $h = 0$, the ground state is degenerate. At nonzero values of the field, this degeneracy is lifted, but the spectral gap is very small. Nonetheless, the overlap of VQE solution in this regime is never smaller than $1/2$, and gradually moves to $1$ across the phase transition. At $h = 1.2$, the point where the energy error is the worst, the overlap of the VQE solution with the ground state is equal to 0.92. Just as observed in more detail for the Hubbard model in Section \ref{sec:hubbard}, the convergence in energy looks exponential in the depth of the circuit: the peaks of the curves are nearly halved with each new layer. 

What is unexpected is that the location of the peak is different for each depth. The offset of the peaks alone could be attributed to the finite-size effect, but this dependendce on the depth suggests that the entanglement structure is different across the range of $h$.

\begin{figure}
    \centering
    \includegraphics[width=0.7\textwidth]{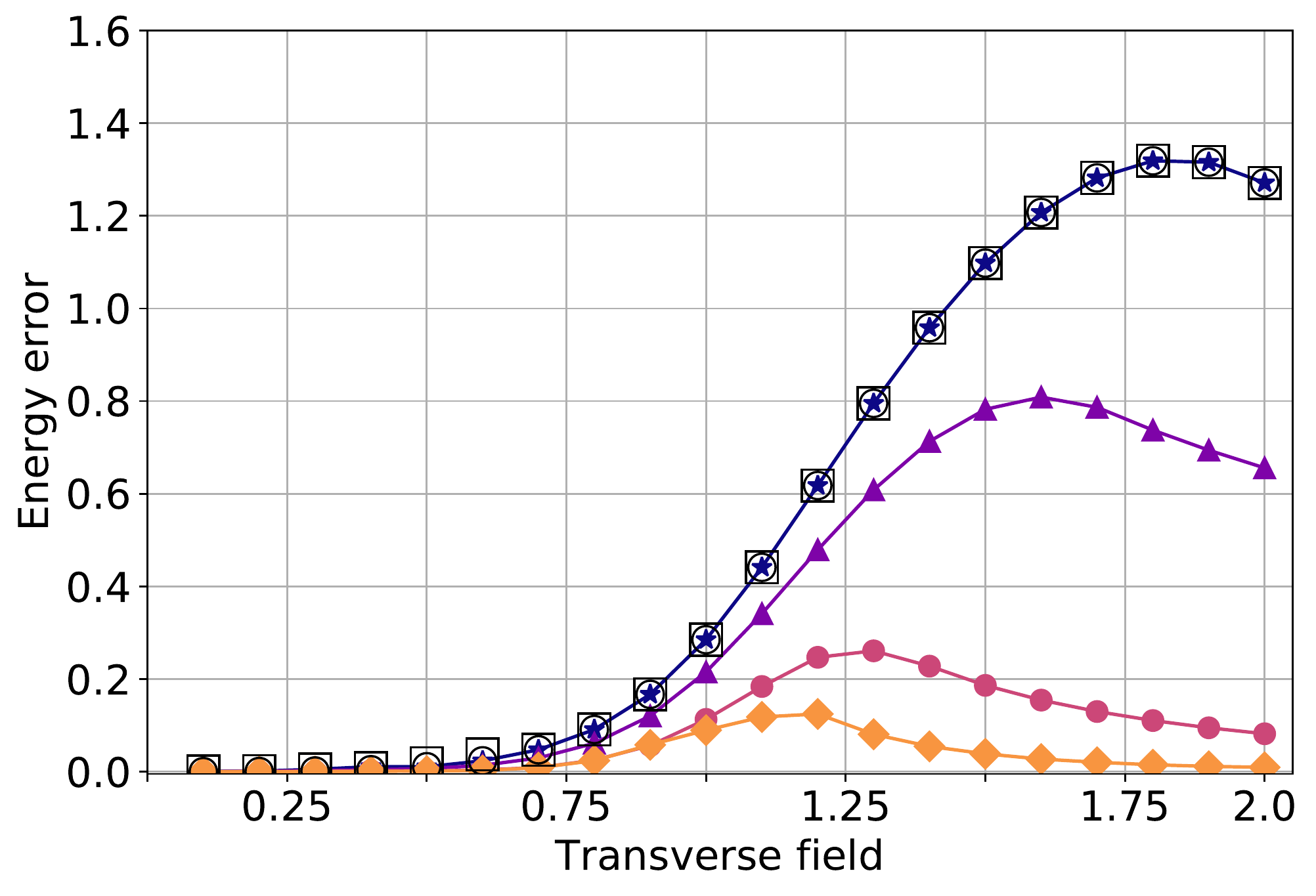}
    \caption{Absolute value of the difference in energy between the exact solution and VQE solutions for the transverse field Ising model. Hollow squares: rank-1 ansatz, hollow circles: tree tensor network, filled markers: checkerboard states ($\bigstar$: 1 layer, $\blacktriangle$: 2 layers, $\bullet$: 3 layers, $\blacklozenge$: 4 layers). Reprinted from \cite{uvarov_machine_2020}.}
    \label{fig:dE_ising}
\end{figure}

\subsection{Transverse-field Heisenberg model}

\begin{figure}
    \centering
    \includegraphics[width=0.7\textwidth]{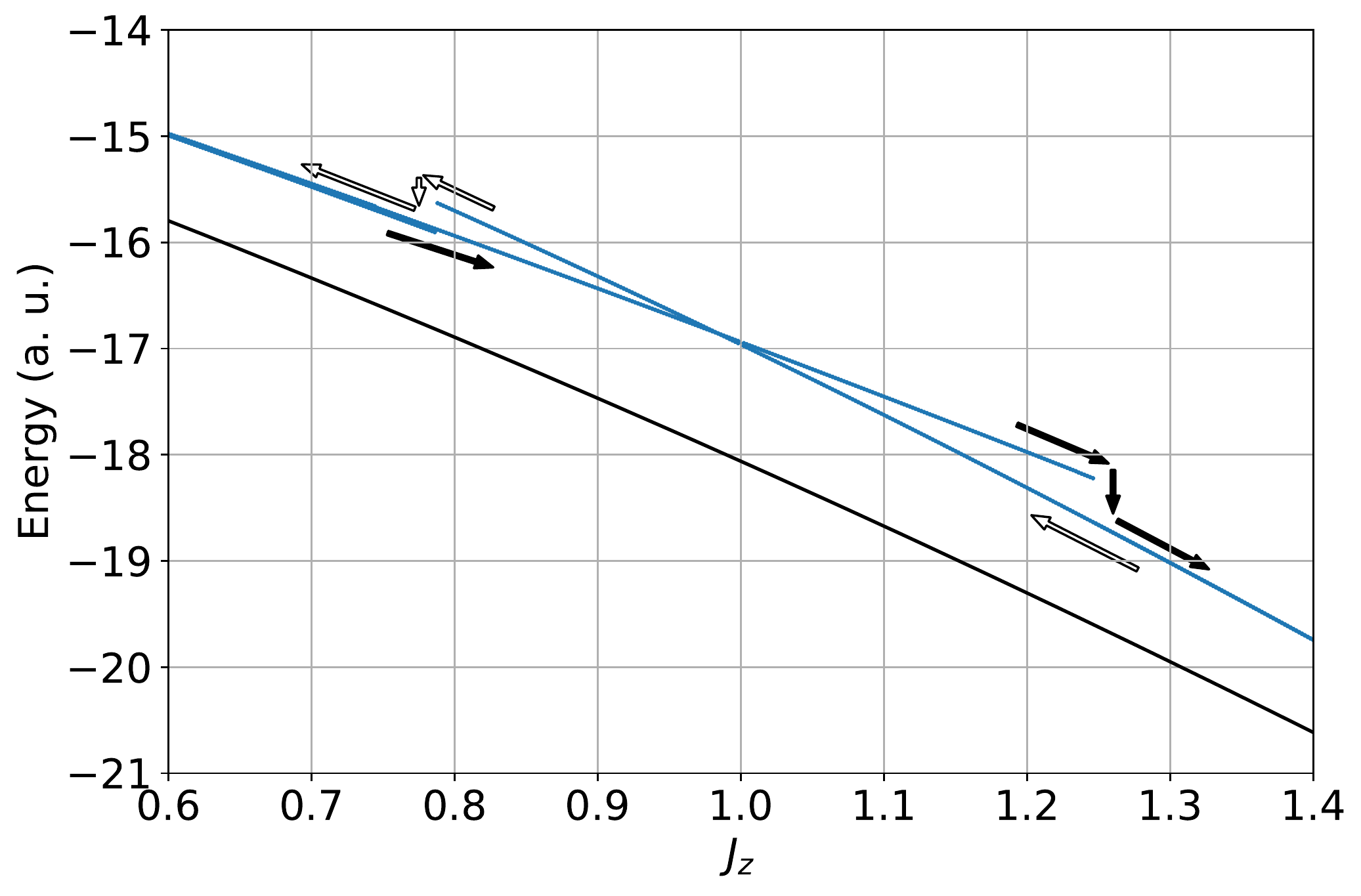}
    \caption{Ground-state energy estimate for the XXZ model found
    in VQE sweeps. Filled (empty) arrows guide the eye along the
    “up” (“down”) sweep. The best solution out of two sweeps was
    subsequently used to train the classifier. The black line denotes the energy of the exact solution. Reprinted from \cite{uvarov_machine_2020}.}
    \label{fig:vqe_hysteresis}
\end{figure}

\begin{figure}
    \centering
    \includegraphics[width=0.7\textwidth]{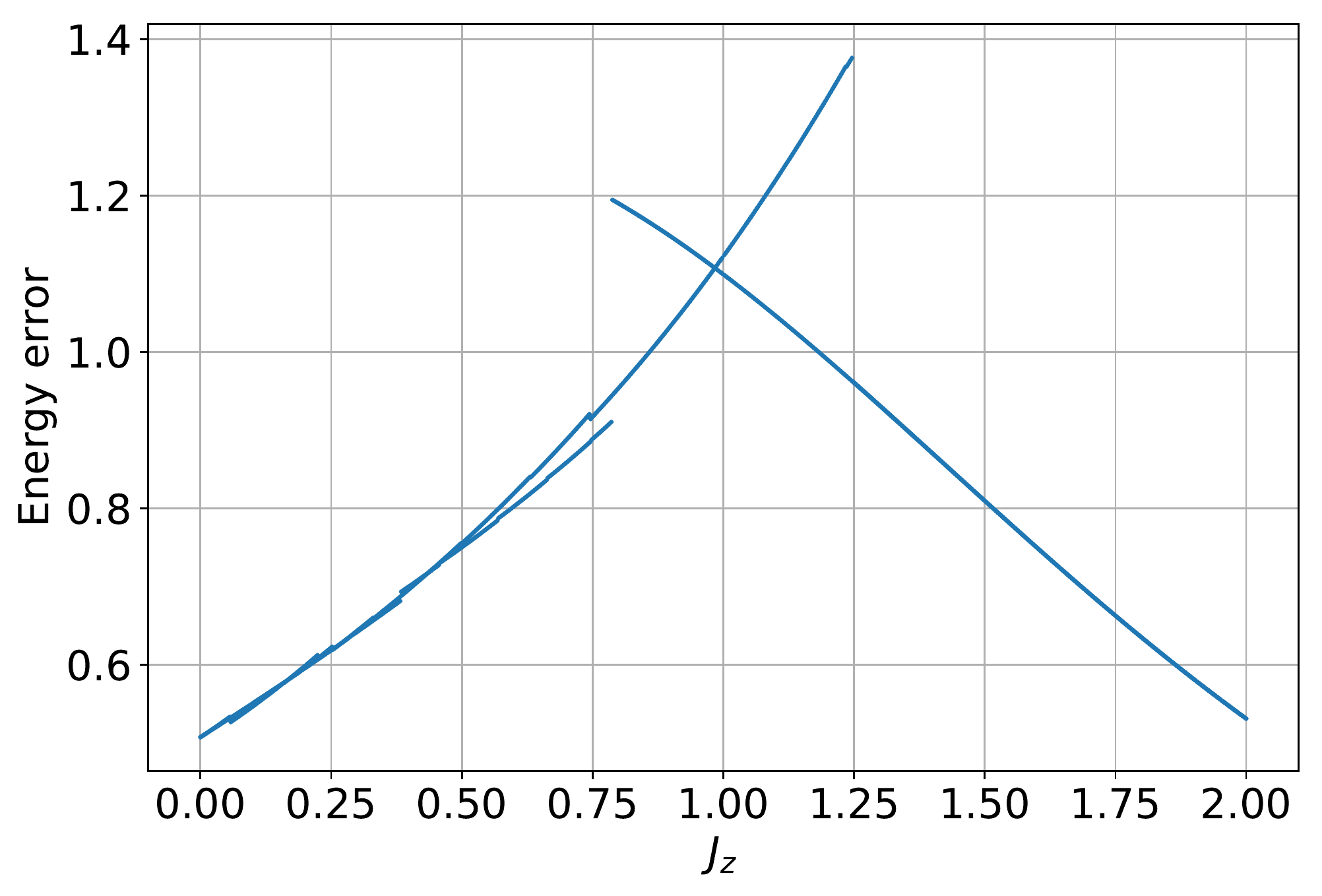}
    \caption{Energy difference between the exact solution and the ansatz solutions for the XXZ model.}
    \label{fig:dE_xxz}
\end{figure}

The XXZ Heisenberg model is defined as follows:

\begin{equation}
    \label{eq:heisenberg_xxz}
    H = \sum_{i=1}^n \left[J_\perp\left(X_i X_{i+1} + Y_i Y_{i+1}\right)
        + J_z Z_i Z_{i+1}\right].
\end{equation}

This model exhibits a more complicated behavior than the TFI model. However, for now we will only note that it, too, has a phase transition at $J_z = 1$. In Chapter \ref{chap:qml}, we study this model alongside the TFI model in context of quantum machine learning. Here we show the results of the same numerical experiments as for the TFI model, except that we only use the best ansatz (checkerboard with 4 layers).

In the experiments with the XXZ model, we observed an interesting peculiarity in the behavior of AAVQE. We found that the direction of the AAVQE sweep matters a lot. In the setup of AAVQE discussed in the original paper \cite{garcia-saez_addressing_2018}, the Hamiltonian is an easy one at the start and a difficult one at the finish. For models like ours, the behavior of the Hamiltonian along the change of $J_z$ is ``easy-hard-easy''. As such, to understand the low-energy properties of the model at different values of $J_z$, we could in principle run AAVQE in either direction. 

Figs.~\ref{fig:vqe_hysteresis}, \ref{fig:dE_xxz} show the results of running AAVQE in different directions. It turns out that the solutions in this case are susceptible to a hysteresis behavior: for a while after crossing the phase transition point, the solver stays in a suboptimal solution. Afterwards, the solution is suddenly changed to the optimal. Surprisingly enough, that behavior was not observed for the TFI model.

\section{VQE for Hubbard-like models}
\label{sec:hubbard}

\subsection{Next-nearest neighbor Hubbard model}

In condensed matter physics, frustrated systems with inhomogeneous interactions are hard to analyze owing to extra degrees of freedom to show up. On the one hand, strong electron-electron interactions precludes perturbative expansion over single-electron wave functions. On the other hand, more advanced numerical approaches to strongly correlated systems, e.g.~based on dynamical mean-field theory, treat the systems on a purely local manner. Whether modern quantum algorithms can give an edge in analyzing these models is an intensely studied question \cite{cade_strategies_2019,rungger_dynamical_2019,jaderberg_minimum_2020}. In this section, we analyze the performance of VQE for a one-dimensional model of spinless electrons with nearest- and next-nearest-neighbor interactions. This model represents a simple theoretical testbed to explore the physical properties of frustrated systems. Resulting from the competition between two types of interactions, a metallic state emerges even for strongly interacting systems. Interestingly, results of numerical simulations for finite size clusters unambiguously reveal that the ground state does not belong to the Luttinger liquid universality class \cite{Zhuravlev1997,zhuravlev_breakdown_2000,zhuravlev_one-dimensional_2001,Hohenadler2012,Karrasch2012}.

The frustrated systems are typically modeled using different variations of the Hubbard model. Physically, the model represents a lattice of atomic sites. Each site can have one or more energy levels for each spin. Electrons can hop between identical levels in different sites. The complicated part of the model is the Coulomb repulsion between the electrons sharing the same site or residing in neighboring sites. For simplicity, we consider the sites that have only one energy level to occupy. The Hamiltonian for such a model is as follows:

\begin{equation}
    H_{\text{Hubbard}} = \sum_{i \in [n], \sigma \in \{\uparrow, \downarrow\}} \epsilon_{i \sigma} \hat{n}_{i \sigma}
    - \sum_{<i, j>, \sigma \in \{\uparrow, \downarrow\}} t_{ij, \sigma} (a^\dagger_{i \sigma} a_{j \sigma} + a^\dagger_{j \sigma} a_{i \sigma})
    + U \sum_{i \in [n]} \hat{n}_{i \uparrow} \hat{n}_{i \downarrow},
\end{equation}
where $<i,j>$ denotes summation over all pairs of neighboring sites, and $\hat{n}_{i \sigma} = a^\dagger_{i \sigma} a_{i \sigma}$ denotes the population of the site $(i, \sigma)$. The first term of this Hamiltonian represents the on-site energy of sites. If the number of particles is conserved (i.e.~we are not considering an open system), this term is not constant only if the sites represent different atoms. In our simulations we assume that all sites are identical, so in what follows we omit this term. The second term represents hopping between the sites, and the third term represents Coulomb repulsion.

To better catch long-range order using few qubits, we will consider spinless sites. This enables us to simulate chain that is twice as long, while using the same amount of qubits.
In such a model, the Coulomb repulsion is then added as an energy affecting electrons in neighboring sites. Such a model would in general look as follows:

\begin{equation}
    \label{eq:hubbard_nnn}
    H = - \sum_{<i, j>} t_{ij} (a^\dagger_{i} a_{j} + a^\dagger_{j} a_{i})
    + \sum_{i, j} U_{ij} \hat{n}_{i} \hat{n}_{j}.
\end{equation}

The simplest nontrivial model of that kind is a 1D chain of sites with identical hopping matrix elements $t_{i, i+1} = t$ and nearest-neighbor Coulomb repulsion $U_{i, i+1} = V_1$. In the numerical experiments detailed in this section, we add add next-nearest Coulomb repulsion $U_{i, i+2} = V_2$.

Importantly, the number of particles in such a model must be conserved. One way to ensure that the optimization finds the correct number of particles is to add a corresponding penalty term to the Hamiltonian. That is, let the correct number of particles be equal to $m$. Then the Hamiltonian we want to treat with VQE is equal to 
$H + M (\sum_{i \in [n]} a^\dagger_{i} a_{i} - m)^2$, taking $M$ to be a sufficiently large positive constant. An alternative method --- that works only when the Jordan--Wigner transformation is used --- consists in tailoring the ansatz circuit so that the particle number constraint is preserved automatically. In such an ansatz, single-qubit gates can only add relative phases to computational basis states, but cannot map $\ket{0}$ to $\ket{1}$ or vice versa. Two-qubit gates have to be a direct sum of operators acting independently on spaces $W_0 = \operatorname{span} \{ \ket{00} \}$, $W_1 = \operatorname{span} \{\ket{01}, \ket{10}\}$, and $W_2 = \operatorname{span} \{ \ket{11} \}$.

\subsection{Numerical setup}

In our numerical experiments, we used the Jordan--Wigner encoding (Section \ref{sec:fermion-transforms}). We used a checkerboard ansatz whose two-qubit blocks were the particle-conserving gates introduced in \cite{barkoutsos_quantum_2018}:

\begin{equation}
\label{eq:particle-conserving_gate}
    U(\theta_1, \theta_2) = 
    \begin{pmatrix}
1 & 0 & 0 & 0 \\
0 & \cos{\theta_1} & e^{\imath \theta_2} \sin \theta_1 & 0 \\
0  & e^{-\imath \theta_2} \sin \theta_1  & -\cos{\theta_1} & 0  \\
0 & 0 & 0 & 1
\end{pmatrix}.
\end{equation}
The initial state supplied to the ansatz is a product state $\ket{1010...}$ that has the amount of $1$'s equal to the number of electrons. We study the model at half-filling, so the number of electrons is equal to $\lfloor n / 2\rfloor$. Numerical simulations were performed using Qiskit. The Jordan--Wigner transform and the Bravyi--Kitaev transform (used further) are implemented using OpenFermion \cite{mcclean_openfermion_2020}. The optimization routine in VQE uses the limited-memory  Broyden–Fletcher–Goldfarb–Shanno (L-BFGS) algorithm. This method consists in approximately evaluating the Hessian matrix of the cost function and performing the step of the Newton method.

\subsection{Results}

The results of VQE optimization are shown in Figure \ref{fig:vqe_hubbard_nnn}. For the most part, the energy exhibits exponential convergence with the depth of the ansatz circuit, although the exact value of the exponent is different for different numbers of qubits. The jagged appearance of the curves is due to the fact that the specific particle-conserving gate used in the ansatz does not possess the identity gate in its configuration space. It is also clear that for larger $n$, the convergence tends to go slower. In particular, for $n=11$ qubits, the energy error remains roughly the same despite the increase in the number of layers.

\begin{figure}
    \centering
    \includegraphics[width=0.7\textwidth]{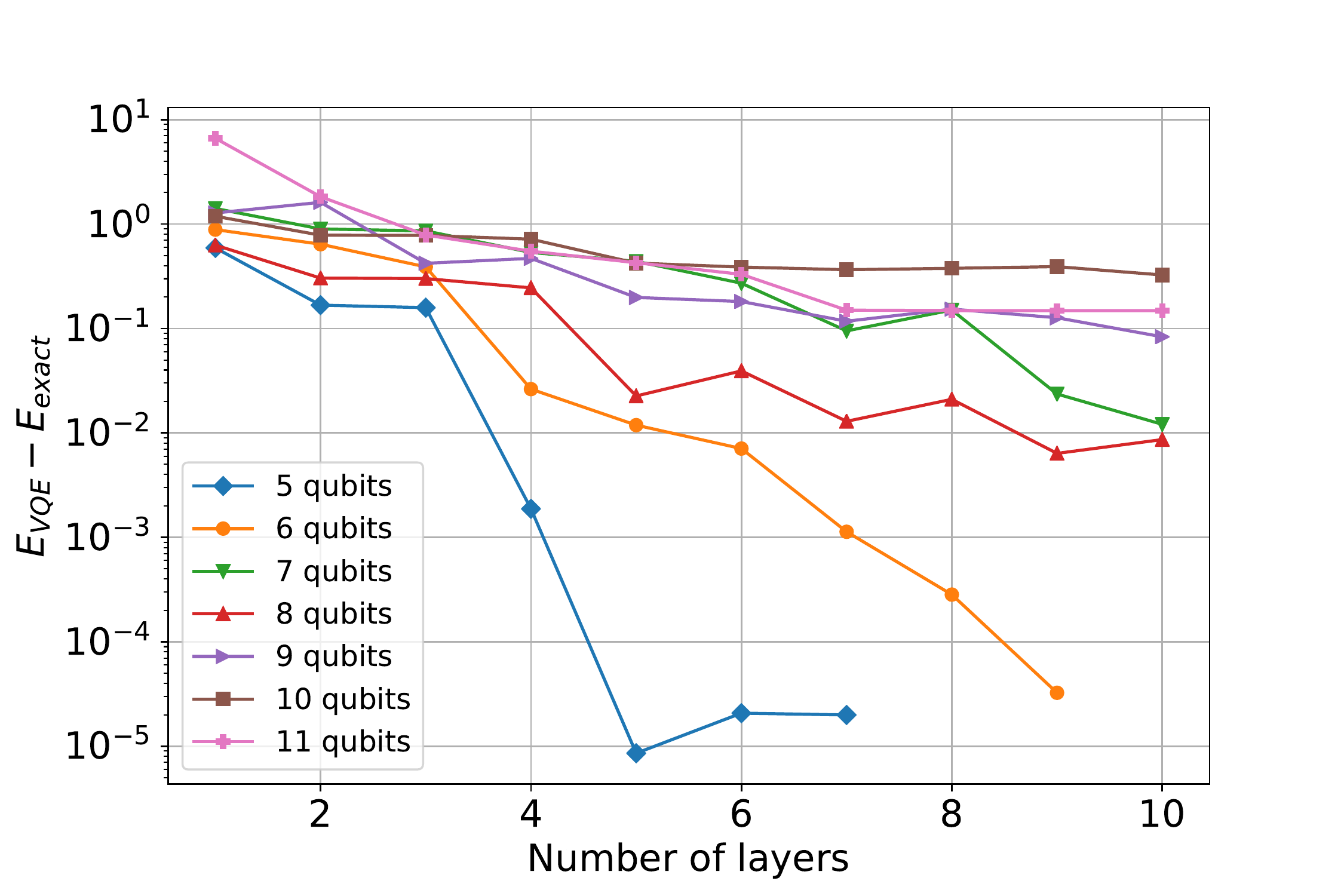}
    \caption{Convergence of the VQE solution to the true ground state for the 1D next-nearest-neighbour repulsion Hubbard model versus the number of layers, on condition $V_1 = 2t, V_2 = t$. Cases of $n \leq 4$ qubits are not shown as they converged to exact solution within 2 layers. Reprinted from \cite{uvarov_variational_2020}.}
    \label{fig:vqe_hubbard_nnn}
\end{figure}

The energy error can be shown to be closely related to the infidelity between the true ground state and the variational approximation \cite{biamonte_universal_2021}. In other words, the error in energy is close to zero if the variational solution lies close to the ground state subspace. However, both of these metrics are useful only when we possess enough information on the exact solution. Since in real applications we want to find the true energy in the first place, we cannot assume this knowledge \textit{a priori}. It is therefore also interesting to consider the convergence of other physically relevant observables. Specifically, we consider the convergence of the density-density correlation function $C(m) = \langle n_0 n_m\rangle - 
\langle n_0\rangle \langle n_m\rangle$. Figure \ref{fig:correlation} shows the behavior of the density-density correlation function in the true ground state and VQE approximations.
To provide a quantitative estimate we depict relative error of the correlation function as implemented in the VQE routine with respect to the exact solution in Fig.~\ref{fig:corr_errors}. For a more serious problem, the exact solution is, of course, unavailable. Nonetheless, the convergence of the correlation function to some apparent limit can be used as an indirect sign of overall convergence. Although we did find qualitative agreement, the accuracy of the approximation, as well as the size of the spin chain, were too low to produce any reasonable quantitative estimates of the asymptotical behavior of the correlation function. 

\begin{figure}
    \centering
    \includegraphics*[width=0.7\textwidth]{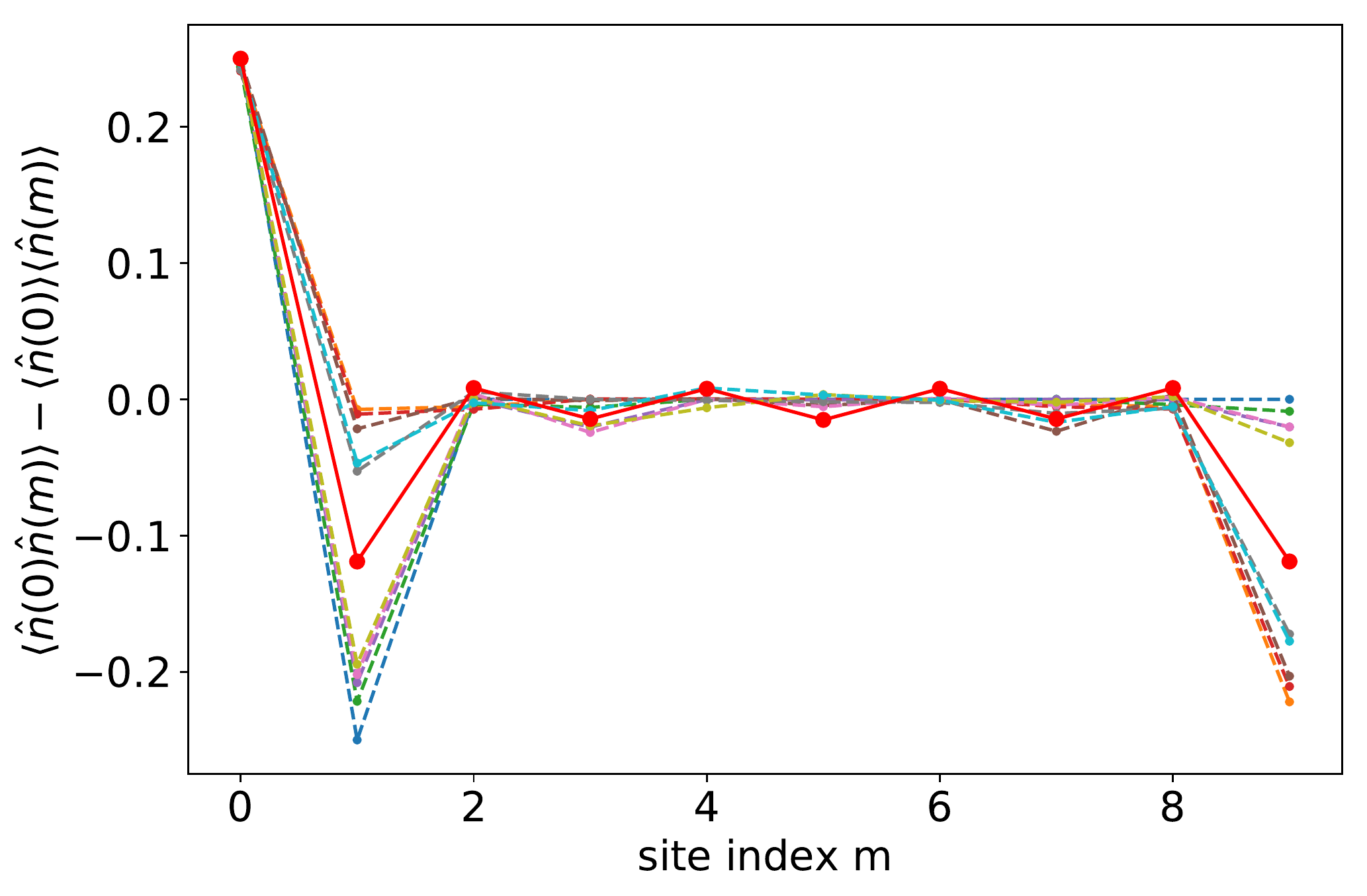}
    \caption{Density-density correlation function between spatially separated lattice sites. Filled dots denote exact values as obtained by virtue of exact diagonalization of the Hamiltonian (\ref{eq:hubbard_nnn}), dashed lines denote different approximations. Here $V_1 = 2t, V_2 = t$. Reprinted from \cite{uvarov_variational_2020}.}
    \label{fig:correlation}
\end{figure}

\begin{figure}
    \centering
    \includegraphics[width=0.7\linewidth]{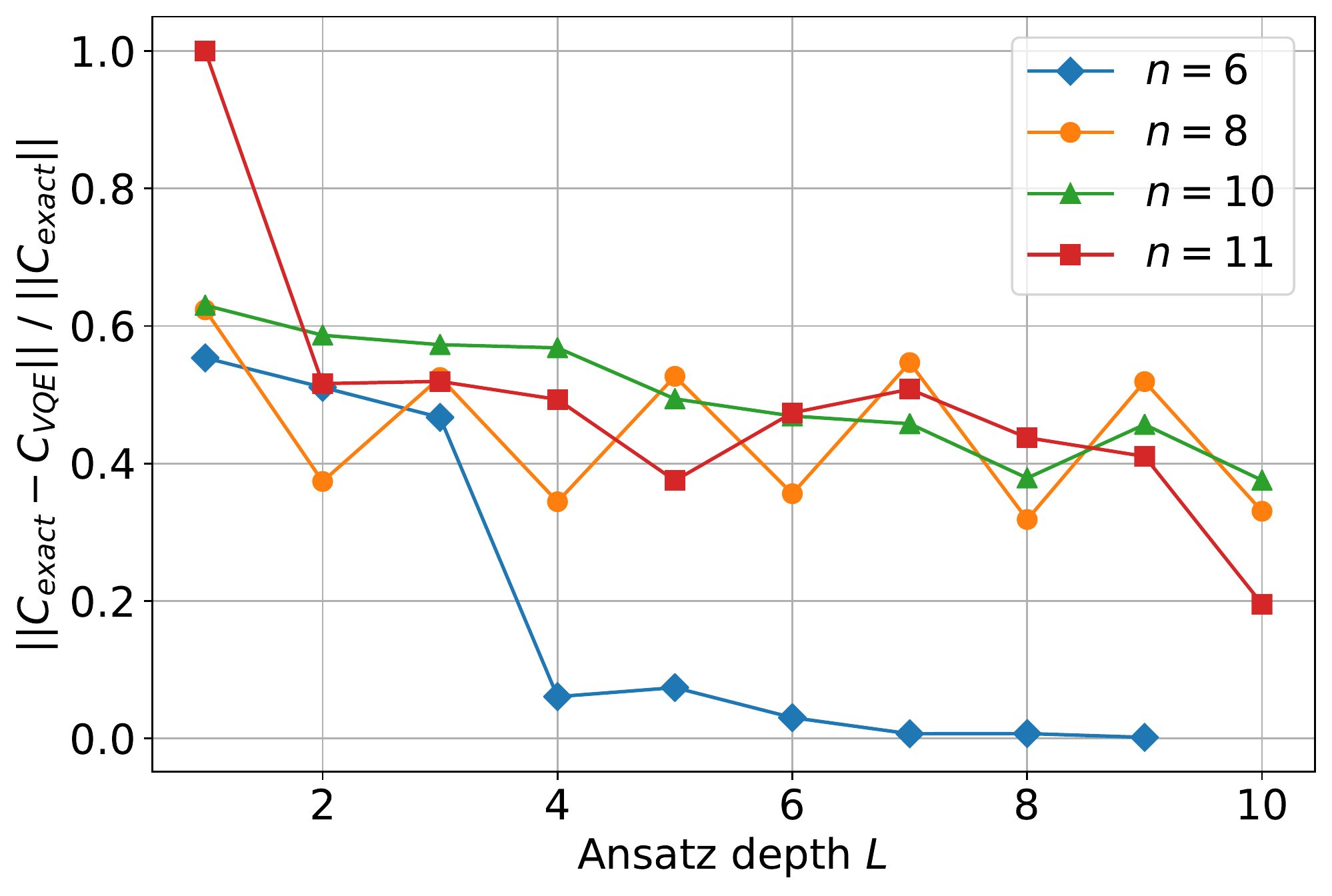}
    \caption{Relative error of the correlation function as obtained with the VQE method, $C_{VQE}$, with respesect to the exact solution, $C_\mathrm{exact}$, as a function of the number of layers $L$. Reprinted from \cite{uvarov_variational_2020}.}
    \label{fig:corr_errors}
    \end{figure}{}

\subsection{Barren plateaus}

Variational algorithms can experience so-called barren plateaus \cite{mcclean_barren_2018}. While in detail this phenomenon is studied in Chapter \ref{chap:plateaus}, we will briefly explain it here as well. When we work with variational algorithms, we would like to know somehting about the optimization landscape of the problem. To do that, we can sample random points in the parameter space and consider the partial derivatives of the cost function. When the ansatz circuit is very deep, sampling unitary operators from that circuit in a certain sense very similar to sampling random unitary operators uniformly from the unitary group $U(2^n)$. Under these conditions (see Theorem \ref{thm:mcclean}), the partial derivatives $\partial_\theta \bra{\psi} H \ket{\psi}$ have zero mean and exponentially vanishing variance (in $n$).

This result, however, is valid for sufficiently long circuits with enough freedom in individual gates. Here we introduced an ansatz which preserves the number of particles in the state. A random unitary operator sampled from such an ansatz will necessarily be a direct sum of operators acting in subspaces corresponding to different particle numbers. This property is, of course, not present for random unitary operators. Another aspect of the problem is that the Hamiltonians in question also change with the number of qubits. Because of that, we performed a number of numerical experiments concerning the onset of barren plateaus in the case of Hubbard model and our ans\"atze.

Figure \ref{fig:plateaus_hubbard_ising}, top left, shows the behavior of the variance of the derivatives (averaged over random selections of $\boldsymbol{\theta}$ and $k$) for the next-nearest-neighbor Hubbard model (\ref{eq:hubbard_nnn}) under the Jordan--Wigner mapping. The particle-conserving ansatz is used for the analysis. For most qubit numbers and regardless of the model parameters, said variance drops to its limiting value even for very shallow circuits. In contrast, the same derivatives under the Bravyi--Kitaev mapping show a more gradual behavior (Figure \ref{fig:plateaus_hubbard_ising}, top right): the variance decays exponentially with the depth of the circuit. The particle-preserving gate in the ansatz was replaced with a more generic gate shown in Fig.~\ref{fig:entangler}, since the former would no longer conserve the particle number.
For comparison, we also performed a similar experiment for the transverse field Ising model. The gradient behavior away from the critical point ($h = 0.1$) and at the critical point ($h = 1$) is shown in Fig.~\ref{fig:plateaus_hubbard_ising}, bottom. In this case, the gradient variance decays exponentially with the number of layers until reaching the plateau regime for the particular number of qubits. Thus, for 4 qubits the plateau is reached right away, while for 10 qubits, 30 layers of the ansatz are still a number belonging to the transition regime. In the meantime, the criticality of the model does not seem to affect this behavior. As we will see in Chapter \ref{chap:plateaus}, a possible reason for the behavior observed in these three cases is that the Jordan--Wigner mapping produces highly nonlocal operators which are difficult to optimize.


\begin{figure}
    \centering
    \begin{subfigure}{.48\linewidth}
        \centering
        \includegraphics[width=\textwidth]{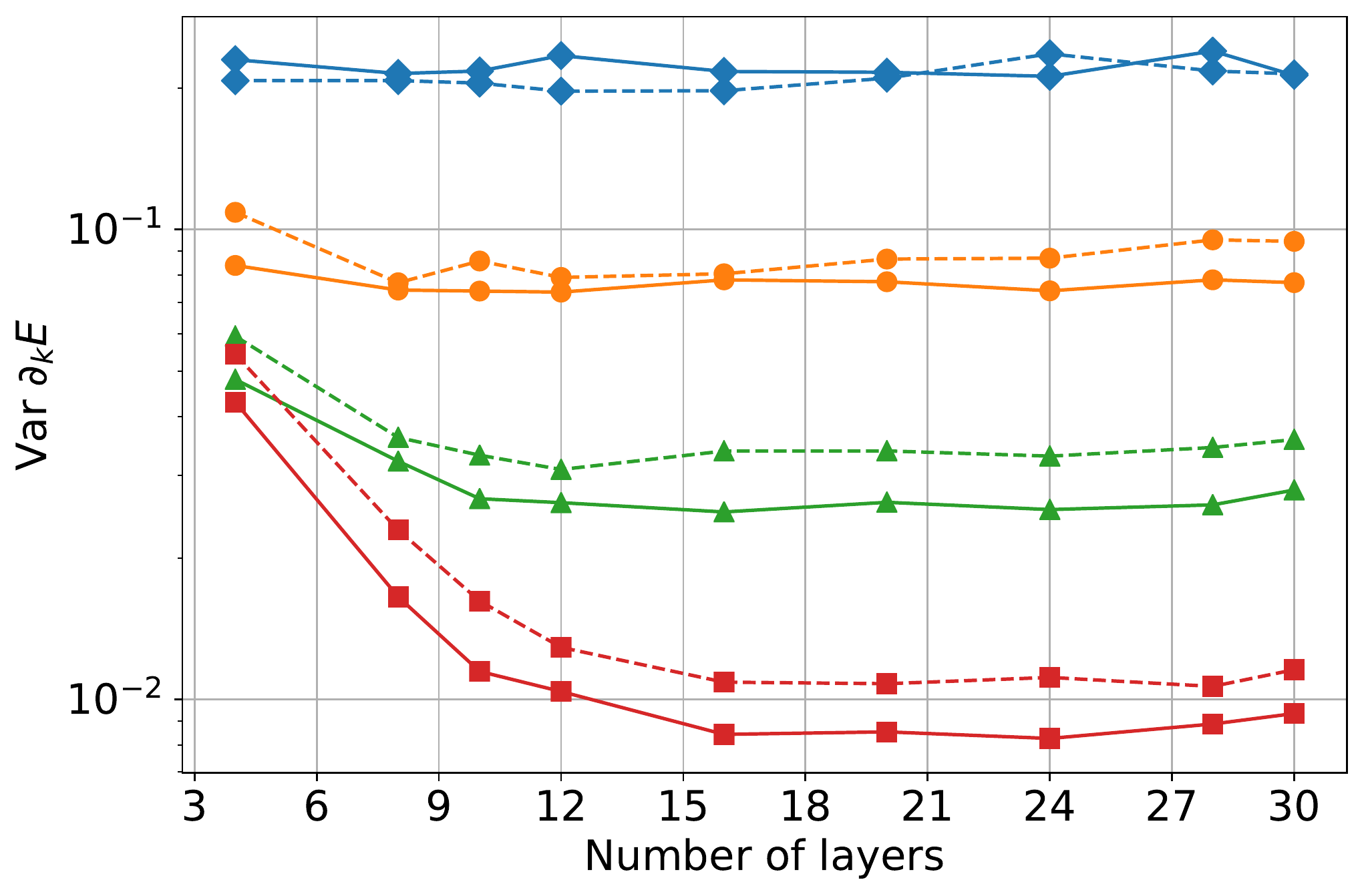}
    \end{subfigure}\begin{subfigure}{.48\linewidth}
        \centering
        \includegraphics[width=\textwidth]{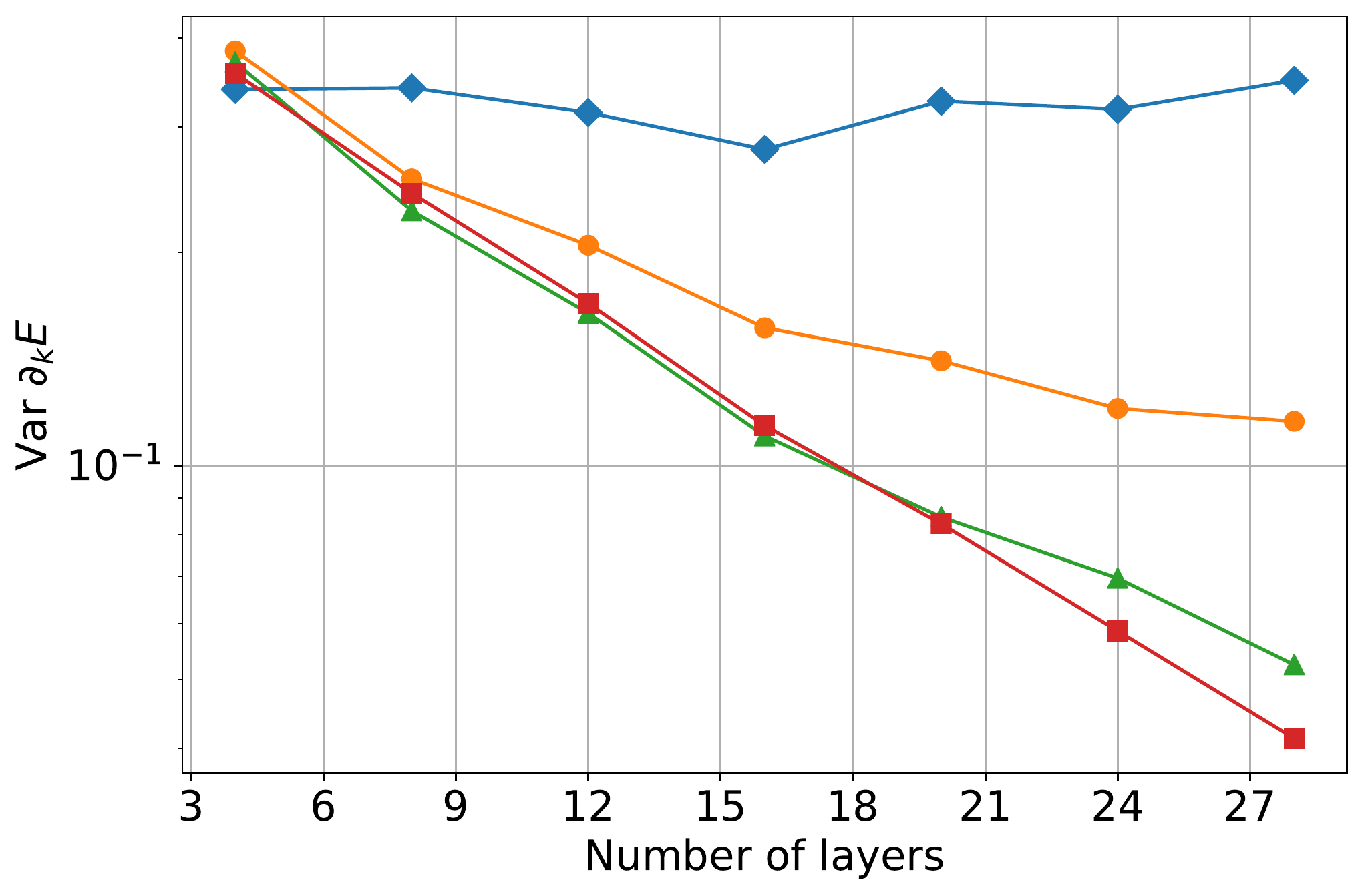}
    \end{subfigure}
    \begin{subfigure}{.48\linewidth}
        \centering
        \includegraphics[width=\textwidth]{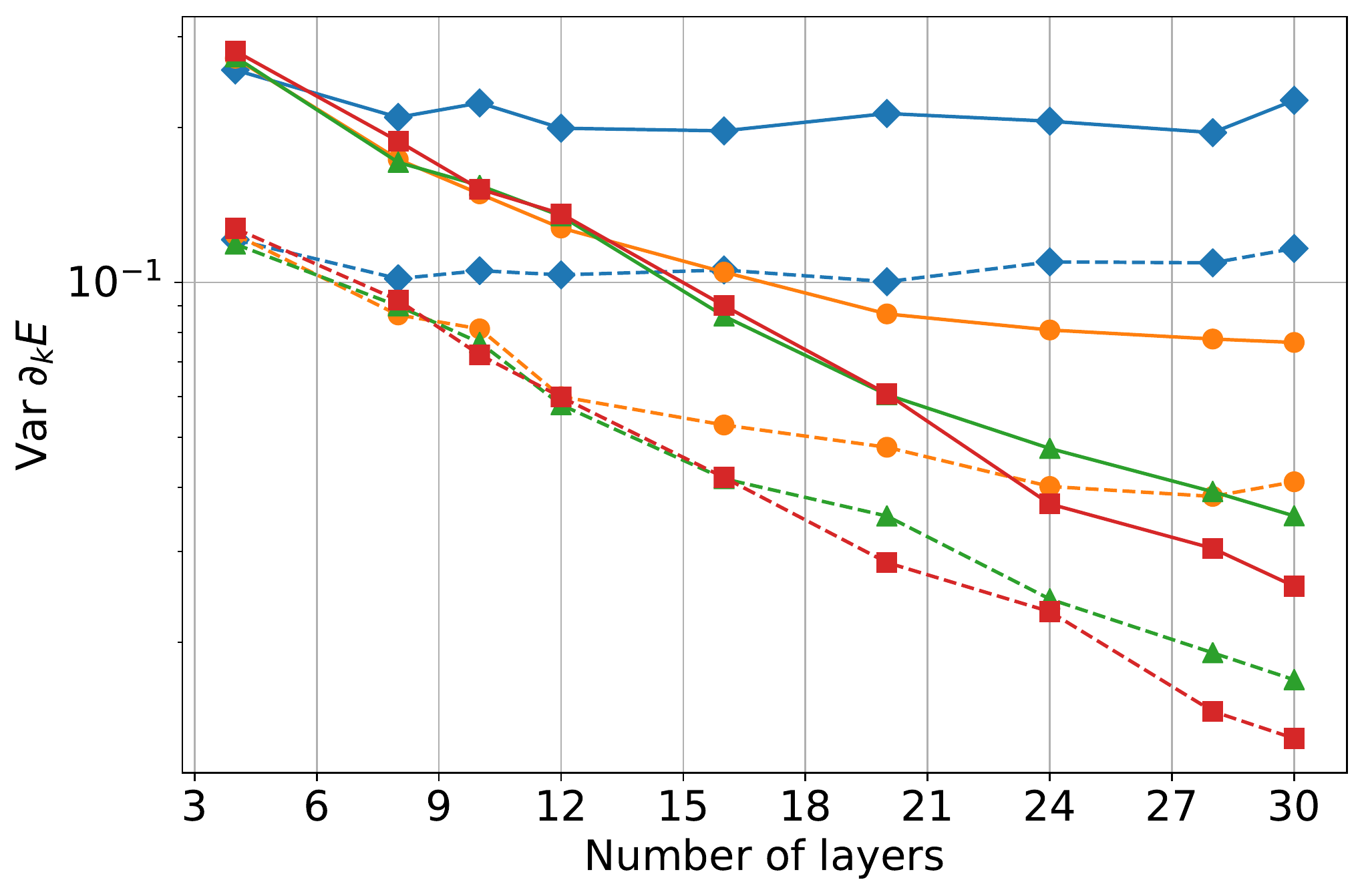}
    \end{subfigure}
    \caption{\textbf{Top left:} barren plateau effect for the Hamiltonian of Eq. (3) with
    $V_1 = t$ and $V_2 = 0$ (dashed lines), as well as $V_1 = 2t$ and $V_2 = t$
    (solid lines) versus the number of qubits as realized by virtue of
    Jordan--Wigner mapping. Diamonds: four qubits; circles: six qubits;
    triangles: eight qubits; squares: 10 qubits. \textbf{Top right:} same effect under Bravyi--Kitaev mapping, $V_1 = 2t$, $V_2 = t$.
    \textbf{Bottom:} same effect for for the transverse field Ising model of Eq. (4) away from criticality with $h = 0.1$ (dashed lines) and at the critical point $h = 1$ (solid lines).    
    Reprinted from \cite{uvarov_variational_2020}.}
    \label{fig:plateaus_hubbard_ising}
\end{figure}

\section{Discussion}

In this chapter, we saw different physical models analyzed using the VQE method. Regardless of the model, we observed a fast convergence of the energy error with the increase of the circuit depth, however, the experiments with the Hubbard model hint that this may just be the consequence of a small number of qubits, as for larger $n$ this convergence becomes less apparent.

There is a lot of choices to be made when choosing the implementation of VQE. We always have to choose the ansatz and/or its update scheme, we need to choose the optimization method, and for fermionic problems, we also have to choose the qubit encoding. A whole other set of choices not covered in this work relate to the experimental setup: how to label qubits, how many measurements to make per step, and so on. All of this affects the quality of the solution.

In the experiments with the Hubbard model we observed the barren plateau phenomenon, in which the partial derivatives converge to exponentially small values when sampled randomly. On the one hand, its behavior is to be expected: as we will see in Chapter \ref{chap:plateaus}, the locality of the operators included in the Hamiltonian is one of the key factors affecting the onset of barren plateaus. On the other hand, the theory of barren plateaus heavily relies on the concept of \textit{t-designs}, i.e.~random quantum circuits that are indistinguishable from uniformly random unitary operators under certain conditions. Clearly, the particle-conserving ansatz is not even close to resembling a uniformly random unitary, hence it is not a \textit{t}-design (again, we will back this up numerically in Chapter \ref{chap:plateaus}). Nonetheless, we observe the plateaus. This implies that in some cases the condition for observing barren plateaus is not necessarily tied to the notion of a \textit{t}-design. Alternatively, this may mean that the ansatz becomes close to a \textit{t}-design when it is restricted to the subspace with the fixed number of particles.

\chapter{Quantum machine learning with VQAs}
\label{chap:qml}

\section{Introduction}

Every discipline in computer science tries to gain some advantage from the usage of quantum computation, or establish whether there is any. Machine learning is no exception. The topic of quantum machine learning has considerable overlap with variational quantum algorithms, which motivates some attention to the topic within this thesis. In this chapter, we will mostly focus on quantum neural networks (QNN) rather than the more general techniques for quantum machine learning. The reason for that, as we shall soon see, is that modern QNNs can be essentially considered equivalent to variational circuits, and training quantum neural networks can be classified as a variational quantum algorithm. For broader reveiws, we refer the reader to Refs. \cite{schuld_introduction_2015,biamonte_quantum_2017, ciliberto_quantum_2018}.

\subsection{Quantum neural networks}
 
What is a QNN? There is a multitude of ways one can come up with quantum analogues of a neural network. To do that, one needs to identify the key properties of a neural network that should be replicated. For example:

\begin{enumerate}
\item The network should consist of individual neurons connected to each other in a nontrivial manner, with tunable connection weights
\item The neurons should implement some kind of nonlinear activation functions.
\item The network is trained to solve its task by minimizing an error measured on the training samples.
\end{enumerate}

A big challenge in designing a quantum neural network is that all evolution that happens in quantum mechanins is linear, except for the measurement part, while classical neural networks are unthinkable without nonlinearity. 

\paragraph{Early QNN proposals.}
Here we will go in roughly chronological order. The first speculations about quantum neural computation can be found in Ref.~\cite{kak_quantum_1995}, but these only outline the concept without being too specific about the implementation. Behrman et al.~\cite{behrman_simulations_2000} proposed a neural network based on the interaction of quantum dots. Another proposal introduces nonlinearity through a dissipative, non-unitary operator \cite{gupta_quantum_2001}.

Some of these proposals can be summarized as follows. A quantum neuron (called `quron' in Ref.~\cite{schuld_quest_2014}) is a two-level quantum system, and different architectures engineer different interactions between the neurons. Such networks were reviewed by Schuld et al.~\cite{schuld_quest_2014}. It was observed that none of these networks can replicate the behavior of a Hopfield network -- a network which exhibits multiple stable states of its neurons. In a Hopfield network, the state of the neurons is iteratively updated based on the previous states and inter-neuron weights. Its key feature is that it has stable states and that the basins of attraction split all the state space, which is impossible if the evolution of a neural network is purely unitary.

\paragraph{Variational QNNs.}

The idea to use a generic variational circuit as a trainable model for a classification problem was independently proposed by Schuld et al.~\cite{schuld_circuit-centric_2020} and by Havlicek et al.~\cite{havlicek_supervised_2019}. The algorithm presented by us in Ref.~\cite{uvarov_machine_2020} and expanded upon in this chapter also belongs to this family. It is of course quite simple to implement feedforward networks in this setup, but one can also go for more complicated constructions, such as recurrent neural networks \cite{bausch_recurrent_2020}, generative adversarial networks \cite{dallaire-demers_quantum_2018} or convolutional neural networks \cite{pesah_absence_2020}.

While some QNNs of that kind do not even have well-defined neurons at this stage, others do define neurons as quantum gates that have several input qubits and one output qubit (possibly with ancilla qubits used in the process) \cite{cao_quantum_2017,bausch_recurrent_2020}. In this fashion, one can talk about layers in a QNN, which is somewhat more specific than layers in a general variational ansatz. Every unitary in the layer acts on all input qubits plus one output qubit. After all unitaries of that layer are applied, the output qubits are treated as input qubits for the next layer. Finally, the output qubits are measured after the last layer. Note that in this architecture, the input qubits from the last layer can be discarded and reused. This way, one can make deep quantum neural networks \cite{beer_training_2020} even with a limited budget of qubits, provided that the quantum processor can reinitialize qubits in runtime. 

\subsection{Quantum machine learning models viewed as kernel models}

An alternative view on the variational QNNs was presented in Ref.~\cite{schuld_quantum_2021}. Recall that our classifier circuit consists of two sub-circuit: one prepares the desired quantum state, another performs the classification. For a more general setup, the first part can be seen as an encoding circuit that maps classical data to quantum states. Let $x_i \in \mathcal{X} = \mathbb{R}^m$ be a classical data point with label $y_i \in \mathbb{R}$. Let us denote the corresponding quantum state $\ket{\phi(x_i)}$. Its density matrix is equal to $\rho(x_i) = \ket{\phi(x_i)}\bra{\phi(x_i)}$. When we prepare the state, apply the classifier circuit, and measure some observable $M$, we extract a prediction $\hat{y}_i \in \mathbb{R}$. The combination of a classifier circuit and a measurement is a function $f: \operatorname{Pos}(2^n) \rightarrow \mathbb{R}$:
\begin{equation}
    f(x) = \Tr(\rho(x) U^\dagger_{\mathrm{class}} (\boldsymbol{\varphi}) M U_{\mathrm{class}} (\boldsymbol{\varphi})).
\end{equation}
Here $\operatorname{Pos}(2^n)$ is the space of $2^n \times 2^n$ density matrices. Importantly, this function is linear in $\rho$, i.e.~$f$ belongs to the dual space. If the vector space has a basis (in our case, the basis of density matrices $\{\rho_i\}$) in with a fixed inner product (here we have $(a, b) = \Tr(a^\dagger b)$), then there is a natural choice of the dual basis: $\rho_i^* = (\rho_i, \placeholder)$.

Let us now introduce kernel models. A \textit{kernel} is a function $\kappa: \mathcal{X} \times  \mathcal{X} \rightarrow \mathbb{R}$. We require that this function takes larger values when its arguments are close to each other. Sometimes $\kappa$ is called a similarity function. Typically, $\kappa$ is also symmetric ($\kappa(x, x') = \kappa(x', x)$) and positive definite, that is, for all integer $l$ $x_1, .., x_l \in \mc{X}$, and any $c_1, ..., c_l \in \mathbb{R}$, the following is true \cite{vert_primer_2004}:
\begin{equation}
    \sum_{i=1}^l  \sum_{j=1}^l c_i c_j \kappa(x_i, x_j) \geq 0.
\end{equation}
Returning to the quantum models, we can choose a kernel as the inner product of the data points encoded in the state space:
\begin{equation}
    \kappa(x, x') = \Tr(\rho(x) \rho(x')).
\end{equation}

A famous result in the kernel theory, called the \textit{representer theorem}, states that any function minimizing the expected loss over the data set is expressed as a linear combination of kernel functions involving the data points:

\begin{equation}
    f_{\text{opt}}(x) = \sum_i \alpha_i \kappa(x_i, x), \quad \quad \alpha_i \in \mathbb{R}
\end{equation}

Moreover, in this form, the optimization of the coefficients is \textit{convex}. This means that in principle, one can replace the QNNs with kernel models of that kind and always find a global optimum. On the other hand, for a data set with $m$ points, one will need to estimate $O(m^2)$ inner products, and the classification of a data point will take $O(m)$ evaluations. For big data applications, this is a disadvantage compared to QNNs. For instance, the classification of a data point for a QNN takes $O(1)$ in the number of data points, while training can sometimes take less than $O(m^2)$.





An interesting question is whether a quantum ML model bring any advantage in predicting the outcomes of quantum experiments. Huang and coworkers \cite{huang_information-theoretic_2021} show that, in a very general setting, a quantum ML model does not bring any advantage in terms of average-case error, but does bring exponential advantage in terms of worst-case error.

\section{Quantum classifier to partition quantum data}

\paragraph{Learning the phase of the transverse-field Ising model. } Here we demonstrate how we solve a machine learning problem that is intrinsically quantum. Recall that in chapter \ref{chap:vqe_numerics} we studied the transverse field Ising model:

\begin{equation}
\label{eq:tfim_2}
    H = J \sum Z_i Z_{i+1} + h \sum X_i.
\end{equation}

This model has a known phase transition at $J=h$. Using this fact, we set up our machine learning problem as follows:

\begin{itemize}
    \item The data points are the approximate ground states of (\ref{eq:tfim_2}) for $J=1$ and different values of $h$. The approximate states none other than the VQE solutions. We took the best set of solutions we had in terms of energy error --- the one found using four layers of the checkerboard ansatz.
    \item Given a quantum state, promised to be a ground state of the TFI model, the task is to tell whether $h<1$ or $h>1$ for this state.
\end{itemize}

To solve this problem, we employ a variational ansatz. Let $U_{\mathrm{VQE}}(\boldsymbol{\theta})$ be the ansatz that prepares the input state $\ket{\psi(\boldsymbol{\theta})}$ that is to be classified. Denote $U_{\mathrm{class}}(\boldsymbol{\varphi})$ the classifier unitary. Then, to evaluate the class of $\ket{\psi(\boldsymbol{\theta})}$, we prepare $U_{\mathrm{class}}(\boldsymbol{\varphi}) \ket{\psi(\boldsymbol{\theta})}$, measure all qubits, and decide the class by majority of the qubits. The quantum circuit is schematically shown in Fig.~\ref{fig:classifier_scheme}. For the transverse-field Ising model, the classifier circuit $U_{\mathrm{class}}(\boldsymbol{\varphi})$ had the same depth as the VQE circuit (i.e.~four layers of the checkerboard ansatz).

Effectively, estimating the majority is equivalent to estimating the energy of the state $U_{\mathrm{class}}(\boldsymbol{\varphi}) \ket{\psi(\boldsymbol{\theta})}$ relative to the following Hamiltonian:
\begin{equation}
    \label{eq:h_vote}
    H_{\text{vote}} = \sum_{w(i) > n/2} \ket{i} \bra{i} + \frac{1}{2}\sum_{w(i) = n/2} \ket{i} \bra{i}.
\end{equation}
Here $w(i)$ denotes the Hamming weight of the basis state $\ket{i}$, i.e.~the number of ones it contains.


\begin{figure}
    \centering
    \includegraphics[width=0.7\linewidth]{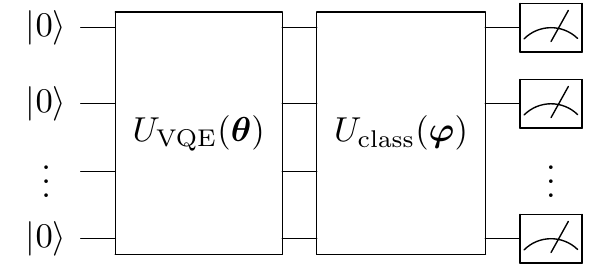}
    \caption{Quantum circuit that implements the classifier. The first part prepares the VQE solution, the second one performs the classification. The assigned label is inferred from the measurements in the $Z$ basis. Both $U_{\mathrm{VQE}}$ and $U_{\mathrm{class}}$ have the checkerboard structure. Reprinted from \cite{uvarov_machine_2020}.}
    \label{fig:classifier_scheme}
\end{figure}

To train the classifier circuit $U_{\mathrm{class}}(\boldsymbol{\varphi})$, we used the log-likelihood cost function. Let $\{ (\boldsymbol{\theta}_i, y_i) \}_{i=1}^{N_{train}}$ be the set of training data points and their labels, $y_i \in \{0, 1\}$. Let $p_i \in [0, 1]$ be the label predicted by the neural network:
\begin{equation}
    p_i = \bra{\psi(\boldsymbol{\theta}_i)} U^\dagger_{\mathrm{class}} (\boldsymbol{\varphi})H_{\text{vote}} U_{\mathrm{class}}(\boldsymbol{\varphi})\ket{\psi(\boldsymbol{\theta}_i)}.
\end{equation}
Then the loss function is:
\begin{equation}
\label{eq:logloss}
    f = -\sum_{i=1}^{N_{train}} \left( y_i \log p_i + (1 - y_i) \log (1 - p_i) \right).
\end{equation}

To minimize $f$, we used the simultaneous perturbation stochastic approximation (SPSA) algorithm \cite{spall_multivariate_1992}. This algorithm estimates the gradient vector by computing a finite difference in random direction, then performs a gradient descent step. We optimized the log loss over 300 epochs, with both finite differences step size and learning rate starting very coarse and decreasing as $1/\sqrt{n_{epoch}}$, where $n_{epoch}$ is the epoch number.

The result of the classification is shown in Fig.~\ref{fig:phase_classification}, left. The horizontal axis depicts the true value of $h$, while the vertical axis shows the label predicted by the classifier. In this case, the classification problem was quite simple, yielding a high accuracy score (97\%). The TFI model, however, can be easily partitioned without any machine learning: the value of total magnetization along the $X$ axis is also a good predictor of the phase of the system: with the increase of $h$, the magnetization gradually increases and cusps around the phase transition point (see Fig. \ref{fig:x_classifies_phases}). For this reason, we also performed the same analysis for a more complicated model.

\begin{figure}
    \centering
    \includegraphics[width=0.7\linewidth]{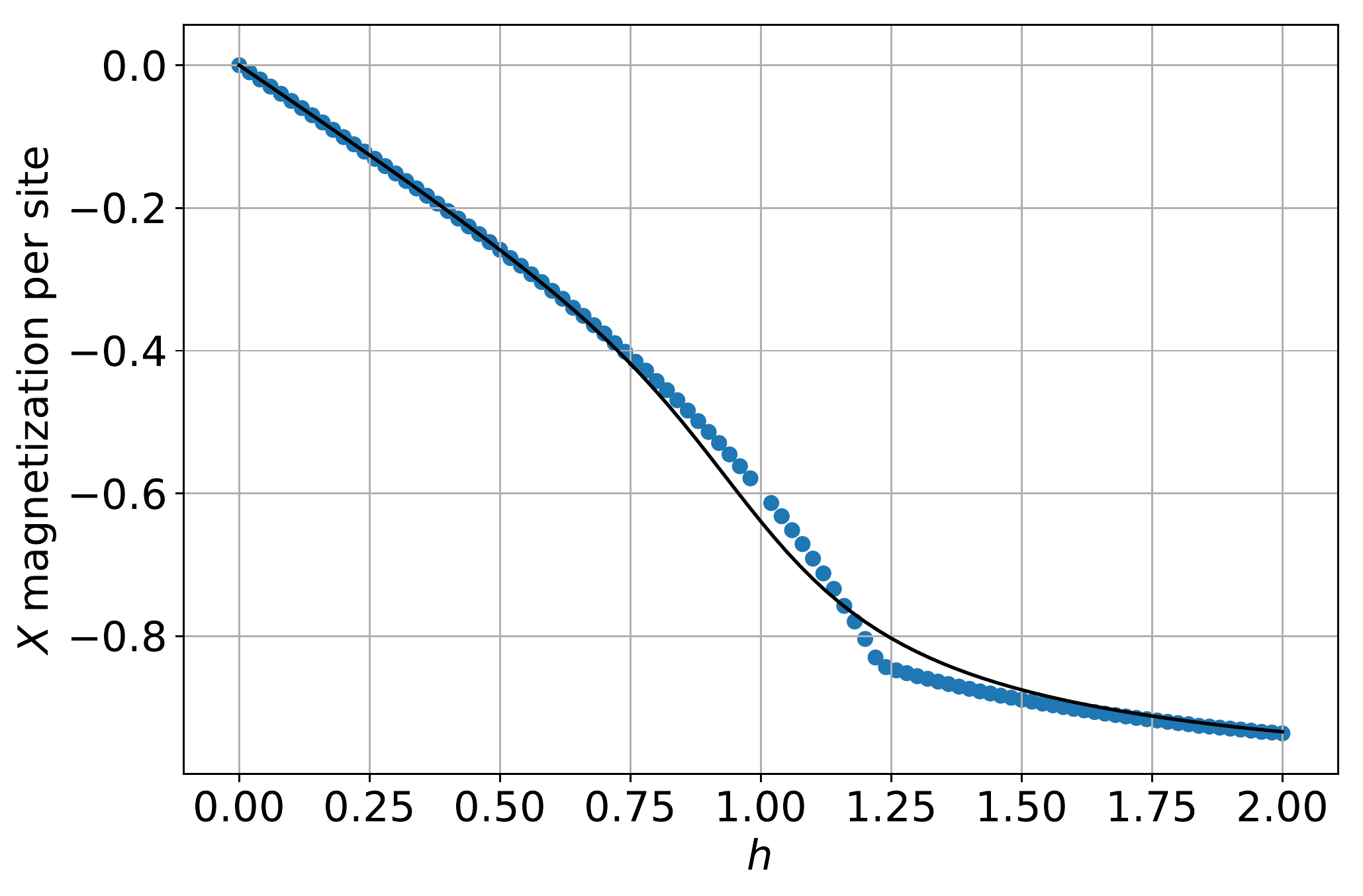}
    \caption{Total magnetization of the ground state as a function of the transverse field $h$. Solid line: exact ground states, markers: VQE solutions.}
    \label{fig:x_classifies_phases}
\end{figure}

\paragraph{Heisenberg $XXZ$ model.} Recall that the Heisenberg $XXZ$ model is described by the following Hamiltonian (\ref{eq:heisenberg_xxz}):

\begin{equation}
H = \sum_{i=1}^n \left[J_\perp\left(X_i X_{i+1} + Y_i Y_{i+1}\right)
    + J_z Z_i Z_{i+1}\right].
\end{equation}

From a physical perspective, Eq.~(\ref{eq:heisenberg_xxz}) corresponds to a uniform exchange coupled system with a uniaxial anisotropy specified by $J_z$. At $|J_z| < J_\perp$, this model is in the XY, or planar, phase which is characterized by algebraic decay of equal-time spin-spin correlation functions. In the regime $J_z > J_\perp$ the Hamiltonian corresponds to the antiferromagnetic Ising state. The system undergoes a Berezinsky--Kosterlitz--Thouless type phase transition at $J_z = J_\perp$   \cite{franchini_introduction_2017}. At the phase transition point, the ground state has the highest nearest-neighbour concurrence and a cusp in nearest-neighbour quantum discord \cite{dillenschneider_quantum_2008}.

For this model, there is some rotational symmetry that preserves the energy of the state. Namely, $H$ commutes with the operator of rotating each qubit around the $Z$ axis by an angle $\varphi$. In addition, $H$ commutes with all-qubit spin flip operators $X^{\otimes n}$ and $Z^{\otimes n}$. This enables us to perform a procedure of data augmentation by applying these operators to the approximate ground states. The new states obtained this way will be just as good in terms of the energy error. Recall that the ansatz we used consists of two-qubit entangler gates shown in Fig. \ref{fig:entangler}:
\begin{equation*}
    \mbox{
        \Qcircuit @C=1.0em @R=1.0em {
               & \gate{e^{-i \tilde{\theta}_1 X}} & \multigate{1}{e^{-i \tilde{\theta}_3 {Z} \otimes {Z}}} & \gate{e^{-i \tilde{\theta}_4 Z}} & \qw \\
               & \gate{e^{-i \tilde{\theta}_2 X}} & \ghost{e^{-i \tilde{\theta}_3 Z \otimes Z}} & \gate{e^{-i \tilde{\theta}_5 Z}} & \qw \\
           }
        }
\end{equation*}
The $Z$ rotation to all qubits can easily be applied by increasing all angles in the last layer of the operators by $\varphi$. The $X$ flip can be performed by inverting the angles of the $Z$ rotations in the last layer (because $Z$ and $X$ anticommute) and increasing the angles of the $X$ rotations by $\pi / 2$. Note that the rotation gates are often written with the angle divided by half, e.g.~$R_y = \exp(-\rmi \frac{\theta}{2} Y)$. In this case, the parameters should instead be incremented by $\pi$.

The result is depicted in in Fig.~\ref{fig:phase_classification}, right. Compared to the TFI model, we had to increase the depth of the classifier circuit to 6 layers. The resulting plot is much more jagged, and the accuracy is somewhat lower (93\%). The reason for the points being spread around is likely the data augmentation procedure. However, without data augmentation, the data points do not represent the state space accurately. The reason for that is the usage of AAVQE, which forces the solution points for neighboring values of $h$ to be close to each other, despite the freedom granted by the rotational symmetry.

\begin{figure}
    \centering
    \begin{subfigure}{.48\linewidth}
        \centering
        \includegraphics[width=\textwidth]{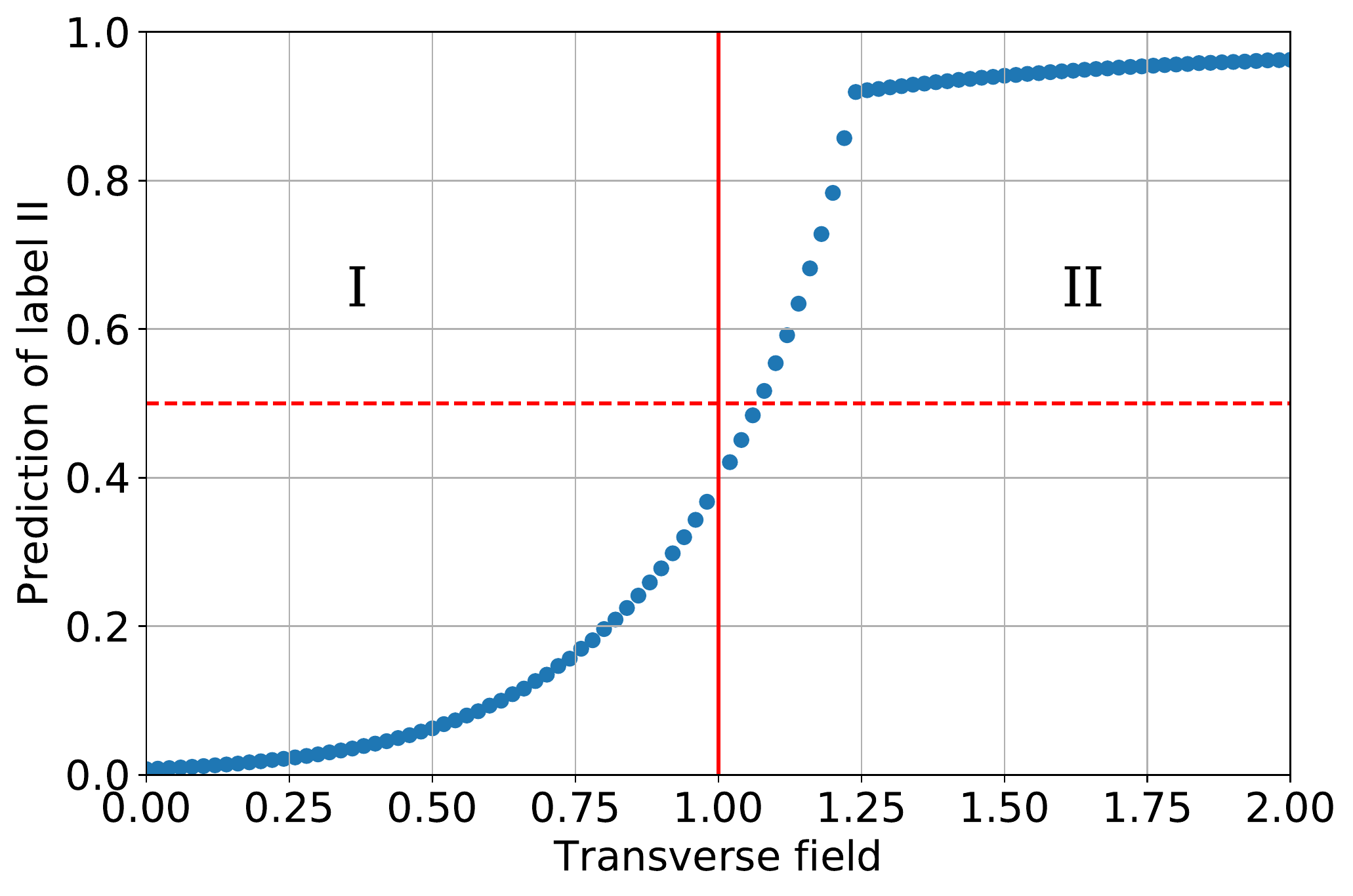}
    \end{subfigure}\begin{subfigure}{.48\linewidth}
        \centering
        \includegraphics[width=\textwidth]{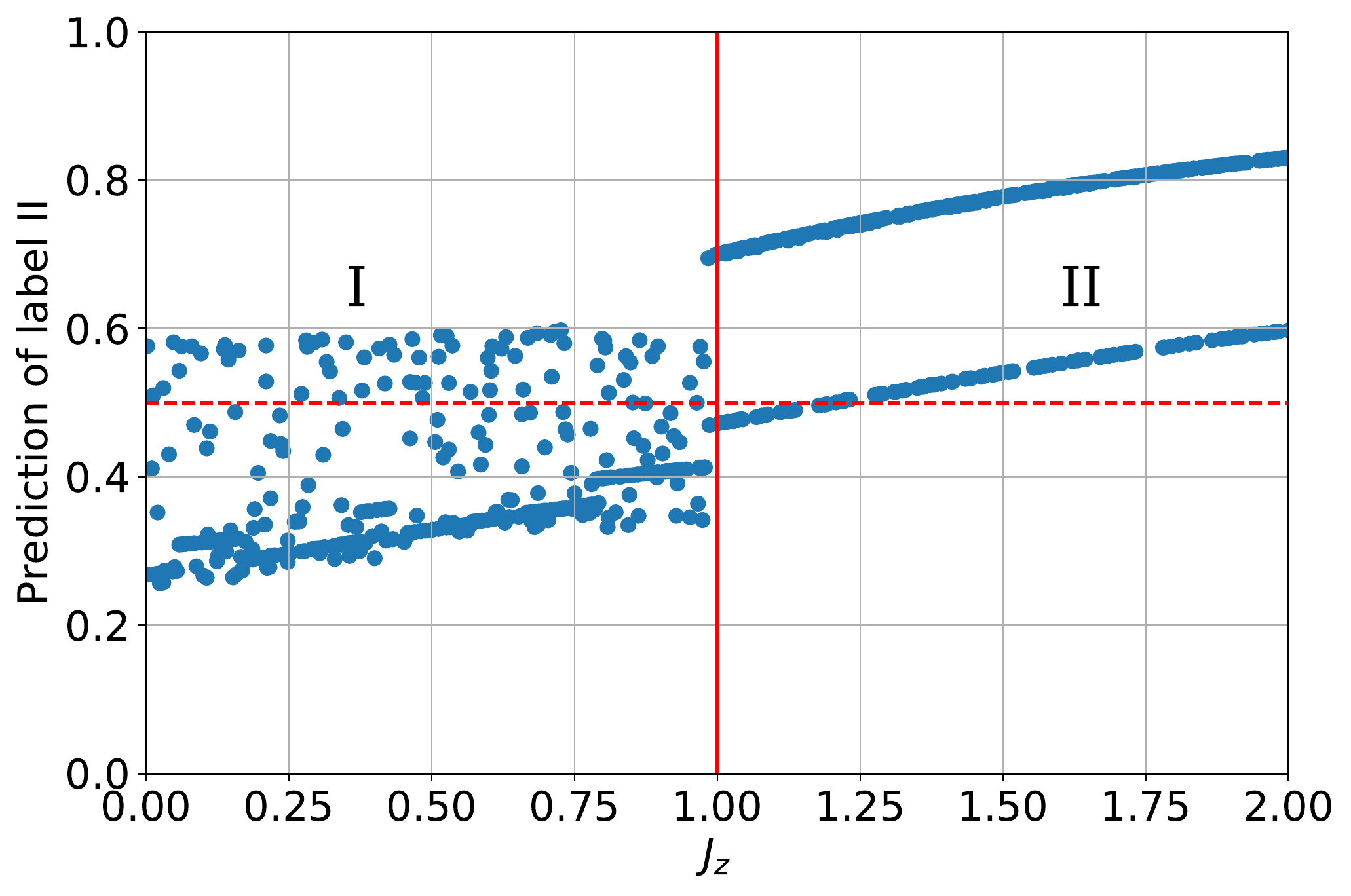}
    \end{subfigure}
    \caption{Left: predicted label of phase II as a function of magnetic field for transverse field Ising model. Right: predicted label of phase II as a function of $J_z$ for the XXZ model. Roman numbers denote the phases I and II of the models.}
    \label{fig:phase_classification}
\end{figure}

\paragraph{Random Hamiltonians.} Simpler toy models, like the TFI model, may have simple classification criteria which do not require application of machine learning. Here we classify the solutions of a randomized model: $H(\alpha) = (1 - \alpha) H_1 + \alpha H_2, \ \alpha \in [0, 1]$, where $H_1$ and $H_2$ are random Hermitian matrices sampled from the Gaussian unitary ensemble. We split solutions in two classes: (i) $\alpha < 0.5$ and (ii) $\alpha > 0.5$. Then we run the optimization routine to train the learning circuit to discern between the two classes. 

The approach was tested for 6 qubits, $n=100$, where $n_\text{train} = 70$, $n_{\text{test}} = 30$. The depth of the VQE circuit and the classifier were both set to four layers. The results are shown in Fig.~\ref{fig:learning_random_hams}. For this configuration, the accuracy of $93 \%$ was reached. This shows that the algorithm works even for a problem where there are no simple physically-motivated criteria. Naturally, having some structure compatible with the circuit topology would greatly benefit the convergence for larger problems.

\begin{figure}
    \centering
    \includegraphics[width=0.7\linewidth]{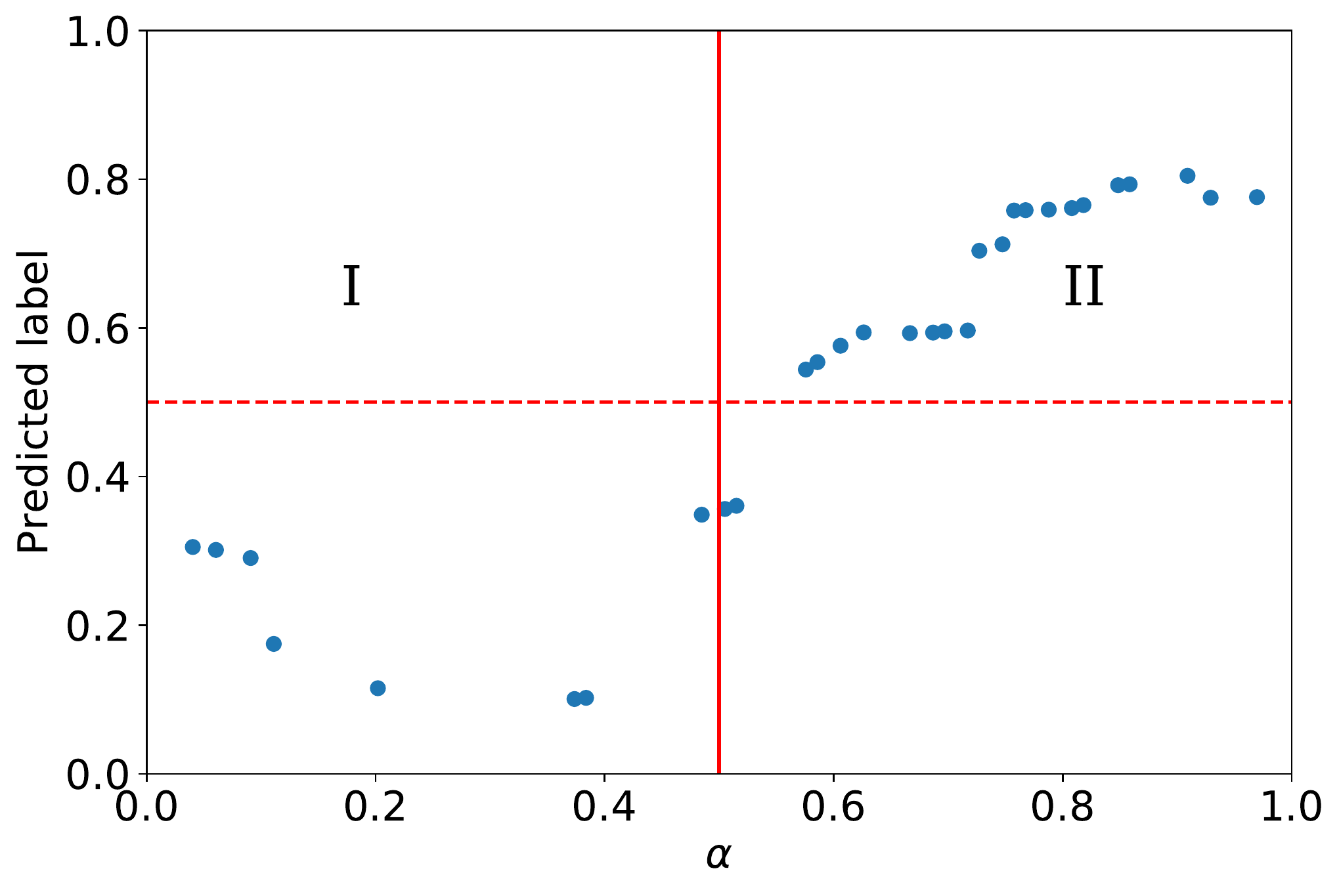}
    \caption{Results of learning on the random Hamiltonians model. Roman letters denote label classes. Reprinted from \cite{uvarov_machine_2020}.}
    \label{fig:learning_random_hams}
\end{figure}

\section{Discussion}

\subsection{Learning by confusion}
In this example, we knew the location of the phase transition point all along. But what if the task is to actually locate this point? This is also possible. The idea is as follows: let us pick a random value $h^*$, partition the data points across this point, and train the classifier. If the value $h^*$ is very far from all data points, all points will belong to one class, and the classifier will be able to partition them with 100\% accuracy. If $h^*$ does split the data points nontrivially, but is in the wrong location, then some points will be very close, but in the different classes, which will lead to a suboptimal training accuracy. Finally, the accuracy will be maximal when $h^*$ is the correct value. This behavior was observed for a classical neural network learning on spin states \cite{van_nieuwenburg_learning_2017}.

We reproduced the same experiment for our quantum neural network for the task of classifying the ground states of the TFI model. The result is shown in Fig. \ref{fig:learning_by_confusion}. Surprisingly enough, even though the accuracy does peak around the correct value of $h$, the characteristic W-shaped curve seen in \cite{van_nieuwenburg_learning_2017} is not observed. 

\begin{figure}
    \centering
    \includegraphics[width=0.7\linewidth]{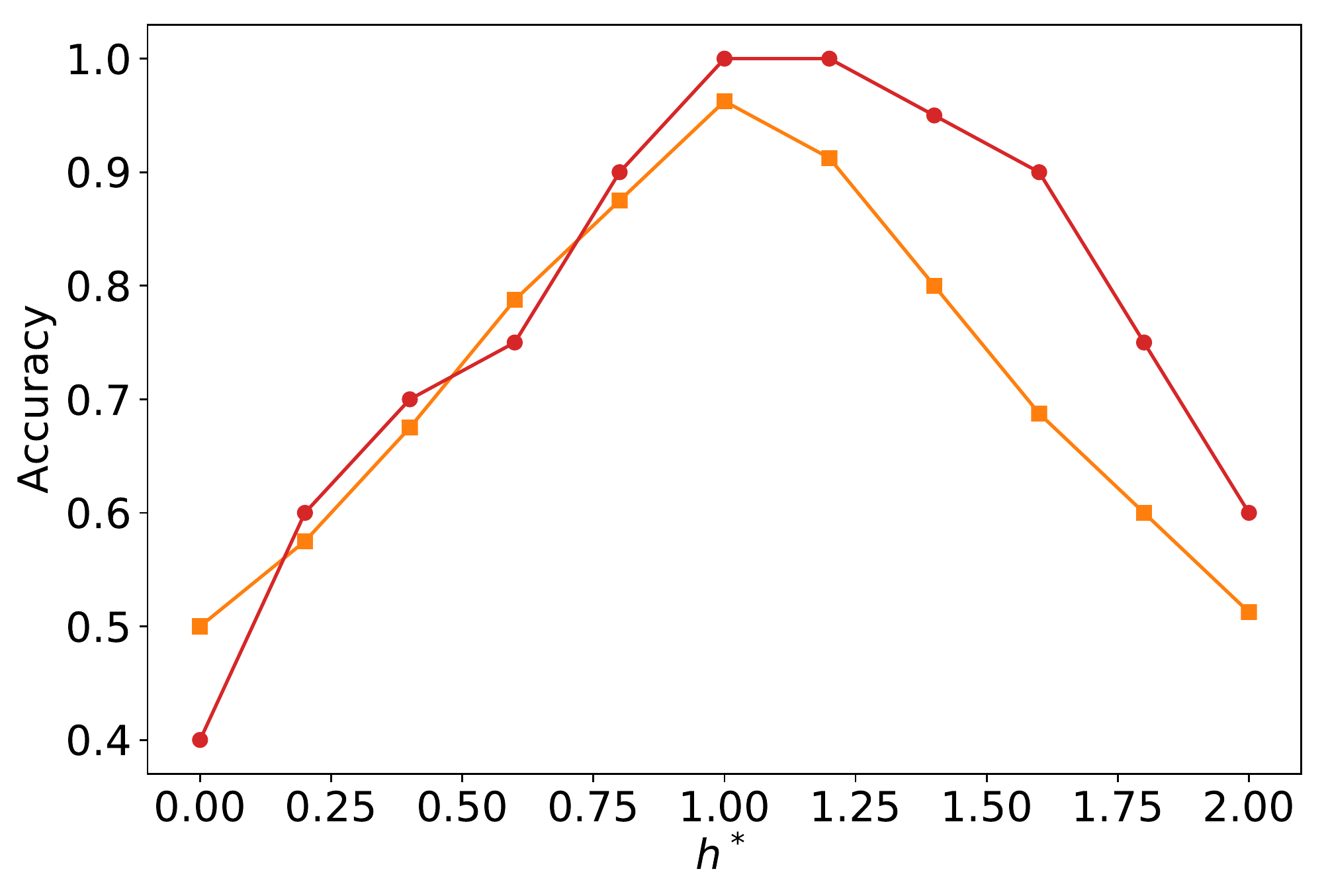}
    \caption{Train (squares) and test (circles) accuracy of the classifier for data marked with proposed threshold $h^*$.}
    \label{fig:learning_by_confusion}
\end{figure}

\subsection{Possible choices for the labeling procedure}

The unitary classifier circuit $U_{\mathrm{class}}(\boldsymbol{\varphi})$ by itself takes an input state and produces an output state. Then the measurements translate them into a probability distribution on the values of the output bits. The choice how to translate this distribution into a final label also has substantial freedom. In the pioneering work on such classifiers \cite{schuld_circuit-centric_2020}, the authors simply measure the output of the first qubit. In this situation, the label is assigned by the following rule:

\begin{equation}
    p_i = \frac{1}{2}\bra{\psi(\boldsymbol{\theta}_i)} U^\dagger_{\mathrm{class}} (\boldsymbol{\varphi})(1 + Z_1) U_{\mathrm{class}}(\boldsymbol{\varphi})\ket{\psi(\boldsymbol{\theta}_i)}.
\end{equation}

While this is a straightforward option, it does introduce some asymmetry in the model, which can only recovered by a long classifier circuit. For example, for circuit of depth $O(1)$, most qubits (and hence most of the input data) would not affect the measurement outcome in any way. To deal with this asymmetry, we proposed the majority-vote classifier described earlier.

New evidence suggests that the choice of the method to translate measurement results into labels may affect the trainability of the circuit. Such a choice can be translated to a Hamiltonian, and the locality of the latter may induce the barren plateaus in the optimization landscape \cite{uvarov_barren_2021,cerezo_cost-function-dependent_2020}. This locality dependence will be studied in more detail in Chapter \ref{chap:plateaus}. Still, the majority vote classifier is not necessarily a bad choice. It's true that the Hamiltonian $H_\text{vote}$ (\ref{eq:h_vote}) is highly non-local, but on the other hand it contains an exponential number of terms, which may offset the barren plateaus effect. Besides, since we count the number of ones in the output state, even flipping one qubit has a substantial effect on the output.


\section{Conclusions}

In this chapter, we gave an overview of methods to transfer the idea of neural networks to quantum computing. We proposed and numerically implemented a quantum classifier and trained it on quantum states constructed by VQE. 
It is a nontrivial fact that the Ising model required fewer layers than the XXZ model. In the transverse field Ising model, the magnetization $\sum \langle \sigma_x^{(i)} \rangle$ as a function of magnetic field clearly points at the location of the phase transition points. This implies that the phases of the model are easy to classify. In the XXZ model, the transition at $J_z=1$ is a transition between a paramagnetic and an antiferromagnetic phase \cite{franchini_introduction_2017}. Neither of these phases shows spontaneous magnetic moment in absence of an external field, making it somewhat harder to discern the two phases. 
The proposed classification technique can be applied to any model that can be expressed as a spin model (e.g.~fermion problems can be mapped to spin problems by using Jordan--Wigner transformation or Bravyi--Kitaev transformation).

In addition, we recreated an experiment for learning by confusion \cite{van_nieuwenburg_learning_2017} in the quantum setup. We found that this technique does pinpoint the location of the phase transition, but the shape of the curve is substantially different from the classical. We hypothesize that in this situation the classifier just did not have enough expressive power to mark all states with the same label.

\chapter{Barren plateaus in variational algorithms}
\label{chap:plateaus}

In this chapter, we discuss the phenomenon of barren plateaus in VQAs. As the name suggests, the idea is that the optimization landscape of VQAs at randomly chosen points can often resemble a flat space with no good directions of search. This behavior of quantum circuits is in stark contrast with classical deep neural networks, where overparametrized neural networks are surprisingly good at finding good minima. We begin with technical details about random unitary operators and averaging over such operators. Then we move on to show how barren plateaus appear in generic overparametrized circuits. Finally, we report on our results regarding the barren plateaus in quantum circuits constructed out of small parametrized blocks, which are common in VQAs.

\section{Integration with respect to the Haar measure}

A typical ansatz quantum circuit is a large and complicated structure which is difficult to analyze. One way to study its properties is to analyze its average behavior with respect to sampling a random point in the parameter space. Such a sampling defines a probability measure on the unitary group $\mc{U}(d)$\footnote{To specify the probability measure, you need the space of elementary outcomes $\Omega$ -- in our case, the group $\mc{U}(d)$ -- and a certain algebra of its subsets called the Borel $\sigma$-algebra $\mathrm{Borel} (\Omega)$. A probability measure $P$ then has to map the subsets from $\mathrm{Borel} (\Omega)$ to $[0, 1]$ in a way that is (i) countably additive (ii) zero on the empty set and (iii) normalized: $P(\Omega) = 1$.}. 
It turns out that for many purposes this probability measure can be approximated by the probability measure on $\mc{U}(d)$ called the Haar measure.

\begin{definition}
    The \textit{Haar measure} $\mu: \mathrm{Borel} (\mc{U}(d)) \rightarrow [0, 1]$ on the unitary group $\mc{U}(d)$ is the unique left- and right-invariant probability measure on that group. That is, let $V \in \mc{U}(d)$ and let $\mathcal{A} \in \mathrm{Borel} (\mc{U}(d))$. Then $\mu(\mathcal{A}) = \mu(V \mathcal{A}) = \mu(\mathcal{A} V)$.
\end{definition}

The proof that such measure is unique can be found e.g.~in \cite{watrous_theory_2018}.

In what follows, we will need to evaluate certain integrals over the unitary group. The integrands are related to matrix multiplications involving unitary matrices. As such, they will have the form of polynomials over the entries of $U$ and $U^*$. 

The simplest integral of that form is $\int U^\dagger A U \mathrm{d} \mu$. In tensor network diagrams, it can be expressed as follows:

\begin{equation}
    \label{eq:uau_picture}
    \int U^\dagger  A U  
    \mathrm{d} \mu
    = \int 
    \adjustbox{raise=2.5pt}{
    \Qcircuit @C=1em @R=.7em 
    {& \gate{U^\dagger} &  \gate{A} 
    & \gate{U} & \qw
    }
    }
    \ \mathrm{d} \mu.
\end{equation}

Here a wire means the vector space $\mathbb{C}^d$, not just a single-qubit space. We will also need a second-order integral of that sort:

\begin{equation}
    \label{eq:uuabuu_picture}
    \int (U^\dagger \otimes U^\dagger) (A \otimes B) (U \otimes U) 
    \mathrm{d} \mu
    = \int 
    \adjustbox{raise=15pt}{
    \Qcircuit @C=1em @R=.7em 
    {& \gate{U^\dagger} &  \gate{A} 
    & \gate{U} & \qw
    \\
    & \gate{U^\dagger} &  \gate{B} 
    & \gate{U} & \qw
    }
    }
    \mathrm{d} \mu.
\end{equation}

For more generality, we can consider integrals that don't use any additional matrices in their construction and just instead consider the following integral:

\begin{equation}
    \label{eq:unitary_integral}
    \mathcal{I}_t = \int U^{\otimes t} \otimes (U^\dagger)^{\otimes t} \mathrm{d} \mu.
\end{equation}

Note that when the tensor power of $U$ is not equal to the tensor power of $U^\dagger$, this integral is zero. This can be seen from the translation invariance of $\mu$. Let $V = e^{\mathrm{i}\theta} I$, then
\begin{equation}
    \int U^{\otimes t} \otimes (U^\dagger)^{\otimes t'} \mathrm{d} \mu
    = \int (VU)^{\otimes t} \otimes (U^\dagger V^\dagger)^{\otimes t'} \mathrm{d} \mu
    = e^{\mathrm{i}\theta (t - t')} \int U^{\otimes t} \otimes (U^\dagger)^{\otimes t'} \mathrm{d} \mu,
\end{equation}
which is possible only if $t = t'$ or if both sides of the equation are zero.

A tensor network diagram for this value is this:

\begin{equation}
    \label{eq:utut_picture}
    \mc{I}_t
    = \int 
    \adjustbox{valign=m}{
    \Qcircuit @C=1em @R=.7em 
    {& \gate{U} & \qw \\
    & \dots \\
    & \gate{U} & \qw \\
    & \gate{U^\dagger} & \qw \\
    & \dots \\
    & \gate{U^\dagger} & \qw \\
    }
    }
    \ \mathrm{d} \mu.
\end{equation}

The translation invariance of the Haar measure is the feature that enables us to calculate these integrals \cite{samuel_u_1980,collins_integration_2006}. Indeed, from the (left and right) translation invariance it follows that for any $V \in \mc{U}(d)$, we have that 
\begin{align}
    \label{eq:haar_translation}
    (V^{\otimes t} \otimes \id^{\otimes t}) \mathcal{I}_t
    (\id^{\otimes t} \otimes (V^\dagger)^{\otimes t})  & = \mathcal{I}_t, \\
    (\id^{\otimes t} \otimes V^{\otimes t}) \mathcal{I}_t
    ((V^\dagger)^{\otimes t} \otimes \id^{\otimes t})  & = \mathcal{I}_t. 
\end{align}

In particular, the translation invariance implies \cite{samuel_u_1980} that $\mc{I}_t$ can only be a linear combination of permutations of tensor factors:

\begin{equation}
    \label{eq:haar_is_permutations}
    \mc{I}_t
    = \sum_{\sigma_A, \sigma_B \in S_t} C_{\sigma_A, \sigma_B}
    = \sum_{\sigma_A, \sigma_B \in S_t} C_{\sigma_A, \sigma_B}
    \adjustbox{raise=-40pt}{
    \includegraphics[width=0.25\linewidth]{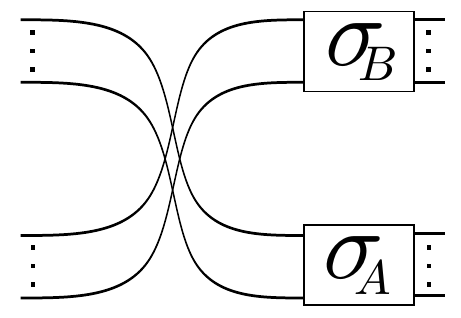}}
\end{equation}

The coefficients themselves are more difficult to obtain. We will restrict our attention to $t = 1$ and $t = 2$. The permutation group $S_1$ is trivial, so $\mathcal{I}_1 = c \cdot \mc{S}$. This implies that $\int U^\dagger A U \mathrm{d} \mu = (c \operatorname{Tr} A) \cdot I$:
\begin{equation}
    \label{eq:one_design_diagrams}
    \int 
    \adjustbox{raise=-3.5pt}{
    \includegraphics[width=0.15\linewidth]{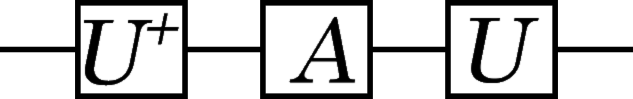}}
     \ \mathrm{d}\mu
    = \int 
    \adjustbox{raise=-8pt}{
    \includegraphics[width=0.1\linewidth]{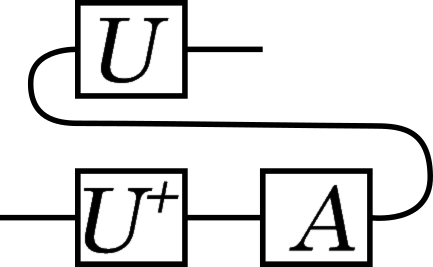}} \ \mathrm{d}\mu
    = c \ \adjustbox{raise=-8pt}{\includegraphics[width=0.12\linewidth]{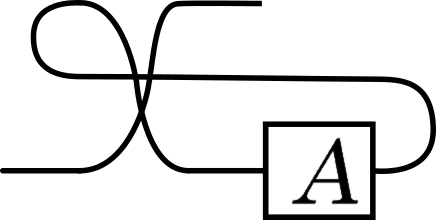}}
    = c \ \adjustbox{raise=-7pt}{\includegraphics[width=0.12\linewidth]{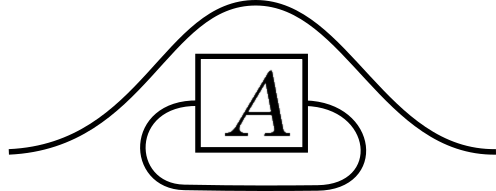}}
\end{equation}
The translation invariance implies that applying this unitary averaging twice will yield the same effect: $\int V^\dagger  U^\dagger A U V \mathrm{d}U \mathrm{d}V = \int U^\dagger A U \mathrm{d}U$. On the other hand, if we use (\ref{eq:one_design_diagrams}), we conclude that $c^2 \Tr \id \Tr A = c \Tr A$, which leads to $c = 1 / \Tr \id = 1/d$.

The case of $t = 2$ is somewhat more complicated. The group $S_2$ has two components, so the summation over $\{\sigma_A \times \sigma_B | \sigma_A, \sigma_B \in S_2\}$ has four terms:

\begin{equation}
    \mathcal{I}_2   
        = c_{II}\adjustbox{raise=-12pt}{
            \includegraphics[width=0.12\linewidth]{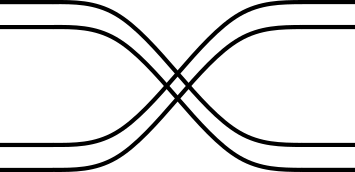}}
        + c_{IS}\adjustbox{raise=-12pt}{
            \includegraphics[width=0.12\linewidth]{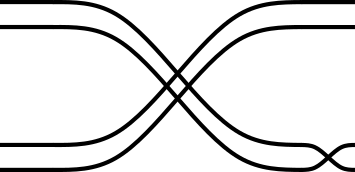}}
        + c_{SI}\adjustbox{raise=-12pt}{
            \includegraphics[width=0.12\linewidth]{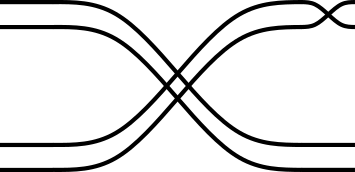}}
        + c_{SS}\adjustbox{raise=-12pt}{
            \includegraphics[width=0.12\linewidth]{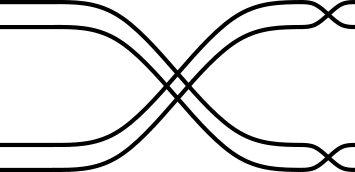}}.
\end{equation}

Now if we want to integrate $(U^\dagger \otimes U^\dagger) (A \otimes B) (U \otimes U)$ over $U$, we obtain the following formula:

\begin{multline}
    \label{eq:integrate_t2_no_coeffs}
    \int (U^\dagger \otimes U^\dagger) (A \otimes B) (U \otimes U) \mathrm{d} \mu = c_{II} (\Tr A \Tr B) \id \otimes \id +
    \\ + c_{SI} (\Tr A \Tr B) \mathcal{S} 
    + c_{IS} (\Tr AB) \id \otimes \id 
    + c_{SS} (\Tr AB) \mathcal{S}.
\end{multline}

The coefficients $c_{\sigma_A \sigma_B}$ here are no longer easy to calculate: the same trick as for $t=1$ will yield a complicated system of quadratic equations. One trick is to consider the shifts by unitaries of the type $\exp(\rmi \epsilon X)$ and take the derivative w.r.t.~$\epsilon$. This method, along with the graphical calculus, yields the result for $t=2$ \cite{poland_no_2020}.

The solution for arbitrary $t$ can be found using the representation theory of $S_n$ and $\mc{U}(d)$. Collins and \'Sniady \cite{collins_integration_2006} show that $c_{\sigma_A \sigma_B}$ are calculated using so-called Weingarten functions that depend on the characters of $\sigma_A \sigma^{-1}_B$ in different representations, and therefore only on its conjugacy class:
\begin{equation}
    c_{\sigma_A \sigma_B} = \operatorname{Wg} (\sigma_A \sigma^{-1}_B) = \frac{1}{t!^2} \sum_{\lambda \vdash t} \frac{\chi_\lambda(e)}{s_{\lambda, d} (1, ..., 1)}\chi_\lambda (\sigma_A \sigma^{-1}_B).
\end{equation}
Here $\chi_\lambda$ is the character of the representation $S_\lambda$, and $s_{\lambda, d}$ is a Schur polynomial in $d$ variables (like the irreducible representations of $S_n$, the Schur polynomials are indexed by Young diagrams).
With this equation, one can obtain the final formulas for $t=1$:
\begin{equation}
    \label{eq:haar_integral_1}
    \int U^\dagger A U \mathrm{d} \mu = \frac{\operatorname{Tr} A}{d} \cdot \id
\end{equation}
and $t=2$:
\begin{multline}
    \label{eq:haar_integral_2}
     \int (U^\dagger \otimes U^\dagger) (A \otimes B) (U \otimes U) \mathrm{d} \mu = \\
      \frac{1}{d^2 - 1} \left[  
         (\Tr A \Tr B  - \frac{1}{d} \Tr AB) \mathbbm{1} \otimes  \mathbbm{1} + (\Tr AB  - \frac{1}{d} \Tr A \Tr B) \mc{S} \right].
\end{multline}
Here $\mc{S}$ denotes a swap of two tensor components: $\mc{S} (v \otimes w) = w \otimes v$.

\section{Unitary \emph{t}-designs}

Random parametrized quantum circuits, such as the ones used in VQE, are known to approximate the entire group of unitary operators $\mc{U}(2^n)$ in the following sense. When circuit is constructed gates randomly picked from a universal set of gates (meaning that it generates $\mc{U}(2^n)$), the distribution of the circuit approaches the Haar measure on the unitary group \cite{emerson_convergence_2005}. Because of this, in analyzing the behavior of circuits, one might be tempted to approximate them with random unitary operators. However, it is also known that sampling from the unitary group requires exponentially long circuits.

Still, there are distributions on the unitary group such that samples from such distributions are similar to those obtained from the Haar measure in the following sense.

\begin{definition}
    A probability distribution $\nu$ on the unitary group $\mc{U}(2^n)$ is a  \textit{unitary} $t$\textit{-design}
    if the expected value of any polynomial of power $t$ in the matrix elements of $U$ and $U^*$ with respect to $\nu$ is the same as that w.r.t.~the Haar measure on $\mc{U}(2^n)$.  
\end{definition}
If a unitary ensemble is a $t$-design, then it is obviously also a $(t-1)$-design. There are a few simple examples of $t$-designs:
\begin{enumerate}
    \item A random quantum circuit constructed by independently applying a random Pauli matrix (picked with equal probability from $\{\id, X, Y, Z\}$) to each qubit is a 1-design \cite{ambainis_private_2000}. To see this, observe that conjugation of a qubit by a random Pauli matrix replaces this qubit with an identity density matrix, exactly like the integration over the Haar measure did in (\ref{eq:one_design_diagrams}).
    \item A Clifford circuit uniformly picked from the Clifford group is a 3-design, but not a 4-design \cite{webb_clifford_2016,zhu_multiqubit_2017}.
\end{enumerate}

Another approximation still is to consider so-called \textit{approximate $t$-designs} and \textit{tensor product expanders}. There are many definitions for different purposes \cite{low_pseudo-randomness_2010}, but we picked those that we found the most convenient for a numerical experiment.

\begin{definition}
    An ensemble of random unitary gates $\nu$ is a $\lambda$-approximate tensor product expander (TPE) if $||\mathbb{E}_{Haar} (U^{\otimes t} \otimes (U^*)^{\otimes t}) - \mathbb{E}_\nu (U^{\otimes t} \otimes (U^*)^{\otimes t}) ||_p \leq \lambda$ for $p=\infty$. When such an equation holds for $p=1$, the ensemble $\nu$ is called a $\lambda$-approximate $t$-design.
\end{definition}


\begin{remark}
    There is a reason why the definition of TPE uses the letter $\lambda$. When the expected value $\mathbb{E}_\nu$ is treated as a quantum channel on $\mathbb{C}_d^{\otimes t}$ (in the sense of conjugating by random $U$), the number $\lambda$ provides an upper bound to its second eigenvalue. This also implies that if $\lambda < 1$, then composing $m$ copies of $\mathbb{E}_\nu$ will yield a $\lambda^m$-approximate TPE.
\end{remark}

Approximate $t$-designs are easier to come by than exact ones. In fact, if one constructs a quantum circuit out of Haar-random two-qubit gates, it will be an $\epsilon$-approximate $t$-design if it has a number of gates that scales polylogarithmically with $1/\epsilon$ and $t$ \cite{brandao_local_2016}. For qubits with connectivity arranged in a $D$-dimensional lattice, an approximate $t$-design appears for depth $\operatorname{poly}(t) \cdot n^{1/D}$ \cite{harrow_approximate_2018}.

\section{Barren plateaus}

We are now ready to formulate the barren plateaus phenomenon \cite{mcclean_barren_2018}. Informally, the observation is that for long enough quantum circuits, running VQAs might become time-inefficient because the derivative of the cost function being minimized will be exponentially small in the number of qubits. This observation assumes that the starting point for the VQA is chosen at random, and that the random selection of parameters leads to an ensemble of unitaries that can be described as an approximate 2-design.

Consider a parametrized quantum circuit in which we distinguish one gate: $U = U_A e^{-\mathrm{i} \theta F} U_B$, where $F$ is a Pauli string. Let our cost function be some local Hamiltonian $H = \sum c_i h_i$, where $h_i$ are Pauli strings, and $h_0 = I$. The energy to be minimized in VQE is then equal to 

\begin{equation}
    E = \bra{\psi_0} U_B^\dagger e^{\mathrm{i} \theta F} U_A^\dagger H U_A e^{-\mathrm{i} \theta F} U_B \ket{\psi_0}.    
\end{equation}

The energy derivative over $\theta$ is now equal to 

\begin{multline}
    \label{eq:partial_E}
    \partial_\theta E = \bra{\psi_0} U_B^\dagger e^{\mathrm{i} \theta F} (\mathrm{i} F) U_A^\dagger H U_A e^{-\mathrm{i} \theta F} U_B \ket{\psi_0} + \\
    +
    \bra{\psi_0} U_B^\dagger e^{\mathrm{i} \theta F} U_A^\dagger H U_A (-\mathrm{i} F) e^{-\mathrm{i} \theta F} U_B \ket{\psi_0} = \\
    = \mathrm{i} \bra{\psi_0} U_B^\dagger e^{\mathrm{i} \theta F}  [F, U_A^\dagger H U_A] e^{-\mathrm{i} \theta F} U_B \ket{\psi_0}.
\end{multline}

In this formula, $U_A$, $U_B$, and their Hermitian conjugates appear in the first power, and the expression for $(\partial_\theta E)^2$ would have all of them appear in the second power at most. The barren plateaus result assumes that $U_A$ and $U_B$ are chosen randomly from ensembles, either of which is a 2-design. This assumption is a good approximation for long parametrized quantum circuits \cite{brandao_local_2016}. 

\begin{proposition}
    Let $U_A$ or $U_B$ form a $1$-design. Then the expected value of $\partial_\theta E$ is equal to zero.
\end{proposition}
\begin{proof}
    If $U_A$ is a 1-design, then $\int U_A^\dagger H U_A \mathrm{d}\mu = C \operatorname{Tr} (H) I$, which commutes with every matrix. If $U_B$ is a 1-design, then $\partial_\theta E$ is proportional to the trace of $[F, \int U_A^\dagger H U_A \mathrm{d} U_A]$, which is equal to zero, since $\operatorname{Tr} [A, B] = \operatorname{\Tr} AB - \operatorname{\Tr} BA = 0$.
\end{proof}

\begin{theorem}[after \cite{mcclean_barren_2018}]
    \label{thm:mcclean}
    Let either $U_A$ or $U_B$ form a $2$-design. If $\operatorname{card} H \in \operatorname{poly}(n)$, and the Pauli coefficients $c_i$ are bounded by a constant, then the variance $\operatorname{Var} \partial_\theta E \in O(2^{-n})$.
\end{theorem}

\begin{proof}[Proof of Theorem \ref{thm:mcclean}]
    Recall that for any random variable $X$ the variance is equal to $\operatorname{Var} X = \mathbb{E} (X^2) - (\mathbb{E} X)^2$. However, since $\mathbb{E} \partial_\theta E = 0$, we only consider the expectation of the square:
    \begin{equation}
        \label{eq:varde_integral}
        \operatorname{Var} \partial_\theta E = \int \mathrm{d} U_A \mathrm{d} U_B
        (\bra{\psi_0} U_B^\dagger e^{\mathrm{i} \theta F}  [\mathrm{i} F, U_A^\dagger H U_A] e^{-\mathrm{i} \theta F} U_B \ket{\psi_0})^2.
    \end{equation}
    \textbf{1. $U_B$ is a 2-design.} We can integrate over $U_B$ using (\ref{eq:haar_integral_2}):
    \begin{equation}
        \int \mathrm{d} U_B 
        (U_B \ket{\psi_0} \bra{\psi_0} U_B^\dagger)^{\otimes 2}
        = \frac{1 - \frac{1}{d}}{d^2 - 1} (I \otimes I + \mc{S}_n).
    \end{equation}
    Here by $\mc{S}_n$ we mean the operator on $\mathbb{C}^{2^n} \otimes \mathbb{C}^{2^n}$ that swaps the copies: $\mc{S}_n (\ket{\phi} \otimes \ket{\zeta}) = \ket{\zeta} \otimes \ket{\phi}$. The dimension $d$ is henceforth equal to $2^n$. Denoting the prefactor $\frac{1 - \frac{1}{d}}{d^2 - 1}$ as $\alpha$, we obtain the following expression for the variance\footnote{Recall that $\bra{\psi}A \ket{\psi} = \Tr (A \ket{\psi} \bra{\psi})$.}:
    \begin{multline}
        \operatorname{Var} \partial_\theta E
        = \alpha \int \mathrm{d}U_A \Tr ([\mathrm{i}F, U_A^\dagger H U_A])^{\otimes 2} + \\
        + \alpha \int \mathrm{d}U_A \Tr ([\mathrm{i}F, U_A^\dagger H U_A][\mathrm{i}F, U_A^\dagger H U_A]).
    \end{multline}
    The first part of this expression is zero as the commutator is traceless. In the second part, expanding the commutator by definition and using the invariance of $\Tr$ to cyclic shifts, we obtain the following\footnote{Note that, because the identity matrix commutes with everything, we can without loss of generality assume that $c_0 = 0$.}:
    \begin{multline}
        \Tr ([\mathrm{i}F, U_A^\dagger H U_A][\mathrm{i}F, U_A^\dagger H U_A]) = -2 \Tr (F U_A^\dagger H U_A F U_A^\dagger H U_A) + \\
        + 2 \Tr (F^2 U_A^\dagger H^2 U_A).
    \end{multline}
    Since $F$ is a Pauli string, $F^2 = 1$. The second term then reduces to $2 \Tr H^2 = 2d \sum c_i^2$. As for the first term, denote $\tilde{H} := U_A^\dagger H U_A$, and denote $\tilde{H_c}$ the sum of terms in $\tilde{H}$ that commute with $F$ and $\tilde{H_a}$ the sum of terms that anticommute with $F$. Then
    \begin{equation}
        \Tr F \tilde{H} F \tilde{H} = \Tr \tilde{H_c} \tilde{H} - \Tr \tilde{H_a} \tilde{H},
    \end{equation}
    where we again used $F^2 = 1$. Now to bound this term, we will use the fact that with a scalar product $\Tr A^\dagger B$ the space of Pauli strings is a real Euclidean space (i.e.~Pythagoras theorem is applicable):
    \begin{align}
        |-2 \Tr F \tilde{H} F \tilde{H}| &= \left| -2 \Tr \tilde{H_c} \tilde{H} + 2 \Tr \tilde{H_a} \tilde{H} \right| \\ 
        &= \left| -2 \Tr \tilde{H_c} \tilde{H_c} -2 \Tr \tilde{H_c} \tilde{H_a} + 2 \Tr \tilde{H_a} \tilde{H_c}  + 2 \Tr \tilde{H_a} \tilde{H_a} \right|  \\
        &= \left| -2 ||H_c||^2 + 2 ||H_a||^2 \right| \\
        & \leq 2 ||H||^2 = 2 d \sum c_i^2.
    \end{align}
    Combining together all the factors, we obtain that 
    \begin{equation}
        \operatorname{Var} \partial_\theta E \leq 4 \alpha \sum c_i^2 d = \frac{4 \sum c_i^2}{d+1} \in O(2^{-n}).
    \end{equation}
    \textbf{2. $U_A$ is a 2-design.} We have to evaluate $\int [\mathrm{i}F, U^\dagger_A H U_A]^{\otimes 2}  \mathrm{d} U_A$. Equation \ref{eq:haar_integral_2} suggests that the integration will yield a term proportional to $I \otimes I$ -- which will vanish under the commutators -- and a term proportional to $\mc{S}_n$. To make sense of this, we will need to expand the commutators by definition, which will yield the following:
    \begin{equation}
        \int [\mathrm{i}F, U^\dagger_A H U_A]^{\otimes 2}  \mathrm{d} U_A
        = \frac{\Tr H^2 - \frac{1}{d} (\Tr H)^2}{d^2 - 1}\left(2(\mathrm{i}F \otimes \mathrm{i}F) \mc{S}_n + 2 \mc{S}_n \right).
    \end{equation}
    Substituting this into \ref{eq:varde_integral} will yield:
    \begin{multline}
        \operatorname{Var} \partial_\theta E
        = 2\frac{\Tr H^2 - \frac{1}{d} (\Tr H)^2}{d^2 - 1}
        \int \mathrm{d} U_B 
        \left(\Tr U_B \ket{\psi_0} \bra{\psi_0} U_B^\dagger \mathrm{i}F U_B \ket{\psi_0} \bra{\psi_0} U_B^\dagger \mathrm{i}F +  \right. \\
        + \left. \Tr U_B \ket{\psi_0} \bra{\psi_0} U_B^\dagger U_B \ket{\psi_0} \bra{\psi_0} U_B^\dagger \right).
    \end{multline}
    The first integrand is in $[-1, 0]$, the second integrand is equal to 1. The enumerator of the fraction is $O(d)$, hence the entire expression is in $O(2^{-n})$.
\end{proof}

\section{Locality dependence of barren plateaus}

In this section, we will discuss the situation when the entire ansatz circuit cannot be treated as a 2-design. However, we will make an assumption that the ansatz consists of smaller blocks that can be described as local 2-designs. In this situation, the key thing that influences the onset of barren plateaus is the locality of the operators comprising the cost function. This was first noted in \cite{cerezo_cost-function-dependent_2020}, but here we approach the problem in a slightly different way and do not impose any limitations on the way these blocks are placed within the circuit.

For this section, we will mostly consider the Heisenberg picture of VQE. We consider the Eq.~\ref{eq:varde_integral} in the Heisenberg picture, that is, we think of all operators as acting on $H \otimes H$, while the state $\ket{\psi_0}$ is kept fixed.

The key assumption that is made for analysis of local circuits is that the blocks comprising such circuits are local 2-designs. A block is nothing more than a set of adjacent gates considered together as a single gate. To avoid confusion, we will also demand that the depiction of a block fits entirely into some rectangle, and that said rectangle does not contain gates not belonging to the block.

\subsection{Mixer channels}

The calculation of variance involves many operations on two copies of the same Hilbert space, so we will often use pairs of Pauli strings $h \otimes h$.

\begin{definition}[Super Pauli strings]
If $h$ is a Pauli string, then we will call $h \otimes h$ the induced \emph{super Pauli string}. If a Pauli string acts on qubits labeled $1, 2, \dots, n$, then a super Pauli string acts on qubits labeled $1, 2, \dots, n, 1', 2', \dots, n'$.
\end{definition}

We will denote super Pauli strings as $(\sigma_1 \otimes ... \otimes \sigma_n)^{\otimes 2}$, omitting the tensor product $\otimes$ when 
there is no ambiguity. When necessary, we will mark the variables related to the second copy (on the reader's right) with an apostrophe.
For example, a Pauli string $h = X \otimes \mathbbm{1} \otimes \mathbbm{1}$ acts nontrivially on the first out of $n = 3$ qubits. A super Pauli string $h \otimes h = (X \otimes \mathbbm{1} \otimes \mathbbm{1})^{\otimes 2}$ acts nontrivially on qubits $1$ and $1'$.

\begin{definition}[Causal cone] 
    Let $U$  be an ansatz, and $h$ a Pauli string. 
    A gate (or a block of gates) $V$ is in the \emph{causal cone} $C(h, U)$ of $h$ under ansatz $U$, if that gate or block cannot be eliminated from the conjugate $U^\dagger h U$. We denote as $|C(h, U)|$ the support of this causal cone, i.e.~the number of qubits on which $U^\dagger h U$ can act nontrivially.
\end{definition}{}


For example, Figure \ref{fig:causal_cone} depicts a checkerboard ansatz \cite{uvarov_machine_2020}, or alternating layered ansatz \cite{cerezo_cost-function-dependent_2020} acting on six qubits and consisting of three layers. Relative to a Pauli string $\mathbbm{1} \otimes \mathbbm{1} \otimes \mathbbm{1} \otimes X \otimes \mathbbm{1} \otimes \mathbbm{1}$, the causal cone for this ansatz consists of blocks $G_1$, $G_2$, $G_3$, $G_4$, $G_5$, and $G_7$. The support of this causal cone consists of all six qubits.

\begin{figure}
    \centering
    \includegraphics[width=0.8\textwidth]{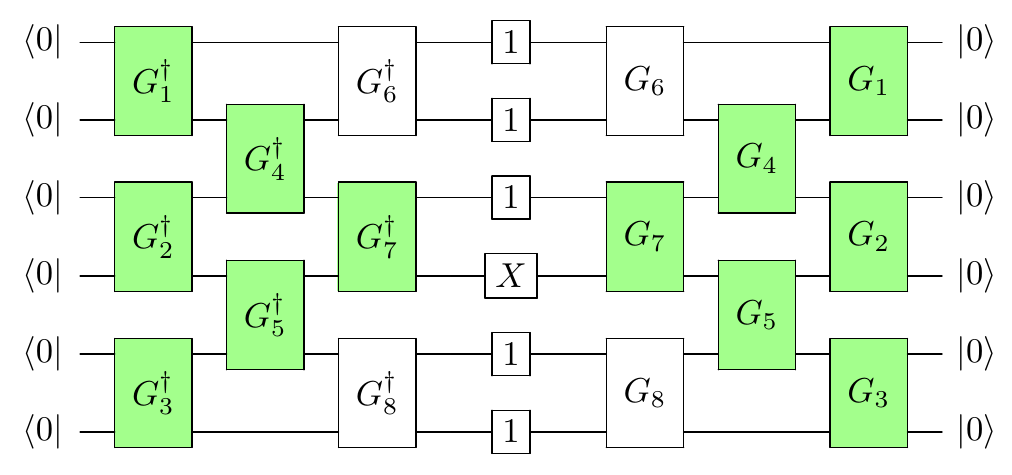}
    \caption{A causal cone of a Pauli string. Highlighted gates do not cancel in $U^\dagger h U$, where $U$ is the quantum circuit pictured. Reprinted from \cite{uvarov_barren_2021}.}
    \label{fig:causal_cone}
\end{figure}{}

\begin{definition}[Mixer]

    Let $\mc{Y}$ be a subset of the qubit registry. Define the \textit{local mixing operator}, or simply a \textit{mixer}\footnote{The mixer is in fact a quantum channel since it is defined by its own Kraus decomposition.}  $M_{\mc{Y}}: \operatorname{End} (\mc{H}\otimes \mc{H}) \rightarrow \operatorname{End} (\mc{H}\otimes \mc{H})$ as follows:
    \begin{equation}
    \begin{aligned}
        \label{eq:m2}
        & M_{ \mc{Y}} (h_1 \otimes h_2) = \int d\mu_{\mc{Y}} (U)
        (U^\dagger \otimes U^\dagger)
        (h_1 \otimes h_2)
        (U \otimes U),
    \end{aligned}{}
    \end{equation}{}
    where $\mu_{\mc{Y}}$ is the Haar distribution of unitaries acting nontrivially on $\mc{Y}$ and trivially on all other qubits.
    
\end{definition}

\begin{proposition}
    \label{prop:m2_decomposed}
    Let $h$ be a Pauli string. If its substring $h_{\mc{Y}}$ is nontrivial, then
    
    \begin{equation}
        M_{\mc{Y}}(h \otimes h) = \frac{1}{4^{|\mc{Y}|} - 1} \left( \sum_{\sigma_\mc{Y} \neq \mathbbm{1}} (\sigma_{\mc{Y}} \otimes h_{\mc{H} \setminus \mc{Y}})^{\otimes 2}\right),
    \end{equation}{}
    where the summation extends over all nontrivial Pauli substrings $\sigma_{\mc{Y}}$.
    Otherwise,  $M_{\mc{Y}}(h \otimes h) = h \otimes h$.
    
\end{proposition}{}

\begin{proof}
    We first apply formula (\ref{eq:haar_integral_2}):
    \begin{equation}
    \begin{aligned}
    \label{eq:m2_proof_1}
        & M_{\mc{Y}}(h \otimes h) = \\
        & = \frac{1}{4^{|\mc{Y}|} - 1} (h \otimes h)_{\mc{H} \setminus \mc{Y}} \otimes \left[  
        \left(\Tr (h \otimes h)_{\mc{Y}}  - \frac{1}{2^{|\mc{Y}|}} \Tr \mc{S}_{\mc{Y}} (h \otimes h)_{\mc{Y}} \right) \mathbbm{1} \otimes  \mathbbm{1} + \right. \\
        & \left. + \left(\Tr \mc{S}_{\mc{Y}} (h \otimes h)_{\mc{Y}}  - \frac{1}{2^{|\mc{Y}|}} \Tr (h \otimes h)_{\mc{Y}}\right) \mc{S}_\mc{Y} \right],
    \end{aligned}{} 
    \end{equation}
    where $\mc{S}_\mc{Y}$ is the swap operator permuting pairs of qubits $(i, i')$ for $i \in \mc{Y}$. It means that $\mc{S}_\mc{Y}$ is a tensor product of two-qubit swap gates $\mc{S}_2$. 
    First, note that $\Tr \mc{S}_{\mc{Y}} (h \otimes h)_{\mc{Y}} = 2^{|\mc{Y}|}$ and $\Tr (h \otimes h)_{\mc{Y}}$ is equal to zero for nontrivial $h_\mc{Y}$ and $4^{|\mc{Y}|}$ for a trivial substring. Next, $\mathcal{S}_2$ is decomposed as 
    \begin{equation}
    \label{eq:swap_decomp}
        \mathcal{S}_2 = \frac{1}{2}\left(X \otimes X + Y \otimes Y + Z \otimes Z + \mathbbm{1} \otimes \mathbbm{1} \right).
    \end{equation}
    Applying this decomposition to $\mc{S}_\mc{Y}$ yields a sum of all possible super Pauli strings (including the trivial string):
    \begin{equation}
        \mc{S}_{\mc{Y}} = \frac{1}{2^\mc{|Y|}} \sum_{\sigma_1, ..., \sigma_\mc{|Y|}}
        (\sigma_1 \otimes \sigma_2 \otimes ... \otimes \sigma_\mc{|Y|})^{\otimes 2}.
    \end{equation}
    Substituting this decomposition into \eqref{eq:m2_proof_1} we recover the desired result.
\end{proof}{}

This means that, after application of the local mixer channel $M_\mc{Y}$ to a super Pauli string $h \otimes h$, we forget the exact Pauli operators in the registry $\mc{Y}$, and replace them will a uniform linear combination of all possible nontrivial Pauli strings.For example, consider a single-qubit mixing channel $M_2$ acting on qubits $(2, 2')$.
Then $M_2 ((XX)^{\otimes 2}) = \frac{1}{3} ( (XX)^{\otimes 2} + (XY)^{\otimes 2} + (XZ)^{\otimes 2})$. On the contrary, a string $(XI)^{\otimes 2}$ will be invariant under the mixing channel: $M_2 ((XI)^{\otimes 2}) = (XI)^{\otimes 2}$.

Importantly, mixer channels do not accept Pauli strings of the type $h_1 \otimes h_2$, where $h_1 \neq h_2$. The following statement can also be proven by application of formula (\ref{eq:haar_integral_2}).

\begin{proposition}
    \label{prop:mixer_kills_asymmetry}
    Let $h_1, h_2$ be Pauli strings such that their restriction on a registry $\mc{Y}$ is different. Then $\mc{M_Y} (h_1 \otimes h_2) = 0$.
\end{proposition}

\subsubsection{Action of many mixer channels}

When we consider a circuit made of many blocks, we essentially model them with a sequence of local mixers. 

\begin{proposition}
    \label{prop:paulis_decouple}
    Let $\mc{Y}_1, ..., \mc{Y}_N$ be a collection of qubit subsets
    such that $\mc{Y}_1 \cup ... \cup \mc{Y}_N$ contains all $n$ qubits (the subsets are allowed to intersect). Let $h_1, h_2$ be two distinct Pauli strings. Then, 
    $M_{ \mc{Y}_N} \circ \dots \circ  M_{ \mc{Y}_1} (h_1 \otimes h_2) = 0.$
\end{proposition}{}
\begin{proof}
    Let the strings be different in position $k$. If $k \notin \mc{Y}_j$, then $M_{ \mc{Y}_j}$ will output a sum of Pauli strings $h_{1, j} \otimes h_{2, j}$ that still have different entries in positions $k$ and $k'$. If $k \in \mc{Y}_j$, the output will be zero by virtue of Proposition~\ref{prop:mixer_kills_asymmetry}. Since every qubit is contained in the support of some mixer, the latter must happen for some $j$.
\end{proof}

\begin{remark}
    The sum of coefficients of the Pauli strings is conserved by the mixers. However, suppose that the mixers act on all qubits in the support of some nontrivial Pauli string $h$, with $C$ being support the causal cone of $h$ under the mixers. After the action of the mixers, $h \otimes h$ is replaced by the sum of all possible super Pauli strings that are nontrivial on $|C|$ trivial on the complement $[n] \backslash |C|$, with the weights equal to $(4^{-|C|} - 1)$.
\end{remark}

\subsubsection{Commutator operator}

Equation (\ref{eq:varde_integral}) contains an expression depending on $([\rmi F, U^\dagger H U])^{\otimes 2}$ for some $U$, $H$, and $F$. For convenience we introduce a commutator operator $\mc{C}_{F} \in \operatorname{End}(\mc{H} \otimes \mc{H})  \rightarrow \operatorname{End}(\mc{H} \otimes \mc{H})$, which maps $A \otimes B$ to $[\rmi F, A] \otimes [\rmi F, B]$. Unlike the mixer operator, $\mc{C}_F$ is no longer a quantum channel since it does not preserve trace. Nonetheless, it is important for our purposes to consider its action on Pauli strings. We will henceforth assume that $F$ is a Pauli string. Graphically, we can express this superoperator like this:

\newcommand{\legw}{1}
\newcommand{\gapw}{2}
\newcommand{\gaph}{1}
\newcommand{\barh}{0.6}
\newcommand{\wireh}{0.2}
\newcommand{\wiregap}{0.5}
\newcommand{\wirel}{0.3}

\begin{equation}
    \mc{C}_F(\star) =
    \adjustbox{raise=-7pt}{
    \begin{tikzpicture}[thick,scale=0.5]
    \node[blank] at (\gapw/2 + \legw, \gaph/2) {$\star$};

    \draw 
    (0, 0) 
    -- (0, \gaph + \barh) 
    -- (\gapw + \legw + \legw,\gaph + \barh)
    -- (\gapw + \legw + \legw,0)
    -- (\gapw + \legw,0)
    -- (\gapw + \legw, \gaph)
    -- (\legw, \gaph)
    -- (\legw,0)
    -- (0, 0)
    
    (0, \wireh) -- (-\wirel, \wireh)
    (0, \wireh + \wiregap) -- (-\wirel, \wireh + \wiregap)
    
    (\legw, \wireh) -- (\legw + \wirel, \wireh)
    (\legw, \wireh + \wiregap) -- (\legw + \wirel, \wireh + \wiregap)
    
    (\legw + \gapw, \wireh) -- (\legw + \gapw - \wirel, \wireh)
    (\legw + \gapw, \wireh + \wiregap) -- (\legw + \gapw - \wirel, \wireh + \wiregap)
    
    (\gapw + \legw + \legw, \wireh) -- (\gapw + \legw + \legw + \wirel, \wireh)
    (\gapw + \legw + \legw, \wireh + \wiregap) -- (\gapw + \legw + \legw + \wirel, \wireh + \wiregap)
    
    ;
    \end{tikzpicture}
    }
\end{equation}

Unlike the mixer, $\mc{C}_{F}$ does not care if the Pauli strings in the tensor copies of $\operatorname{End}(\mc{H})$ are identical: $[\rmi F, A] \otimes [\rmi F, B]$ does not necessarily vanish for distinct $A, B$.

Any two Pauli strings either commute or anticommute with each other. Let $\boldsymbol{\sigma}_1$ and $\boldsymbol{\sigma}_2$ be Pauli strings on $\mc{H}$. Then $\mc{C}_F (\boldsymbol{\sigma}_1 \otimes \boldsymbol{\sigma}_2)$ does not vanish if and only if both $\boldsymbol{\sigma}_1$ and $\boldsymbol{\sigma}_2$ anticommute with $F$. 
A product of two Pauli strings is again a Pauli string, possibly multiplied by a power of $\rmi$. An operator $[\rmi F, \boldsymbol{\sigma}]$ is Hermitian and proportional to a Pauli string, so it is equal to $\pm 2 \boldsymbol{\eta}$ for some Pauli string $\boldsymbol{\eta}$. A remarkable property of $\mc{C}_F$ is that for super Pauli strings $\boldsymbol{\sigma} \otimes \boldsymbol{\sigma}$ such that $F$ anticommutes with $\boldsymbol{\sigma}$, the sign is canceled, so $\mc{C}_F(\boldsymbol{\sigma} \otimes \boldsymbol{\sigma}) = 4 \boldsymbol{\eta} \otimes \boldsymbol{\eta}$ for some $\boldsymbol{\eta}$.

We can now consider what happens when there are mixers before and after $\mc{C}_F$.


\begin{proposition}
    \label{prop:commutator_old}
    The following identities hold:
    \begin{enumerate}
        \item For every $F \in \mathrm{Herm}(\mc{Y})$, $\mc{C}_{F}(\id_{\mc{Y}} \otimes \id_{\mc{Y}})$ vanishes. Thus, $M_{\mc{Y}} \circ \mc{C}_F \circ M_{\mc{Y}} (\id_{\mc{Y}} \otimes \id_{\mc{Y}})=0$.
        \item Let $F$ be a nontrivial Pauli string acting on $\mc{Y}$. Then, for any nontrivial Pauli string $h$ acting on $\mc{Y}$
        \begin{equation}
            \label{eq:sandwiched_commutator}
             M_{\mc{Y}} \circ \mc{C}_{F} \circ M_{\mc{Y}} \left( h^{\otimes 2}\right) = \frac{2 \cdot 4^{|\mc{Y}|}}{4^{|\mc{Y}|} - 1} M\left( h^{\otimes 2}\right).
        \end{equation}
    \end{enumerate}{}
\end{proposition}{}
    
\begin{proof}
    The first part follows directly: identity operator commutes with any other operator.
    To prove the second part, we will sequentially apply the operators in the left-hand side of \eqref{eq:sandwiched_commutator}. First, the local mixing operator returns a linear combination of all nontrivial Pauli strings $\boldsymbol{\sigma}_i$: $M_{\mc{Y}} \left( h^{\otimes 2}\right) = 1/ (4^{|\mc{Y}|} - 1)\sum \boldsymbol{\sigma}_i^{\otimes 2}$. After applying $\mc{C}_{\mc{Y}}$ to each super Pauli string we either get zero for those commuting with $F$ and some other super Pauli string $\boldsymbol{\kappa}_i \otimes \boldsymbol{\kappa}_i$ multiplied by 4 for those anticommuting with $F$:
    \begin{equation}
        ([\rmi F, \boldsymbol{\sigma}_i])^{\otimes 2} = (\pm 2 \boldsymbol{\kappa}_i)^{\otimes 2} = 4 \boldsymbol{\kappa}_i^{\otimes 2}.
    \end{equation}

    For any nontrivial Pauli string $F$, there are exactly $4^{|\mc{Y}|} / 2$ nontrivial Pauli strings that anticommute with $F$. 
    Indeed, let $F$ contain $m$ nontrivial Pauli matrices and let $P$ be some Pauli string that we wish to construct, so that it anticommutes with $F$. How many ways of constructing $P$ are there? There must be an odd number of sites $j$ such that Pauli matrices $F_j$ and $P_j$ are both nontrivial and not equal to each other. We can pick such sites in $2^{m - 1}$ ways. Then, for each of these sites, there is a choice of 2 Pauli matrices not commuting with $F_j$. For all sites where $F_j$ is nontrivial, but which are not included in our selection, $P_j$ is either equal to $\id$ or to $F_j$. Finally, in all sites where $F_k = \id$, we are free to choose any Pauli matrix. Hence, when the choice of sites is fixed, we have $2^m 4^{|\mc{Y}| - m}$ options. Multiplying this by $2^{m - 1}$, we get $4^{|\mc{Y}|} / 2$.
    
    Overall, the result is the following: the first mixer produces a sum of all possible nontrivial super Pauli strings, the commutator $\mc{C}$ keeps $4^{|\mc{Y}|} / 2$ of them and multiplies them by 4, and then the second mixer again turns each super string into a sum of all possible super Pauli strings. Collecting the prefactors yields \eqref{eq:sandwiched_commutator}.

\end{proof}

\subsection{Main statement}

Now we are ready to formulate the main result of this section \cite{uvarov_barren_2021}:

\begin{theorem}
    \label{thm:block_plateaus}
    Let $H$ be an $n$-qubit Hamiltonian consisting of Pauli strings $h_i$: $H = \sum c_i h_i$ with finite $c_i \in \mathbb{R}$. Let the ansatz $U$ consist of $l$ layers, and denote $l_c$ the layer which contains the block $G_k$ depending on parameter $\theta_a$. 
    Let each block of the ansatz be an independently parametrized local 2-design. Let the block $G$ also be decomposable into $G = G_A e^{-i \theta_a F} G_B$, where $G_A$ and $G_B$ are local 2-designs not depending on $\theta_a$. Then, the variance of the gradient of $E$ with respect to that parameter is bounded below as follows:
    
    \begin{equation}
        \operatorname{Var} \partial_a E \geq \frac{2 \cdot 4^{|\mc{Y}_k|}}{4^{|\mc{Y}_k|} - 1} \left( \frac34 \right)^{l - l_c}  \sum_i c_i^2 \cdot 3^{-|C(h_i, U)|},
    \end{equation}{}
    where $|C(h_j, U)|$ is the number of qubits in the causal cone of the $j^{th}$ Pauli string, and the summation is over those Pauli strings whose causal cone contains the block $G$.
\end{theorem}{}

The proof of this theorem uses the fact that $\partial_a E (h_i)$ are uncorrelated random variables:

\begin{proposition}
\label{lemma:expectations_decouple}
In the conditions of Theorem \ref{thm:block_plateaus}, the individual Pauli string coefficients make independent contributions to the total variance:
\begin{equation}
    \operatorname{Var} \partial_a E (H) = \sum_i c_i^2 \operatorname{Var} \partial_a E (h_i).
\end{equation}
\end{proposition}

\subsection{Idea of the proof}

The proof is conceptually not difficult, but quite laborious. To estimate the variance, we will need to integrate $(\partial_aE) ^2$ over all possible assignments of the parameters $\boldsymbol{\theta}$. Since we assumed that the blocks are parametrized independently, we can also take their expected values independently. If $U_p$ depends on parameters $\boldsymbol{\theta}_j =  \theta_{j_1}, ..., \theta_{j_m}$ for some $j_1, ..., j_m$, then the operator we care about takes the following form:

\begin{equation}
\label{eq:true_mixer}
M = \int (U_p^\dagger \otimes U_p^\dagger) (\star) (U_p \otimes U_p) 
\mathrm{d} \boldsymbol{\theta}_j
= \int 
\adjustbox{raise=31pt}{
\Qcircuit @C=1em @R=.7em 
{& \multigate{1}{U^\dagger} &  \multigate{3}{\star} 
& \multigate{1}{U} & \qw
\\
& \ghost{U^\dagger} & \ghost{\star}
& \ghost{U} & \qw
\\
& \multigate{1}{U^\dagger} & \ghost{\star}
& \multigate{1}{U} & \qw
\\
& \ghost{U^\dagger} & \ghost{\star}
& \ghost{U} & \qw
}
}
\mathrm{d} \boldsymbol{\theta}_j
\end{equation}
where the star ($\star$) is a placeholder for a Hermitian operator on $\mc{H} \otimes \mc{H}$. The commutator found in \eqref{eq:partial_E} can also be viewed as a superoperator $[\rmi F, \star]^{\otimes 2}$. 

In this graphical language, the value of $\operatorname{Var} \partial_a E$ is expressed as a diagram shown in Fig.~\ref{fig:variance_as_diagram}.

Instead of evaluating the action of operators like $M$ for a specific ansatz, we instead assume that each block of the ansatz constitutes a local 2-design, in which case one can compute their action exactly. We will refer to such operators as ``mixing operators''.

The action $[\rmi F, \star]^{\otimes 2}$ can be written down explicitly using the assumption that there are two mixing operators around it. In which case, its role reduces to eliminating those Pauli strings that don't share support with $F$, and multiplying all other strings by a constant. We estimate the number of strings that survive this operation by tracing a path along the structure of the ansatz.

After all these operators, we end up with a number of Pauli strings with some coefficients. Taking the expectation w.r.t.\ the zero kets eliminates those strings that contain $X$ or $Y$ Pauli matrices. We estimate the share of Pauli strings that are not eliminated in the process. The sum of their coefficients is the final value that we are after.

\begin{figure}
    \begin{tikzpicture}
    \node[] (formula) {$\operatorname{Var} \partial_a E = \displaystyle{\int} \mathrm{d} \boldsymbol{\theta}$};
    \node (fig1) [above right of =formula, xshift=7cm, yshift=1cm] {\includegraphics[width=0.7\linewidth]{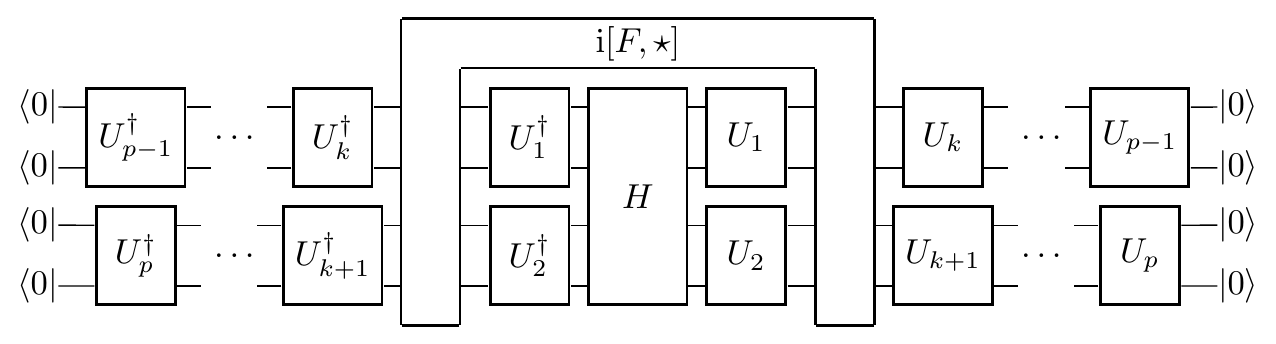}};
    \node (fig2) [below right of =formula, xshift=7cm, yshift=-0.6cm] {\includegraphics[width=0.7\linewidth]{figures/vardE.pdf}};
    \end{tikzpicture}

    \caption{Variance of the derivative of $E$, expressed as an integration over all possible assignments of $\boldsymbol{\theta}$. Reprinted from~\cite{uvarov_barren_2021}.
    }
    \label{fig:variance_as_diagram}
\end{figure}

In the following sections, we will first derive the necessary properties of the mixing operators and of the commutator-induced superoperator. Then we will use them to prove the main statement.

\subsection{Proof of the main statement}

Following the definitions, the variance can be expressed as an average of a certain operator over the zero ket vectors 
$\ket{\boldsymbol{00}} = (\ket{\boldsymbol{0}} \otimes \ket{\boldsymbol{0}})$. To write down that operator, we use the assumption that individual blocks $G_1, ..., G_q$ are local 2-designs, and replace the integration with local mixing operators $M_{\mc{Y}_1} ... M_{\mc{Y}_q}$:
\begin{equation}
\label{eq:all_mixers}
\Var \partial_a E (H) 
= \bra{\boldsymbol{00}}
M_{ \mc{Y}_1} \circ 
\dots \circ  M_{ \mc{Y}_k} \circ \mc{C} \circ M_{ \mc{Y}_k} \circ 
\dots \circ  M_{ \mc{Y}_q} (H \otimes H) \ket{\boldsymbol{00}}.
\end{equation}{}
Note the reverse order of the mixing operators: if the ansatz state is $\ket{\psi} = G_q ... G_1 \ket{\mathbf{0}}$, then $\bra{\psi} H \ket{\psi} = 
\bra{\mathbf{0}} G_1^\dagger ... G_q^\dagger H G_q ... G_1 \ket{\mathbf{0}}$. 
Therefore, in the Heisenberg picture, the conjugation with unitaries reverses.

The action of all operators in the right hand side of \eqref{eq:all_mixers}
is linear, so we can replace $H \otimes H$ by $\sum_{i, j}c_i c_j h_i \otimes h_j$. From Proposition \ref{prop:paulis_decouple} we know that all terms with $i \neq j$ will vanish, so we arrive at
\begin{equation}
    \label{eq:paulis_decouple}
    \Var \partial_a E (H) = \sum_i c_i^2 \Var \partial_a E (h_i),
\end{equation}
effectively reducing the problem to the case when the Hamiltonian of interest is a single Pauli string $h$.

In the next three subsections, we estimate the value of $\Var \partial_a E (H)$ by sequentially applying the superoperators shown in \eqref{eq:all_mixers}.

\subsubsection{First portion of mixing operators}

Let us now follow the evolution of some Pauli string $h_i$ along the application of the mixing operators before the commutation operator $\mc{C}$. Each mixing operator replaces the super Pauli string with the sum of super Pauli strings with all possible nontrivial substrings in its support. For example, a two qubit mixing operator takes one super Pauli string and returns 15 super Pauli strings (see Proposition \ref{prop:m2_decomposed}):

\begin{equation}
    M\left((X \otimes X)^{\otimes 2} \right) = \frac{1}{15}
    \left( \sum_{i,j} (\sigma_i \otimes \sigma_j)^{\otimes 2}
    + \sum_{i} \left((\sigma_i \otimes \mathbbm{1})^{\otimes 2}
    + (\mathbbm{1} \otimes \sigma_i)^{\otimes 2} \right)
    \right).
\end{equation}

After a sequence of such mixing operators, the super Pauli string $h \otimes h$ is transformed into a sum of some other super Pauli strings $g_\alpha$: $M_{\mc{Y}_k} \circ ... \circ M_{\mc{Y}_q} (h_i \otimes h_i) = \sum c'_\alpha g_\alpha \otimes g_\alpha$. This collection is rather difficult to describe, however, from the properties of the mixers we know that (i) the coefficients $c'_\alpha$ sum up to one, and that (ii) the support of every Pauli string $g_\alpha$ is bounded by the support of the causal cone $|C(h, U)|$. Also, one can show that every qubit in the causal cone is in the support of some Pauli string $g_\alpha$.

\subsubsection{Elimination of terms by commutator}

The block $G_k$ corresponds to a triple of operators $M_{ \mc{Y}_k} \circ \mc{C} \circ M_{ \mc{Y}_k}$, of which we already used the first one. The pair of operators $M_{\mc{Y}_k} \mc{C}$ now acts in the following way: the strings whose support does not intersect that of $\mc{C}$ are eliminated, while all other strings are multiplied by a constant coefficient depending on the size of the block containing $\mc{C}$ (see Proposition \ref{prop:commutator_old}). To bound the number of surviving terms from below, we explicitly track a subset of such terms. 

If the block $G_k$ is outside the causal cone of $h \otimes h$, then the derivative $\partial_\theta E$ is trivially zero. 
So, excluding this case, let us assume that the block $G_k$ is within the causal cone of $h$. This means that we can find a sequence of blocks $G_{j_l}, ... G_{j_{l_c + 1}}$ situated in layers $l, ..., l_c + 1$, such that the first one shares support with $h$, and each $G_{j_{k-1}}$ shares support with $G_{j_k}$. Finally, we require that $G_{j_{l_c + 1}}$ shares support with the block $G_k$. Thus, we have established a causal path from $h$ to the commutator (see Fig. \ref{fig:comm_path}). 

Now, the output of the mixer $M_{j_l}$ contains $4^{|\mc{Y}_{j_l}|} - 1$ super Pauli strings with equal coefficients, at least $3/4$ of which act nontrivially in the support of $G_{j_{l-1}}$. For example, if $G_{j_l}$ is a two-qubit block, it outputs 15 super Pauli strings, of which 12 share support with the next block $G_{j_{l - 1}}$. The next mixer 
$M_{j_{l - 1}}$ takes those super Pauli strings, and for each of them, outputs $4^{|\mc{Y}_{j_{l - 1}}|} - 1$ super Pauli strings, of which at least $3/4$ again have nontrivial action in the support of the next block. Continuing on, we find that the total weight of such super Pauli strings is at least $(3/4)^{l - l_c}$. Then it gets multiplied by $\frac{2 \cdot 4^{|\mc{Y}_k|}}{4^{|\mc{Y}_k|} - 1}$. 

Note that the blocks outside the specified path can increase this number, but not decrease it. For example, in the notation of Fig. \ref{fig:comm_path}, we first acted with the mixer $M_{10}$, corresponding to the block $G_{10}$, and got a collection of super Pauli strings in the output, some of which act nontrivially on qubits 3 and 4, the support of $G_8$. The action of the mixer $M_{11}$ cannot make those strings lose the nontrivial action on that support. In principle, such block could bring more strings to act nontrivially there, but for a lower bound this is not important.

\begin{figure}
    \centering
    \begin{pgfpicture}{0em}{0em}{0em}{0em}
    \color{gray!70}
    \pgfrect[fill]{\pgfpoint{11.32em}{-4.25em}}{\pgfpoint{2.65em}{2.8em}}
    \pgfrect[fill]{\pgfpoint{8.1em}{-6.15em}}{\pgfpoint{2.25em}{2.8em}}
    \pgfrect[fill]{\pgfpoint{4.9em}{-8.05em}}{\pgfpoint{2.2em}{2.8em}}
    \color{green!80!yellow!45!white}
    \pgfrect[fill]{\pgfpoint{1.65em}{-6.15em}}{\pgfpoint{2.25em}{2.8em}}
    \end{pgfpicture}
    \mbox{
            \Qcircuit @C=1.0em @R=1.0em {
            \lstick{\ket{0}} 
            & \multigate{1}{G_1} 
            & \qw 
            & \multigate{1}{G_7}    
            & \qw 
            & \qw
            & \rstick{1}
            \\
            \lstick{\ket{0}} 
            & \ghost{G_1} 
            & \multigate{1}{G_4} 
            & \ghost{G_7} 
            & \multigate{1}{G_{10}} 
            & \qw  
            & \rstick{1}
            \\
            \lstick{\ket{0}} 
            & \multigate{1}{G_2}  
            & \ghost{G_4}
            & \multigate{1}{G_8} 
            & \ghost{G_{10}}
            & \qw  
            & \rstick{X}
            \\
            \lstick{\ket{0}} 
            & \ghost{G_2} 
            & \multigate{1}{G_5} 
            & \ghost{G_8}
            & \multigate{1}{G_{11}} 
            & \qw  
            & \rstick{X}
            \\
            \lstick{\ket{0}} 
            & \multigate{1}{G_3}
            & \ghost{G_5} 
            & \multigate{1}{G_9}
            & \ghost{G_{11}} 
            & \qw  
            & \rstick{1}
            \\
            \lstick{\ket{0}} 
            & \ghost{G_3} 
            & \multigate{1}{G_6}
            & \ghost{G_8}
            & \multigate{1}{G_{12}}
            & \qw
            & \rstick{1}
            \\
            \lstick{\ket{0}} 
            & \qw
            & \ghost{G_6}
            & \qw
            & \ghost{G_{12}}
            & \qw
            & \rstick{1}
            \\
              }
        }
    \caption{An example of a path constructed out of blocks with nontrivial inputs. Dark blocks correspond to mixing operators, highlighted block contains the commutator.}
    \label{fig:comm_path}
\end{figure}
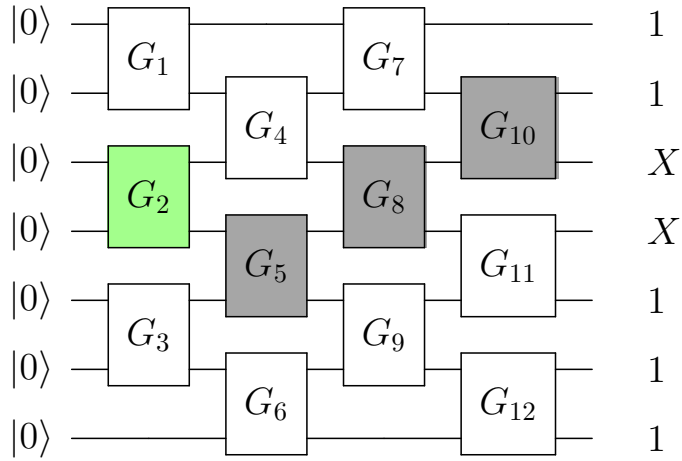{}

\subsubsection{Second portion of the mixing operators}

Let us apply all the remaining operators except the set of mixers $B = \{ M_{\mc{Y}_{1}}, ..., M_{\mc{Y}_m} \}$ corresponding to the first layer of the ansatz. Upon doing that, we get a linear combination of super Pauli strings $\sum_\alpha c''_\alpha g'_\alpha \otimes  g'_\alpha$. The coefficients $c''_\alpha$ sum up to at least $\frac{2 \cdot 4^{|\mc{Y}_k|}}{4^{|\mc{Y}_k|} - 1} \cdot (3/4)^{l - l_c} $. Each super string $g'_\alpha \otimes  g'_\alpha$ will then go through this layer of mixers and then the output will get averaged over the zero ket vector $\ket{\boldsymbol{00}}$.

For every $g'_\alpha \otimes  g'_\alpha$, the number resulting from this series of operations is greater or equal to $\prod_{M_{\mc{Y}} \in B} \frac{2^{|\mc{Y}|} - 1}{4^{|\mc{Y}|} - 1}$. The super Pauli string $g'_\alpha \otimes  g'_\alpha$ can act trivially or nontrivially on the support of each $M_{\mc{Y}} \in B$. If it acts trivially, then $M_{\mc{Y}}$ does nothing. In the opposite case, $M_{\mc{Y}}$ yields $4^{|\mc{Y}|} - 1$ super Pauli strings with equal weights. Of these strings, only $2^{|\mc{Y}|} - 1$ consist entirely of identity matrices and Pauli $Z$ matrices, and only these strings yield a 1 when averaged over $\ket{\boldsymbol{00}}$. Repeating for all mixers in $B$ yields the lower bound $\prod_{M_{\mc{Y}} \in B} \frac{2^{|\mc{Y}|} - 1}{4^{|\mc{Y}|} - 1}$.

Summing up all of the above, we write the lower bound on $\Var \partial_\theta E(h_i)$:
\begin{equation}
\label{eq:var_theta}
    \partial_\theta E (h) \geq \frac{2 \cdot 4^{|\mc{Y}_k|}}{4^{|\mc{Y}_k|} - 1} \left( \frac34 \right)^{l - l_c} 
    \prod_{M_{\mc{Y}} \in B} 
    \frac{2^{|\mc{Y}|} - 1}{4^{|\mc{Y}|} - 1} = 
    \frac{2 \cdot 4^{|\mc{Y}_k|}}{4^{|\mc{Y}_k|} - 1} \left( \frac34 \right)^{l - l_c}
    \prod_{M_{\mc{Y}} \in B} 
    \frac{1}{2^{|\mc{Y}|} + 1}.
\end{equation}{}

The last product can be bounded from below by $3^{-|C(h, U)|}$. Here is how: $1/ (2^{|\mc{Y}|} + 1) = (1/2^{|\mc{Y}|}) \cdot 1/(1 + 2^{-|\mc{Y}|})$. Since $|\mc{Y}| \geq 1$, then the second factor of the right-hand side of this equality is greater or equal than $2/3$. The product can be then transformed as follows:

\begin{equation}
\label{eq:massage_var_theta}
     \prod_{M_{\mc{Y}} \in B} 
    \frac{1}{2^{|\mc{Y}|} + 1} \geq      
    \prod_{M_{\mc{Y}} \in B}  \frac{1}{2^{|\mc{Y}|}} \frac{2}{3}
    = \left( \frac{2}{3} \right)^{|B|} \frac{1}{2^{|C(h, U)|}}.
\end{equation}{}
Now observe that the number of blocks is not greater than the number of qubits, and hence the last part of \eqref{eq:massage_var_theta} is lower bounded by $3^{-|C(h, U)|}$, which concludes the proof.

\subsection{Extension of the theorem}

In the formulation of Theorem~\ref{thm:block_plateaus}, we required that the partial blocks $G_A$, $G_B$ form 2-designs themselves. This assumption can be partially relieved. We know from Proposition~\ref{prop:ad_is_pauli_orthogonal} that $\mathrm{Ad}_{G_B \otimes G_B}$ preserves the 2-norm of its input. This means that the 1-norm --- the quantity preserved by the mixers --- will be multiplied by a factor in $[4^{-|\mc{Y}|}, 4^{|\mc{Y}|}]$ due to the equivalence of vector norms. The precise composition of its output does not matter: the subsequent mixer operators will essentially erase it, only caring about the causal cone structure. Therefore, we can remove the requirement that $G_B$ forms a 2-design at the cost of a $\Theta(1)$ multiple to the lower bound.

Removing the requirement that $G_A$ is a 2-design is more difficult. Suppose that just before $G_A$ we have some sum of super Pauli strings $\sum_\sigma \tilde{c}_\sigma \sigma \otimes \sigma$. If $G_A$ is a 2-design, we can reason about the proportion of terms that will end up anticommuting with $F$. For arbitrary $G_A$, the worst case it that all terms $\textrm{Ad}_{G_A} (\sigma)$ commute with $F$, causing the overall variance to go to zero. To rule out this scenario, one will have to understand the distribution of coefficients $\tilde{c}_\sigma$.

\section{Numerical results}

\subsection{Proximity of local blocks to 2-designs}
\label{subsec:proximity_to_designs}

It is known that approximate 2-designs can be prepared by a polynomial depth random circuit \cite{harrow_approximate_2018,brandao_local_2016}. However, here we are interested in local blocks whose properties are not guaranteed by asymptotic estimates.

We performed a series of numerical experiments to compare certain two-qubit blocks to exact unitary designs. A simple way of evaluating the proximity of the gate families to the Haar measure is to measure the distance to the so-called quantum $t$-tensor product expander (TPE) \cite{brandao_local_2016,low_pseudo-randomness_2010}. A family of random unitary gates $\nu$ is a $\lambda$-approximate TPE if $||\mathbb{E}_{Haar} (U^{\otimes t} \otimes (U^*)^{\otimes t}) - \mathbb{E}_\nu (U^{\otimes t} \otimes (U^*)^{\otimes t}) ||_p \leq \lambda$ for $p=\infty$. 
The trace definition of an approximate $t$-design involves the same quantity for $p = 1$. 

Finally, for $p=2$ this quantity can be related to the coefficients of Pauli decomposition of a Hamiltonian going through a mixing operator.  Let $H = \sum c_i \sigma_i$ be a Hamiltonian on $n$ qubits. Recall that Pauli strings form an orthogonal basis. Since the Hilbert-Schmidt inner product is the same as the scalar product of matrices as vectors in $\mathbb{R}^{2^n \times 2^n}$, this also applies to their reshaping to vectors. Then one can verify that $||\mathrm{vec}(H)||_2 = 
2^{\frac{n}{2}} \sqrt{\sum_i |c_i|^2}
\equiv ||\mathbf{c}||_2 \cdot 2^{\frac{n}{2}}$. This works when $\sigma_i$ are super Pauli strings. The operator 2-norm of $\mathbb{E}_{Haar} (U^{\otimes t} \otimes (U^*)^{\otimes t}) - \mathbb{E}_{\mu} (U^{\otimes t} \otimes (U^*)^{\otimes t})$, provides an upper bound on the vector norm of the output of this operator, meaning that this is the maximum norm of the discrepancy from the perfect output for an input of unit norm. For a super Pauli string $h \otimes h$, this error $\lambda_2$ upper bounds the 2-norm of the vector $(\mathbf{c} - \mathbf{c}_{\text{Haar}})$.

We will denote the aforementioned $p$-norms as $\lambda_1$, $\lambda_2$ and $\lambda_\infty$. 

\begin{table}
    \centering
    \begin{tabularx}{\textwidth}{|>{\centering}X|c|c|c|c|}
    \hline
        Block & Circuit diagram or matrix & $\lambda_1$ & $\lambda_\infty$ & $\lambda_2$\\
        \hline
        $X$, $Z$, and $ZZ$ rotations &  
            $\Qcircuit @C=1.0em @R=1.0em {
                   \quad & \gate{R_Z} & \multigate{1}{R_{ZZ}} & \gate{R_X} & \qw \\
                   \quad & \gate{R_Z} & \ghost{R_{ZZ}} & \gate{R_X} & \qw \\
               }$
        & 0.95 & 1.80 & 0.87\\
        \hline 
        Universal gates and a CNOT &
            $\Qcircuit @C=0.8em @R=0.8em {
           \quad & \gate{U_3} & \ctrl{1} & \gate{U_3} & \qw \\
           \quad & \gate{U_3} & \targ & \gate{U_3} & \qw \\
            }$
        & 0.68 & 0.69 & 0.42\\
        \hline
        $Y$ rotations and a CZ~\cite{cerezo_cost-function-dependent_2020} &            
            $\Qcircuit @C=0.8em @R=0.8em {
           \quad & \gate{R_Y} & \ctrl{1} & \gate{R_Y} & \qw \\
           \quad & \gate{R_Y} & \ctrl{-1} & \gate{R_Y} & \qw \\
            }$ & 1.76 & 1.76 & 1.00\\
        \hline
        Number-conserving~\cite{barkoutsos_quantum_2018} & 
        $
        \begin{pmatrix}
        1 & 0 & 0 & 0 \\
        0 & \cos(\theta_1) & e^{i\theta_2} \sin(\theta_1) & 0 \\
        0 & e^{-i\theta_2} \sin(\theta_1) & -\cos(\theta_1) & 0 \\
        0 & 0 & 0 & 1 \\
        \end{pmatrix}
        $
        & 2.40 & 2.40 & 1.00\\
        \hline
        Cartan decomposition \cite{khaneja_cartan_2000,khaneja_time_2001} &
        $\Qcircuit @C=0.8em @R=0.8em {
       \quad & \gate{U_3} & \multigate{1}{R_{XX}} & \multigate{1}{R_{YY}} & \multigate{1}{R_{ZZ}}& \gate{U_3} & \qw \\
       \quad & \gate{U_3} & \ghost{R_{XX}} & \ghost{R_{YY}} & \ghost{R_{ZZ}} & \gate{U_3} & \qw \\
        }$
            & 0.25 & 0.25 & 0.17\\
    \hline
    \end{tabularx}
    \caption{Proximity to the 2-tensor product expander for different two-qubit blocks, estimated by random sampling.}
    \label{tab:local_designs}
\end{table}

We estimated the values of $\lambda$ for $t=2$ different gate families by the following numerical procedure. The Haar-averaged tensor product $\mathbb{E}_{\text{Haar}} (U^{\otimes t} \otimes (U^*)^{\otimes t})$ is constructed explicitly using exact formulas \cite{mcclean_barren_2018,poland_no_2020}. 
For the two-qubit blocks, we pick the parameters uniformly at random and average the resulting tensor product over $N=500000$ trials. From this average, we estimate $\lambda$.

To estimate the error of this method, we also evaluated $\lambda$ for an ensemble of matrices distributed according to the Haar measure. For $N \rightarrow \infty$, the estimate should converge to zero as $1/\sqrt{N}$. Hence, the numerical value of $\lambda$ shows the typical scale of sampling error. For $N=500000$ trials, we observed the values of $\lambda_1 \approx \lambda_\infty = 0.022$, $\lambda_2 = 0.0028$.

The results of numerical experiments are summarized in Table \ref{tab:local_designs}. As expected, the more sophisticated ans\"atze are usually better at approximating a 2-design. Note further that the block implemented according to the Cartan decomposition of $SU(4)$ \cite{khaneja_cartan_2000,khaneja_time_2001} is not an exact 2-design, although it is capable of preparing any two-qubit gate. Nonetheless, among the gate families studied, this block is the closest approximation of a 2-design.

With a similar numerical experiment for $t=1$, we found that all blocks, except the particle-conserving block \cite{barkoutsos_quantum_2018}, are also exact 1-designs up to sampling tolerance.

\subsection{Plateau dependence}
\label{subsec:plateau_numeric}

\begin{figure}
    \centering
    \begin{subfigure}{.48\linewidth}
        \centering
        \includegraphics[width=\textwidth]{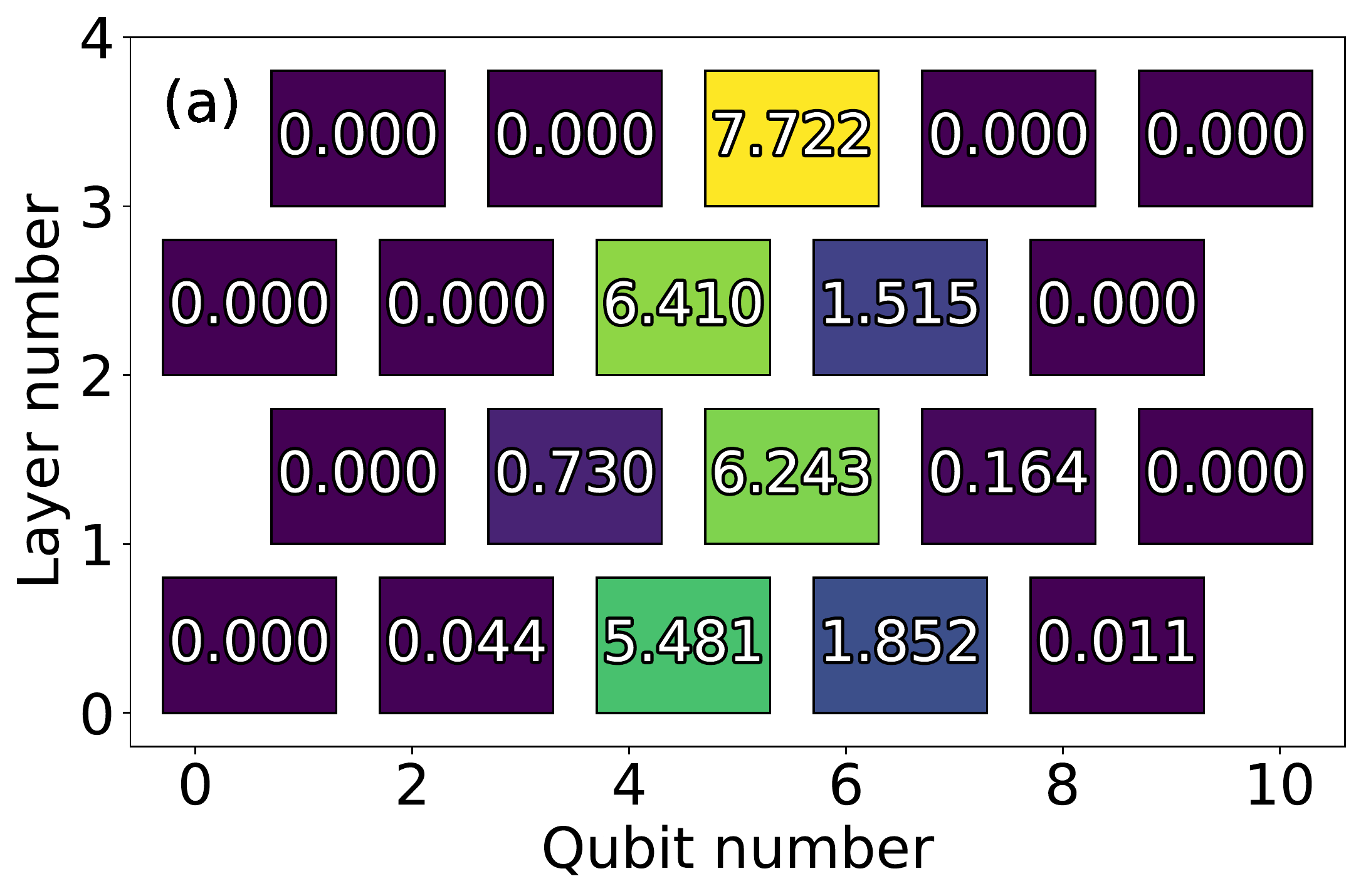}
    \end{subfigure}\begin{subfigure}{.48\linewidth}
        \centering
        \includegraphics[width=\textwidth]{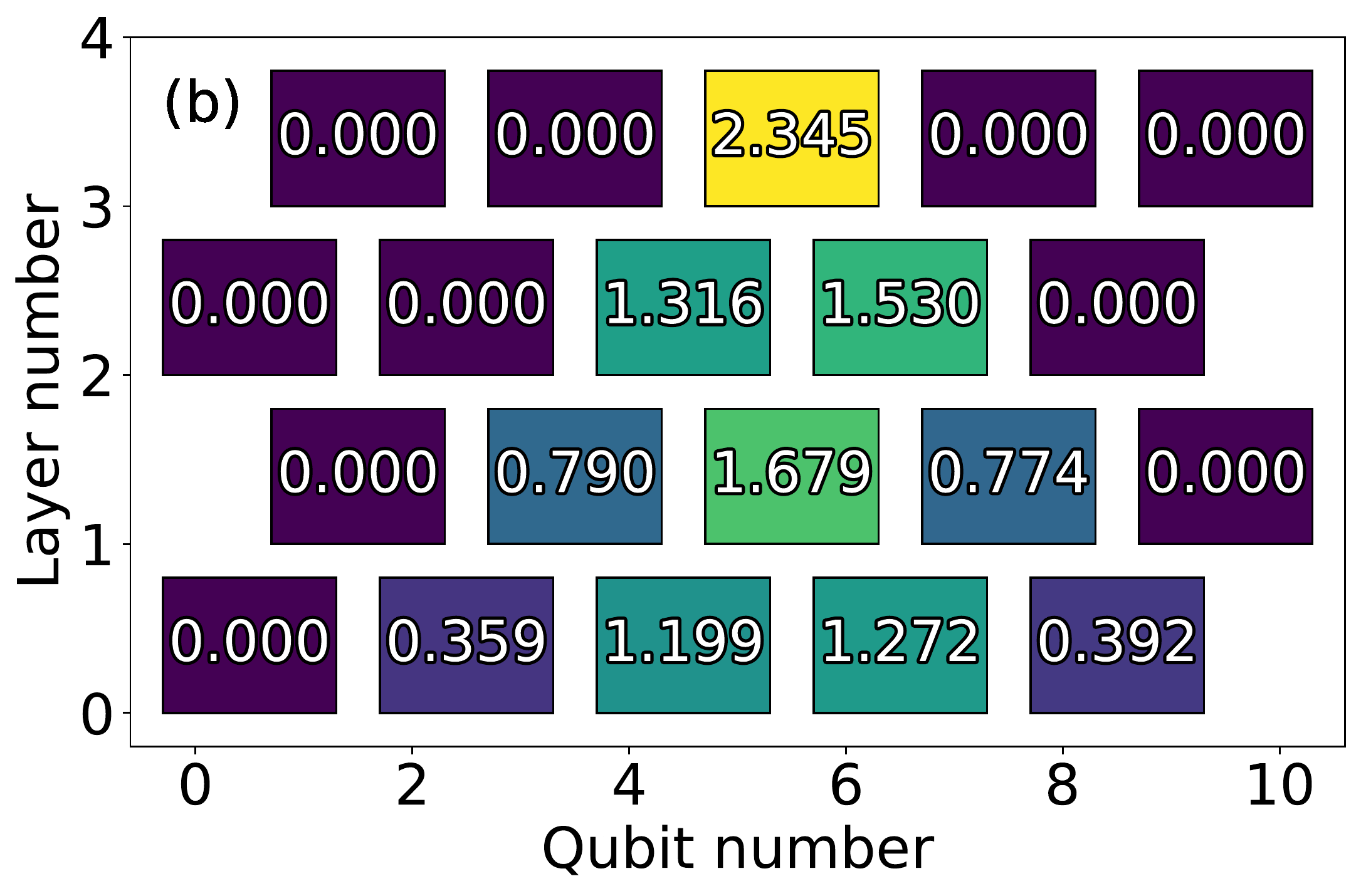}
    \end{subfigure}
    \begin{subfigure}{.48\linewidth}
        \centering
        \includegraphics[width=\textwidth]{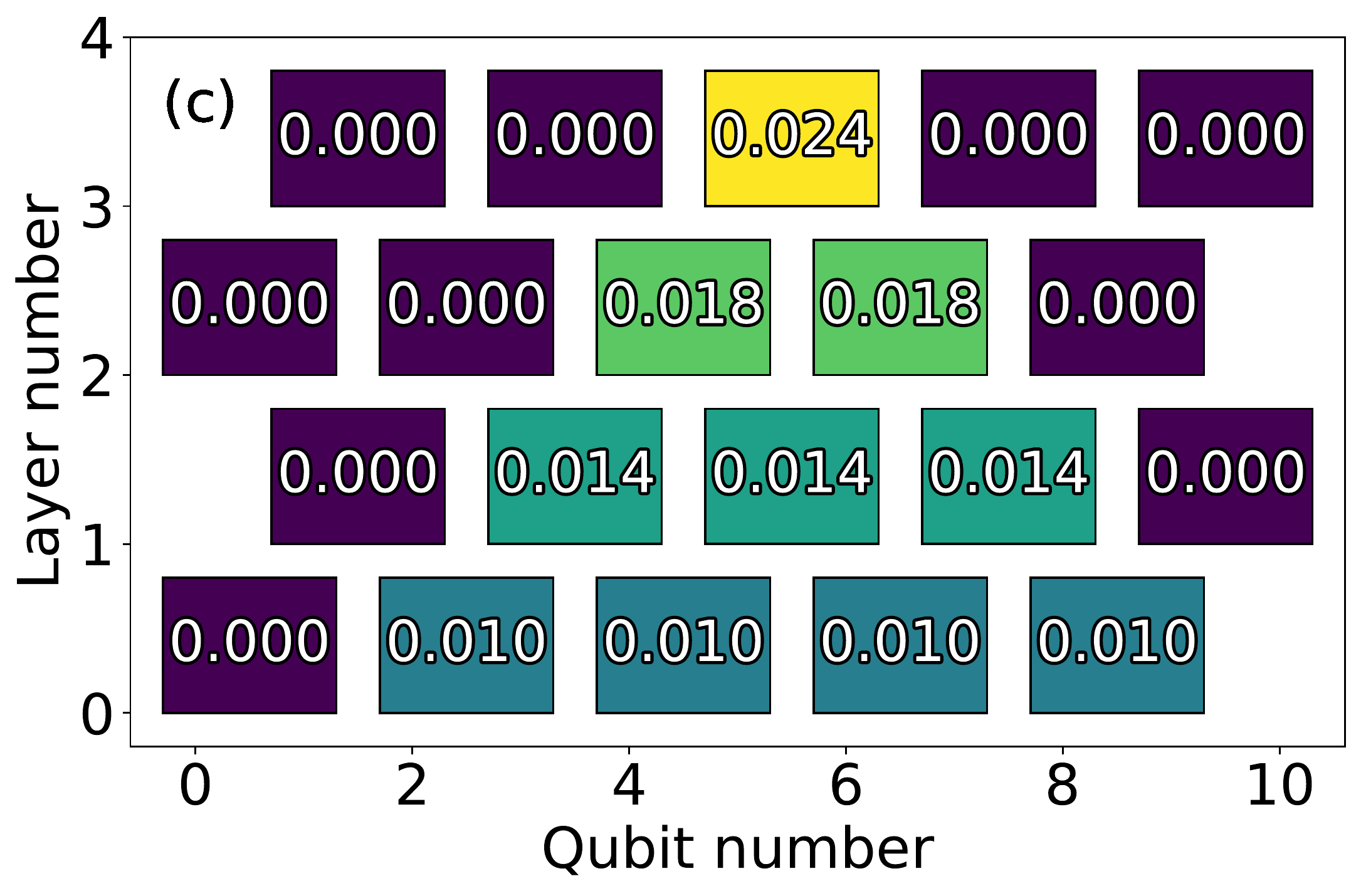}
    \end{subfigure}
    \caption{Derivative variances for $H = X_5$, averaged over parameters in each ansatz block. The numbers in the boxes denote $\Var \partial_\theta E \cdot 100$. Qubit number 10 is identified with qubit number 0. (a) Numerical result for an ansatz with blocks of $X,Z,ZZ$ rotations. (b) Numerical result for blocks implemented according to the Cartan decomposition. (c) Lower bound given in Theorem \ref{thm:block_plateaus}. Reprinted from~\cite{uvarov_barren_2021}.}
    \label{fig:one-local}
\end{figure}

To estimate gradients, we used the following analytical procedure \cite{mitarai_quantum_2018,schuld_evaluating_2019}: let $f(\theta)$ be the cost function, and $\theta$ a parameter which is included in the quantum circuit in a gate like $\exp(\rmi F\theta / 2)$ for some Pauli operator $F$. Then the derivative w.r.t.\ this parameter is equal to $(f(\theta + \pi /2) - f(\theta - \pi / 2))/2$. The simulations assume noise-free conditions and use the statevector simulator provided by Qiskit~\cite{aleksandrowicz_qiskit:_2019}.

The first Hamiltonian we tested our predictions on is the single-qubit Hamiltonian $H = X_5$ acting on $n=10$ qubits. The qubits are enumerated starting from zero. 
For $N = 400$ samples, the derivative with respect to each parameter was evaluated, then the variances were averaged over each block. We performed two numerical experiments with different two-qubit blocks from Table \ref{tab:local_designs}: one with blocks of $X$, $Z$, and $ZZ$ rotations, and the other with blocks implemented according to the Cartan decomposition of $\mc{SU}(4)$. The ansatz used was the checkerboard ansatz with ring connectivity. The result of the numerical test is shown in Fig.~\ref{fig:one-local}. 

\begin{figure}
    \centering
    \begin{subfigure}{.48\linewidth}
        \centering
        \includegraphics[width=\linewidth]{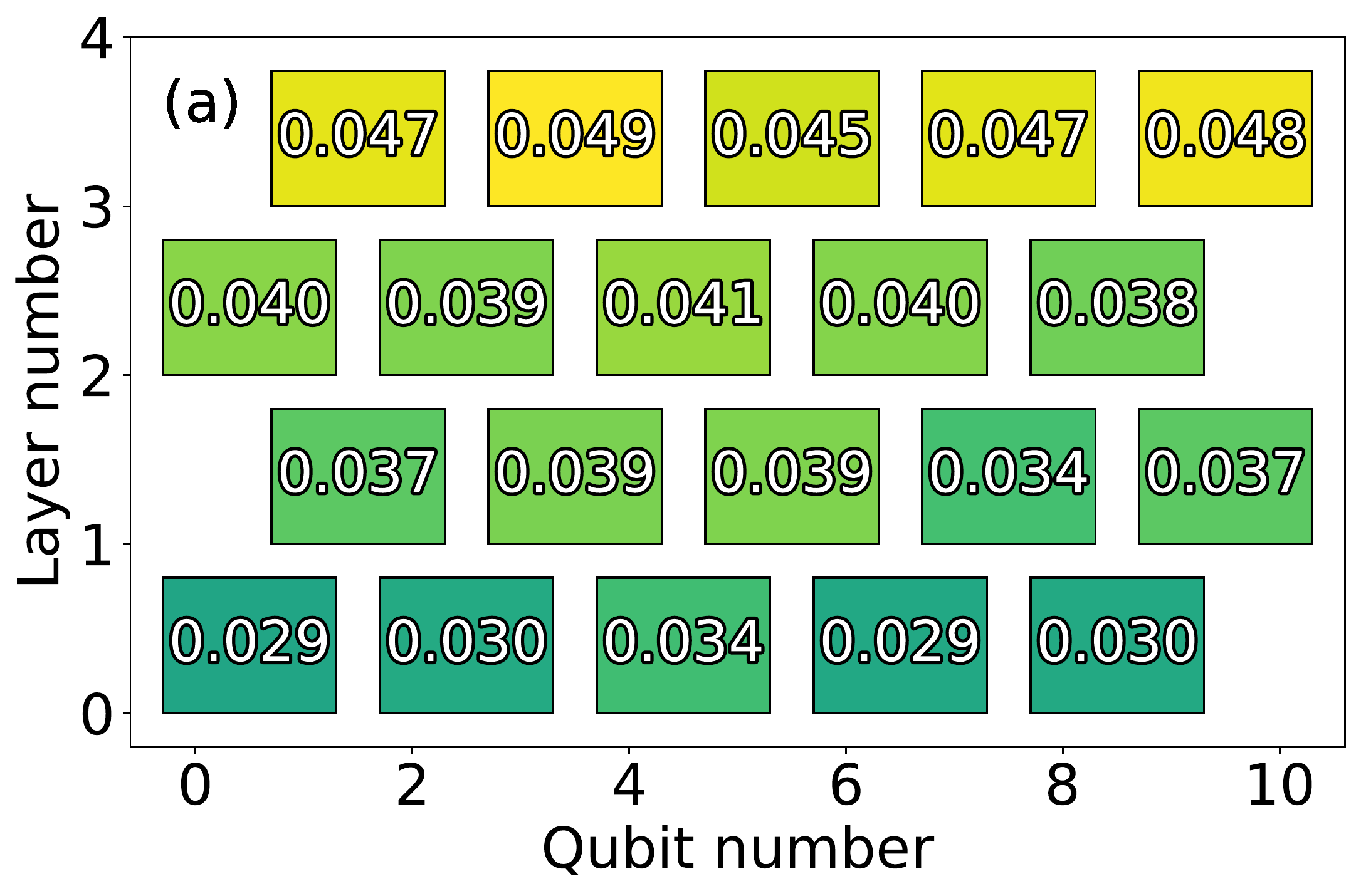}
    \end{subfigure}
    \begin{subfigure}{.48\linewidth}
        \centering
        \includegraphics[width=\linewidth]{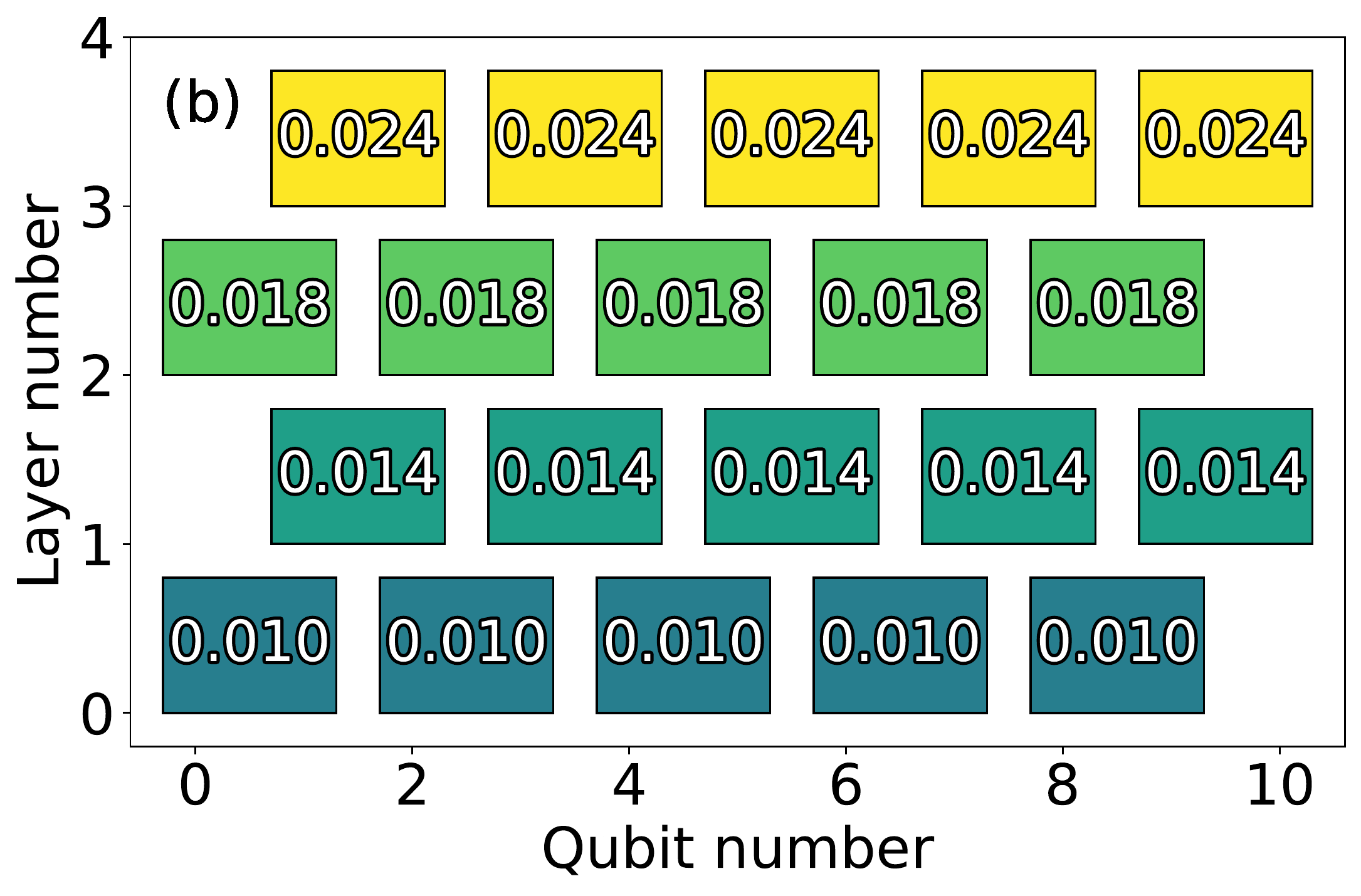}
    \end{subfigure}
    \caption{Derivative variances for $H = X^{\otimes n}$, numerical estimate (a) and the lower bound (b). Reprinted from~\cite{uvarov_barren_2021}.}
    \label{fig:n-local}
\end{figure}

In both experiments, the causal cone structure is evident, as is evident the tendency of the gradients to decrease with the decreasing number of layer. However, the Cartan decomposed blocks show smoother results. This result is consistent with the fact that the first block type is further from a 2-design. The condition that the block can be further decomposed into two independent local 2-designs is also violated in the first case. Because of these factors, the gradients are uneven.

The theoretical lower bound is fulfilled by a large margin in both cases. The lower bound also does not catch the difference of gradients within one layer, which tend to be more significant in the middle of the causal cone as opposed to the edges of the cone, where the gradients are much smaller.

Figure \ref{fig:n-local} shows the results of a similar numerical test for $H = X^{\otimes n}$ for $n = 10$. Here, the ``Cartan decomposition'' blocks were used in the ansatz. As implied by the lower bound and in accordance with the results previously found in the literature \cite{cerezo_cost-function-dependent_2020}, this $n$-local Hamiltonian exhibits barren plateaus even for a very shallow ansatz.

According to \eqref{eq:paulis_decouple}, variances for a Hamiltonian consisting of several Pauli strings are equal to the sum of variances computed for each Pauli string independently. We tested that prediction on a pair of Hamiltonians $H_1 = X_4 X_5$, $H_2 = X_5 X_6$. In this test, the ansatz acts on 10 qubits and consists of 4 layers of ``Cartan decomposition'' blocks. The variances for $H_1$ and $H_2$ separately are shown as a stacked bar chart in Fig.~\ref{fig:additive}a. Each bar corresponds to a parameter $\theta_i$ in the ansatz. Fig.~\ref{fig:additive}b shows the variances for $H_1 + H_2$. The qualitative agreement between the graphs is evident, and the differences for each parameter of the ansatz (shown in Fig.~\ref{fig:additive_delta}) are close to zero, up to the standard errors of the samples.

\subsection{Alternative ansatz architectures}
\label{subsec:alt_ansatz}

The width of the causal cone depends on the ansatz structure. Conversely, some ansatz structures may be less prone to barren plateaus. We performed the same numerical tests for two more circuit architectures that are better suited for NISQ devices.

\begin{figure}
    \centering
    \begin{subfigure}{.48\linewidth}
        \centering
        \includegraphics[width=\linewidth]{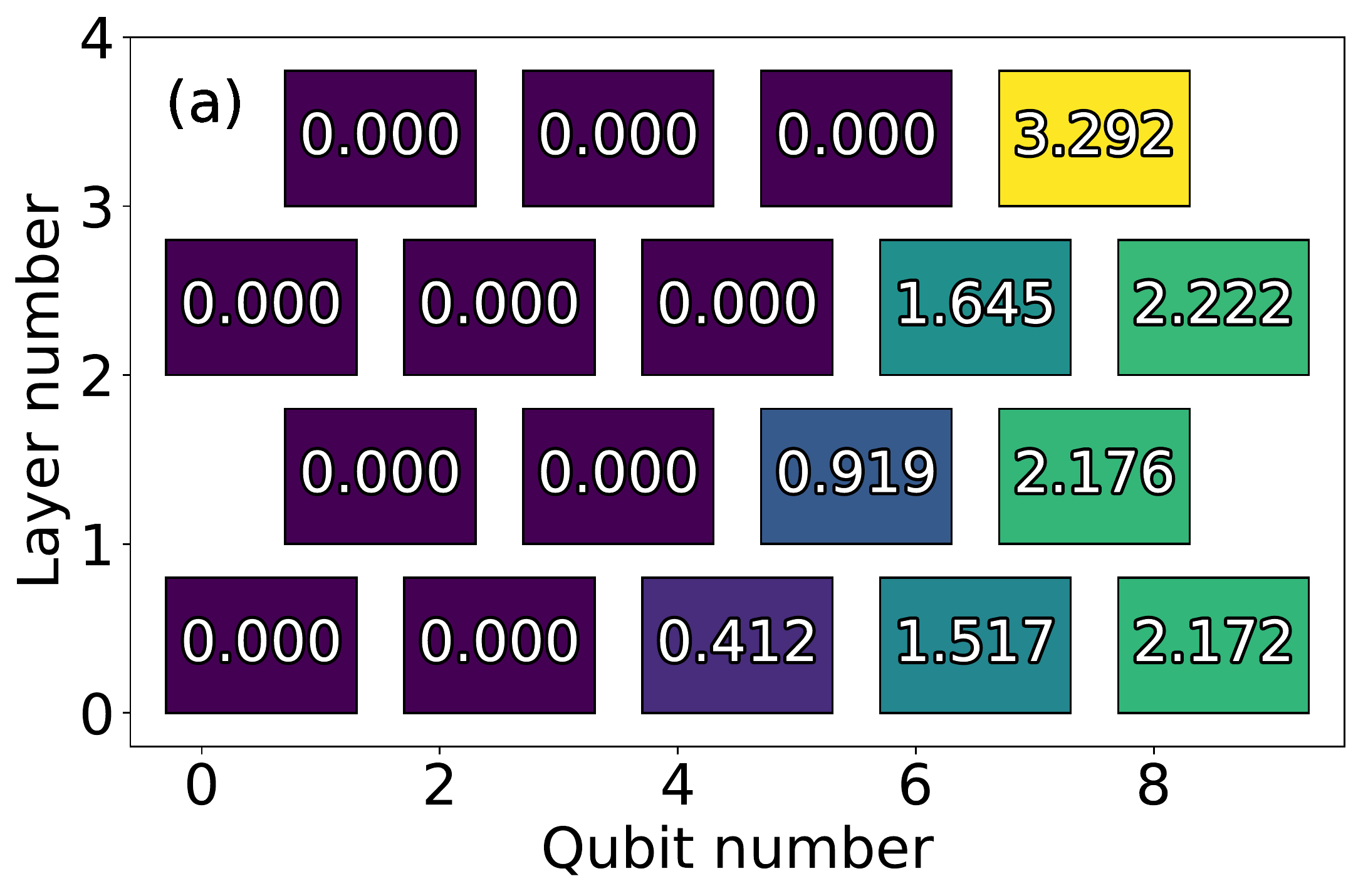}
    \end{subfigure}
    \begin{subfigure}{.48\linewidth}
        \centering
        \includegraphics[width=\linewidth]{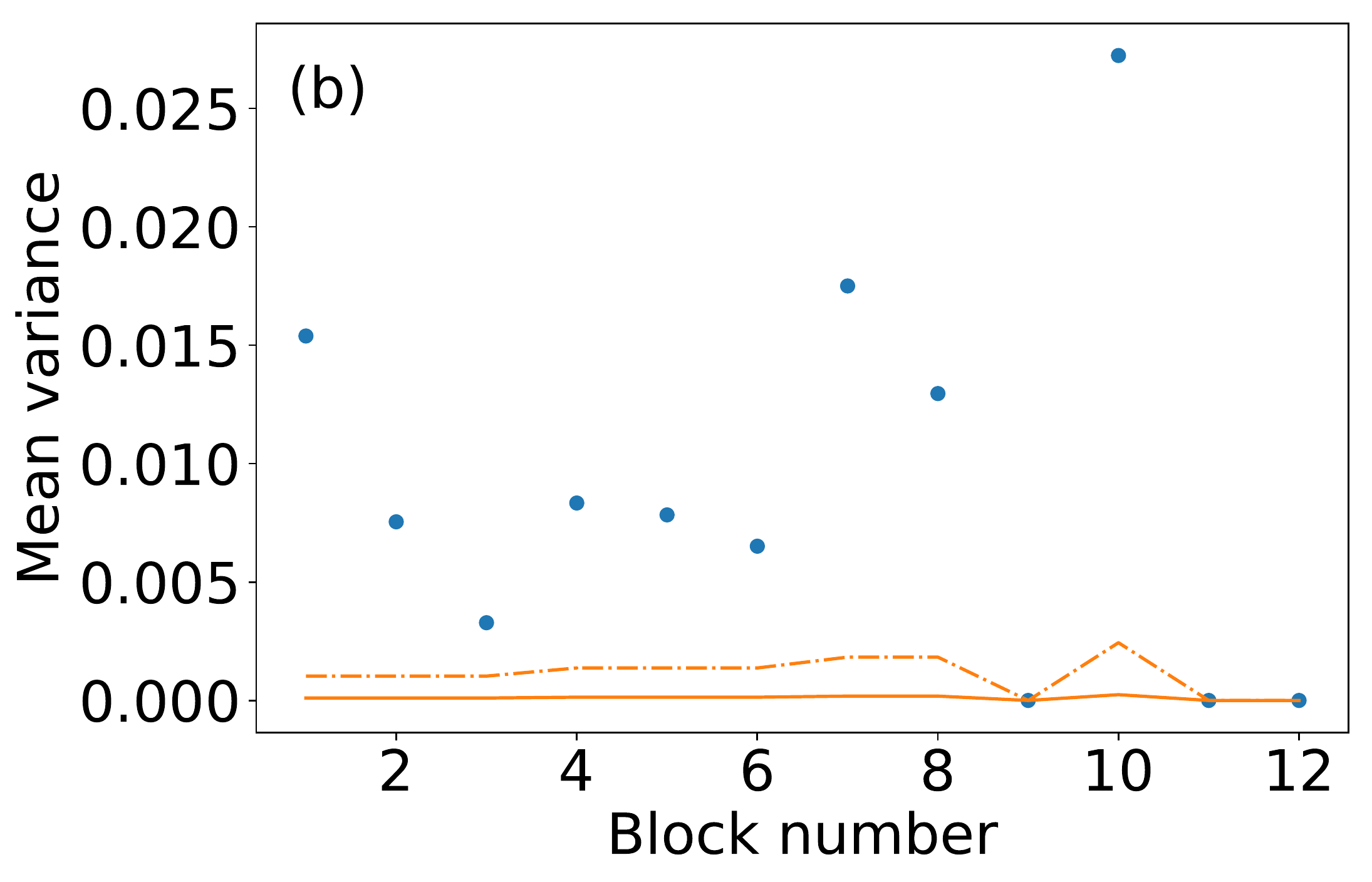}
    \end{subfigure}
    \caption{Blockwise averaged values of derivative variances with respect to one-local Pauli strings. (a) Numerical result for a line-connected checkerboard ansatz. (b) Numerical result for a two-dimensional lattice ansatz. Dots: numerical values, line: lower bound,  dot-dashed line: lower bound multiplied by 10. Reprinted from~\cite{uvarov_barren_2021}.}
    \label{fig:alt_connect}
\end{figure}

\subsubsection{Checkerboard with open boundary conditions}

For certain quantum computing platforms, e.g.~Calcium ions and Rydberg atoms, it is easiest to arrange qubits in a line and perform entangling gates acting on adjacent qubits. Unlike ring connectivity, this structure does not use direct coupling of the first qubit with the last qubit. Thus, the qubits closer to the edge will have narrower causal cones, and possibly higher values of the gradients. Fig.~\ref{fig:alt_connect}a shows the behavior of derivatives for such an architecture, for $H = X_8$. In comparison with the ring connectivity (Fig. \ref{fig:one-local}b), the gradient variances are significantly larger.

\begin{figure}
    \centering
    \includegraphics[width=\linewidth]{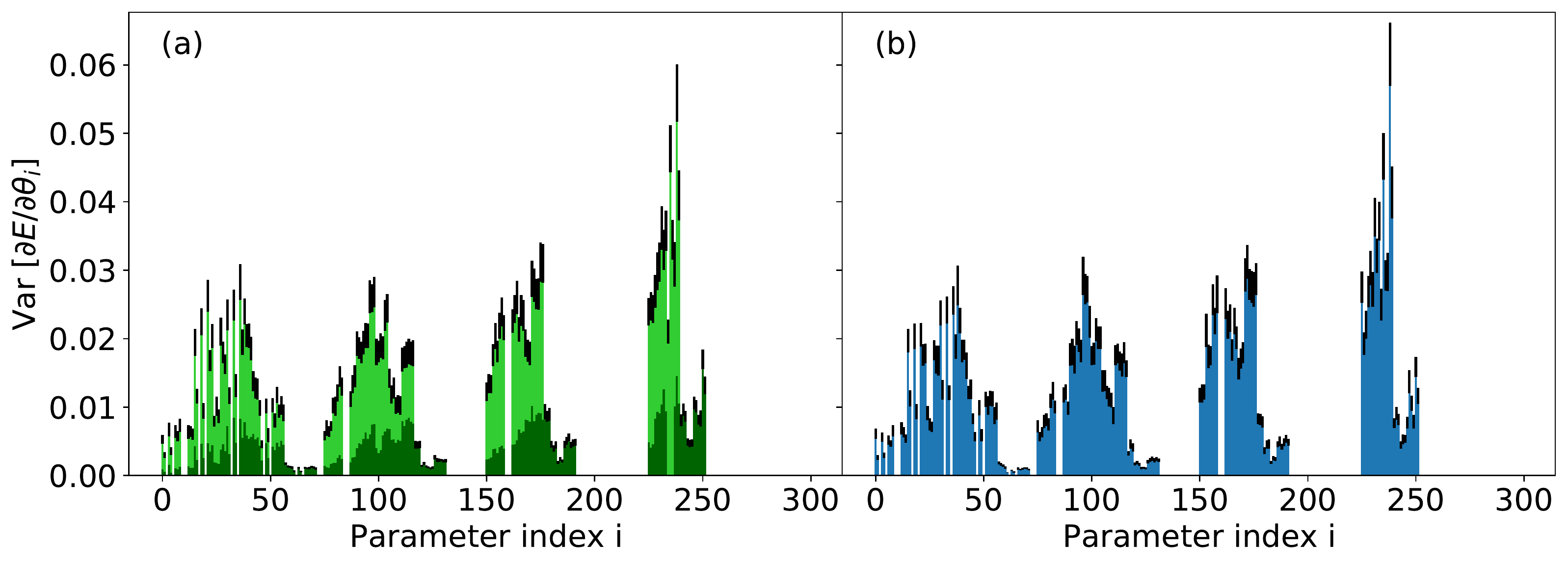}
    \caption{Variances of the cost function derivatives with respect to different ansatz parameters for $H_1 = X_5 X_6$, $H_2 = X_4 X_5$ (a), and their sum $H_1 + H_2$ (b). Reprinted from~\cite{uvarov_barren_2021}.}
    \label{fig:additive}
\end{figure}

\begin{figure}
    \centering
    \includegraphics[width=0.7\linewidth]{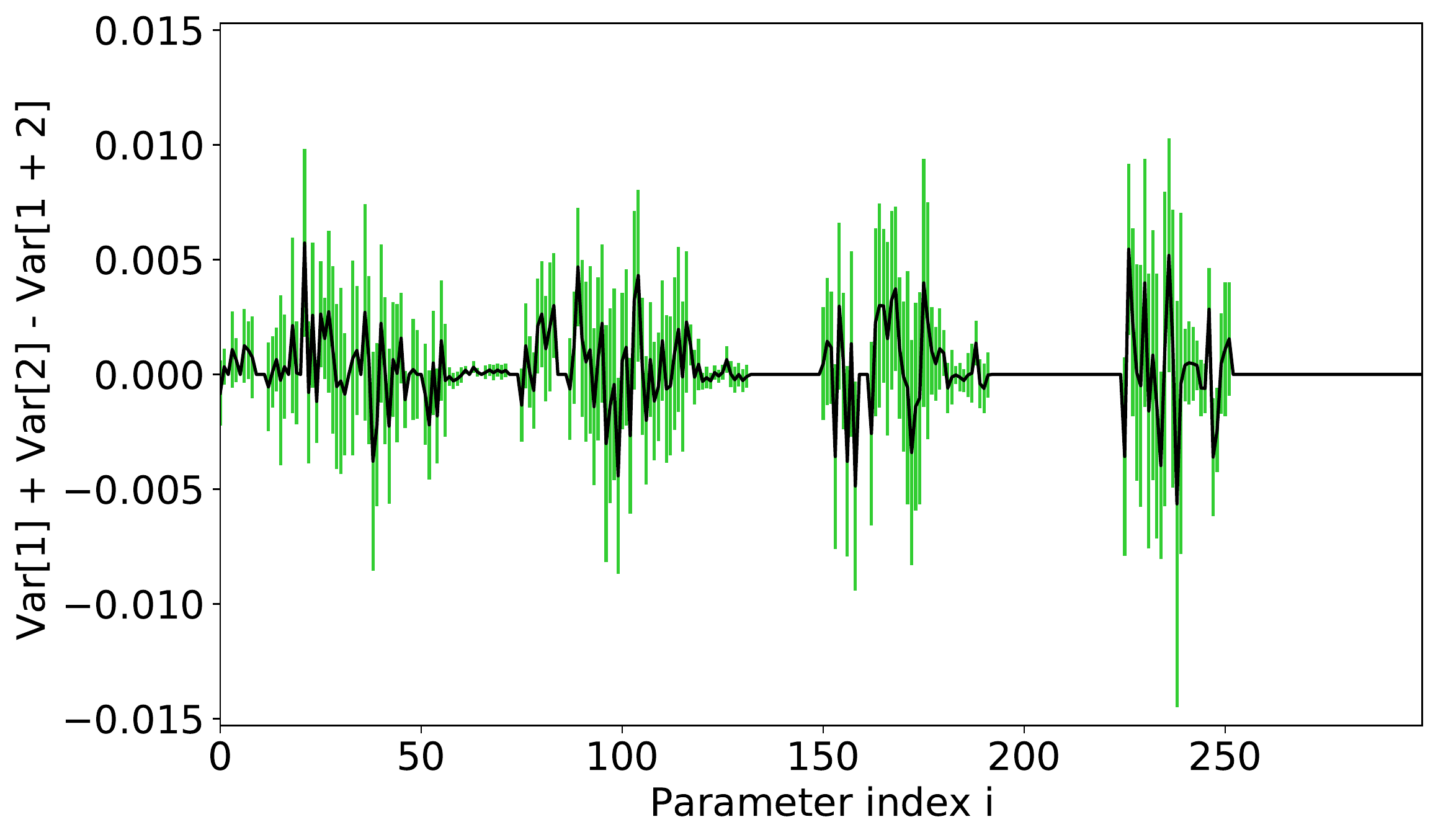}
    \caption{Difference between the variances plotted in Fig.~\ref{fig:additive}. Error bars denote one standard error. Reprinted from~\cite{uvarov_barren_2021}.}
    \label{fig:additive_delta}
\end{figure}

\subsubsection{Two-dimensional lattice}

We also tested the predictions of Theorem \ref{thm:block_plateaus} on a two-dimensional $3 \times 3$ lattice. The two-dimensional lattice ansatz is constructed as follows. The first layer of the ansatz is schematically depicted in Fig.~\ref{fig:2d_ansatz_scheme}. Every next layer is obtained from the last by rotating the layout 90 degrees clockwise. The two-dimensional ansatz consisted of four layers of two-qubit blocks. 

The numerical values of the derivative variances, as well as their lower bounds, are depicted in Fig.~\ref{fig:alt_connect}b. The causal cone structure for this ansatz is more convoluted, but it is possible to tell that some ansatz blocks are not included in the causal cone, and hence, the derivative over their parameters is equal to zero. 

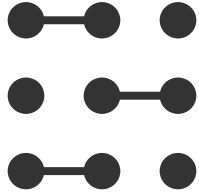
\begin{figure}
    \centering
    \begin{tikzpicture}[thick]
    \node at ( 0,0) [circle,draw=black!80,fill=black!80] {};
    \node at ( 0,1) [circle,draw=black!80,fill=black!80] {};
    \node at ( 0,2) [circle,draw=black!80,fill=black!80] {};
    \node at ( 1,0) [circle,draw=black!80,fill=black!80] {};
    \node at ( 1,1) [circle,draw=black!80,fill=black!80] {};
    \node at ( 1,2) [circle,draw=black!80,fill=black!80] {};
    \node at ( 2,0) [circle,draw=black!80,fill=black!80] {};
    \node at ( 2,1) [circle,draw=black!80,fill=black!80] {};
    \node at ( 2,2) [circle,draw=black!80,fill=black!80] {};
    
    \draw [draw=black!80,line width=3] (0,0) -- (1,0);
    \draw [draw=black!80,line width=3] (1,1) -- (2,1);
    \draw [draw=black!80,line width=3] (0,2) -- (1,2);
    \end{tikzpicture}    
    \caption{Connectivity of one layer in the 2D lattice ansatz. Other layers are formed by rotating this pattern by 90 degrees. Reprinted from~\cite{uvarov_barren_2021}.}
    \label{fig:2d_ansatz_scheme}
\end{figure}



\section{Discussion}
Given a Hamiltonian, we can now estimate its susceptibility to the barren plateaus. One can hence preprocess Hamiltonians in order to make the optimization more viable. For example, a method similar to that of Ref.~\cite{ryabinkin_iterative_2020} could be employed.

Our results indicate that the severity of plateaus also depends on the structure of the ansatz. This may mean that some hardware topologies are more suited for VQE than others. For example, our numerical tests demonstrate that in line connectivity of the qubits, the gates on the edges are potentially less prone to vanishing derivatives than those in the middle.

The numerical tests provided in this chapter can be implemented in hardware as well. The computational cost of simulating a quantum computer using the best known methods is exponential either in time or in memory, while estimation of the gradients in quantum hardware (up to fixed absolute tolerance) is linear in the number of ansatz parameters, and possibly even sub-linear in the cardinality of the problem Hamiltonian, if clever simultaneous measurement strategies are used \cite{verteletskyi_measurement_2020}. 

At the first glance, the barren plateaus phenomenon looks like a ``no-go'' result that prevents training any quantum circuits of sufficiently high depth. However, since the condition is that the parameters are random, it is more of a caution to not initialize circuits at random. This was first noted in \cite{grant_initialization_2019}, where the authors proposed initializing most gates with zero rotation angles. In this situation, the starting point is not chosen uniformly randomly across the entire search space, but instead it is sampled from a small neighborhood of zero. 

Other plentiful variants of VQE (see Section \ref{sec:vqe_variants}) rely on dynamically adjusting the ansatz. This can also be seen as a convoluted ``initialization strategy'': when the process produces a long ansatz, it already suggests a good starting point. 

\chapter{Penalty Hamiltonians for state verification}
\label{chap:ghz}

In this chapter, we propose and study a new method to estimate the fidelity of Clifford states. Primarily we focus on the Greenberger--Horne--Zeilinger (GHZ) state, as preparation of this state with high fidelity is an important benchmark for quantum hardware. We will start by reviewing the existing techniques tailored to verification of the GHZ state. Then we will introduce our proposal based on measuring the expected value of a certain Hamiltonian constructed using the description of the quantum circuit. Finally, we will present our numerical experiments and discuss their results.

\section{Background}

There are various methods to characterize quantum devices. The choice of a technique depends on the information we need to obtain, the amount of resources we can spend, and the assumptions made about the system of interest. The most dierct method to characterize a state prepared by the device is called quantum tomography~\cite{dariano_quantum_2003,straupe_adaptive_2016}. In quantum tomography, one repeatedly prepares a state of interest and measures it in different bases to completely reconstruct its density matrix. Unfortunately, full quantum tomography takes exponential time in the number of qubits, so it is only applicable to small systems and gates. Another recently proposed method, shadow tomography~\cite{aaronson_shadow_2018,huang_predicting_2020,koh_classical_2020}, enables the estimation functions linear in the density matrix $\rho$ without characterizing the state completely. Average properties of quantum circuits prepared by the device can also be inferred using a process called randomized benchmarking~\cite{magesan_robust_2011-1,knill_randomized_2008}.

One more metric for quantum computers is the size of the largest entangled state it can reliably prepare. In what follows we focus on the Greenberger--Horne--Zeilinger (GHZ) state $\ket{\Psi} = \frac{1}{\sqrt{2}} (\ket{0...0} + \ket{1...1})$. This state is a textbook example of a quantum state exhibiting genuine multipartite entanglement, meaning that there is no bipartition of qubits that would make this state separable. In addition, any state whose fidelity with the GHZ state above $1/2$ is guaranteed to be entangled (this is true even for mixed states)~\cite{sackett_experimental_2000}. Because of that, the fidelity of $1/2$ is used as a threshold benchmark for real quantum devices. Due to the simple form of the GHZ state, it is possible to estimate its fidelity without resorting to full state tomography. In the following subsection, we describe the existing techniques developed specifically for the GHZ state.

\subsection{Fidelity measurement techniques for the GHZ state}

\subsubsection{Parity oscillations}

A common technique to evaluate the fidelity of a quantum state with the GHZ state is based on measuring the \emph{parity oscillations} \cite{sackett_experimental_2000,leibfried_toward_2004,leibfried_creation_2005,monz_14-qubit_2011,song_observation_2019,omran_generation_2019}. 
The diagonal components $\rho_{0...0, 0...0}$ and $\rho_{1...1, 1...1}$ are measured in the computational basis. If we measure the off-diagonal components $\rho_{0...0, 1...1}$ and $\rho_{1...1, 0...0}$ (which are complex conjugates of each other), we can calculate the fidelity. The latter is usually expressed as a sum of two unknown values, called population $P$ and coherence $C$:
\begin{equation}
    \label{eq:f_is_p_plus_c}
    F = \frac{1}{2} (P + C).
\end{equation}
Here $P = \rho_{0...0, 0...0} + \rho_{1...1, 1...1}$, and 
$C = 2 |\rho_{1...1, 0...0}|$. 
\begin{remark}
    Strictly speaking, the fidelity is equal to $\frac{1}{2} (P + 2 \operatorname{Re} \rho_{1...1, 0...0})$. However, if one only cares about the multipartite entanglement, then one can estimate fidelity with the closest state of the type $\frac{1}{\sqrt{2}} (\ket{0...0} + e^{\rmi \gamma}\ket{1...1})$. In this case, the estimate is given by Eq.~(\ref{eq:f_is_p_plus_c}).
\end{remark}
The population is straightforward to measure. Evaluating coherence, on the other hand, requires a more complicated setup. The implementation is different depending on the platform, but the overall idea is to induce a phase difference $\varphi$ between the components of the GHZ state, and then measure an observable that would be sensitive to this phase difference.

The method employed in trapped-ion processors  consists in applying a local unitary gate $U = e^{-\rmi \frac{\pi}{4} (\cos \varphi X + \sin \varphi Y)}$ to every qubit. Then, the state is measured in the $Z$ basis. The value of interest is the parity, i.e.~the ratio of measurements with an even number of ones minus the ratio of measurements with an odd number of ones. Effectively, this means measuring the expected value of $Z^{\otimes n}$. If we denote $\rho_{0...0, 1...1} = re^{i \gamma}$, then the value of parity will have an oscillating component equal to $2r \cos(n \varphi - \gamma)$. Fitting the parity with a sine curve and extracting $\rho$ and $\gamma$ enables the calculation of fidelity.

\subsubsection{Multiple quantum coherence}

The value of $C$ can also be measured by method of multiple quantum coherences (MQC) \cite{wei_verifying_2020}. Excluding the technical error-mitigating steps, the method goes as follows:
\begin{enumerate}
    \item Prepare the GHZ state;
    \item Apply a $Z$ rotation with angle $\varphi$ to each qubit;
    \item Unprepare the GHZ state, i.e.~apply the entangling gates in reverse order;
    \item Measure the first qubit in the $Z$ basis. Record the  probability of measuring $\ket{0...0}$ as the overlap signal $S_\varphi$.
\end{enumerate}

The steps 1-4 must be repeated for $\varphi = \frac{\pi j}{n+1}$ for $j = 0, 1, ..., 2n+1$. Finally, by computing the Fourier transform of $S_\varphi$ and taking the highest-frequency component $I_n$, we can calculate the coherence as $C = 2 \sqrt{I_n}$. 

To mitigate errors, the authors also apply an $X$ gate to every qubit after step one. For a perfect GHZ state, this step does nothing. However, in experimental conditions there is a coherent error that results in the state drifting from the desired position. The application of $X$ gates effectively reverses the direction of the drift, enabling better overall coherence of the state for a while. This is analogous to the spin echo effect~\cite{hahn_spin_1950}.

It is interesting to compare MQC to parity oscillations. An obvious downside is that MQC requires the circuit to be twice as long. Another disadvantage is that, strictly speaking, the method depends on certain assumptions about the noise in the device~\cite{garttner_relating_2018}. However, there are advantages as well. First, MQC enables the aforementioned error mitigation by $\pi$-pulses. Second, MQC is easier to correct for readout errors \cite{wei_verifying_2020}. In MQC, the ideal measurement contains only the strings $00...0$ and $10...0$, as opposed to parity oscillations, where all strings should be expected.

The current record of preparing the GHZ state for 27 qubits \cite{mooney_generation_2021} uses the MQC technique for verification.


\section{Stability lemma and telescope construction}

Quantum states can be encoded into ground states of local Hamiltonians. If one has such a Hamiltonian with a unique ground state $\ket{\psi}$, then measuring its energy with respect to the state $\ket{\phi}$ enables us to estimate the fidelity of $|\braket{\phi}{\psi}|^2$ via the following lemma, extending the stability lemma from \cite{biamonte_universal_2021}.

\begin{proposition}[Stability lemma]
    \label{prop:stability}
    Let $H$ be a Hamiltonian with eigenvalues $0 = \lambda_0 < \Delta = \lambda_1 \leq \lambda_2 \leq ... \leq \lambda_{\mathrm{max}}$ and corresponding eigenvectors $\ket{\lambda_i}$. Let $\rho$ be a density operator such that $E = \Tr \rho H \leq \Delta$. Then
    \begin{equation}
        \label{eq:stability_lemma}
        1 - \frac{E}{\Delta} 
        \leq \bra{\lambda_0} \rho \ket{\lambda_0}
        \leq 1 - \frac{E}{\lambda_{\mathrm{max}}}.
    \end{equation}
\end{proposition}
\begin{proof}
    The trace of $\rho$ is equal to one, so the desired fidelity is a function of all other diagonal elements:
    \begin{equation}
        \label{eq:stability_proof}
        \rho_{00} = \bra{\lambda_0} \rho \ket{\lambda_0} = 1 - \sum_{i > 0} \rho_{ii}.
    \end{equation}
    The condition that $E$ is below the gap lets us find an upper bound on $\sum_{i > 0} \rho_{ii}$:
    \begin{equation}
        \sum_{i > 0} \rho_{ii} = \frac{1}{\Delta} \sum_{i > 0} \rho_{ii} \Delta \leq \frac{1}{\Delta} \sum_{i > 0} \rho_{ii} \lambda_{i} = \frac{E}{\Delta}.
    \end{equation}
    Another bound can be obtained in a similar fashion:
    \begin{equation}
        \sum_{i > 0} \rho_{ii} = \frac{1}{\lambda_{\text{max}}} \sum_{i > 0} \rho_{ii} \lambda_{\text{max}} \geq \frac{1}{\lambda_{\text{max}}} \sum_{i > 0} \rho_{ii} \lambda_{i} = \frac{E}{\lambda_{\text{max}}}.
    \end{equation}
    Substituting these two bounds into (\ref{eq:stability_proof}) yields the desired bounds.
\end{proof}

So, if we want to estimate the fidelity of quantum state $\rho$ with respect to the pure state $\ket{\psi}$ that was intended to be prepared, what we can do is find a Hamiltonian that contains that state in its ground space, measure the expected value of said Hamiltonian, and then find the bounds using Lemma~\ref{prop:stability}. In general, it is not clear how --- and whether it is possible --- to construct such a Hamiltonian so that it contains a polynomial number of Pauli terms, although the state can be encoded in the ground state of a local Hamiltonian using ancilla qubits~\cite{kitaev_classical_2002,biamonte_universal_2021}. However, for a class of states called Clifford states the task of finding a good Hamiltonian is rather easy. To explain the construction, we first define the Clifford group.

\begin{definition}[Clifford group]
    The \emph{Clifford group} $\mathcal{C} \subset \mc{U}(2^n)$ consists of operators $U$ such that their action by conjugation maps Pauli strings to Pauli strings. In other words, the Clifford group is the normalizer of the Pauli group. A quantum state of the type $U \ket{0...0}$ for $U \in \mc{C}$ is called a \emph{Clifford state}.
\end{definition}

What makes this group relevant is that it includes a large class of quantum circuits called \emph{Clifford circuits}.

\begin{proposition}[\cite{nielsen_quantum_2010}]
    The Clifford group is generated by the following gates, called respectively Hadamard, Phase and CNOT:
    \begin{equation}
        H = \frac{1}{\sqrt{2}}\begin{pmatrix}
            1 & 1 \\ 1 & -1
        \end{pmatrix},
        \quad
        P = \begin{pmatrix}
            1 & 0 \\ 0 & \rmi
        \end{pmatrix},
        \quad
        CNOT = \begin{pmatrix}
            1 & 0 & 0 & 0 \\ 
            0 & 1 & 0 & 0 \\ 
            0 & 0 & 0 & 1 \\ 
            0 & 0 & 1 & 0 
        \end{pmatrix}.
    \end{equation}
\end{proposition}

These gates act by conjugation on Pauli strings as follows (we write here only the action on generators of the Pauli group): 

\begin{align}
    H&: X \mapsto Z, Z \mapsto X \\
    P&: X \mapsto Y, Z \mapsto Z \\
    CNOT&: XI \mapsto XX, IX \mapsto IX, ZI \mapsto ZI, IZ \mapsto ZZ
\end{align}

The key fact about the Clifford group is that Clifford circuits can be simulated classically in polynomial time \cite{gottesman_heisenberg_1998,aaronson_improved_2004}. The idea is to keep track not of the Clifford state itself, but of the Pauli strings that stabilize it. The stabilizer subgroup of any Clifford state contains $2^n$ Pauli strings, however it is sufficient to consider just $n$ generating Pauli strings. With each gate, we act on them by conjugation and again receive Pauli strings. Ref.~\cite{aaronson_improved_2004} develops a more advanced algorithm to calculate the stabilizers and to evaluate the observables w.r.t.~the Clifford state.

The telscoping construction uses the idea of Clifford stabilizers. First we consider the state $\ket{0...0}$. Its stabilizer subgroup is generated by one-local Pauli strings $Z_i$. It is straightforward to check that the following Hamiltonian has $\ket{0...0}$ as the only ground state:

\begin{equation}
    H_0 = -\sum_{i = 1}^n Z_i + n \id
\end{equation}

Here the constant operator is added to shift the spectrum of the Hamiltonian, so that the ground state has zero energy. For any unitary $U$, it is true that the state $U \ket{0...0}$ is the unique ground state of the Hamiltonian $U H_0 U^\dagger$. When $U$ is a Clifford circuit, $H_0$ consists of $(n+1)$ Pauli strings. Thus we define this to be the \emph{telescope Hamiltonian} for the given Clifford circuit.

\begin{definition}
    A \emph{telescope Hamiltonian} for a Clifford circuit $U_L ... U_1$ is the following Hamiltonian:
    \begin{equation}
        \label{eq:telescope}
        H_{\text{tele}} = U_L ... U_1 H_0 U_1^\dagger ... U_L^\dagger.
    \end{equation}
\end{definition}

In particular, we can explicitly find the telescope Hamiltonian for the GHZ state. Recall that a standard circuit for preparing it looks like this:
\begin{equation*}
    \Qcircuit @C=1.0em @R=1.0em {
       & \lstick{\ket{0}} & \gate{H} & \ctrl{1} 
       & \qw & \qw & \qw  & \qw
       \\
       & \lstick{\ket{0}} & \qw & \targ 
       & \ctrl{1} & \qw & \qw & \qw
       \\
       & \lstick{\ket{0}} & \qw & \qw
       & \targ & \ctrl{1} & \qw & \qw
       \\
       & \lstick{\ket{0}} & \qw & \qw
       & \qw & \targ & \qw & \qw
       \\ & \ & \ & \ & ... & \ 
       \\
       & \lstick{\ket{0}} & \qw & \qw
       & \qw & \qw & \ctrl{1} & \qw
       \\
       & \lstick{\ket{0}} & \qw & \qw
       & \qw  & \qw & \targ & \qw
    }
\end{equation*}
By acting with Clifford gates by conjugation on $H_0$, we arrive to the following Hamiltonian:

\begin{equation}
    H = -\sum_{i=1}^{n-1} Z_i Z_{i+1} - X^{\otimes n} + n.
\end{equation}

Its only ground state is the GHZ state, and the constant term sets the ground state energy to zero. Thus, we can estimate the fidelity of a candidate GHZ state. Using Lemma~\ref{prop:stability}, we obtain
\begin{equation}
    1 - \frac{\langle H_\mathrm{tele} \rangle}{2}  \leq F \leq 1 - \frac{\langle H_\mathrm{tele} \rangle}{2n}.
\end{equation}
Evaluating this energy requires two series of measurements: in the $Z$ basis and in the $X$ basis. Unlike the parity oscillations technique, there is no need to scan through different values of $\varphi$.

\section{Numerical experiments}

\subsection{Local depolarizing errors}

It is interesting to compare the scaling of errors for different fidelity estimation techniques, e.g.~as a function of the number of shots and of the number of qubits. There is a certain difficulty with that: the results also heavily depend on the noise model. In the following, we adopt a simple noise model. We assume that every single-qubit (two-qubit) gate is followed by a single-qubit (two-qubit) depolarizing noise channel with depolarization probability $p_1$ ($p_2$).

In discussion of errors, we need to stress that we can estimate random errors, but not necessarily systematic errors. For example, the parity oscillation method requires us to measure qubits in the $X$ basis, but this means applying an extra Hadamard gate to every qubit before measurement. As a result, in the depolarizing noise model, this leads to a systematic underestimation of coherence. Random errors are, on the other hand, analyzable using the measurement statistics as usual.

We estimate the random errors of population and telescopic bounds by standard techniques. That is, the standard error of the mean is the square root of the sample variance divided by the number of measurements $\sqrt{N_{shots}}$. The coherence measurements are more complicated. In the case of parity oscillations, the coherence is a parameter obtained from a curve fit, while for MQC the coherence is obtained as a specific Fourier coefficient of the data. For these two estimates, the standard error is calculated by parametrized bootstrap resampling (using $B=100$ bootstraps). More specifically, for each value of $\varphi$, we treat the measurement result as having two outcomes (''even'' or ''odd'' for parity, ''0'' or ''1'' for multiple quantum coherences). Then, we infer the success probability $q_\varphi$ for each point. Then, we prepare the bootstrap measurement data by sampling from the binomial distributions $\mathcal{B}(q_\varphi, N_{shots})$. Finally, we obtain the coherence estimate. Repeating this $B$ times yields $B$ bootstrap coherence estimates $\hat{C}_i$. The variance of these coherences is our estimate of the standard error.

Figure \ref{fig:fidelity_terms} shows the behavior of the estimates of different terms as we increase the number of shots that are used in the experiment. The experiment was conducted for $n=8$ qubits with the depolarizing noise model ($p_1 = 0.001, p_2 = 0.01$). The telescope bounds in this case are substantially looser than the estimates from the other techniques, except for the small shot counts.

Figure \ref{fig:errors_as_nq} shows the dependence of the standard error on the number of qubits $n$. The number of shots per point was fixed at $2^{15}$. Surprisingly, the multiple quantum coherences method appears to show the worst scaling of the error, while the parity oscillations method gives an almost constant random error. The telescopic lower bound has a random error steadily growing with the size of the system. The upper bound behaves as $1/N$ by construction, so its error appears to decrease.

\begin{figure}
    \centering
    \includegraphics[width=0.7\textwidth]{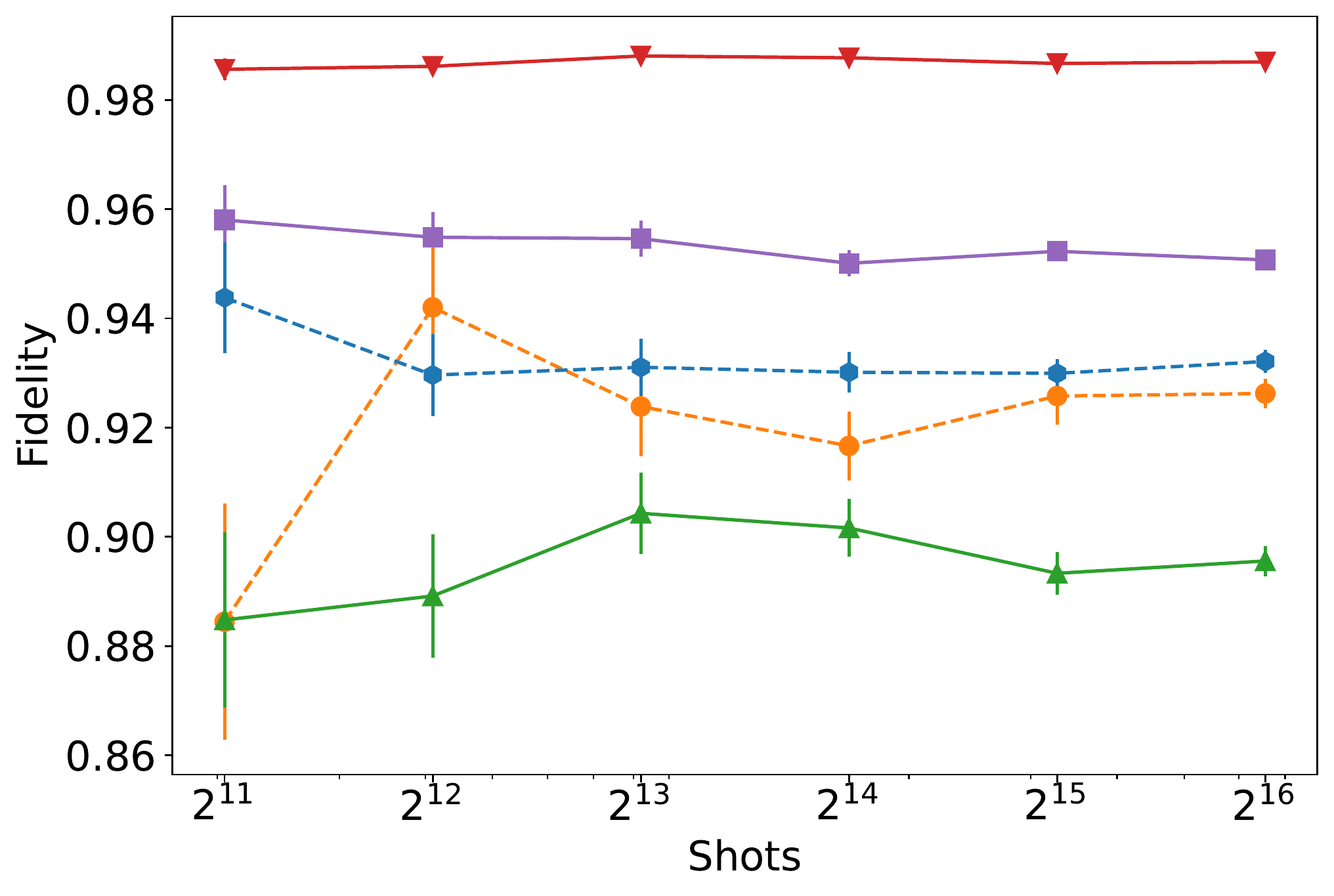}
    \caption{Sample size dependence of telescope lower ($\blacktriangle$) and upper ($\blacktriangledown$) bounds, population estimate ($\blacksquare$), and coherence estimates using parity oscillations ($\bullet$) and MQC ($\hexagofill$). Vertical lines denote one standard error.}
    \label{fig:fidelity_terms}
\end{figure}

Finally, Figure \ref{fig:errors_as_nshots} shows the dependence of random errors on the number of shots, while everything else is fixed. As expected, for all methods, the error behaves as $1/\sqrt{N_{shots}}$.

\begin{figure}
    \centering
    \includegraphics[width=0.7\textwidth]{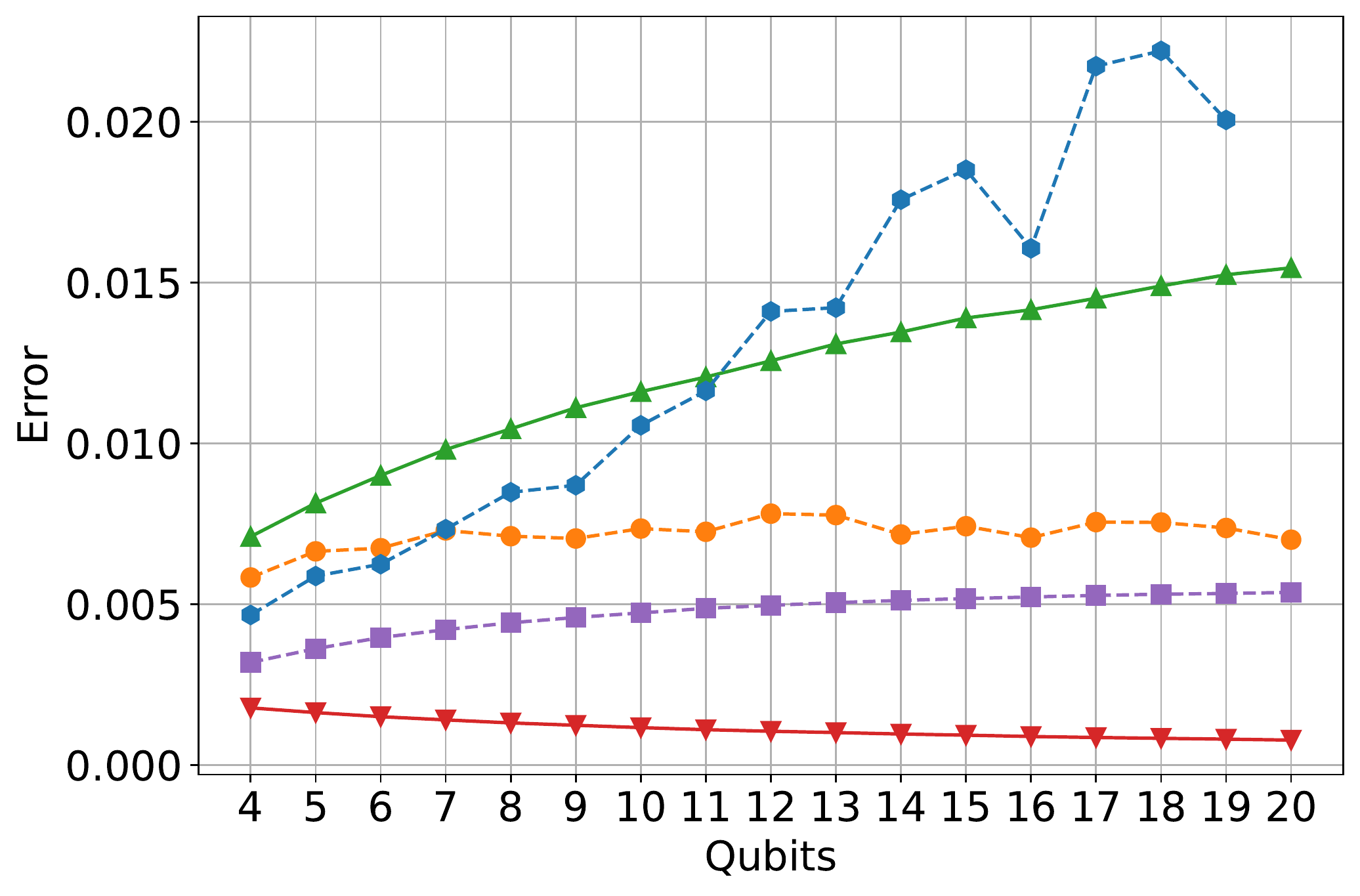}
    \caption{Estimates of random errors of telescope lower ($\blacktriangle$) and upper ($\blacktriangledown$) bounds, population estimate ($\blacksquare$), and coherence estimates using parity oscillations ($\bullet$) and MQC ($\hexagofill$). The errors for population and telescopic bounds are estimated by taking the standard error of the mean, the coherence errors are estimated by parametrized bootstrap resampling (see main text).}
    \label{fig:errors_as_nq}
\end{figure}

\begin{figure}
    \centering
    \includegraphics[width=0.7\textwidth]{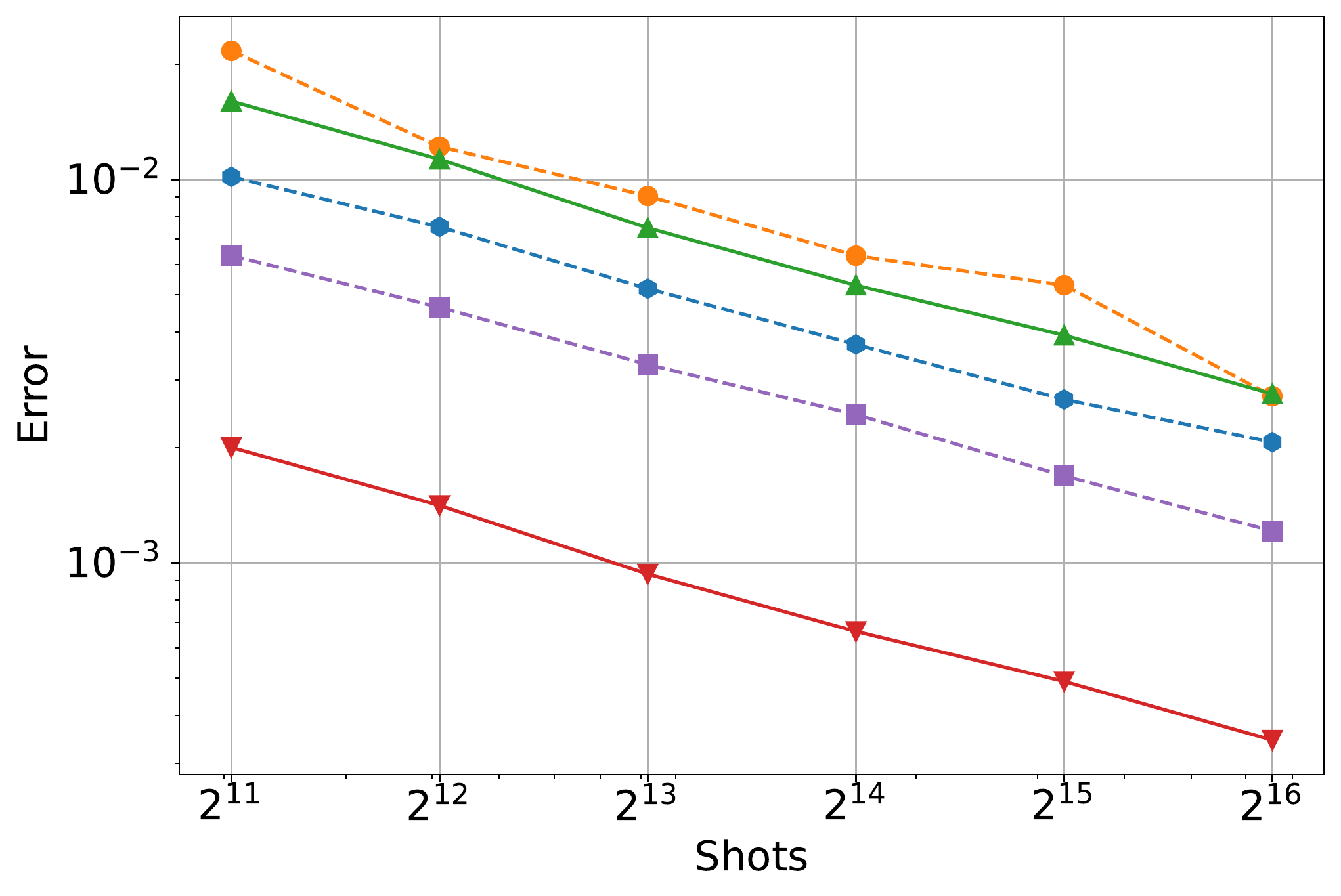}
    \caption{Sample size dependence of the standard errors in the terms comprising different fidelity estimates. The markers match those in Fig.~\ref{fig:fidelity_terms}.}
    \label{fig:errors_as_nshots}
\end{figure}

\subsection{Global depolarizing error}

We analyzed the behavior of all three methods for the following family of states:
\begin{equation}
    \label{eq:partially_mixed}
    \rho(\alpha) = (1-\alpha) \ket{\Phi} \bra{\Phi} + \frac{\alpha}{2^n} \id.
\end{equation}
Here $\ket{\Phi}$ is the GHZ state. For that density matrix, we can exactly calculate the fidelity and the bounds provided by the telescope construction. The fidelity with the clean GHZ state is equal to 
\begin{equation}
    F = \bra{\Phi} \rho \ket{\Phi} = 1 - \alpha + \frac{\alpha}{2^n}.
\end{equation}
The energy of this state with respect to the telescope Hamiltonian is equal to $\frac{\alpha}{2^n} \Tr H = n \alpha$. Hence, the fidelity bounds provided by Proposition~\ref{prop:stability} are as follows:
\begin{equation}
    1 - \frac{n \alpha}{2} \leq F \leq 1 - \frac{\alpha}{2}.
\end{equation}
The lower bound deteriorates with the increase of $n$, becoming completely uninfomative at $\alpha = 2/n$. The upper bound, on the other hand, does not depend on $n$.

All three methods were tested numerically for determining the fidelity of $\rho$. In case of parity oscillations and multiple quantum coherences techniques, the simulation is somewhat artifical. The matter is that both of these techniques are substantially affected by measurement errors, which are not accounted for in this model. Instead, the methods receive the state $\rho$, then perform all necessary actions with perfect fidelity.

For every tested value of $\alpha$, 1000 measurements were allocated for measuring coherence (or, for the telescope, the expected value of $X^{\otimes n}$), and 1000 measurements were allocated for measuring population (respectively, the expectations of $Z_i Z_{i+1}$). For coherence measurements, the values of $\varphi$ were equal to $\frac{2 \pi j}{2 n + 2}$, $j \in \{0, 1, ..., 2n+1\}$.
Figure~\ref{fig:fidelity_depo} shows the results of the simulations. The experiments were performed for $n=10$ qubits. As expected from the analytical formulas, the bounds provided by the telescope construction are quite loose in this case, while the parity oscillations technique ends up being much closer to the correct value. Interestingly enough, the multiple quantum coherences method in this case consistently overestimates the fidelity: we repeated the same experiment multiple times and obtaied qualitatively the same picture.

\begin{figure}
    \centering
    \includegraphics[width=0.7\textwidth]{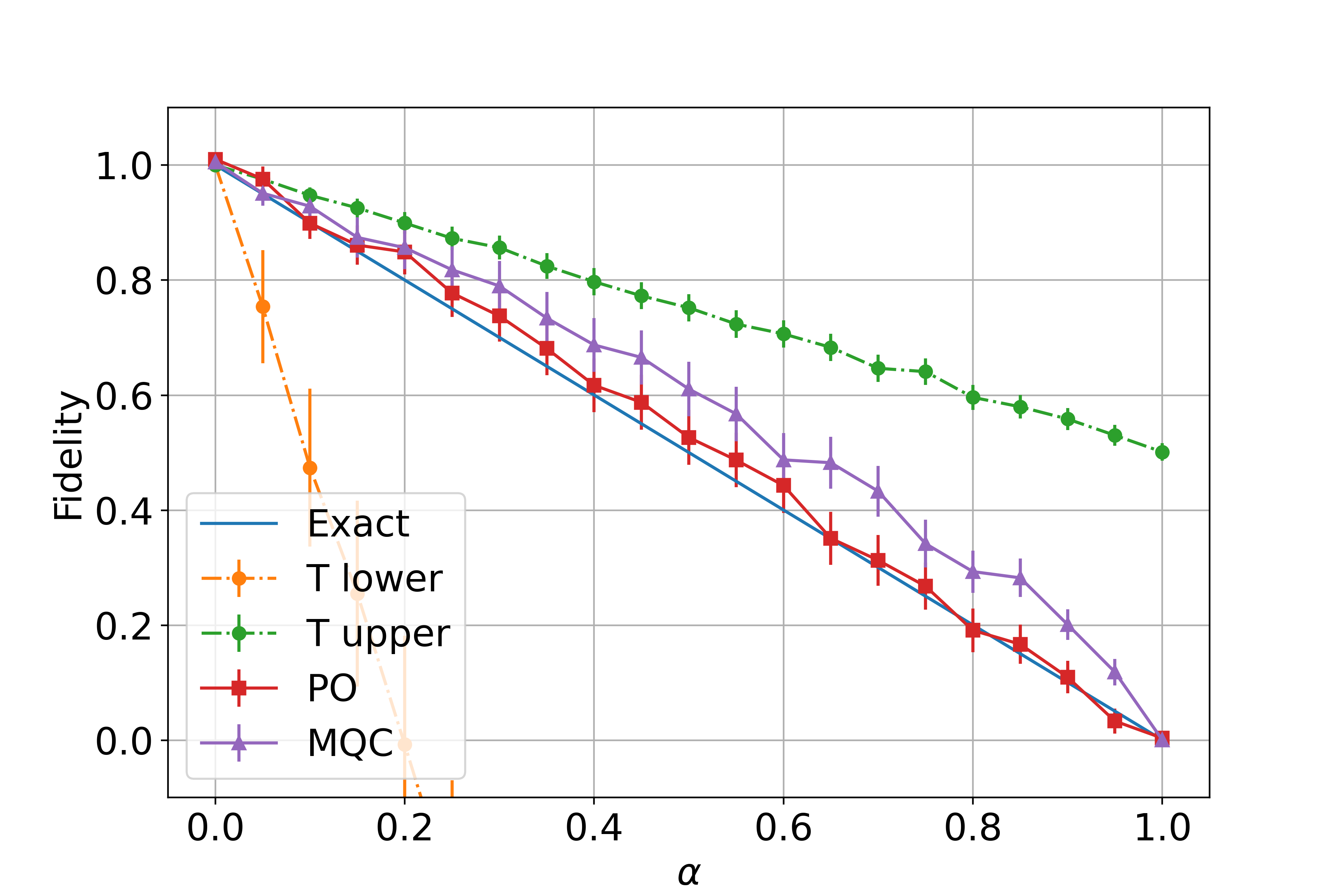}
    \caption{Fidelity estimates for mixed states of the form~(\ref{eq:partially_mixed}).}
    \label{fig:fidelity_depo}
\end{figure}

\section{Discussion}

The method we proposed relies on the properties of the Clifford group. Another important method of characterizing quantum devices, called randomized benchmarking (RB), also relies on the properties of the Clifford group. The goal of RB is to estimate the average error rate per gate. The standard variant of RB, described in \cite{magesan_robust_2011-1}, works as follows. 

\begin{enumerate}
    \item Pick $m$ elements from the Clifford group uniformly at random\footnote{This is not the same as $m$ Hadamard, Phase or CNOT gates.}. Calculate their inverse and append to the list of elements. 
    \item Prepare the circuit and estimate the probability of observing $\ket{0...0}$
    \item Repeat steps 1-2 and calculate the average probability $F$, called average sequence fidelity.
    \item Repeat steps 1-3 for different values of $m$ and fit parameters to the law 
    \begin{equation}
        F = Ap^m + B
    \end{equation}
\end{enumerate}

The resulting $p$ is then used to calculate the average error rate per gate $r$: 

\begin{equation}
    r = (1-p)\left(1 - \frac{1}{2^n}\right).
\end{equation}

A more complicated expression is shown in Ref.~\cite{magesan_robust_2011-1} for gate-dependent error models. The constants $A, B$ account for state preparation and measurement (SPAM) errors.

A slightly different protocol for RB is described by Knill et al.~\cite{knill_randomized_2008}. The difference is that, after generating a Clifford circuit, we append it not by its inverse, but by a local circuit that would diagonalize a randomly picked Pauli string from the stabilizer generators of the state. Effectively, the figure calculated is then equal to the average expectation of the Pauli string $\sigma$, picked from the stabilizer generators of the state, averaged over the Clifford group. For every fixed Clifford state, this expectation is equal to $\langle \sigma \rangle = 1 - (1/n)\langle H_\mathrm{tele} \rangle$. The fidelity of that Clifford state then can be estimated using Lemma \ref{prop:stability}.
\begin{gather}
    1 - \frac{\langle H_\mathrm{tele} \rangle}{2}  \leq F \leq 1 - \frac{\langle H_\mathrm{tele} \rangle}{2n} \\
    1 + \frac{n \langle \sigma \rangle}{2} - \frac{n}{2}  \leq F \leq \frac{1}{2} + \frac{\langle \sigma \rangle}{2}.
\end{gather}

\section{Conclusions}

To conclude, we proposed a method of bounding the fidelity of the GHZ state by means of measuring the terms from its parent Hamiltonian. Although the GHZ state is the most salient witness of genuine multipartite entanglement, the proposed verification technique works for any Clifford state. Moreover, since unitary transformations preserve the spectrum of Hamiltonians, the method still works for non-Clifford circuits, albeit the number of observables to be evaluated grows exponentially with the number of non-Clifford gates~\cite{biamonte_universal_2021}.

Compared to the existing methods of estimating the fidelity of the GHZ state, our method has both advantages and drawbacks. The obvious drawback is that the method only gives bounds to fidelity and not the fidelity itself. On the other hand, the bounds give some information even for small series of shots. Another benefit of the proposed technique lies in its simplicity: both the parity oscillations and multiple quantum coherences methods require implementing rotation gates with angle resolution scaling as $1/n$, whereas in our method there is no such requirement. Finally, the method makes no assumptions about the nature of the noise in the device.
\addcontentsline{toc}{chapter}{Conclusions}
\chapter*{Conclusions}

In the age of NISQ computers, variational quantum algorithms are one of the central research topics. In this thesis, we addressed this topic.

In chapter \ref{chap:quantum_basics}, we introduced the basic notions of quantum computation. In chapter \ref{chap:vqas}, we reviewed the state of the art in variational quantum algorithms. Our contribution starts from chapter \ref{chap:vqe_numerics}, where we investigated the ability of VQE to find the ground state of different physical models. We studied how quality of the solution depends on the depth of the circuit and on the parameters of the model. We also touched the subject of barren plateaus, although a more detailed investigation is postponed to chapter \ref{chap:plateaus}.
In chapter \ref{chap:qml}, we presented our work on machine learning with variational quantum algorithms. The goal was to train a variational circuit on quantum states obtained from VQE. In chapter \ref{chap:plateaus}, we investigated the emergence of barren plateaus in optimization landscapes. Our main tools were the approximation of individual circuit blocks with unitary 2-designs and investigation of the phenomenon in the Heisenberg picture.
Finally, in chapter \ref{chap:ghz}, we proposed an algorithm to estimate the fidelity of the GHZ state prepared in a quantum register. This algorithm works for any Clifford state and, although it tends to give loose bounds, it is applicable to any Clifford state.

The results presented in different chapters give an interesting insight when put together. On the one hand, VQE for physical problems shows exponential convergence with the number of layers. This observation matches independent observations that are made in \cite{cade_strategies_2019} and \cite{bravo-prieto_scaling_2020}, although in the latter two substantially different regimes were found, with the tipping point for a critical problem being approximately linear in the number of qubits. On the other hand, the barren plateaus become a significant issue after logarithmically many layers. Combining the two observations, we can see that non-critical problems should be relatively easy for VQE, while for critical problems one will have to employ a more clever strategy of optimization than naive VQE.

The main results of the work are as follows:
\begin{enumerate}
    \item We developed a numerical implementation of the VQE algorithm. Using this implementation, we investigated the behavior of the solutions for the transverse-field Ising model, anisotropic Heisenberg model, and a spinless variant of the Hubbard model with next-nearest-neighbor interactions. 
    \item The states near the phase transition point are the hardest to approximate with a variational circuit. We found a hysteresis effect in the adiabatic-assisted VQE, meaning that going between two easy Hamiltonians through a difficult region yields two different results depending on the direction. In all models, the scaling of the error with circuit depth was close to exponential, agreeing with existing literature.
    \item The barren plateaus effect for short-depth circuits sets on at a different pace depending on the choice of the fermion-to-qubit encoding. In particular, the derivatives vanish essentially immediately for the Jordan-Wigner transform and more gradually for the Bravyi-Kitaev transform.
    \item Variational quantum circuits can be optimized (trained) to distinguish the phases of quantum models. This works both for a simple Ising model, where the phase transition can be detected with a simple observable, and for the Heisenberg model, where the transition is harder to detect.
    \item We derived a lower bound on the variance of cost function derivatives in variational quantum circuits. The bound mainly depends on two things: the size of the causal cones of the operators in the cost function Hamiltonian, and the position of the gate in the circuit.
    \item We proposed a technique for bounding the GHZ state fidelity. Unlike state-of-the-art methods, this technique does not require the ability to fine-tune the angles of the rotation gates, nor does it rely on any assumptions about the noise.
\end{enumerate}



\addcontentsline{toc}{chapter}{List of symbols and abbreviations}
\chapter*{List of symbols and abbreviations}

\begin{tabularx}{0.9\textwidth}{lX}
    $\id$ & Identity matrix \\
    $A^*$ & Complex conjugate of a matrix \\
    $A^\dagger$ & Conjugate transpose (Hermitian conjugate) of a matrix \\
    AAVQE & Adiabatically-assisted variational quantum eigensolver \\
    $\operatorname{Ad}_M$ & The operator of conjugation by an invertible matrix $M$: $X \mapsto M X M^{-1}$. \\
    $||A||_p$ & Operator $p$-norm of a linear operator \\
    BFGS & Broyden-Fletcher-Goldfarb-Shanno algorithm \\
    $\mathbbm{C}$ & The field of complex numbers \\
    $\operatorname{card} H$ & Number of nonzero terms in the Pauli decomposition of a matrix $H$ \\
    $\operatorname{dim} V$ & Dimension of a vector space $V$ \\
    $\operatorname{End}(V)$ & The space of linear operators on a vector space $V$ \\
    $\operatorname{Herm}(d)$ & The space of Hermitian operators on $\mathbbm{C}^d$ \\
    $\rmi$ & Imaginary unit \\
    $\operatorname{loc} H$ & Maximum algebraic locality of Pauli strings in the Pauli decomposition of a matrix $H$ \\
    NISQ & Noisy intermediate-size quantum (computing)\\
    $\ket{\psi}$ & A pure quantum state \\
    $\mathbbm{R}$ & The field of real numbers \\
    SGD & Stochastic gradient descent \\
    TFI & Transverse field Ising (model) \\
    $\mc{U}(d)$ & The group of $d \times d$ unitary matrices \\
    VQA & Variational quantum algorithms \\
    VQE & Variational quantum eigensolver \\
    $X_i, Y_i, Z_i$ & Pauli matrices acting on qubit number $i$ \\
\end{tabularx}        
\clearpage                                  
\urlstyle{rm}                               
\ifdefmacro{\microtypesetup}{\microtypesetup{protrusion=false}}{} 
\insertbibliofull                           
\ifdefmacro{\microtypesetup}{\microtypesetup{protrusion=true}}{}
\urlstyle{tt}                               
\clearpage
\ifdefmacro{\microtypesetup}{\microtypesetup{protrusion=false}}{} 
\listoffigures  

\clearpage
\listoftables   
\ifdefmacro{\microtypesetup}{\microtypesetup{protrusion=true}}{}
\newpage           

\setcounter{totalchapter}{\value{chapter}} 

\appendix
\setlength{\midchapskip}{20pt}
\renewcommand*{\afterchapternum}{\par\nobreak\vskip \midchapskip}

\ifnumequal{\value{englishthesis}}{0}{
    \renewcommand\thechapter{\Asbuk{chapter}} 
}{}

\setcounter{totalappendix}{\value{chapter}} 


@article{horodecki_five_2022,
  title = {Five {{Open Problems}} in {{Quantum Information Theory}}},
  author = {Horodecki, Paweł and Rudnicki, Łukasz and Życzkowski, Karol},
  date = {2022-03-03},
  journaltitle = {PRX Quantum},
  shortjournal = {PRX Quantum},
  volume = {3},
  number = {1},
  pages = {010101},
  issn = {2691-3399},
  doi = {10.1103/PRXQuantum.3.010101},
  url = {https://link.aps.org/doi/10.1103/PRXQuantum.3.010101},
  urldate = {2022-04-06},
  langid = {english}
}

@article{aaronson_shadow_2018,
title = {Shadow {{Tomography}} of {{Quantum States}}},
author = {Aaronson, Scott},
date = {2018},
journaltitle = {Los Angeles},
pages = {14},
langid = {english}
}

@article{huang_predicting_2020,
title = {Predicting {{Many Properties}} of a {{Quantum System}} from {{Very Few Measurements}}},
author = {Huang, Hsin-Yuan and Kueng, Richard and Preskill, John},
date = {2020-10},
journaltitle = {Nature Physics},
shortjournal = {Nat. Phys.},
volume = {16},
number = {10},
eprint = {2002.08953},
eprinttype = {arxiv},
pages = {1050--1057},
issn = {1745-2473, 1745-2481},
doi = {10.1038/s41567-020-0932-7},
url = {http://arxiv.org/abs/2002.08953},
urldate = {2020-12-24},
archiveprefix = {arXiv},
langid = {english}
}

@unpublished{koh_classical_2020,
title = {Classical {{Shadows}} with {{Noise}}},
author = {Koh, Dax Enshan and Grewal, Sabee},
date = {2020-11-23},
eprint = {2011.11580},
eprinttype = {arxiv},
primaryclass = {math-ph, physics:quant-ph},
url = {http://arxiv.org/abs/2011.11580},
urldate = {2021-04-15},
archiveprefix = {arXiv},
langid = {english}
}

@misc{dariano_quantum_2003,
title = {Quantum {{Tomography}}},
author = {D'Ariano, G. Mauro and Paris, Matteo G. A. and Sacchi, Massimiliano F.},
date = {2003-02-04},
number = {arXiv:quant-ph/0302028},
eprint = {quant-ph/0302028},
eprinttype = {arxiv},
publisher = {{arXiv}},
url = {http://arxiv.org/abs/quant-ph/0302028},
urldate = {2022-06-13},
archiveprefix = {arXiv}
}

@article{straupe_adaptive_2016,
title = {Adaptive Quantum Tomography},
author = {Straupe, S. S.},
date = {2016-10},
journaltitle = {JETP Letters},
shortjournal = {Jetp Lett.},
volume = {104},
number = {7},
pages = {510--522},
issn = {0021-3640, 1090-6487},
doi = {10.1134/S0021364016190024},
url = {http://link.springer.com/10.1134/S0021364016190024},
urldate = {2021-04-02},
langid = {english}
}

@article{knill_randomized_2008,
  title = {Randomized Benchmarking of Quantum Gates},
  author = {Knill, E. and Leibfried, D. and Reichle, R. and Britton, J. and Blakestad, R. B. and Jost, J. D. and Langer, C. and Ozeri, R. and Seidelin, S. and Wineland, D. J.},
  date = {2008-01-08},
  journaltitle = {Physical Review A},
  shortjournal = {Phys. Rev. A},
  volume = {77},
  number = {1},
  pages = {012307},
  issn = {1050-2947, 1094-1622},
  doi = {10.1103/PhysRevA.77.012307},
  url = {https://link.aps.org/doi/10.1103/PhysRevA.77.012307},
  urldate = {2022-02-15},
  langid = {english}
}

@article{magesan_robust_2011-1,
  title = {Robust Randomized Benchmarking of Quantum Processes},
  author = {Magesan, Easwar and Gambetta, J. M. and Emerson, Joseph},
  date = {2011-05-06},
  journaltitle = {Physical Review Letters},
  shortjournal = {Phys. Rev. Lett.},
  volume = {106},
  number = {18},
  eprint = {1009.3639},
  eprinttype = {arxiv},
  primaryclass = {quant-ph},
  pages = {180504},
  issn = {0031-9007, 1079-7114},
  doi = {10.1103/PhysRevLett.106.180504},
  url = {http://arxiv.org/abs/1009.3639},
  urldate = {2022-05-25},
  archiveprefix = {arXiv},
  langid = {english}
}

@article{hahn_spin_1950,
	title = {Spin Echoes},
	volume = {80},
	issn = {0031-899X},
	url = {https://link.aps.org/doi/10.1103/PhysRev.80.580},
	doi = {10.1103/PhysRev.80.580},
	pages = {580--594},
	number = {4},
	journaltitle = {Physical Review},
	shortjournal = {Phys. Rev.},
	author = {Hahn, E. L.},
	urldate = {2022-05-04},
	date = {1950-11-15},
	langid = {english}
}

@article{garttner_relating_2018,
	title = {Relating Out-of-Time-Order Correlations to Entanglement via Multiple-Quantum Coherences},
	volume = {120},
	issn = {0031-9007, 1079-7114},
	url = {https://link.aps.org/doi/10.1103/PhysRevLett.120.040402},
	doi = {10.1103/PhysRevLett.120.040402},
	pages = {040402},
	number = {4},
	journaltitle = {Physical Review Letters},
	shortjournal = {Phys. Rev. Lett.},
	author = {Gärttner, Martin and Hauke, Philipp and Rey, Ana Maria},
	urldate = {2022-03-22},
	date = {2018-01-24},
	langid = {english}
}

@article{cerezo_variational_2020,
	title = {Variational Quantum Algorithms},
	url = {http://arxiv.org/abs/2012.09265},
	journaltitle = {{arXiv}:2012.09265 [quant-ph, stat]},
	author = {Cerezo, M. and Arrasmith, Andrew and Babbush, Ryan and Benjamin, Simon C. and Endo, Suguru and Fujii, Keisuke and {McClean}, Jarrod R. and Mitarai, Kosuke and Yuan, Xiao and Cincio, Lukasz and Coles, Patrick J.},
	urldate = {2020-12-21},
	date = {2020-12-16},
	langid = {english},
	eprinttype = {arxiv},
	eprint = {2012.09265}
}

@article{bharti_noisy_2021,
	title = {Noisy intermediate-scale quantum ({NISQ}) algorithms},
	url = {http://arxiv.org/abs/2101.08448},
	journaltitle = {{arXiv}:2101.08448 [cond-mat, physics:quant-ph]},
	author = {Bharti, Kishor and Cervera-Lierta, Alba and Kyaw, Thi Ha and Haug, Tobias and Alperin-Lea, Sumner and Anand, Abhinav and Degroote, Matthias and Heimonen, Hermanni and Kottmann, Jakob S. and Menke, Tim and Mok, Wai-Keong and Sim, Sukin and Kwek, Leong-Chuan and Aspuru-Guzik, Alán},
	urldate = {2021-01-22},
	date = {2021-01-21},
	langid = {english},
	eprinttype = {arxiv},
	eprint = {2101.08448}
}

@article{mehta_high-bias_2019,
	title = {A high-bias, low-variance introduction to Machine Learning for physicists},
	volume = {810},
	issn = {03701573},
	url = {https://linkinghub.elsevier.com/retrieve/pii/S0370157319300766},
	doi = {10.1016/j.physrep.2019.03.001},
	pages = {1--124},
	journaltitle = {Physics Reports},
	shortjournal = {Physics Reports},
	author = {Mehta, Pankaj and Bukov, Marin and Wang, Ching-Hao and Day, Alexandre G.R. and Richardson, Clint and Fisher, Charles K. and Schwab, David J.},
	urldate = {2020-04-30},
	date = {2019-05},
	langid = {english}
}

@article{ruder_overview_2017,
	title = {An overview of gradient descent optimization algorithms},
	url = {http://arxiv.org/abs/1609.04747},
	journaltitle = {{arXiv}:1609.04747 [cs]},
	author = {Ruder, Sebastian},
	urldate = {2022-03-29},
	date = {2017-06-15},
	eprinttype = {arxiv},
	eprint = {1609.04747}
}

@article{keskar_large-batch_2017,
	title = {On Large-Batch Training for Deep Learning: Generalization Gap and Sharp Minima},
	url = {http://arxiv.org/abs/1609.04836},
	shorttitle = {On Large-Batch Training for Deep Learning},
	journaltitle = {{arXiv}:1609.04836 [cs, math]},
	author = {Keskar, Nitish Shirish and Mudigere, Dheevatsa and Nocedal, Jorge and Smelyanskiy, Mikhail and Tang, Ping Tak Peter},
	urldate = {2022-03-29},
	date = {2017-02-09},
	eprinttype = {arxiv},
	eprint = {1609.04836}
}

@article{ho_efficient_2019,
	title = {Efficient variational simulation of non-trivial quantum states},
	volume = {6},
	issn = {2542-4653},
	url = {https://scipost.org/10.21468/SciPostPhys.6.3.029},
	doi = {10.21468/SciPostPhys.6.3.029},
	pages = {029},
	number = {3},
	journaltitle = {{SciPost} Physics},
	shortjournal = {{SciPost} Phys.},
	author = {Ho, Wen Wei and Hsieh, Timothy H.},
	urldate = {2020-07-31},
	date = {2019-03-07}
}

@article{anand_quantum_2022,
	title = {A quantum computing view on unitary coupled cluster theory},
	volume = {51},
	issn = {0306-0012, 1460-4744},
	url = {http://xlink.rsc.org/?DOI=D1CS00932J},
	doi = {10.1039/D1CS00932J},
	pages = {1659--1684},
	number = {5},
	journaltitle = {Chemical Society Reviews},
	shortjournal = {Chem. Soc. Rev.},
	author = {Anand, Abhinav and Schleich, Philipp and Alperin-Lea, Sumner and Jensen, Phillip W. K. and Sim, Sukin and Díaz-Tinoco, Manuel and Kottmann, Jakob S. and Degroote, Matthias and Izmaylov, Artur F. and Aspuru-Guzik, Alán},
	urldate = {2022-03-23},
	date = {2022},
	langid = {english}
}

@article{barison_efficient_2021,
	title = {An efficient quantum algorithm for the time evolution of parameterized circuits},
	volume = {5},
	issn = {2521-327X},
	url = {https://quantum-journal.org/papers/q-2021-07-28-512/},
	doi = {10.22331/q-2021-07-28-512},
	pages = {512},
	journaltitle = {Quantum},
	shortjournal = {Quantum},
	author = {Barison, Stefano and Vicentini, Filippo and Carleo, Giuseppe},
	urldate = {2021-08-03},
	date = {2021-07-28},
	langid = {english}
}

@article{mooney_generation_2021,
	title = {Generation and verification of 27-qubit Greenberger-Horne-Zeilinger states in a superconducting quantum computer},
	volume = {5},
	issn = {2399-6528},
	url = {https://iopscience.iop.org/article/10.1088/2399-6528/ac1df7},
	doi = {10.1088/2399-6528/ac1df7},
	pages = {095004},
	number = {9},
	journaltitle = {Journal of Physics Communications},
	shortjournal = {J. Phys. Commun.},
	author = {Mooney, Gary J and White, Gregory A L and Hill, Charles D and Hollenberg, Lloyd C L},
	urldate = {2022-02-17},
	date = {2021-09-01},
	langid = {english}
}

@article{wei_verifying_2020,
	title = {Verifying multipartite entangled Greenberger-Horne-Zeilinger states via multiple quantum coherences},
	volume = {101},
	issn = {2469-9926, 2469-9934},
	url = {https://link.aps.org/doi/10.1103/PhysRevA.101.032343},
	doi = {10.1103/PhysRevA.101.032343},
	pages = {032343},
	number = {3},
	journaltitle = {Physical Review A},
	shortjournal = {Phys. Rev. A},
	author = {Wei, Ken X. and Lauer, Isaac and Srinivasan, Srikanth and Sundaresan, Neereja and {McClure}, Douglas T. and Toyli, David and {McKay}, David C. and Gambetta, Jay M. and Sheldon, Sarah},
	urldate = {2022-02-18},
	date = {2020-03-25},
	langid = {english}
}

@article{leibfried_toward_2004,
	title = {Toward Heisenberg-Limited Spectroscopy with Multiparticle Entangled States},
	volume = {304},
	issn = {0036-8075, 1095-9203},
	url = {https://www.science.org/doi/10.1126/science.1097576},
	doi = {10.1126/science.1097576},
	pages = {1476--1478},
	number = {5676},
	journaltitle = {Science},
	shortjournal = {Science},
	author = {Leibfried, D. and Barrett, M. D. and Schaetz, T. and Britton, J. and Chiaverini, J. and Itano, W. M. and Jost, J. D. and Langer, C. and Wineland, D. J.},
	urldate = {2022-01-25},
	date = {2004-06-04},
	langid = {english}
}

@article{sackett_experimental_2000,
	title = {Experimental entanglement of four particles},
	volume = {404},
	issn = {0028-0836, 1476-4687},
	url = {http://www.nature.com/articles/35005011},
	doi = {10.1038/35005011},
	pages = {256--259},
	number = {6775},
	journaltitle = {Nature},
	shortjournal = {Nature},
	author = {Sackett, C. A. and Kielpinski, D. and King, B. E. and Langer, C. and Meyer, V. and Myatt, C. J. and Rowe, M. and Turchette, Q. A. and Itano, W. M. and Wineland, D. J. and Monroe, C.},
	urldate = {2022-02-15},
	date = {2000-03},
	langid = {english}
}

@article{verteletskyi_measurement_2020,
	title = {Measurement optimization in the variational quantum eigensolver using a minimum clique cover},
	volume = {152},
	issn = {0021-9606, 1089-7690},
	url = {http://aip.scitation.org/doi/10.1063/1.5141458},
	doi = {10.1063/1.5141458},
	pages = {124114},
	number = {12},
	journaltitle = {The Journal of Chemical Physics},
	shortjournal = {J. Chem. Phys.},
	author = {Verteletskyi, Vladyslav and Yen, Tzu-Ching and Izmaylov, Artur F.},
	urldate = {2021-03-22},
	date = {2020-03-31},
	langid = {english}
}

@article{zhu_multiqubit_2017,
	title = {Multiqubit Clifford groups are unitary 3-designs},
	volume = {96},
	issn = {2469-9926, 2469-9934},
	url = {https://link.aps.org/doi/10.1103/PhysRevA.96.062336},
	doi = {10.1103/PhysRevA.96.062336},
	pages = {062336},
	number = {6},
	journaltitle = {Physical Review A},
	shortjournal = {Phys. Rev. A},
	author = {Zhu, Huangjun},
	urldate = {2022-02-11},
	date = {2017-12-29},
	langid = {english}
}

@article{harrow_approximate_2018,
	title = {Approximate unitary t-designs by short random quantum circuits using nearest-neighbor and long-range gates},
	url = {http://arxiv.org/abs/1809.06957},
	journaltitle = {{arXiv}:1809.06957 [quant-ph]},
	author = {Harrow, Aram and Mehraban, Saeed},
	urldate = {2020-01-24},
	date = {2018-09-18},
	eprinttype = {arxiv},
	eprint = {1809.06957}
}

@article{webb_clifford_2016,
	title = {The Clifford group forms a unitary 3-design},
	url = {http://arxiv.org/abs/1510.02769},
	journaltitle = {{arXiv}:1510.02769 [quant-ph]},
	author = {Webb, Zak},
	urldate = {2022-02-11},
	date = {2016-11-08},
	langid = {english},
	eprinttype = {arxiv},
	eprint = {1510.02769}
}

@inproceedings{ambainis_private_2000,
	location = {Redondo Beach, {CA}, {USA}},
	title = {Private quantum channels},
	isbn = {978-0-7695-0850-4},
	url = {http://ieeexplore.ieee.org/document/892142/},
	doi = {10.1109/SFCS.2000.892142},
	eventtitle = {41st Annual Symposium on Foundations of Computer Science},
	pages = {547--553},
	booktitle = {Proceedings 41st Annual Symposium on Foundations of Computer Science},
	publisher = {{IEEE} Comput. Soc},
	author = {Ambainis, A. and Mosca, M. and Tapp, A. and De Wolf, R.},
	urldate = {2020-05-14},
	date = {2000},
	langid = {english}
}

@article{mcclean_openfermion_2020,
	title = {{OpenFermion}: the electronic structure package for quantum computers},
	volume = {5},
	issn = {2058-9565},
	url = {https://iopscience.iop.org/article/10.1088/2058-9565/ab8ebc},
	doi = {10.1088/2058-9565/ab8ebc},
	shorttitle = {{OpenFermion}},
	pages = {034014},
	number = {3},
	journaltitle = {Quantum Science and Technology},
	shortjournal = {Quantum Sci. Technol.},
	author = {{McClean}, Jarrod R and Rubin, Nicholas C and Sung, Kevin J and Kivlichan, Ian D and Bonet-Monroig, Xavier and Cao, Yudong and Dai, Chengyu and Fried, E Schuyler and Gidney, Craig and Gimby, Brendan and Gokhale, Pranav and Häner, Thomas and Hardikar, Tarini and Havlíček, Vojtěch and Higgott, Oscar and Huang, Cupjin and Izaac, Josh and Jiang, Zhang and Liu, Xinle and {McArdle}, Sam and Neeley, Matthew and O’Brien, Thomas and O’Gorman, Bryan and Ozfidan, Isil and Radin, Maxwell D and Romero, Jhonathan and Sawaya, Nicolas P D and Senjean, Bruno and Setia, Kanav and Sim, Sukin and Steiger, Damian S and Steudtner, Mark and Sun, Qiming and Sun, Wei and Wang, Daochen and Zhang, Fang and Babbush, Ryan},
	urldate = {2022-02-04},
	date = {2020-06-15}
}

@article{jaderberg_minimum_2020,
  title={Minimum hardware requirements for hybrid quantum-classical {DMFT}},
  url={http://arxiv.org/abs/2002.04612},
  journal={arXiv:2002.04612},
  author={Jaderberg, B. and Agarwal, A. and Leonhardt, K. and Kiffner, M. and Jaksch, D.}, 
  year={2020}, 
  month={Feb} 
}

@article{rungger_dynamical_2019,
  title={Dynamical mean field theory algorithm and experiment on quantum computers},
  url={http://arxiv.org/abs/1910.04735},
  journal={arXiv:1910.04735},
  author={Rungger, I. and Fitzpatrick, N. and Chen, H. and Alderete, C.~H. and Apel, H. and Cowtan, A. and Patterson, A. and Mu\~noz Ramo, D. and Zhu, Y. and Nguyen, N.~H. and Grant, E. and Chretien, S. and Wossnig, L. and Linke, N.~M. and Duncan, R.}, 
  year={2019}, 
  month={Oct} 
}

@article{zhuravlev_breakdown_2000,
  title = {Breakdown of {L}uttinger liquid state in a one-dimensional frustrated spinless fermion model},
  volume = {61},
  url = {https://link.aps.org/doi/10.1103/PhysRevB.61.15534},
  doi = {10.1103/PhysRevB.61.15534},
  number = {23},
  journal = {Phys. Rev. B},
  author = {Zhuravlev, A.~K. and Katsnelson, M.~I.},
  month = {jun},
  year = {2000},
  pages = {15534},
}

@article{zhuravlev_one-dimensional_2001,
  title = {One-dimensional spinless fermion model with competing interactions beyond half filling},
  volume = {64},
  url = {https://link.aps.org/doi/10.1103/PhysRevB.64.033102},
  doi = {10.1103/PhysRevB.64.033102},
  number = {3},
  journal = {Phys. Rev. B},
  author = {Zhuravlev, A.~K. and Katsnelson, M.~I.},
  month = {jun},
  year = {2001},
  pages = {033102},
}

@article{Zhuravlev1997,
  title = {Electronic phase transitions in a one-dimensional spinless fermion model with competing interactions},
  author = {Zhuravlev, A.~K. and Katsnelson, M.~I. and Trefilov, A.~V.},
  journal = {Phys. Rev. B},
  volume = {56},
  issue = {20},
  pages = {12939--12946},
  numpages = {0},
  year = {1997},
  month = {Nov},
  publisher = {American Physical Society},
  doi = {10.1103/PhysRevB.56.12939},
  url = {https://link.aps.org/doi/10.1103/PhysRevB.56.12939}
}

@article{Hohenadler2012,
  title = {Interaction-range effects for fermions in one dimension},
  author = {Hohenadler, M. and Wessel, S. and Daghofer, M. and Assaad, F.~F.},
  journal = {Phys. Rev. B},
  volume = {85},
  issue = {19},
  pages = {195115},
  numpages = {10},
  year = {2012},
  month = {May},
  publisher = {American Physical Society},
  doi = {10.1103/PhysRevB.85.195115},
  url = {https://link.aps.org/doi/10.1103/PhysRevB.85.195115}
}

@article{Karrasch2012,
  title = {Luttinger liquid physics from the infinite-system density matrix renormalization group},
  author = {Karrasch, C. and Moore, J.~E.},
  journal = {Phys. Rev. B},
  volume = {86},
  issue = {15},
  pages = {155156},
  numpages = {8},
  year = {2012},
  month = {Oct},
  publisher = {American Physical Society},
  doi = {10.1103/PhysRevB.86.155156},
  url = {https://link.aps.org/doi/10.1103/PhysRevB.86.155156}
}

@article{kiani_learning_2020,
	title = {Learning Unitaries by Gradient Descent},
	url = {http://arxiv.org/abs/2001.11897},
	journaltitle = {{arXiv}:2001.11897 [math-ph, physics:quant-ph]},
	author = {Kiani, Bobak Toussi and Lloyd, Seth and Maity, Reevu},
	urldate = {2020-02-03},
	date = {2020-01-31},
	langid = {english},
	eprinttype = {arxiv},
	eprint = {2001.11897}
}

@article{farhi_quantum_2017,
	title = {Quantum Algorithms for Fixed Qubit Architectures},
	url = {http://arxiv.org/abs/1703.06199},
	journaltitle = {{arXiv}:1703.06199 [quant-ph]},
	author = {Farhi, E. and Goldstone, J. and Gutmann, S. and Neven, H.},
	urldate = {2018-12-04},
	date = {2017-03-17},
	langid = {english},
	eprinttype = {arxiv},
	eprint = {1703.06199}
}

@article{parrish_quantum_2019,
	title = {Quantum Computation of Electronic Transitions using a Variational Quantum Eigensolver},
	url = {http://arxiv.org/abs/1901.01234},
	journaltitle = {{arXiv}:1901.01234 [quant-ph]},
	author = {Parrish, Robert M. and Hohenstein, Edward G. and {McMahon}, Peter L. and Martinez, Todd J.},
	urldate = {2019-01-14},
	date = {2019-01-04},
	langid = {english},
	eprinttype = {arxiv},
	eprint = {1901.01234}
}

@article{hempel_quantum_2018,
	title = {Quantum Chemistry Calculations on a Trapped-Ion Quantum Simulator},
	volume = {8},
	issn = {2160-3308},
	url = {https://link.aps.org/doi/10.1103/PhysRevX.8.031022},
	doi = {10.1103/PhysRevX.8.031022},
	number = {3},
	journaltitle = {Physical Review X},
	author = {Hempel, Cornelius and Maier, Christine and Romero, Jonathan and {McClean}, Jarrod and Monz, Thomas and Shen, Heng and Jurcevic, Petar and Lanyon, Ben P. and Love, Peter and Babbush, Ryan and Aspuru-Guzik, Alán and Blatt, Rainer and Roos, Christian F.},
	urldate = {2019-01-31},
	date = {2018-07-24},
	langid = {english}
}

@article{bauer_hybrid_2016,
	title = {Hybrid quantum-classical approach to correlated materials},
	volume = {6},
	issn = {2160-3308},
	url = {http://arxiv.org/abs/1510.03859},
	doi = {10.1103/PhysRevX.6.031045},
	number = {3},
	journaltitle = {Physical Review X},
	author = {Bauer, Bela and Wecker, Dave and Millis, Andrew J. and Hastings, Matthew B. and Troyer, M.},
	urldate = {2018-12-05},
	date = {2016-09-21},
	langid = {english},
	eprinttype = {arxiv},
	eprint = {1510.03859}
}

@article{heya_variational_2018,
	title = {Variational Quantum Gate Optimization},
	url = {http://arxiv.org/abs/1810.12745},
	journaltitle = {{arXiv}:1810.12745 [quant-ph]},
	author = {Heya, Kentaro and Suzuki, Yasunari and Nakamura, Yasunobu and Fujii, Keisuke},
	urldate = {2018-11-06},
	date = {2018-10-30},
	langid = {english},
	eprinttype = {arxiv},
	eprint = {1810.12745}
}

@article{seeley_bravyi-kitaev_2012,
	title = {The Bravyi-Kitaev transformation for quantum computation of electronic structure},
	volume = {137},
	issn = {0021-9606, 1089-7690},
	url = {http://aip.scitation.org/doi/10.1063/1.4768229},
	doi = {10.1063/1.4768229},
	pages = {224109},
	number = {22},
	journaltitle = {The Journal of Chemical Physics},
	shortjournal = {The Journal of Chemical Physics},
	author = {Seeley, Jacob T. and Richard, Martin J. and Love, Peter J.},
	urldate = {2019-06-20},
	date = {2012-12-14},
	langid = {english}
}

@article{ryabinkin_qubit_2018,
	title = {Qubit Coupled Cluster Method: A Systematic Approach to Quantum Chemistry on a Quantum Computer},
	volume = {14},
	issn = {1549-9618, 1549-9626},
	url = {https://pubs.acs.org/doi/10.1021/acs.jctc.8b00932},
	doi = {10.1021/acs.jctc.8b00932},
	shorttitle = {Qubit Coupled Cluster Method},
	pages = {6317--6326},
	number = {12},
	journaltitle = {Journal of Chemical Theory and Computation},
	shortjournal = {J. Chem. Theory Comput.},
	author = {Ryabinkin, Ilya G. and Yen, Tzu-Ching and Genin, Scott N. and Izmaylov, Artur F.},
	urldate = {2020-04-16},
	date = {2018-12-11},
	langid = {english}
}

@article{sapova_variational_2021,
	title = {Variational quantum eigensolver techniques for simulating carbon monoxide oxidation},
	url = {http://arxiv.org/abs/2108.11167},
	journaltitle = {{arXiv}:2108.11167 [physics, physics:quant-ph]},
	author = {Sapova, M. D. and Fedorov, A. K.},
	urldate = {2022-02-01},
	date = {2021-08-25},
	langid = {english},
	eprinttype = {arxiv},
	eprint = {2108.11167}
}

@article{veis_relativistic_2012,
	title = {Relativistic quantum chemistry on quantum computers},
	volume = {85},
	issn = {1050-2947, 1094-1622},
	url = {https://link.aps.org/doi/10.1103/PhysRevA.85.030304},
	doi = {10.1103/PhysRevA.85.030304},
	pages = {030304},
	number = {3},
	journaltitle = {Physical Review A},
	shortjournal = {Phys. Rev. A},
	author = {Veis, Libor and Višňák, Jakub and Fleig, Timo and Knecht, Stefan and Saue, Trond and Visscher, Lucas and Pittner, Jiří},
	urldate = {2022-02-01},
	date = {2012-03-23},
	langid = {english}
}

@article{yalouz_state-averaged_2021,
	title = {A state-averaged orbital-optimized hybrid quantum–classical algorithm for a democratic description of ground and excited states},
	volume = {6},
	issn = {2058-9565},
	url = {https://iopscience.iop.org/article/10.1088/2058-9565/abd334},
	doi = {10.1088/2058-9565/abd334},
	pages = {024004},
	number = {2},
	journaltitle = {Quantum Science and Technology},
	shortjournal = {Quantum Sci. Technol.},
	author = {Yalouz, Saad and Senjean, Bruno and Günther, Jakob and Buda, Francesco and O’Brien, Thomas E and Visscher, Lucas},
	urldate = {2021-07-30},
	date = {2021-01-21},
	langid = {english}
}

@article{kempe_3-local_2003,
	title = {3-Local Hamiltonian is {QMA}-complete},
	url = {http://arxiv.org/abs/quant-ph/0302079},
	journaltitle = {{arXiv}:quant-ph/0302079},
	author = {Kempe, Julia and Regev, Oded},
	urldate = {2022-01-26},
	date = {2003-05-20},
	eprinttype = {arxiv},
	eprint = {quant-ph/0302079}
}

@article{levin_universal_1973,
	title = {Universal Sequential Search Problems},
	volume = {9},
	pages = {265--266},
	number = {3},
	journaltitle = {Problems of Information Transmission},
	author = {Levin, L A},
	date = {1973},
	langid = {russian}
}

@inproceedings{cook_complexity_1971,
	location = {Shaker Heights, Ohio, United States},
	title = {The complexity of theorem-proving procedures},
	url = {http://portal.acm.org/citation.cfm?doid=800157.805047},
	doi = {10.1145/800157.805047},
	eventtitle = {the third annual {ACM} symposium},
	pages = {151--158},
	booktitle = {Proceedings of the third annual {ACM} symposium on Theory of computing  - {STOC} '71},
	publisher = {{ACM} Press},
	author = {Cook, Stephen A.},
	urldate = {2022-01-26},
	date = {1971},
	langid = {english}
}

@article{omran_generation_2019,
	title = {Generation and manipulation of Schrödinger cat states in Rydberg atom arrays},
	volume = {365},
	issn = {0036-8075, 1095-9203},
	url = {http://arxiv.org/abs/1905.05721},
	doi = {10.1126/science.aax9743},
	pages = {570--574},
	number = {6453},
	journaltitle = {Science},
	shortjournal = {Science},
	author = {Omran, Ahmed and Levine, Harry and Keesling, Alexander and Semeghini, Giulia and Wang, Tout T. and Ebadi, Sepehr and Bernien, Hannes and Zibrov, Alexander S. and Pichler, Hannes and Choi, Soonwon and Cui, Jian and Rossignolo, Marco and Rembold, Phila and Montangero, Simone and Calarco, Tommaso and Endres, Manuel and Greiner, Markus and Vuletić, Vladan and Lukin, Mikhail D.},
	urldate = {2022-01-25},
	date = {2019-08-09},
	langid = {english},
	eprinttype = {arxiv},
	eprint = {1905.05721}
}

@article{song_observation_2019,
	title = {Observation of multi-component atomic Schrödinger cat states of up to 20 qubits},
	volume = {365},
	issn = {0036-8075, 1095-9203},
	url = {http://arxiv.org/abs/1905.00320},
	doi = {10.1126/science.aay0600},
	pages = {574--577},
	number = {6453},
	journaltitle = {Science},
	shortjournal = {Science},
	author = {Song, Chao and Xu, Kai and Li, Hekang and Zhang, Yuran and Zhang, Xu and Liu, Wuxin and Guo, Qiujiang and Wang, Zhen and Ren, Wenhui and Hao, Jie and Feng, Hui and Fan, Heng and Zheng, Dongning and Wang, Dawei and Wang, H. and Zhu, Shiyao},
	urldate = {2022-01-25},
	date = {2019-08-09},
	langid = {english},
	eprinttype = {arxiv},
	eprint = {1905.00320}
}

@article{monz_14-qubit_2011,
	title = {14-Qubit Entanglement: Creation and Coherence},
	volume = {106},
	issn = {0031-9007, 1079-7114},
	url = {https://link.aps.org/doi/10.1103/PhysRevLett.106.130506},
	doi = {10.1103/PhysRevLett.106.130506},
	shorttitle = {14-Qubit Entanglement},
	pages = {130506},
	number = {13},
	journaltitle = {Physical Review Letters},
	shortjournal = {Phys. Rev. Lett.},
	author = {Monz, Thomas and Schindler, Philipp and Barreiro, Julio T. and Chwalla, Michael and Nigg, Daniel and Coish, William A. and Harlander, Maximilian and Hänsel, Wolfgang and Hennrich, Markus and Blatt, Rainer},
	urldate = {2022-01-25},
	date = {2011-03-31},
	langid = {english}
}

@article{leibfried_creation_2005,
	title = {Creation of a six-atom ‘Schrödinger cat’ state},
	volume = {438},
	issn = {0028-0836, 1476-4687},
	url = {http://www.nature.com/articles/nature04251},
	doi = {10.1038/nature04251},
	pages = {639--642},
	number = {7068},
	journaltitle = {Nature},
	shortjournal = {Nature},
	author = {Leibfried, D. and Knill, E. and Seidelin, S. and Britton, J. and Blakestad, R. B. and Chiaverini, J. and Hume, D. B. and Itano, W. M. and Jost, J. D. and Langer, C. and Ozeri, R. and Reichle, R. and Wineland, D. J.},
	urldate = {2022-01-25},
	date = {2005-12},
	langid = {english}
}

@article{aaronson_improved_2004,
	title = {Improved simulation of stabilizer circuits},
	volume = {70},
	issn = {1050-2947, 1094-1622},
	url = {https://link.aps.org/doi/10.1103/PhysRevA.70.052328},
	doi = {10.1103/PhysRevA.70.052328},
	number = {5},
	journaltitle = {Physical Review A},
	author = {Aaronson, Scott and Gottesman, Daniel},
	urldate = {2018-10-10},
	date = {2004-11-30},
	langid = {english}
}

@article{kak_quantum_1995,
  title = {Quantum Neural Computing},
	volume = {94},
	isbn = {978-0-12-014736-6},
	url = {https://linkinghub.elsevier.com/retrieve/pii/S1076567008701472},
	pages = {259--313},
  journal = {Advances in Imaging and Electron Physics},
	booktitle = {Advances in Imaging and Electron Physics},

	publisher = {Elsevier},
	author = {Kak, Subhash C.},
	urldate = {2021-11-24},
	date = {1995},
	langid = {english},
	doi = {10.1016/S1076-5670(08)70147-2}
}

@article{vert_primer_2004,
  title={A primer on kernel methods},
  author={Vert, Jean-Philippe and Tsuda, Koji and Sch{\"o}lkopf, Bernhard},
  journal={Kernel methods in computational biology},
  volume={47},
  pages={35--70},
  year={2004},
  publisher={MIT press Cambridge, MA}
}

@article{schuld_introduction_2015,
	title = {An introduction to quantum machine learning},
	volume = {56},
	issn = {0010-7514, 1366-5812},
	url = {http://arxiv.org/abs/1409.3097},
	doi = {10.1080/00107514.2014.964942},
	pages = {172--185},
	number = {2},
	journaltitle = {Contemporary Physics},
	author = {Schuld, M. and Sinayskiy, I. and Petruccione, F.},
	urldate = {2019-04-09},
	date = {2015-04-03},
	langid = {english},
	eprinttype = {arxiv},
	eprint = {1409.3097}
}

@article{ciliberto_quantum_2018,
	title = {Quantum machine learning: a classical perspective},
	volume = {474},
	issn = {1364-5021, 1471-2946},
	url = {https://royalsocietypublishing.org/doi/10.1098/rspa.2017.0551},
	doi = {10.1098/rspa.2017.0551},
	shorttitle = {Quantum machine learning},
	pages = {20170551},
	number = {2209},
	journaltitle = {Proceedings of the Royal Society A: Mathematical, Physical and Engineering Sciences},
	shortjournal = {Proc. R. Soc. A.},
	author = {Ciliberto, Carlo and Herbster, Mark and Ialongo, Alessandro Davide and Pontil, Massimiliano and Rocchetto, Andrea and Severini, Simone and Wossnig, Leonard},
	urldate = {2021-11-23},
	date = {2018-01},
	langid = {english}
}

@article{gupta_quantum_2001,
	title = {Quantum Neural Networks},
	volume = {63},
	issn = {00220000},
	url = {https://linkinghub.elsevier.com/retrieve/pii/S0022000001917696},
	doi = {10.1006/jcss.2001.1769},
	pages = {355--383},
	number = {3},
	journaltitle = {Journal of Computer and System Sciences},
	shortjournal = {Journal of Computer and System Sciences},
	author = {Gupta, Sanjay and Zia, R.K.P.},
	urldate = {2021-11-19},
	date = {2001-11},
	langid = {english}
}

@article{cao_quantum_2017,
	title = {Quantum Neuron: an elementary building block for machine learning on quantum computers},
	url = {http://arxiv.org/abs/1711.11240},
	shorttitle = {Quantum Neuron},
	journaltitle = {{arXiv}:1711.11240 [quant-ph]},
	author = {Cao, Yudong and Guerreschi, Gian Giacomo and Aspuru-Guzik, Alán},
	urldate = {2019-03-12},
	date = {2017-11-30},
	langid = {english},
	eprinttype = {arxiv},
	eprint = {1711.11240}
}

@article{schuld_quest_2014,
	title = {The quest for a Quantum Neural Network},
	volume = {13},
	issn = {1570-0755, 1573-1332},
	url = {http://link.springer.com/10.1007/s11128-014-0809-8},
	doi = {10.1007/s11128-014-0809-8},
	pages = {2567--2586},
	number = {11},
	journaltitle = {Quantum Information Processing},
	shortjournal = {Quantum Inf Process},
	author = {Schuld, Maria and Sinayskiy, Ilya and Petruccione, Francesco},
	urldate = {2021-11-23},
	date = {2014-11},
	langid = {english}
}

@article{bausch_recurrent_2020,
	title = {Recurrent Quantum Neural Networks},
	url = {http://arxiv.org/abs/2006.14619},
	journaltitle = {{arXiv}:2006.14619 [quant-ph, stat]},
	author = {Bausch, Johannes},
	urldate = {2021-11-24},
	date = {2020-06-25},
	langid = {english},
	eprinttype = {arxiv},
	eprint = {2006.14619}
}

@article{behrman_simulations_2000,
	title = {Simulations of quantum neural networks},
	volume = {128},
	issn = {00200255},
	url = {https://linkinghub.elsevier.com/retrieve/pii/S0020025500000566},
	doi = {10.1016/S0020-0255(00)00056-6},
	pages = {257--269},
	number = {3},
	journaltitle = {Information Sciences},
	shortjournal = {Information Sciences},
	author = {Behrman, E.C. and Nash, L.R. and Steck, J.E. and Chandrashekar, V.G. and Skinner, S.R.},
	urldate = {2021-11-26},
	date = {2000-10},
	langid = {english}
}

@article{dallaire-demers_quantum_2018,
	title = {Quantum generative adversarial networks},
	volume = {98},
	issn = {2469-9926, 2469-9934},
	url = {https://link.aps.org/doi/10.1103/PhysRevA.98.012324},
	doi = {10.1103/PhysRevA.98.012324},
	pages = {012324},
	number = {1},
	journaltitle = {Physical Review A},
	shortjournal = {Phys. Rev. A},
	author = {Dallaire-Demers, Pierre-Luc and Killoran, Nathan},
	urldate = {2021-11-26},
	date = {2018-07-23},
	langid = {english}

}

@article{khaneja_cartan_2000,
	title = {Cartan Decomposition of {SU}(2n), Constructive Controllability of Spin systems and Universal Quantum Computing},
	url = {http://arxiv.org/abs/quant-ph/0010100},
	journaltitle = {{arXiv}:quant-ph/0010100},
	author = {Khaneja, Navin and Glaser, Steffen},
	urldate = {2018-12-13},
	date = {2000-10-28},
	langid = {english},
	eprinttype = {arxiv},
	eprint = {quant-ph/0010100},
}

@article{huang_information-theoretic_2021,
	title = {Information-Theoretic Bounds on Quantum Advantage in Machine Learning},
	volume = {126},
	issn = {0031-9007, 1079-7114},
	url = {https://link.aps.org/doi/10.1103/PhysRevLett.126.190505},
	doi = {10.1103/PhysRevLett.126.190505},
	pages = {190505},
	number = {19},
	journaltitle = {Physical Review Letters},
	shortjournal = {Phys. Rev. Lett.},
	author = {Huang, Hsin-Yuan and Kueng, Richard and Preskill, John},
	urldate = {2021-11-22},
	date = {2021-05-14},
	langid = {english}
}

@article{poland_no_2020,
	title = {No Free Lunch for Quantum Machine Learning},
	url = {http://arxiv.org/abs/2003.14103},
	journaltitle = {{arXiv}:2003.14103 [quant-ph]},
	author = {Poland, Kyle and Beer, Kerstin and Osborne, Tobias J.},
	urldate = {2020-04-01},
	date = {2020-03-31},
	langid = {english},
	eprinttype = {arxiv},
	eprint = {2003.14103},
}

@article{pesah_absence_2020,
	title = {Absence of Barren Plateaus in Quantum Convolutional Neural Networks},
	url = {http://arxiv.org/abs/2011.02966},
	journaltitle = {{arXiv}:2011.02966 [quant-ph, stat]},
	author = {Pesah, Arthur and Cerezo, M. and Wang, Samson and Volkoff, Tyler and Sornborger, Andrew T. and Coles, Patrick J.},
	urldate = {2020-11-25},
	date = {2020-11-05},
	langid = {english},
	eprinttype = {arxiv},
	eprint = {2011.02966}
}

@article{kyriienko_generalized_2021,
	title = {Generalized quantum circuit differentiation rules},
	url = {http://arxiv.org/abs/2108.01218},
	journaltitle = {{arXiv}:2108.01218 [cond-mat, physics:quant-ph]},
	author = {Kyriienko, Oleksandr and Elfving, Vincent E.},
	urldate = {2021-08-04},
	date = {2021-08-02},
	langid = {english},
	eprinttype = {arxiv},
	eprint = {2108.01218}
}

@article{spall_multivariate_1992,
	title = {Multivariate stochastic approximation using a simultaneous perturbation gradient approximation},
	volume = {37},
	issn = {00189286},
	url = {http://ieeexplore.ieee.org/document/119632/},
	doi = {10.1109/9.119632},
	pages = {332--341},
	number = {3},
	journaltitle = {{IEEE} Transactions on Automatic Control},
	shortjournal = {{IEEE} Trans. Automat. Contr.},
	author = {Spall, J.C.},
	urldate = {2019-09-12},
	date = {1992-03}
}

@book{franchini_introduction_2017,
	location = {Cham},
	title = {An Introduction to Integrable Techniques for One-Dimensional Quantum Systems},
	volume = {940},
	isbn = {978-3-319-48486-0 978-3-319-48487-7},
	url = {http://link.springer.com/10.1007/978-3-319-48487-7},
	series = {Lecture Notes in Physics},
	publisher = {Springer International Publishing},
	author = {Franchini, Fabio},
	urldate = {2019-04-30},
	date = {2017},
	langid = {english},
	doi = {10.1007/978-3-319-48487-7}
}

@article{dillenschneider_quantum_2008,
	title = {Quantum discord and quantum phase transition in spin chains},
	volume = {78},
	issn = {1098-0121, 1550-235X},
	url = {https://link.aps.org/doi/10.1103/PhysRevB.78.224413},
	doi = {10.1103/PhysRevB.78.224413},
	pages = {224413},
	number = {22},
	journaltitle = {Physical Review B},
	author = {Dillenschneider, Raoul},
	urldate = {2019-04-26},
	date = {2008-12-16},
	langid = {english}
}

@article{beer_training_2020,
	title = {Training deep quantum neural networks},
	volume = {11},
	issn = {2041-1723},
	url = {http://www.nature.com/articles/s41467-020-14454-2},
	doi = {10.1038/s41467-020-14454-2},
	pages = {808},
	number = {1},
	journaltitle = {Nature Communications},
	shortjournal = {Nat Commun},
	author = {Beer, Kerstin and Bondarenko, Dmytro and Farrelly, Terry and Osborne, Tobias J. and Salzmann, Robert and Scheiermann, Daniel and Wolf, Ramona},
	urldate = {2020-12-17},
	date = {2020-12},
	langid = {english}
}

@article{van_nieuwenburg_learning_2017,
	title = {Learning phase transitions by confusion},
	volume = {13},
	issn = {1745-2473, 1745-2481},
	url = {http://www.nature.com/articles/nphys4037},
	doi = {10.1038/nphys4037},
	pages = {435--439},
	number = {5},
	journaltitle = {Nature Physics},
	author = {van Nieuwenburg, Evert P. L. and Liu, Ye-Hua and Huber, Sebastian D.},
	urldate = {2019-04-09},
	date = {2017-05},
	langid = {english}
}

@article{schuld_quantum_2021,
	title = {Quantum machine learning models are kernel methods},
	url = {http://arxiv.org/abs/2101.11020},
	journaltitle = {{arXiv}:2101.11020 [quant-ph, stat]},
	author = {Schuld, Maria},
	urldate = {2021-01-28},
	date = {2021-01-26},
	langid = {english},
	eprinttype = {arxiv},
	eprint = {2101.11020}
}

@article{larocca_theory_2021,
	title = {Theory of overparametrization in quantum neural networks},
	url = {http://arxiv.org/abs/2109.11676},
	journaltitle = {{arXiv}:2109.11676 [quant-ph, stat]},
	author = {Larocca, Martin and Ju, Nathan and García-Martín, Diego and Coles, Patrick J. and Cerezo, M.},
	urldate = {2021-09-27},
	date = {2021-09-23},
	langid = {english},
	eprinttype = {arxiv},
	eprint = {2109.11676}
}

@article{wiersema_exploring_2020,
	title = {Exploring Entanglement and Optimization within the Hamiltonian Variational Ansatz},
	volume = {1},
	issn = {2691-3399},
	url = {https://link.aps.org/doi/10.1103/PRXQuantum.1.020319},
	doi = {10.1103/PRXQuantum.1.020319},
	pages = {020319},
	number = {2},
	journaltitle = {{PRX} Quantum},
	shortjournal = {{PRX} Quantum},
	author = {Wiersema, Roeland and Zhou, Cunlu and de Sereville, Yvette and Carrasquilla, Juan Felipe and Kim, Yong Baek and Yuen, Henry},
	urldate = {2021-10-24},
	date = {2020-12-08},
	langid = {english}
}

@article{kattemolle_variational_2021,
	title = {Variational quantum eigensolver for the Heisenberg antiferromagnet on the kagome lattice},
	url = {http://arxiv.org/abs/2108.02175},
	journaltitle = {{arXiv}:2108.02175 [cond-mat, physics:quant-ph]},
	author = {Kattemölle, Joris and van Wezel, Jasper},
	urldate = {2021-11-09},
	date = {2021-08-04},
	langid = {english},
	eprinttype = {arxiv},
	eprint = {2108.02175}
}

@article{bharti_iterative_2020,
	title = {Iterative Quantum Assisted Eigensolver},
	url = {http://arxiv.org/abs/2010.05638},
	journaltitle = {{arXiv}:2010.05638 [quant-ph]},
	author = {Bharti, Kishor and Haug, Tobias},
	urldate = {2020-10-13},
	date = {2020-10-12},
	eprinttype = {arxiv},
	eprint = {2010.05638}
}

@article{chivilikhin_mog-vqe_2020,
	title = {{MoG}-{VQE}: Multiobjective genetic variational quantum eigensolver},
	url = {http://arxiv.org/abs/2007.04424},
	shorttitle = {{MoG}-{VQE}},
	journaltitle = {{arXiv}:2007.04424 [cond-mat, physics:quant-ph]},
	author = {Chivilikhin, D. and Samarin, A. and Ulyantsev, V. and Iorsh, I. and Oganov, A. R. and Kyriienko, O.},
	urldate = {2020-07-10},
	date = {2020-07-08},
	langid = {english},
	eprinttype = {arxiv},
	eprint = {2007.04424},
}

@article{fujii_deep_2020,
	title = {Deep Variational Quantum Eigensolver: a divide-and-conquer method for solving a larger problem with smaller size quantum computers},
	url = {http://arxiv.org/abs/2007.10917},
	shorttitle = {Deep Variational Quantum Eigensolver},
	journaltitle = {{arXiv}:2007.10917 [cond-mat, physics:quant-ph]},
	author = {Fujii, Keisuke and Mitarai, Kosuke and Mizukami, Wataru and Nakagawa, Yuya O.},
	urldate = {2021-04-09},
	date = {2020-07-21},
	eprinttype = {arxiv},
	eprint = {2007.10917}
}

@article{pfeuty_one-dimensional_1970,
	title = {The one-dimensional Ising model with a transverse field},
	volume = {57},
	issn = {00034916},
	url = {https://linkinghub.elsevier.com/retrieve/pii/0003491670902708},
	doi = {10.1016/0003-4916(70)90270-8},
	pages = {79--90},
	number = {1},
	journaltitle = {Annals of Physics},
	author = {Pfeuty, Pierre},
	urldate = {2019-03-06},
	date = {1970-03},
	langid = {english},
}

@article{kohn_nobel_1999,
	title = {Nobel Lecture: Electronic structure of matter—wave functions and density functionals},
	volume = {71},
	issn = {0034-6861, 1539-0756},
	url = {https://link.aps.org/doi/10.1103/RevModPhys.71.1253},
	doi = {10.1103/RevModPhys.71.1253},
	shorttitle = {Nobel Lecture},
	pages = {1253--1266},
	number = {5},
	journaltitle = {Reviews of Modern Physics},
	shortjournal = {Rev. Mod. Phys.},
	author = {Kohn, W.},
	urldate = {2021-11-08},
	date = {1999-10-01},
	langid = {english},
}

@article{wang_accelerated_2019,
	title = {Accelerated Variational Quantum Eigensolver},
	volume = {122},
	issn = {0031-9007, 1079-7114},
	url = {https://link.aps.org/doi/10.1103/PhysRevLett.122.140504},
	doi = {10.1103/PhysRevLett.122.140504},
	pages = {140504},
	number = {14},
	journaltitle = {Physical Review Letters},
	shortjournal = {Phys. Rev. Lett.},
	author = {Wang, Daochen and Higgott, Oscar and Brierley, Stephen},
	urldate = {2019-10-08},
	date = {2019-04-12},
	langid = {english}
}

@article{tubman_postponing_2018,
	title = {Postponing the orthogonality catastrophe: efficient state preparation for electronic structure simulations on quantum devices},
	url = {http://arxiv.org/abs/1809.05523},
	shorttitle = {Postponing the orthogonality catastrophe},
	journaltitle = {{arXiv}:1809.05523 [cond-mat, physics:physics, physics:quant-ph]},
	author = {Tubman, Norm M. and Mejuto-Zaera, Carlos and Epstein, Jeffrey M. and Hait, Diptarka and Levine, Daniel S. and Huggins, William and Jiang, Zhang and {McClean}, Jarrod R. and Babbush, Ryan and Head-Gordon, Martin and Whaley, K. Birgitta},
	urldate = {2021-05-20},
	date = {2018-09-14},
	langid = {english},
	eprinttype = {arxiv},
	eprint = {1809.05523}
}

@article{romero_strategies_2017,
	title = {Strategies for quantum computing molecular energies using the unitary coupled cluster ansatz},
	url = {http://arxiv.org/abs/1701.02691},
	journaltitle = {{arXiv}:1701.02691 [quant-ph]},
	author = {Romero, Jonathan and Babbush, Ryan and {McClean}, Jarrod R. and Hempel, Cornelius and Love, Peter and Aspuru-Guzik, Alán},
	urldate = {2019-02-07},
	date = {2017-01-10},
	langid = {english},
	eprinttype = {arxiv},
	eprint = {1701.02691}
}

@article{taube_new_2006,
	title = {New perspectives on unitary coupled-cluster theory},
	volume = {106},
	issn = {0020-7608, 1097-461X},
	url = {http://doi.wiley.com/10.1002/qua.21198},
	doi = {10.1002/qua.21198},
	pages = {3393--3401},
	number = {15},
	journaltitle = {International Journal of Quantum Chemistry},
	author = {Taube, Andrew G. and Bartlett, Rodney J.},
	urldate = {2019-02-07},
	date = {2006},
	langid = {english}
}

@article{shen_quantum_2017,
	title = {Quantum Implementation of Unitary Coupled Cluster for Simulating Molecular Electronic Structure},
	volume = {95},
	issn = {2469-9926, 2469-9934},
	url = {http://arxiv.org/abs/1506.00443},
	doi = {10.1103/PhysRevA.95.020501},
	number = {2},
	journaltitle = {Physical Review A},
	author = {Shen, Yangchao and Zhang, Xiang and Zhang, Shuaining and Zhang, Jing-Ning and Yung, Man-Hong and Kim, Kihwan},
	urldate = {2019-02-07},
	date = {2017-02-15},
	langid = {english},
	eprinttype = {arxiv},
	eprint = {1506.00443}
}

@article{xu_test_2020,
	title = {Test of the unitary coupled-cluster variational quantum eigensolver for a simple strongly correlated condensed-matter system},
	url = {http://arxiv.org/abs/2001.06957},
	journaltitle = {{arXiv}:2001.06957 [cond-mat, physics:quant-ph]},
	author = {Xu, Luogen and Lee, Joseph T. and Freericks, J. K.},
	urldate = {2020-04-07},
	date = {2020-01-19},
	langid = {english},
	eprinttype = {arxiv},
	eprint = {2001.06957}
}

@article{xia_coupled_2020,
	title = {Coupled cluster singles and doubles variational quantum eigensolver ansatz for electronic structure calculations},
	url = {http://arxiv.org/abs/2005.08451},
	journaltitle = {{arXiv}:2005.08451 [quant-ph]},
	author = {Xia, Rongxin and Kais, Sabre},
	urldate = {2020-05-19},
	date = {2020-05-18},
	langid = {english},
	eprinttype = {arxiv},
	eprint = {2005.08451}
}

@article{bosse_probing_2021,
	title = {Probing ground state properties of the kagome antiferromagnetic Heisenberg model using the Variational Quantum Eigensolver},
	url = {http://arxiv.org/abs/2108.08086},
	journaltitle = {{arXiv}:2108.08086 [quant-ph]},
	author = {Bosse, Jan Lukas and Montanaro, Ashley},
	urldate = {2021-10-29},
	date = {2021-10-04},
	langid = {english},
	eprinttype = {arxiv},
	eprint = {2108.08086},
}

@article{omalley_scalable_2016,
	title = {Scalable Quantum Simulation of Molecular Energies},
	volume = {6},
	issn = {2160-3308},
	url = {http://arxiv.org/abs/1512.06860},
	doi = {10.1103/PhysRevX.6.031007},
	number = {3},
	journaltitle = {Physical Review X},
	author = {O'Malley, P. J. J. and Babbush, R. and Kivlichan, I. D. and Romero, J. and {McClean}, J. R. and Barends, R. and Kelly, J. and Roushan, P. and Tranter, A. and Ding, N. and Campbell, B. and Chen, Y. and Chen, Z. and Chiaro, B. and Dunsworth, A. and Fowler, A. G. and Jeffrey, E. and Megrant, A. and Mutus, J. Y. and Neill, C. and Quintana, C. and Sank, D. and Vainsencher, A. and Wenner, J. and White, T. C. and Coveney, P. V. and Love, P. J. and Neven, H. and Aspuru-Guzik, A. and Martinis, J. M.},
	urldate = {2018-12-04},
	date = {2016-07-18},
	langid = {english},
	eprinttype = {arxiv},
	eprint = {1512.06860}
}

@article{samuel_u_1980,
	title = {U( \textit{N} ) Integrals, 1/ \textit{N} , and the De Wit–’t Hooft anomalies},
	volume = {21},
	issn = {0022-2488, 1089-7658},
	url = {http://aip.scitation.org/doi/10.1063/1.524386},
	doi = {10.1063/1.524386},
	pages = {2695--2703},
	number = {12},
	journaltitle = {Journal of Mathematical Physics},
	shortjournal = {Journal of Mathematical Physics},
	author = {Samuel, Stuart},
	urldate = {2020-03-31},
	date = {1980-12},
	langid = {english},
}

@article{lieb_two_1961,
	title = {Two soluble models of an antiferromagnetic chain},
	volume = {16},
	issn = {00034916},
	url = {https://linkinghub.elsevier.com/retrieve/pii/0003491661901154},
	doi = {10.1016/0003-4916(61)90115-4},
	pages = {407--466},
	number = {3},
	journaltitle = {Annals of Physics},
	author = {Lieb, Elliott and Schultz, Theodore and Mattis, Daniel},
	urldate = {2019-03-06},
	date = {1961-12},
	langid = {english}
}

@book{watrous_theory_2018,
	location = {Cambridge, United Kingdom},
	title = {The theory of quantum information},
	isbn = {978-1-107-18056-7},
	pagetotal = {590},
	publisher = {Cambridge University Press},
	author = {Watrous, John},
	date = {2018}
}

@misc{aleksandrowicz_qiskit:_2019,
	title = {Qiskit: An Open-source Framework for Quantum Computing},
	rights = {Apache License 2.0, Open Access},
	url = {https://zenodo.org/record/2562110},
	shorttitle = {Qiskit},
	version = {0.7.2},
	publisher = {Zenodo},
	author = {Aleksandrowicz, Gadi and Alexander, Thomas and Barkoutsos, Panagiotis and Bello, Luciano and Ben-Haim, Yael and Bucher, David and Cabrera-Hernández, Francisco Jose and Carballo-Franquis, Jorge and Chen, Adrian and Chen, Chun-Fu and Chow, Jerry M. and Córcoles-Gonzales, Antonio D. and Cross, Abigail J. and Cross, Andrew and Cruz-Benito, Juan and Culver, Chris and González, Salvador De La Puente and Torre, Enrique De La and Ding, Delton and Dumitrescu, Eugene and Duran, Ivan and Eendebak, Pieter and Everitt, Mark and Sertage, Ismael Faro and Frisch, Albert and Fuhrer, Andreas and Gambetta, Jay and Gago, Borja Godoy and Gomez-Mosquera, Juan and Greenberg, Donny and Hamamura, Ikko and Havlicek, Vojtech and Hellmers, Joe and Łukasz Herok and Horii, Hiroshi and Shaohan Hu and Imamichi, Takashi and Toshinari Itoko and Javadi-Abhari, Ali and Kanazawa, Naoki and Karazeev, Anton and Krsulich, Kevin and Liu, Peng and Luh, Yang and Yunho Maeng and Marques, Manoel and Martín-Fernández, Francisco Jose and {McClure}, Douglas T. and {McKay}, David and Srujan Meesala and Mezzacapo, Antonio and Moll, Nikolaj and Rodríguez, Diego Moreda and Nannicini, Giacomo and Nation, Paul and Ollitrault, Pauline and O'Riordan, Lee James and Hanhee Paik and Pérez, Jesús and Phan, Anna and Pistoia, Marco and Prutyanov, Viktor and Reuter, Max and Rice, Julia and Abdón Rodríguez Davila and Rudy, Raymond Harry Putra and Mingi Ryu and Ninad Sathaye and Schnabel, Chris and Schoute, Eddie and Kanav Setia and Yunong Shi and Adenilton Silva and Siraichi, Yukio and Seyon Sivarajah and Smolin, John A. and Soeken, Mathias and Takahashi, Hitomi and Tavernelli, Ivano and Taylor, Charles and Taylour, Pete and Kenso Trabing and Treinish, Matthew and Turner, Wes and Vogt-Lee, Desiree and Vuillot, Christophe and Wildstrom, Jonathan A. and Wilson, Jessica and Winston, Erick and Wood, Christopher and Wood, Stephen and Wörner, Stefan and Akhalwaya, Ismail Yunus and Zoufal, Christa},
	urldate = {2019-04-24},
	date = {2019-01-23},
	doi = {10.5281/zenodo.2562110}
}

@article{tang_qubit-adapt-vqe_2021,
	title = {Qubit-{ADAPT}-{VQE}: An Adaptive Algorithm for Constructing Hardware-Efficient Ansätze on a Quantum Processor},
	volume = {2},
	issn = {2691-3399},
	url = {https://link.aps.org/doi/10.1103/PRXQuantum.2.020310},
	doi = {10.1103/PRXQuantum.2.020310},
	shorttitle = {Qubit-{ADAPT}-{VQE}},
	pages = {020310},
	number = {2},
	journaltitle = {{PRX} Quantum},
	shortjournal = {{PRX} Quantum},
	author = {Tang, Ho Lun and Shkolnikov, V.O. and Barron, George S. and Grimsley, Harper R. and Mayhall, Nicholas J. and Barnes, Edwin and Economou, Sophia E.},
	urldate = {2021-05-20},
	date = {2021-04-28},
	langid = {english},
}

@article{uvarov_barren_2021,
	title = {On barren plateaus and cost function locality in variational quantum algorithms},
	volume = {54},
	issn = {1751-8113, 1751-8121},
	url = {https://iopscience.iop.org/article/10.1088/1751-8121/abfac7},
	doi = {10.1088/1751-8121/abfac7},
	pages = {245301},
	number = {24},
	journaltitle = {Journal of Physics A: Mathematical and Theoretical},
	shortjournal = {J. Phys. A: Math. Theor.},
	author = {Uvarov, A V and Biamonte, J D},
	urldate = {2021-06-09},
	date = {2021-06-18},
	langid = {english}
}

@article{bilkis_semi-agnostic_2021,
	title = {A semi-agnostic ansatz with variable structure for quantum machine learning},
	url = {http://arxiv.org/abs/2103.06712},
	journaltitle = {{arXiv}:2103.06712 [quant-ph, stat]},
	author = {Bilkis, M. and Cerezo, M. and Verdon, Guillaume and Coles, Patrick J. and Cincio, Lukasz},
	urldate = {2021-03-17},
	date = {2021-03-11},
	langid = {english},
	eprinttype = {arxiv},
	eprint = {2103.06712},
}

@article{mcclean_hybrid_2017,
	title = {Hybrid quantum-classical hierarchy for mitigation of decoherence and determination of excited states},
	volume = {95},
	issn = {2469-9926, 2469-9934},
	url = {http://link.aps.org/doi/10.1103/PhysRevA.95.042308},
	doi = {10.1103/PhysRevA.95.042308},
	number = {4},
	journaltitle = {Physical Review A},
	author = {{McClean}, Jarrod R. and Kimchi-Schwartz, Mollie E. and Carter, Jonathan and de Jong, Wibe A.},
	urldate = {2019-03-07},
	date = {2017-04-06},
	langid = {english}
}

@article{sim_adaptive_2021,
	title = {Adaptive pruning-based optimization of parameterized quantum circuits},
	volume = {6},
	issn = {2058-9565},
	url = {https://iopscience.iop.org/article/10.1088/2058-9565/abe107},
	doi = {10.1088/2058-9565/abe107},
	pages = {025019},
	number = {2},
	journaltitle = {Quantum Science and Technology},
	shortjournal = {Quantum Sci. Technol.},
	author = {Sim, Sukin and Romero, Jonathan and Gonthier, Jérôme F and Kunitsa, Alexander A},
	urldate = {2021-08-10},
	date = {2021-04-01}
}

@article{whitfield_ground_2012,
	title = {Ground State Spin Logic},
	volume = {99},
	issn = {0295-5075, 1286-4854},
	url = {http://arxiv.org/abs/1205.1742},
	doi = {10.1209/0295-5075/99/57004},
	pages = {57004},
	number = {5},
	journaltitle = {{EPL} (Europhysics Letters)},
	author = {Whitfield, J. D. and Faccin, M. and Biamonte, J. D.},
	urldate = {2018-10-30},
	date = {2012-09-01},
	langid = {english},
	eprinttype = {arxiv},
	eprint = {1205.1742}
}

@article{lanyon_towards_2010,
	title = {Towards Quantum Chemistry on a Quantum Computer},
	volume = {2},
	issn = {1755-4330, 1755-4349},
	url = {http://arxiv.org/abs/0905.0887},
	doi = {10.1038/nchem.483},
	pages = {106--111},
	number = {2},
	journaltitle = {Nature Chemistry},
	author = {Lanyon, Benjamin P. and Whitfield, James D. and Gillet, Geoff G. and Goggin, Michael E. and Almeida, Marcelo P. and Kassal, Ivan and Biamonte, Jacob D. and Mohseni, Masoud and Powell, Ben J. and Barbieri, Marco and Aspuru-Guzik, Alán and White, Andrew G.},
	urldate = {2018-11-23},
	date = {2010-02},
	langid = {english},
	eprinttype = {arxiv},
	eprint = {0905.0887}
}

@article{biamonte_quantum_2017,
	title = {Quantum Machine Learning},
	volume = {549},
	issn = {0028-0836, 1476-4687},
	url = {http://arxiv.org/abs/1611.09347},
	doi = {10.1038/nature23474},
	pages = {195--202},
	number = {7671},
	journaltitle = {Nature},
	author = {Biamonte, Jacob and Wittek, Peter and Pancotti, Nicola and Rebentrost, Patrick and Wiebe, Nathan and Lloyd, Seth},
	urldate = {2018-12-30},
	date = {2017-09-13},
	langid = {english},
	eprinttype = {arxiv},
	eprint = {1611.09347}
}

@article{uvarov_machine_2020,
	title = {Machine learning phase transitions with a quantum processor},
	volume = {102},
	issn = {2469-9926, 2469-9934},
	url = {https://link.aps.org/doi/10.1103/PhysRevA.102.012415},
	doi = {10.1103/PhysRevA.102.012415},
	pages = {012415},
	number = {1},
	journaltitle = {Physical Review A},
	shortjournal = {Phys. Rev. A},
	author = {Uvarov, A. V. and Kardashin, A. S. and Biamonte, J. D.},
	urldate = {2020-07-17},
	date = {2020-07-15},
	langid = {english}
}

@article{uvarov_variational_2020,
	title = {Variational quantum eigensolver for frustrated quantum systems},
	volume = {102},
	issn = {2469-9950, 2469-9969},
	url = {https://link.aps.org/doi/10.1103/PhysRevB.102.075104},
	doi = {10.1103/PhysRevB.102.075104},
	pages = {075104},
	number = {7},
	journaltitle = {Physical Review B},
	shortjournal = {Phys. Rev. B},
	author = {Uvarov, Alexey and Biamonte, Jacob D. and Yudin, Dmitry},
	urldate = {2020-09-14},
	date = {2020-08-04},
	langid = {english}
}

@article{biamonte_lectures_2020,
	title = {Lectures on Quantum Tensor Networks},
	url = {http://arxiv.org/abs/1912.10049},
	journaltitle = {{arXiv}:1912.10049 [cond-mat, physics:math-ph, physics:quant-ph]},
	author = {Biamonte, Jacob},
	urldate = {2020-09-15},
	date = {2020-01-04},
	eprinttype = {arxiv},
	eprint = {1912.10049}
}

@article{campos_abrupt_2020,
	title = {Abrupt Transitions in Variational Quantum Circuit Training},
	url = {http://arxiv.org/abs/2010.09720},
	journaltitle = {{arXiv}:2010.09720 [cond-mat, physics:quant-ph]},
	author = {Campos, Ernesto and Nasrallah, Aly and Biamonte, Jacob},
	urldate = {2021-01-21},
	date = {2020-10-19},
	eprinttype = {arxiv},
	eprint = {2010.09720}
}

@article{bittel_training_2021,
	title = {Training variational quantum algorithms is {NP}-hard -- even for logarithmically many qubits and free fermionic systems},
	url = {http://arxiv.org/abs/2101.07267},
	journaltitle = {{arXiv}:2101.07267 [quant-ph]},
	author = {Bittel, Lennart and Kliesch, Martin},
	urldate = {2021-01-20},
	date = {2021-01-18},
	langid = {english},
	eprinttype = {arxiv},
	eprint = {2101.07267},
}

@article{emerson_convergence_2005,
	title = {Convergence conditions for random quantum circuits},
	volume = {72},
	issn = {1050-2947, 1094-1622},
	url = {https://link.aps.org/doi/10.1103/PhysRevA.72.060302},
	doi = {10.1103/PhysRevA.72.060302},
	pages = {060302},
	number = {6},
	journaltitle = {Physical Review A},
	shortjournal = {Phys. Rev. A},
	author = {Emerson, Joseph and Livine, Etera and Lloyd, Seth},
	urldate = {2019-12-24},
	date = {2005-12-02},
	langid = {english}
}

@article{gottesman_heisenberg_1998,
	title = {The Heisenberg Representation of Quantum Computers},
	url = {http://arxiv.org/abs/quant-ph/9807006},
	journaltitle = {{arXiv}:quant-ph/9807006},
	author = {Gottesman, Daniel},
	urldate = {2018-09-26},
	date = {1998-07-01},
	langid = {english},
	eprinttype = {arxiv},
	eprint = {quant-ph/9807006}
}

@article{farhi_quantum_2014,
	title = {A Quantum Approximate Optimization Algorithm},
	url = {http://arxiv.org/abs/1411.4028},
	journaltitle = {{arXiv}:1411.4028 [quant-ph]},
	author = {Farhi, Edward and Goldstone, Jeffrey and Gutmann, Sam},
	urldate = {2018-12-04},
	date = {2014-11-14},
	langid = {english},
	eprinttype = {arxiv},
	eprint = {1411.4028}
}

@article{peruzzo_variational_2014,
	title = {A variational eigenvalue solver on a photonic quantum processor},
	volume = {5},
	issn = {2041-1723},
	url = {http://www.nature.com/articles/ncomms5213},
	doi = {10.1038/ncomms5213},
	pages = {4213},
	number = {1},
	journaltitle = {Nature Communications},
	author = {Peruzzo, Alberto and {McClean}, Jarrod and Shadbolt, Peter and Yung, Man-Hong and Zhou, Xiao-Qi and Love, Peter J. and Aspuru-Guzik, Alán and O’Brien, Jeremy L.},
	urldate = {2018-12-11},
	date = {2014-12},
	langid = {english}
}

@article{bravyi_fermionic_2002,
	title = {Fermionic quantum computation},
	volume = {298},
	issn = {00034916},
	url = {http://arxiv.org/abs/quant-ph/0003137},
	doi = {10.1006/aphy.2002.6254},
	pages = {210--226},
	number = {1},
	journaltitle = {Annals of Physics},
	author = {Bravyi, Sergey and Kitaev, Alexei},
	urldate = {2018-12-29},
	date = {2002-05},
	langid = {english},
	eprinttype = {arxiv},
	eprint = {quant-ph/0003137}
}

@article{barkoutsos_quantum_2018,
	title = {Quantum algorithms for electronic structure calculations: particle/hole Hamiltonian and optimized wavefunction expansions},
	volume = {98},
	issn = {2469-9926, 2469-9934},
	url = {http://arxiv.org/abs/1805.04340},
	doi = {10.1103/PhysRevA.98.022322},
	shorttitle = {Quantum algorithms for electronic structure calculations},
	number = {2},
	journaltitle = {Physical Review A},
	author = {Barkoutsos, Panagiotis Kl and Gonthier, Jerome F. and Sokolov, Igor and Moll, Nikolaj and Salis, Gian and Fuhrer, Andreas and Ganzhorn, Marc and Egger, Daniel J. and Troyer, Matthias and Mezzacapo, Antonio and Filipp, Stefan and Tavernelli, Ivano},
	urldate = {2018-12-29},
	date = {2018-08-20},
	langid = {english},
	eprinttype = {arxiv},
	eprint = {1805.04340}
}

@article{kandala_hardware-efficient_2017,
	title = {Hardware-efficient Variational Quantum Eigensolver for Small Molecules and Quantum Magnets},
	volume = {549},
	issn = {0028-0836, 1476-4687},
	url = {http://arxiv.org/abs/1704.05018},
	doi = {10.1038/nature23879},
	pages = {242--246},
	number = {7671},
	journaltitle = {Nature},
	author = {Kandala, Abhinav and Mezzacapo, Antonio and Temme, Kristan and Takita, Maika and Brink, Markus and Chow, Jerry M. and Gambetta, Jay M.},
	urldate = {2018-12-29},
	date = {2017-09-13},
	langid = {english},
	eprinttype = {arxiv},
	eprint = {1704.05018}
}

@article{harrow_low-depth_2019,
	title = {Low-depth gradient measurements can improve convergence in variational hybrid quantum-classical algorithms},
	url = {http://arxiv.org/abs/1901.05374},
	journaltitle = {{arXiv}:1901.05374 [quant-ph]},
	author = {Harrow, Aram and Napp, John},
	urldate = {2019-01-21},
	date = {2019-01-16},
	langid = {english},
	eprinttype = {arxiv},
	eprint = {1901.05374}
}

@article{kempe_complexity_2006,
	title = {The Complexity of the Local Hamiltonian Problem},
	volume = {35},
	issn = {0097-5397, 1095-7111},
	url = {http://epubs.siam.org/doi/10.1137/S0097539704445226},
	doi = {10.1137/S0097539704445226},
	pages = {1070--1097},
	number = {5},
	journaltitle = {{SIAM} Journal on Computing},
	author = {Kempe, Julia and Kitaev, Alexei and Regev, Oded},
	urldate = {2019-01-30},
	date = {2006-01},
	langid = {english}
}

@book{nielsen_quantum_2010,
  title = {Quantum Computation and Quantum Information},
  author = {Nielsen, Michael A and Chuang, Isaac L},
  date = {2010},
  publisher = {{Cambridge University Press}},
  location = {{Cambridge; New York}},
  url = {http://public.eblib.com/choice/publicfullrecord.aspx?p=647366},
  urldate = {2020-03-20},
  isbn = {978-1-107-00217-3 978-0-511-97666-7 978-1-282-96729-8},
  langid = {english}
}

@article{nielsen_quantum_2006,
	title = {Quantum Computation as Geometry},
	volume = {311},
	issn = {0036-8075, 1095-9203},
	url = {http://www.sciencemag.org/cgi/doi/10.1126/science.1121541},
	doi = {10.1126/science.1121541},
	pages = {1133--1135},
	number = {5764},
	journaltitle = {Science},
	author = {Nielsen, M. A.},
	urldate = {2019-03-01},
	date = {2006-02-24},
	langid = {english}
}

@article{colless_computation_2018,
	title = {Computation of Molecular Spectra on a Quantum Processor with an Error-Resilient Algorithm},
	volume = {8},
	issn = {2160-3308},
	url = {https://link.aps.org/doi/10.1103/PhysRevX.8.011021},
	doi = {10.1103/PhysRevX.8.011021},
	number = {1},
	journaltitle = {Physical Review X},
	author = {Colless, J. I. and Ramasesh, V. V. and Dahlen, D. and Blok, M. S. and Kimchi-Schwartz, M. E. and {McClean}, J. R. and Carter, J. and de Jong, W. A. and Siddiqi, I.},
	urldate = {2019-03-07},
	date = {2018-02-12},
	langid = {english}
}

@article{biamonte_universal_2021,
	title = {Universal variational quantum computation},
	volume = {103},
	issn = {2469-9926, 2469-9934},
	url = {https://link.aps.org/doi/10.1103/PhysRevA.103.L030401},
	doi = {10.1103/PhysRevA.103.L030401},
	pages = {L030401},
	number = {3},
	journaltitle = {Physical Review A},
	shortjournal = {Phys. Rev. A},
	author = {Biamonte, Jacob},
	urldate = {2021-03-17},
	date = {2021-03-10},
	langid = {english}
}

@article{garcia-saez_addressing_2018,
	title = {Addressing hard classical problems with Adiabatically Assisted Variational Quantum Eigensolvers},
	url = {http://arxiv.org/abs/1806.02287},
	journaltitle = {{arXiv}:1806.02287 [cond-mat, physics:quant-ph]},
	author = {Garcia-Saez, A. and Latorre, J. I.},
	urldate = {2019-03-22},
	date = {2018-06-06},
	langid = {english},
	eprinttype = {arxiv},
	eprint = {1806.02287}
}

@article{mcclean_barren_2018,
	title = {Barren plateaus in quantum neural network training landscapes},
	volume = {9},
	issn = {2041-1723},
	url = {http://www.nature.com/articles/s41467-018-07090-4},
	doi = {10.1038/s41467-018-07090-4},
	pages = {4812},
	number = {1},
	journaltitle = {Nature Communications},
	author = {{McClean}, Jarrod R. and Boixo, Sergio and Smelyanskiy, Vadim N. and Babbush, Ryan and Neven, Hartmut},
	urldate = {2019-04-23},
	date = {2018-12},
	langid = {english}
}

@article{havlicek_supervised_2019,
	title = {Supervised learning with quantum-enhanced feature spaces},
	volume = {567},
	issn = {0028-0836, 1476-4687},
	url = {http://www.nature.com/articles/s41586-019-0980-2},
	doi = {10.1038/s41586-019-0980-2},
	pages = {209--212},
	number = {7747},
	journaltitle = {Nature},
	shortjournal = {Nature},
	author = {Havlíček, Vojtěch and Córcoles, Antonio D. and Temme, Kristan and Harrow, Aram W. and Kandala, Abhinav and Chow, Jerry M. and Gambetta, Jay M.},
	urldate = {2019-05-06},
	date = {2019-03},
	langid = {english}
}

@article{kokail_self-verifying_2019,
	title = {Self-verifying variational quantum simulation of lattice models},
	volume = {569},
	issn = {0028-0836, 1476-4687},
	url = {http://www.nature.com/articles/s41586-019-1177-4},
	doi = {10.1038/s41586-019-1177-4},
	pages = {355--360},
	number = {7756},
	journaltitle = {Nature},
	shortjournal = {Nature},
	author = {Kokail, C. and Maier, C. and van Bijnen, R. and Brydges, T. and Joshi, M. K. and Jurcevic, P. and Muschik, C. A. and Silvi, P. and Blatt, R. and Roos, C. F. and Zoller, P.},
	urldate = {2019-05-22},
	date = {2019-05},
	langid = {english}
}

@article{mitarai_quantum_2018,
	title = {Quantum Circuit Learning},
	volume = {98},
	issn = {2469-9926, 2469-9934},
	url = {http://arxiv.org/abs/1803.00745},
	doi = {10.1103/PhysRevA.98.032309},
	pages = {032309},
	number = {3},
	journaltitle = {Physical Review A},
	shortjournal = {Phys. Rev. A},
	author = {Mitarai, Kosuke and Negoro, Makoto and Kitagawa, Masahiro and Fujii, Keisuke},
	urldate = {2019-05-31},
	date = {2018-09-10},
	langid = {english},
	eprinttype = {arxiv},
	eprint = {1803.00745}
}

@article{wecker_progress_2015,
	title = {Progress towards practical quantum variational algorithms},
	volume = {92},
	issn = {1050-2947, 1094-1622},
	url = {https://link.aps.org/doi/10.1103/PhysRevA.92.042303},
	doi = {10.1103/PhysRevA.92.042303},
	pages = {042303},
	number = {4},
	journaltitle = {Physical Review A},
	shortjournal = {Phys. Rev. A},
	author = {Wecker, Dave and Hastings, Matthew B. and Troyer, Matthias},
	urldate = {2019-06-11},
	date = {2015-10-02},
	langid = {english}
}

@article{grimsley_adaptive_2019,
	title = {An adaptive variational algorithm for exact molecular simulations on a quantum computer},
	volume = {10},
	issn = {2041-1723},
	url = {http://www.nature.com/articles/s41467-019-10988-2},
	doi = {10.1038/s41467-019-10988-2},
	pages = {3007},
	number = {1},
	journaltitle = {Nature Communications},
	shortjournal = {Nat Commun},
	author = {Grimsley, Harper R. and Economou, Sophia E. and Barnes, Edwin and Mayhall, Nicholas J.},
	urldate = {2019-08-23},
	date = {2019-12},
	langid = {english}
}

@article{nannicini_performance_2019,
	title = {Performance of hybrid quantum-classical variational heuristics for combinatorial optimization},
	volume = {99},
	issn = {2470-0045, 2470-0053},
	url = {https://link.aps.org/doi/10.1103/PhysRevE.99.013304},
	doi = {10.1103/PhysRevE.99.013304},
	pages = {013304},
	number = {1},
	journaltitle = {Physical Review E},
	shortjournal = {Phys. Rev. E},
	author = {Nannicini, Giacomo},
	urldate = {2019-09-13},
	date = {2019-01-14},
	langid = {english}
}

@article{sweke_stochastic_2019,
	title = {Stochastic gradient descent for hybrid quantum-classical optimization},
	url = {http://arxiv.org/abs/1910.01155},
	journaltitle = {{arXiv}:1910.01155 [quant-ph]},
	author = {Sweke, Ryan and Wilde, Frederik and Meyer, Johannes and Schuld, Maria and Fährmann, Paul K. and Meynard-Piganeau, Barthélémy and Eisert, Jens},
	urldate = {2019-10-04},
	date = {2019-10-02},
	langid = {english},
	eprinttype = {arxiv},
	eprint = {1910.01155}
}

@article{schuld_evaluating_2019,
	title = {Evaluating analytic gradients on quantum hardware},
	volume = {99},
	issn = {2469-9926, 2469-9934},
	url = {https://link.aps.org/doi/10.1103/PhysRevA.99.032331},
	doi = {10.1103/PhysRevA.99.032331},
	pages = {032331},
	number = {3},
	journaltitle = {Physical Review A},
	shortjournal = {Phys. Rev. A},
	author = {Schuld, Maria and Bergholm, Ville and Gogolin, Christian and Izaac, Josh and Killoran, Nathan},
	urldate = {2019-10-04},
	date = {2019-03-21},
	langid = {english}
}

@article{collins_integration_2006,
	title = {Integration with Respect to the Haar Measure on Unitary, Orthogonal and Symplectic Group},
	volume = {264},
	issn = {0010-3616, 1432-0916},
	url = {http://link.springer.com/10.1007/s00220-006-1554-3},
	doi = {10.1007/s00220-006-1554-3},
	pages = {773--795},
	number = {3},
	journaltitle = {Communications in Mathematical Physics},
	shortjournal = {Commun. Math. Phys.},
	author = {Collins, Benoît and Śniady, Piotr},
	urldate = {2019-10-08},
	date = {2006-06},
	langid = {english}
}

@article{khaneja_time_2001,
	title = {Time optimal control in spin systems},
	volume = {63},
	issn = {1050-2947, 1094-1622},
	url = {https://link.aps.org/doi/10.1103/PhysRevA.63.032308},
	doi = {10.1103/PhysRevA.63.032308},
	pages = {032308},
	number = {3},
	journaltitle = {Physical Review A},
	shortjournal = {Phys. Rev. A},
	author = {Khaneja, Navin and Brockett, Roger and Glaser, Steffen J.},
	urldate = {2019-12-04},
	date = {2001-02-14},
	langid = {english}
}

@article{cade_strategies_2019,
	title = {Strategies for solving the Fermi-Hubbard model on near-term quantum computers},
	url = {http://arxiv.org/abs/1912.06007},
	journaltitle = {{arXiv}:1912.06007 [quant-ph]},
	author = {Cade, Chris and Mineh, Lana and Montanaro, Ashley and Stanisic, Stasja},
	urldate = {2019-12-20},
	date = {2019-12-12},
	langid = {english},
	eprinttype = {arxiv},
	eprint = {1912.06007}
}

@article{brandao_local_2016,
	title = {Local Random Quantum Circuits are Approximate Polynomial-Designs},
	volume = {346},
	issn = {0010-3616, 1432-0916},
	url = {http://link.springer.com/10.1007/s00220-016-2706-8},
	doi = {10.1007/s00220-016-2706-8},
	pages = {397--434},
	number = {2},
	journaltitle = {Communications in Mathematical Physics},
	shortjournal = {Commun. Math. Phys.},
	author = {Brandão, Fernando G. S. L. and Harrow, Aram W. and Horodecki, Michał},
	urldate = {2020-01-24},
	date = {2016-09},
	langid = {english}
}

@article{cerezo_cost-function-dependent_2020,
	title = {Cost-Function-Dependent Barren Plateaus in Shallow Quantum Neural Networks},
	url = {http://arxiv.org/abs/2001.00550},
	journaltitle = {{arXiv}:2001.00550 [quant-ph]},
	author = {Cerezo, M. and Sone, Akira and Volkoff, Tyler and Cincio, Lukasz and Coles, Patrick J.},
	urldate = {2020-02-07},
	date = {2020-02-04},
	langid = {english},
	eprinttype = {arxiv},
	eprint = {2001.00550}
}

@book{kitaev_classical_2002,
	location = {Providence, {RI}},
	title = {Classical and quantum computation},
	isbn = {978-0-8218-3229-5 978-0-8218-2161-9},
	series = {Graduate studies in mathematics},
	pagetotal = {257},
	number = {47},
	publisher = {American Mathematical Society},
	author = {Kitaev, Aleksei Yur'evich and Shen, Aleksandr Ch and Vyalyi, Michail N.},
	date = {2002},
	note = {{OCLC}: 611603364}
}

@article{schuld_circuit-centric_2020,
	title = {Circuit-centric quantum classifiers},
	volume = {101},
	issn = {2469-9926, 2469-9934},
	url = {https://link.aps.org/doi/10.1103/PhysRevA.101.032308},
	doi = {10.1103/PhysRevA.101.032308},
	pages = {032308},
	number = {3},
	journaltitle = {Physical Review A},
	shortjournal = {Phys. Rev. A},
	author = {Schuld, Maria and Bocharov, Alex and Svore, Krysta M. and Wiebe, Nathan},
	urldate = {2020-05-06},
	date = {2020-03-06},
	langid = {english}
}

@article{ryabinkin_iterative_2020,
	title = {Iterative Qubit Coupled Cluster Approach with Efficient Screening of Generators},
	volume = {16},
	issn = {1549-9618, 1549-9626},
	url = {https://pubs.acs.org/doi/10.1021/acs.jctc.9b01084},
	doi = {10.1021/acs.jctc.9b01084},
	pages = {1055--1063},
	number = {2},
	journaltitle = {Journal of Chemical Theory and Computation},
	shortjournal = {J. Chem. Theory Comput.},
	author = {Ryabinkin, Ilya G. and Lang, Robert A. and Genin, Scott N. and Izmaylov, Artur F.},
	urldate = {2020-04-21},
	date = {2020-02-11},
	langid = {english}
}

@article{khatri_quantum-assisted_2019,
	title = {Quantum-assisted quantum compiling},
	volume = {3},
	issn = {2521-327X},
	url = {https://quantum-journal.org/papers/q-2019-05-13-140/},
	doi = {10.22331/q-2019-05-13-140},
	pages = {140},
	journaltitle = {Quantum},
	shortjournal = {Quantum},
	author = {Khatri, Sumeet and {LaRose}, Ryan and Poremba, Alexander and Cincio, Lukasz and Sornborger, Andrew T. and Coles, Patrick J.},
	urldate = {2020-04-20},
	date = {2019-05-13},
	langid = {english}
}

@article{grant_initialization_2019,
	title = {An initialization strategy for addressing barren plateaus in parametrized quantum circuits},
	volume = {3},
	issn = {2521-327X},
	url = {http://arxiv.org/abs/1903.05076},
	doi = {10.22331/q-2019-12-09-214},
	pages = {214},
	journaltitle = {Quantum},
	shortjournal = {Quantum},
	author = {Grant, Edward and Wossnig, Leonard and Ostaszewski, Mateusz and Benedetti, Marcello},
	urldate = {2020-04-16},
	date = {2019-12-09},
	eprinttype = {arxiv},
	eprint = {1903.05076}
}

@article{bernstein_quantum_1997,
	title = {Quantum Complexity Theory},
	volume = {26},
	issn = {0097-5397, 1095-7111},
	url = {http://epubs.siam.org/doi/10.1137/S0097539796300921},
	doi = {10.1137/S0097539796300921},
	pages = {1411--1473},
	number = {5},
	journaltitle = {{SIAM} Journal on Computing},
	shortjournal = {{SIAM} J. Comput.},
	author = {Bernstein, Ethan and Vazirani, Umesh},
	urldate = {2021-03-12},
	date = {1997-10},
	langid = {english}
}

@article{skolik_layerwise_2020,
	title = {Layerwise learning for quantum neural networks},
	url = {http://arxiv.org/abs/2006.14904},
	journaltitle = {{arXiv}:2006.14904 [quant-ph]},
	author = {Skolik, Andrea and {McClean}, Jarrod R. and Mohseni, Masoud and van der Smagt, Patrick and Leib, Martin},
	urldate = {2020-06-30},
	date = {2020-06-26},
	langid = {english},
	eprinttype = {arxiv},
	eprint = {2006.14904}
}

@article{bravo-prieto_scaling_2020,
	title = {Scaling of variational quantum circuit depth for condensed matter systems},
	volume = {4},
	issn = {2521-327X},
	url = {https://quantum-journal.org/papers/q-2020-05-28-272/},
	doi = {10.22331/q-2020-05-28-272},
	pages = {272},
	journaltitle = {Quantum},
	shortjournal = {Quantum},
	author = {Bravo-Prieto, Carlos and Lumbreras-Zarapico, Josep and Tagliacozzo, Luca and Latorre, José I.},
	urldate = {2020-06-08},
	date = {2020-05-28},
	langid = {english}
}

@article{low_pseudo-randomness_2010,
	title = {Pseudo-randomness and Learning in Quantum Computation},
	url = {http://arxiv.org/abs/1006.5227},
	journaltitle = {{arXiv}:1006.5227 [quant-ph]},
	author = {Low, Richard A.},
	urldate = {2020-05-08},
	date = {2010-06-27},
	langid = {english},
	eprinttype = {arxiv},
	eprint = {1006.5227}
}

@article{endo_calculation_2020,
	title = {Calculation of the Green's function on near-term quantum computers},
	volume = {2},
	issn = {2643-1564},
	url = {https://link.aps.org/doi/10.1103/PhysRevResearch.2.033281},
	doi = {10.1103/PhysRevResearch.2.033281},
	pages = {033281},
	number = {3},
	journaltitle = {Physical Review Research},
	shortjournal = {Phys. Rev. Research},
	author = {Endo, Suguru and Kurata, Iori and Nakagawa, Yuya O.},
	urldate = {2020-09-29},
	date = {2020-08-20},
	langid = {english}
}

@article{elfving_how_2020,
	title = {How will quantum computers provide an industrially relevant computational advantage in quantum chemistry?},
	url = {http://arxiv.org/abs/2009.12472},
	journaltitle = {{arXiv}:2009.12472 [physics, physics:quant-ph]},
	author = {Elfving, V. E. and Broer, B. W. and Webber, M. and Gavartin, J. and Halls, M. D. and Lorton, K. P. and Bochevarov, A.},
	urldate = {2020-09-30},
	date = {2020-09-25},
	langid = {english},
	eprinttype = {arxiv},
	eprint = {2009.12472}
}

@article{reiher_elucidating_2017, title={Elucidating reaction mechanisms on quantum computers}, volume={114}, ISSN={0027-8424, 1091-6490}, DOI={10.1073/pnas.1619152114}, abstractNote={With rapid recent advances in quantum technology, we are close to the threshold of quantum devices whose computational powers can exceed those of classical supercomputers. Here, we show that a quantum computer can be used to elucidate reaction mechanisms in complex chemical systems, using the open problem of biological nitrogen fixation in nitrogenase as an example. We discuss how quantum computers can augment classical computer simulations used to probe these reaction mechanisms, to significantly increase their accuracy and enable hitherto intractable simulations. Our resource estimates show that, even when taking into account the substantial overhead of quantum error correction, and the need to compile into discrete gate sets, the necessary computations can be performed in reasonable time on small quantum computers. Our results demonstrate that quantum computers will be able to tackle important problems in chemistry without requiring exorbitant resources.}, number={29}, journal={Proceedings of the National Academy of Sciences}, author={Reiher, Markus and Wiebe, Nathan and Svore, Krysta M. and Wecker, Dave and Troyer, Matthias}, year={2017}, month={Jul}, pages={7555–7560} }

@incollection{coecke_physics_2010,
  title = {Physics, {{Topology}}, {{Logic}} and {{Computation}}: {{A Rosetta Stone}}},
  shorttitle = {Physics, {{Topology}}, {{Logic}} and {{Computation}}},
  booktitle = {New {{Structures}} for {{Physics}}},
  author = {Baez, J. and Stay, M.},
  editor = {Coecke, Bob},
  date = {2010},
  volume = {813},
  pages = {95--172},
  publisher = {{Springer Berlin Heidelberg}},
  location = {{Berlin, Heidelberg}},
  doi = {10.1007/978-3-642-12821-9_2},
  url = {http://link.springer.com/10.1007/978-3-642-12821-9_2},
  urldate = {2018-10-30},
  isbn = {978-3-642-12820-2 978-3-642-12821-9},
  langid = {english}
}

@article{bridgeman_hand-waving_2017,
  title = {Hand-Waving and {{Interpretive Dance}}: {{An Introductory Course}} on {{Tensor Networks}}},
  shorttitle = {Hand-Waving and {{Interpretive Dance}}},
  author = {Bridgeman, Jacob C. and Chubb, Christopher T.},
  date = {2017-06-02},
  journaltitle = {Journal of Physics A: Mathematical and Theoretical},
  volume = {50},
  number = {22},
  eprint = {1603.03039},
  eprinttype = {arxiv},
  pages = {223001},
  issn = {1751-8113, 1751-8121},
  doi = {10.1088/1751-8121/aa6dc3},
  url = {http://arxiv.org/abs/1603.03039},
  urldate = {2018-10-30},
  archiveprefix = {arXiv},
  langid = {english}
}

@article{orus_practical_2014,
  title = {A {{Practical Introduction}} to {{Tensor Networks}}: {{Matrix Product States}} and {{Projected Entangled Pair States}}},
  shorttitle = {A {{Practical Introduction}} to {{Tensor Networks}}},
  author = {Orus, Roman},
  date = {2014-10},
  journaltitle = {Annals of Physics},
  volume = {349},
  eprint = {1306.2164},
  eprinttype = {arxiv},
  pages = {117--158},
  issn = {00034916},
  doi = {10.1016/j.aop.2014.06.013},
  url = {http://arxiv.org/abs/1306.2164},
  urldate = {2018-10-30},
  archiveprefix = {arXiv},
  langid = {english}
}

@article{joyal_geometry_1991,
  title = {The Geometry of Tensor Calculus, {{I}}},
  author = {Joyal, André and Street, Ross},
  date = {1991-07},
  journaltitle = {Advances in Mathematics},
  volume = {88},
  number = {1},
  pages = {55--112},
  issn = {00018708},
  doi = {10.1016/0001-8708(91)90003-P},
  url = {http://linkinghub.elsevier.com/retrieve/pii/000187089190003P},
  urldate = {2018-11-15},
  langid = {english}
}

@article{werner_quantum_1989,
  title = {Quantum States with {{Einstein-Podolsky-Rosen}} Correlations Admitting a Hidden-Variable Model},
  author = {Werner, Reinhard F.},
  date = {1989-10-01},
  journaltitle = {Physical Review A},
  shortjournal = {Phys. Rev. A},
  volume = {40},
  number = {8},
  pages = {4277--4281},
  issn = {0556-2791},
  doi = {10.1103/PhysRevA.40.4277},
  url = {https://link.aps.org/doi/10.1103/PhysRevA.40.4277},
  urldate = {2022-07-25},
  langid = {english}
}
\end{document}